%% file: PCA_weak_factor_model_v20240930.tex
\documentclass[english]{article}
\usepackage[T1]{fontenc}
\usepackage[latin9]{inputenc}
\usepackage{geometry}
\usepackage{babel}
\usepackage{array}
\usepackage{float}
\usepackage{mathrsfs}
\usepackage{mathtools}
\usepackage{bm}
\usepackage{multirow}
\usepackage{amsmath}
\usepackage{amsthm}
\usepackage{amssymb}
\usepackage{graphicx}
\usepackage[authoryear]{natbib}
\usepackage[unicode=true,
 bookmarks=true,bookmarksnumbered=false,bookmarksopen=true,bookmarksopenlevel=1,
 breaklinks=false,pdfborder={0 0 1},colorlinks=false]
 {hyperref}
\hypersetup{
 colorlinks,linkcolor=red,anchorcolor=blue,citecolor=blue}

\makeatletter

\providecommand{\tabularnewline}{\\}
\floatstyle{ruled}
\newfloat{algorithm}{tbp}{loa}
\providecommand{\algorithmname}{Algorithm}
\floatname{algorithm}{\protect\algorithmname}

\usepackage{babel}

\usepackage{cite}\usepackage{amsthm}\usepackage{dsfont}\usepackage{array}\usepackage{mathrsfs}\usepackage{comment}\onecolumn
\usepackage{color}

\allowdisplaybreaks

\usepackage{enumitem}
\setlist[itemize]{leftmargin=1.5em}
\setlist[enumerate]{leftmargin=1.5em}

\usepackage{algorithm}
\usepackage{algorithmic}
\usepackage{arydshln}

\DeclareMathOperator{\Var}{{\rm Var}}

\numberwithin{equation}{section}

\definecolor{yly}{RGB}{0,150,0}

\definecolor{yhz}{RGB}{0,0,90}

\makeatother

\begin{document}
\theoremstyle{plain} \newtheorem{lemma}{\textbf{Lemma}} \newtheorem{prop}{\textbf{Proposition}}\newtheorem{theorem}{\textbf{Theorem}}\setcounter{theorem}{0}
\newtheorem{corollary}{\textbf{Corollary}} \newtheorem{assumption}{\textbf{Assumption}}
\newtheorem{example}{\textbf{Example}} \newtheorem{definition}{\textbf{Definition}}
\newtheorem{fact}{\textbf{Fact}} \newtheorem{condition}{\textbf{Condition}}\theoremstyle{definition}\theoremstyle{remark}\newtheorem{remark}{\textbf{Remark}}\newtheorem{claim}{\textbf{Claim}}\newtheorem{conjecture}{\textbf{Conjecture}} 
\title{\textbf{When can weak latent factors be statistically inferred?}}
\author{Jianqing Fan\thanks{Department of Operations Research and Financial Engineering, Princeton
University, Princeton, NJ 08544, USA; Email: \texttt{\{jqfan,yuheng\}@princeton.edu}.} \and Yuling Yan\thanks{Institute for Data, Systems, and Society, Massachusetts Institute
of Technology, Cambridge, MA 02142, USA; Email: \texttt{yulingy@mit.edu}.} \and Yuheng Zheng\footnotemark[1]}

\maketitle
\input{abstract.tex}


\setcounter{tocdepth}{2}

\tableofcontents{}

\input{intro.tex}\input{model_assumption_notation.tex}

\input{main_results.tex}\input{compare_SNR_rate.tex}\input{applications.tex}\input{main_numerical_experiments.tex}\input{discussion.tex}

\section*{Acknowledgements}

J.~Fan is supported in part by the NSF grants  DMS-2053832 and DMS-2210833, and ONR grant N00014-22-1-2340. Y.~Yan is supported in part by the Norbert Wiener Postdoctoral Fellowship
from MIT.

\appendix
\input{appendix_UV_main_results.tex}

\input{appendix_factor_inference.tex}\input{appendix_factor_test.tex}

\input{appendix_covariance_thresholding.tex}

\input{appendix_beta_structure_test.tex}

\input{appendix_two_rows_of_B.tex}

\input{appendix_row_norm_of_B.tex}

\input{appendix_technical.tex}

\bibliographystyle{apalike}
\bibliography{newadded_PCA_factor,bibfileNonconvex}

\end{document}

%% file: abstract.tex
\begin{abstract}
This article establishes a new and comprehensive estimation and inference theory for principal component analysis (PCA) under the weak factor model that allow for cross-sectional dependent idiosyncratic components under the nearly minimal factor strength relative to the noise level or signal-to-noise ratio. Our theory is applicable regardless of the relative growth rate between the cross-sectional dimension $N$ and temporal dimension $T$. This more realistic assumption and noticeable result require completely new technical device, as the commonly-used leave-one-out trick is no longer applicable to the case with cross-sectional dependence. Another notable advancement of our theory is on PCA inference --- for example, under the regime where $N\asymp T$, we show that the asymptotic normality for the PCA-based estimator holds as long as the signal-to-noise ratio (SNR) grows faster than a polynomial rate of $\log N$. This finding significantly surpasses prior work that required a polynomial rate of $N$. Our theory is entirely non-asymptotic, offering finite-sample characterizations for both the estimation error and the uncertainty level of statistical inference. A notable technical innovation is our closed-form first-order approximation of PCA-based estimator, which paves the way for various statistical tests. Furthermore, we apply our theories to design easy-to-implement statistics for validating whether given factors fall in the linear spans of unknown latent factors, testing structural breaks in the factor loadings for an individual unit, checking whether two units have the same risk exposures, and constructing confidence intervals for systematic risks. Our empirical studies uncover insightful correlations between our test results and economic cycles.
\end{abstract}

%% file: intro.tex
\section{Introduction }

The factor model, a pivotal tool for analyzing large panel data, has
become a significant topic in finance and economics research \citep[e.g.,][]{cham1983factor,FamaFrench1993factors,StockWatson2002PCA,BaiNg2002ECTA,GigXiu2021JPE,fan2021recent}.
The estimation and inference for factor models are crucial in economic
studies, particularly in areas like asset pricing and return forecasting.
In the era of big data, the factor model has gained increased prominence
in capturing the latent common structure for large panel data, where
both the cross-sectional and temporal dimensions are ultra-high \citep[see, e.g., recent surveys][]{BaiPeng2016ARFEfactor,FanLiLiao2021ARFEfactor}.
Principal component analysis (PCA), known for its simplicity and effectiveness,
is closely connected with the factor model and has long been a key
research topic of interest in the econometric community \citep[e.g.,][]{StockWatson2002PCA,BaiNg2002ECTA,Bai2003ECTA,Onatski2012PCAweak,POET2013,BaiNg2013identifiPCA,BaiNg2023PCA}.

As pointed out by \citet{GiglioXiu2023weakFactorPredict}, most theoretical
guarantees for the PCA approach to factor analysis rely on the pervasiveness
assumption \citep[e.g.,][]{BaiNg2002ECTA,Bai2003ECTA}. This assumption
requires the signal-to-noise ratio (SNR), which measures the factor
strength relative to the noise level, to grow with the rate of $\sqrt{N}$
-- the square root of the cross-sectional dimension. However, many
real datasets in economics do not exhibit sufficiently strong factors
to meet this pervasiveness assumption. When the SNR grows slower than
$\sqrt{N}$, the resulting model is often called the weak factor model
\citep[e.g.,][]{Onatski2009number,Onatski2010number}. Extensive research
has been dedicated to the weak factor model \citep[e.g.,][]{Onatski2012PCAweak,BKP2021weakFactor,frey2022weakFactorNumber,YoTa2022sWF_estimate,YoTa2022sWF_infer,BaiNg2023PCA,Jiang2023PCA,ChoiMing2024PCA},
among which the PC estimators (the estimators via the PCA approach)
have been a primary subject. Recently, \citet{BaiNg2023PCA,Jiang2023PCA,ChoiMing2024PCA}
studied the consistency and asymptotic normality of the PC estimators
for factors and factor loadings in the weak factor model. In the extreme
case (also called the super-weak factor model) where the SNR is $O(1)$,
\citet{Onatski2012PCAweak} showed that the PC estimators are inconsistent.

This paper establishes a novel and comprehensive theory for PCA in
the weak factor model. Our theory is non-asymptotic and can be easily
translated to asymptotic results. The conditions we propose for the
asymptotic normality of PC estimators are optimal in the sense that,
surprisingly different from the existing literature, the required
growth rate of the SNR for consistency aligns with that for asymptotic
normality, differing only by a logarithmic factor. In particular,
in the regime $N\asymp T$, where $T$ is the temporal dimension,
we prove that asymptotic normality holds as long as the SNR grows
faster than a polynomial rate of $\log N$, and this result is a substantial
advance compared with the existing results that require the SNR to
grow with a polynomial rate of $N$ \citep[e.g.,][]{BaiNg2023PCA,Jiang2023PCA,ChoiMing2024PCA}.

The most innovative part of our theory lies in establishing a closed-form,
first-order approximation of the PC estimator, that is, we decompose
the PC estimator into three components --- the ground truth, a first-order
term, and the higher-order negligible term. We express the first-order
term explicitly using the parameters in the factor model. This closed-form
characterization paves the way for us to establish the asymptotic
normality and design various test statistics for practical applications.
Our theory is based on novel applications of the leave-one-out analysis
and matrix concentration inequalities. Moreover, our findings provide
valuable insights and practical implications for a range of econometric
problems related with PCA, e.g., macroeconomic forecasting based on
factor-augmented regressions \citep[e.g.,][]{BaiNg2006factorAugmented}.

We demonstrate the practical applications of our theories using both
synthetic and real datasets. First, we design an easy-to-implement
statistic for factor specification test \citep[e.g.,][]{BaiNg2006factor}
--- whether an observed factor is in the linear space spanned by
the latent common factors or not. A key innovation of our approach
is the ability to conduct this test in any flexible subset of the
whole period, owing to the row-wise error bound in our theory. We
utilize the monthly return data of the S\&P 500 constituents from
1995 to 2024, along with the data of the Fama-French three factors,
then run the factor test in a rolling window manner \citep[e.g.,][]{FanLiaoYao2015power}.
Our results uncover a notable decline in the importance and explanatory
power of the size factor during the 2008 financial crisis, and of
both the size and value factors during the COVID-19 pandemic around
2019. These findings are supported by our test results, which reject
the null hypothesis that these factors are in the linear space of
latent factors.

Then, we design a test statistic for the structural break of betas
\citep[e.g.,][]{StockWatson2009beta_break,breitung2011beta_break},
and apply it to the aforementioned S\&P 500 constituents data. Our
novelty is that our test statistic works in the weak factor model
without the pervasiveness assumption, which is required in prior work.
We test for each stock to determine if the beta has changed before
and after the three recessions covered by our data --- the Early
2000s Recession, the 2008 Great Recession, and the COVID-19 Recession.
We find that in each recession, the sectors most affected by the shocks,
where many stocks exhibited structural breaks in betas, correspond
reasonably to the causes of these economic recessions. For example,
during the 2008 Great Recession, marked by the subprime mortgage crisis,
the financial sector experienced a strong impact, which aligns with
its exposure to mortgage-backed securities and other related financial
instruments. During the COVID-19 Recession, the Health Care and Real
Estate sectors were significantly impacted, reflecting the uncertainties
brought about by lockdowns and health crises due to the pandemic.
Additionally, we develop statistical tests for the betas to evaluate
the similarity in risk exposure between two stocks and construct valid
confidence interval for the systematic risk of each stock.

The rest of the paper is organized as follows. Section \ref{sec:model_assump_notation}
introduces the model setup, basic assumptions, and notation. In Section
\ref{sec:main_1st_approx}, we show our main results on the first-order
approximations for the PC estimators, and present the asymptotic normality
results as corollaries. Section \ref{sec:compare_SNR} provides a
detailed comparison of our results with related work. In Section \ref{sec:3_econ_applications},
we showcase four applications in econometrics based on our main results.
Section \ref{sec:Numerical-experiments} collects the numerical results
in both simulated data and real data. More related works are discussed
in Section \ref{sec:related work}, and the paper concludes with a
discussion on future directions in Section \ref{sec:Discussion}. 

%% file: model_assumption_notation.tex
\section{Model, assumptions, and notation\label{sec:model_assump_notation}}

\subsection{Model setup}

Let $N$ be the number of cross-sectional units and $T$ be the number
of observations. Consider the factor model for a panel of data $\{x_{i,t}\}:$
\begin{equation}
x_{i,t}=\bm{b}_{i}^{\top}\bm{f}_{t}+\varepsilon_{i,t},\qquad1\leq i\leq N,1\leq t\leq T,\label{factor model entrywise}
\end{equation}
where $\bm{f}_{t}=(f_{1,t},f_{2,t},\ldots,f_{r,t})^{\top}$ is the
latent factor, $r$ is the number of factors, $\bm{b}_{i}=(b_{i,1},b_{i,2},\ldots,b_{i,r})^{\top}$
is a vector of factor loadings, and $\varepsilon_{i,t}$ represents
the idiosyncratic noise. Viewing $x_{i,t}$ as the excess return of
the $i$-th asset at time $t$, the model (\ref{factor model entrywise})
is intimately linked with the multi-factor pricing model. This model
originates from the Arbitrage Pricing Theory (APT) developed by \citet{Ross1976APT}
and finds extensive applications in finance.

To compact the notation, we denote by $\bm{B}=(\bm{b}_{1},\bm{b}_{2},\ldots,\bm{b}_{N})^{\top}$
the $N\times r$ factor loading matrix. Let $\bm{x}_{t}=(x_{1,t},x_{2,t},\ldots,x_{N,t})^{\top}$
and $\bm{e}_{t}=(\varepsilon_{1,t},\varepsilon_{2,t},\ldots,\varepsilon_{N,t})^{\top}$.
Then, the factor model (\ref{factor model entrywise}) can be expressed
as $\bm{x}_{t}=\bm{Bf}_{t}+\bm{e}_{t}$, or written in a matrix form
as follows 
\begin{equation}
\bm{X}=\bm{BF}^{\top}+\bm{E},\label{factor model matrix form}
\end{equation}
where $\bm{X}=(\bm{x}_{1},\bm{x}_{2},\ldots,\bm{x}_{T})$, $\bm{F}=(\bm{f}_{1},\bm{f}_{2},\ldots,\bm{f}_{T})^{\top}$,
and $\bm{E}=(\bm{e}_{1},\bm{e}_{2},\ldots,\bm{e}_{T})$ are the $N\times T$
panel data, the $T\times r$ factor realizations, and the $N\times T$
idiosyncratic noise matrix, respectively.

In our setup, the only observable part is the panel data $\bm{X}$.
We are interested in the estimation and inference for both the latent
factors and factor loadings via the PCA approach.

\subsection{Basic assumptions}

We present some basic assumptions as follows. Note that the factor
model has the rotation (indeed affine transform) ambiguity \citep[e.g.,][]{BaiNg2013identifiPCA},
that is, the factor loadings and latent factors are not identifiable
since $\bm{BF}^{\top}=(\bm{B}\bm{H}^{-1})(\bm{F\bm{H}^{\top}})^{\top}$
holds for any invertible matrix $\bm{H}$. 
Without loss of generality, we assume that the columns of $\bm{B}$
are orthogonal and the covariance of $\bm{f}_{t}$ is the identity
matrix, as stated in Assumption \ref{Assump_Bf_identification} below.

\begin{assumption} \label{Assump_Bf_identification}For $t=1,2,\ldots,T$,
the factor $\bm{f}_{t}$ has mean zero and the identity covariance
matrix. The factor loading matrix $\bm{B}$ has orthogonal columns:
\begin{equation}
\bm{B}^{\top}\bm{B}=\bm{\Sigma}^{2}\text{\qquad with\qquad}\bm{\Sigma}=\text{\ensuremath{\mathsf{diag}}}(\sigma_{1},\sigma_{2},\bm{\ldots},\sigma_{r}),\label{factor loading diagonal}
\end{equation}
where $\sigma_{1}\geq\sigma_{2}\geq\cdots\geq\sigma_{r}>0$. \end{assumption}

Assumption \ref{Assump_Bf_identification} is a standard identifiability
condition for the factor model \citep[e.g.,][]{POET2013}. The singular
values $\{\sigma_{i}\}$ of the factor loading matrix $\bm{B}$ characterize
the strengths of the latent factors.

Next, to accommodate the cross-sectional correlation in the noise
and to facilitate leave-one-out analysis in our technical proof, we
propose Assumption \ref{Assump_noise_Z_entries} for the noise matrix
$\bm{E}$, specifying its structure and the distribution of its entries.

\begin{assumption} \label{Assump_noise_Z_entries}The idiosyncratic
noise matrix $\bm{E}$ is given by 
\begin{equation}
\bm{E}=\bm{\Sigma}_{\varepsilon}^{1/2}\bm{Z}.\label{noise matrix formula}
\end{equation}
Here, $\bm{\Sigma}_{\varepsilon}$ is a $N\times N$ positive definite
matrix and $\bm{\Sigma}_{\varepsilon}^{1/2}$ is the symmetric square
root of $\bm{\Sigma}_{\varepsilon}$; $\bm{Z}=(Z_{i,t})_{i=1}^{N}{}_{t=1}^{T}$
is a $N\times T$ matrix; The entries of $\bm{Z}$ are independent
sub-Gaussian random variables that satisfy 
\[
\mathbb{E}[Z_{i,t}]=0,\text{ }\mathbb{E}[Z_{i,t}^{2}]=1,\text{ }\left\Vert Z_{i,t}\right\Vert _{\psi_{2}}=O(1),
\]
for $1\leq i\leq N,1\leq t\leq T$, where $\left\Vert \cdot\right\Vert _{\psi_{2}}$
is the sub-Gaussian norm \citep[see Definition 2.5.6 in][]{vershynin2016high}.\end{assumption}

The nonzero off-diagonal entries of $\bm{\Sigma}_{\varepsilon}$ characterize
the cross-sectional correlations in the idiosyncratic noise $\bm{e}_{t}=(\varepsilon_{1,t},\varepsilon_{2,t},\ldots,\varepsilon_{N,t})^{\top}$.
Though the noise terms $\bm{e}_{1},\bm{e}_{2},\ldots,\bm{e}_{T}$
are independent under Assumption \ref{Assump_noise_Z_entries}, our
theory could potentially be generalized to the cases where the temporal
correlations are present in the noise matrix $\bm{E}$, and this generalization
is an interesting future direction. The formulation (\ref{noise matrix formula})
imposes additional structural constraints on the noise matrix. However,
such assumptions are commonplace in the study of the weak factor model;
see similar assumptions in \citet{Onatski2010number,Onatski2012PCAweak}.

Then, to establish the non-asymptotic results via the concentration
inequalities, we propose the following assumption on the distribution
of the factors.

\begin{assumption}\label{Assump_factor_f}The factor
$\bm{F}=(\bm{f}_{1},\bm{f}_{2},\ldots,\bm{f}_{T})^{\top}$ is independent
with the noise matrix $\bm{E}$; The factors $\bm{f}_{1},\bm{f}_{2},\ldots,\bm{f}_{T}$
are independent and sub-Gaussian random vectors in $\mathbb{R}^{r}$
satisfying that 
\[
\left\Vert \bm{f}_{t}\right\Vert _{\psi_{2}}=O(1)\text{\qquad for\qquad}1\leq t\leq T.
\]
\end{assumption}

The conditions in Assumption~\ref{Assump_factor_f} are slightly
stronger than those in the literature, e.g., \citet{BaiNg2023PCA},
but remain standard and simplify our technical proofs.

\subsection{Notation\label{subsec:Paper-organization-notation}}

We introduce some notation that will be used throughout the paper.
For two sequences $a_{n}$ and $b_{n}$, we write $a_{n}\lesssim b_{n}$
(or equivalently, $b_{n}\gtrsim a_{n}$) if $a_{n}=O(b_{n})$, i.e.,
there exist a constant $C>0$ and an integer $N>0$ such that, $a_{n}\leq Cb_{n}$
holds for any $n>N$; write $a_{n}\asymp b_{n}$ if both $a_{n}\lesssim b_{n}$
and $a_{n}\gtrsim b_{n}$ hold and $a_{n}\ll b_{n}$ (or equivalently,
$b_{n}\gg a_{n}$) if $a_{n}=o(b_{n})$.

For a symmetric matrix $\bm{M}$, we denote by $\lambda_{\min}(\bm{M})$
and $\lambda_{\max}(\bm{M})$ its minimum and maximum eigenvalues.
For a vector $\bm{v}$, we denote $\Vert\bm{v}\Vert_{2}$, $\Vert\bm{v}\Vert_{1}$,
and $\Vert\bm{v}\Vert_{\infty}$ as the $\ell_{2}$-norm, $\ell_{1}$-norm,
and supremum norm, respectively. Consider any matrix $\bm{A}\in\mathbb{R}^{m\times n}$.
We denote by $\bm{A}_{i,\cdot}$ and $\bm{A}_{\cdot,j}$ the $i$-th
row and the $j$-th column of $\bm{A}$. We let $\Vert\bm{A}\Vert_{2},$
$\Vert\bm{A}\Vert_{\mathrm{F}},$ and $\Vert\bm{A}\Vert_{2,\infty}$
denote the spectral norm, the Frobenius norm, and the $\ell_{2,\infty}$-norm
(i.e., $\Vert\bm{A}\Vert_{2,\infty}:=\sup_{1\leq i\leq m}\Vert\bm{A}_{i,\cdot}\Vert_{2}$),
respectively. For an index set $S\subseteq\{1,2,\ldots,m\}$ (resp.
$S\subseteq\{1,2,\ldots,n\}$), we use $|S|$ to denote its cardinality,
and use $\bm{A}_{S,\cdot}$ (resp. $\bm{A}_{\cdot,S}$) to denote
a submatrix of $\bm{A}$ whose rows (resp. columns) are indexed by
$S$. We let $col(\bm{A})$, $\bm{A}^{+}$, and $\bm{P}_{\bm{A}}:=\bm{A}\bm{A}^{+}$
denote the column subspace of $\bm{A}$, the generalized inverse of
$\bm{A}$, 
and the projection matrix onto the column space of $\bm{A}$.

We denote by $\mathsf{diag}(a_{1},a_{2},\ldots,a_{r})$ the $r\times r$
diagonal matrix whose diagonal entries are given by $a_{1},a_{2},\ldots,a_{r}$.
Let $\bm{I}_{r}$ be the $r\times r$ identity matrix. We denote by
$\mathcal{O}^{r\times r}$ the set of all $r\times r$ orthonormal
(or rotation) matrices. For a non-singular $n\times n$ matrix $\bm{H}$
with SVD $\bm{U}_{H}\bm{\Sigma}_{H}\bm{V}_{H}^{\top}$, we denote
by $\mathsf{sgn}(\bm{H})$ the following orthogonal matrix $\mathsf{sgn}(\bm{H})\coloneqq\bm{U}_{H}\bm{V}_{H}^{\top}$.
Then we have that, for any two matrices $\widehat{\bm{U}},\bm{U}\in\mathbb{R}^{n\times r}$
with $r\leq n$, among all rotation matrices, the one that best aligns
$\widehat{\bm{U}}$ and $\bm{U}$ is precisely $\mathsf{sgn}(\widehat{\bm{U}}^{\top}\bm{U})$
\citep[see, e.g., Appendix D.2.1 in][]{ma2017implicit}, namely, 
\[
\mathsf{sgn}(\widehat{\bm{U}}^{\top}\bm{U})=\arg\min_{\bm{O}\in\mathcal{O}^{r\times r}}\Vert\widehat{\bm{U}}\bm{O}-\bm{U}\Vert_{\mathrm{F}}.
\]

%% file: main_results.tex
\section{Main results\label{sec:main_1st_approx}}

In this section, we demonstrate the desirable statistical performance
of PCA in weak factor models. We will first present a master result
(cf.~Theorem~\ref{Thm UV 1st approx row-wise error}) on subspace
error decomposition. Based on this key result, we derive non-asymptotic
distributional characterization for our estimators for the factors
and factor loadings, paving the way to data-driven statistical inference.

First, we formally introduce the PC estimators of the factor and factor
loadings in Algorithm \ref{alg:PCA-SVD}.

\begin{algorithm}[h]
\caption{The PCA-based method.}

\label{alg:PCA-SVD}\begin{algorithmic}

\STATE \textbf{{Input}}: panel data $\bm{X}$, rank $r$.

\STATE \textbf{{Compute} }the truncated rank-$r$ SVD $\widehat{\bm{U}}\widehat{\bm{\Sigma}}\widehat{\bm{V}}^{\top}$
of $T^{-1/2}\bm{X}$, where $\widehat{\bm{U}}\in\mathbb{R}^{N\times r}$
and $\widehat{\bm{V}}\in\mathbb{R}^{T\times r}$ have orthonormal
columns, and $\widehat{\bm{\Sigma}}=\mathsf{diag}(\widehat{\sigma}_{1},\widehat{\sigma}_{2},\ldots,\widehat{\sigma}_{r})\in\mathbb{R}^{r\times r}$
satisfies that $\widehat{\sigma}_{1}\geq\widehat{\sigma}_{2}\geq\cdots\geq\widehat{\sigma}_{r}$.

\STATE \textbf{{Output}} $\widehat{\bm{F}}:=T^{1/2}\widehat{\bm{V}}$
as the estimator of factors $\bm{F}$, and $\widehat{\bm{B}}:=\widehat{\bm{U}}\widehat{\bm{\bm{\Sigma}}}$
as the estimator of factor loadings $\bm{B}$.

\end{algorithmic} 
\end{algorithm}

Our master result focuses on the subspace estimates $\widehat{\bm{U}}$
and $\widehat{\bm{V}}$. We define some relevant quantities first.
Denote by $\bm{U}\bm{\Lambda}\bm{V}^{\top}$ the SVD of $T^{-1/2}\bm{BF}^{\top}$,
where both $\bm{U}\in\mathbb{R}^{N\times r}$ and $\bm{V}\in\mathbb{R}^{T\times r}$
have orthonormal columns, and $\bm{\Lambda}=\mathsf{diag}(\lambda_{1},\lambda_{2},\ldots,\lambda_{r})\in\mathbb{R}^{r\times r}$
satisfies that $\lambda_{1}\geq\lambda_{2}\geq\cdots\geq\lambda_{r}\geq0$.
We define $n:=\max(N,T)$.

\subsection{A first-order characterization of subspace perturbation errors}

Before presenting our master results, we list two key assumptions
on the SNR. To characterize the SNR, we define 
\[
\theta:=\sigma_{r}/\sqrt{\Vert\bm{\Sigma}_{\varepsilon}\Vert_{2}},\qquad\vartheta_{k}:=\sigma_{r}/\Vert(\bm{\Sigma}_{\varepsilon}^{1/2})_{k,\cdot}\Vert_{1},\text{\qquad and\qquad}\vartheta:=\sigma_{r}/\Vert\bm{\Sigma}_{\varepsilon}^{1/2}\Vert_{1}.
\]
Here, the signal strength, which is the factor strength in our case,
is effectively represented by $\sigma_{r}$, the smallest singular
value of the factor loading matrix $\bm{B}$. The noise level is captured
by the norm of the covariance matrix $\bm{\Sigma}_{\varepsilon}$
for the idiosyncratic noise $\bm{E}=\bm{\Sigma}_{\varepsilon}^{1/2}\bm{Z}$
(cf.~(\ref{noise matrix formula})). The SNRs $\theta$, $\vartheta_{k}$,
and $\vartheta$ are similar, though their denominators adopt different
norms of $\bm{\Sigma}_{\varepsilon}$. The inclusion of $\ell_{1}$-norms
$\Vert(\bm{\Sigma}_{\varepsilon}^{1/2})_{k,\cdot}\Vert_{1}$ and $\Vert\bm{\Sigma}_{\varepsilon}^{1/2}\Vert_{1}=\max_{1\leq k\leq N}\Vert(\bm{\Sigma}_{\varepsilon}^{1/2})_{k,\cdot}\Vert_{1}$,
though less common than the spectral norm $\sqrt{\Vert\bm{\Sigma}_{\varepsilon}\Vert_{2}}=\Vert\bm{\Sigma}_{\varepsilon}^{1/2}\Vert_{2}$,
is that our main results are on the row-wise error bound and we need
these norms in the matrix inequalities for technical reasons. Then,
two crucial assumptions on SNR are as follows.

\begin{assumption} \label{Assump_SNR 1 norm}There exists a sufficiently
large constant $C_{1}>0$ such that
\begin{equation}
\vartheta\equiv\frac{\sigma_{r}}{\Vert\bm{\Sigma}_{\varepsilon}^{1/2}\Vert_{1}}\geq C_{1}\sqrt{\frac{n}{T}\log n}.\label{SNR 1-norm logn}
\end{equation}
\end{assumption}

\begin{assumption} \label{Assump_SNR 2 norm}There exists a sufficiently
large constant $C_{2}>0$ such that 
\begin{equation}
\theta\equiv\frac{\sigma_{r}}{\Vert\bm{\Sigma}_{\varepsilon}^{1/2}\Vert_{2}}\geq C_{2}\sqrt{\frac{n}{T}\log n}.\label{SNR 2-norm logn}
\end{equation}
\end{assumption}

Assumption \ref{Assump_SNR 1 norm} implies Assumption \ref{Assump_SNR 2 norm}
due to the elementary inequality $\Vert\bm{\Sigma}_{\varepsilon}^{1/2}\Vert_{2}\leq\Vert\bm{\Sigma}_{\varepsilon}^{1/2}\Vert_{1}$.
On the other hand, $\Vert\bm{\Sigma}_{\varepsilon}^{1/2}\Vert_{1}\leq s\Vert\bm{\Sigma}_{\varepsilon}^{1/2}\Vert_{2}$,
where $s$ is the maximum number of nonzero elements of $\bm{\Sigma}_{\varepsilon}^{1/2}$
in each row, which is small when $\bm{\Sigma}_{\varepsilon}$ is sparse.
In particular, when the noise covariance matrix $\bm{\Sigma}_{\varepsilon}$
is diagonal (i.e.,~no cross-sectional correlations, also known as
the strict factor model \citep{Ross1976APT,fanfanlv2008covFactor,baishi2011Cov}),
we have that $\Vert\bm{\Sigma}_{\varepsilon}^{1/2}\Vert_{2}=\Vert\bm{\Sigma}_{\varepsilon}^{1/2}\Vert_{1}$.
In this case, Assumption \ref{Assump_SNR 1 norm} is equivalent with
Assumption \ref{Assump_SNR 2 norm}. Throughout the paper, we denote
by $C_{0}$, $C_{1}$, $C_{2}$, $C_{U}$, $C_{V}$, $c_{0}$, etc.~the
generic constants that may vary from place to place. We are now ready
to present our main results as follows.

\begin{theorem}\label{Thm UV 1st approx row-wise error} Assume that
$T\geq C_{0}(r+\log n)$ and $n\geq C_{0}r\log n$ for some sufficiently
large constant $C_{0}>0$. Consider the first-order expansions 
\begin{subequations}
\begin{align}
\widehat{\bm{U}}\bm{R}_{U}-\bm{U} & =\bm{G}_{U}+\bm{\Psi}_{U}\text{\qquad with\qquad}\bm{G}_{U}:=T^{-1/2}\bm{E}\bm{V}\bm{\Lambda}^{-1},\label{U 1st-order approx hat til}\\
\widehat{\bm{V}}\bm{R}_{V}-\bm{V} & =\bm{G}_{V}+\bm{\Psi}_{V}\text{\qquad with\qquad}\bm{G}_{V}:=T^{-1/2}\bm{E}^{\top}\bm{U}\bm{\Lambda}^{-1},\label{V 1st-order approx hat til}
\end{align}
\end{subequations}
 where $\bm{R}_{U}:=\mathsf{sgn}(\widehat{\bm{U}}^{\top}\bm{U})$
and $\bm{R}_{V}:=\mathsf{sgn}(\widehat{\bm{V}}^{\top}\bm{V})$ are
two global rotation matrices. Under Assumptions \ref{Assump_Bf_identification},
\ref{Assump_noise_Z_entries}, and \ref{Assump_factor_f}, we have,
with probability at least $1-O(n^{-2})$, the remainder terms $\bm{\Psi}_{U}$
and $\bm{\Psi}_{V}$ are higher-order negligible terms that satisfy
the following bounds:

(i) Under Assumption \ref{Assump_SNR 1 norm}, there exists some universal
constant $C_{U}>0$ such that, uniformly for all $k=1,2,\ldots,N$,
\[
\Vert(\bm{\Psi}_{U})_{k,\cdot}\Vert_{2}\leq C_{U}\frac{\sqrt{n}}{\vartheta_{k}\vartheta T}\sqrt{r}\log^{3/2}n+C_{U}(\frac{n}{\theta^{2}T}+\frac{1}{\theta\sqrt{T}}\sqrt{r}\log n)\Vert\bm{U}_{k,\cdot}\Vert_{2}+C_{U}\frac{n}{\vartheta_{k}\vartheta T}\log n\Vert\bm{U}\Vert_{2,\infty}.
\]

(ii) Under Assumption \ref{Assump_SNR 2 norm}, there exists some
universal constant $C_{V}>0$ such that, uniformly for all $l=1,2,\ldots,T$,
\[
\Vert(\bm{\Psi}_{V})_{l,\cdot}\Vert_{2}\leq C_{V}\frac{\sqrt{n}}{\theta^{2}T}\sqrt{r}\log^{3/2}n+C_{V}(\frac{n}{\theta^{2}T}+\frac{1}{\theta\sqrt{T}})\sqrt{r}\log n\Vert\bm{V}_{l,\cdot}\Vert_{2}.
\]

\end{theorem}

The perturbation bounds in Theorem \ref{Thm UV 1st approx row-wise error}
pave the way for the statistical inference of the factor loadings
and factors. The reasons why the quantities presented in Theorem \ref{Thm UV 1st approx row-wise error}
relate to the factor loadings $\bm{B}$ and factors $\bm{F}$ are
that, as demonstrated in Lemma \ref{Lemma SVD BF good event}, the
column subspaces of $\bm{U}$ and $\bm{B}$ are identical, as are
those of $\bm{V}$ and $\bm{F}$.

To show how strong our SNR condition is, we examine scenarios where
there is no correlation in the noise matrix $\bm{E}$, i.e., $\bm{\Sigma}_{\varepsilon}$
is diagonal as in the strict factor model. In this case, the SNRs
satisfy $\vartheta=\theta=\sigma_{r}/\sigma_{\varepsilon}$, where
$\sigma_{\varepsilon}=\Vert\bm{\Sigma}_{\varepsilon}^{1/2}\Vert_{2}=\Vert\bm{\Sigma}_{\varepsilon}^{1/2}\Vert_{1}$.
Then Assumptions \ref{Assump_SNR 1 norm} and \ref{Assump_SNR 2 norm}
match the minimal condition $\theta\gg\sqrt{n/T}$ required for the
consistent estimations of the left and right singular subspaces \citep[e.g.,][]{yan2021inference},
up to log factors.

To explain that why $\bm{\Psi}_{U}$ and $\bm{\Psi}_{V}$ are higher-order
negligible terms, we compare their bounds with the magnitude of the
first-order terms $\bm{G}_{U}$ and $\bm{G}_{V}$. For simplicity,
we consider the setting where $r\asymp1$, $N\asymp T$, $\bm{\Sigma}_{\varepsilon}$
is well-conditioned (i.e., $\kappa_{\varepsilon}:=\lambda_{\text{max}}(\bm{\Sigma}_{\varepsilon})/\lambda_{\text{min}}(\bm{\Sigma}_{\varepsilon})\asymp1$),
and $\bm{U}$ is incoherent \citep[see, e.g., Definition 3.1 in][]{chen2021Monograph};
the matrix $\bm{V}$ can be proven to be incoherent since $\Vert\bm{V}\Vert_{2,\infty}\lesssim T^{-1/2}\log^{1/2}n$.
Then, the typical size of the $k$-th row of the first-order term
$\bm{G}_{U}$ can be measured by $\mathbb{E}^{1/2}[\Vert(\bm{G}_{U})_{k,\cdot}\Vert_{2}^{2}]$.
Standard computations yield $\mathbb{E}^{1/2}[\Vert(\bm{G}_{U})_{k,\cdot}\Vert_{2}^{2}]\geq T^{-1/2}\sigma_{r}^{-1}\sqrt{\lambda_{\text{min}}(\bm{\Sigma}_{\varepsilon})}$.
Similarly, for the first-order term $\bm{G}_{V}$, we have that $\mathbb{E}^{1/2}[\Vert(\bm{G}_{V})_{l,\cdot}\Vert_{2}^{2}]\geq T^{-1/2}\sigma_{r}^{-1}\sqrt{\lambda_{\text{min}}(\bm{\Sigma}_{\varepsilon})}$.
Then, using the row-wise bounds for $\bm{\Psi}_{U}$ and $\bm{\Psi}_{V}$
in Theorem \ref{Thm UV 1st approx row-wise error}, we obtain that
\begin{subequations}
\begin{equation}
(\mathbb{E}^{1/2}[\left\Vert (\bm{G}_{U})_{k,\cdot}\right\Vert _{2}^{2}])^{-1}\left\Vert (\bm{\Psi}_{U})_{k,\cdot}\right\Vert _{2}\lesssim(\frac{\theta}{\vartheta_{k}\vartheta}+\frac{1}{\theta})\log^{3/2}n+\frac{1}{\sqrt{N}}\log n\ll1,\label{U first-order term size ratio}
\end{equation}
and 
\begin{equation}
(\mathbb{E}^{1/2}[\left\Vert (\bm{G}_{V})_{l,\cdot}\right\Vert _{2}^{2}])^{-1}\left\Vert (\bm{\Psi}_{V})_{l,\cdot}\right\Vert _{2}\lesssim(\frac{1}{\theta}+\frac{1}{\sqrt{T}})\log^{3/2}n\ll1.\label{V first-order term size ratio}
\end{equation}
\end{subequations}
 The last inequality in (\ref{U first-order term size ratio}) (resp.
(\ref{V first-order term size ratio})) implies that the first-order
term $(\bm{G}_{U})_{k,\cdot}$ (resp. $(\bm{G}_{V})_{l,\cdot}$) dominates
the higher-order term $(\bm{\Psi}_{U})_{k,\cdot}$ (resp. $(\bm{\Psi}_{V})_{l,\cdot}$),
provided that the SNRs grow at polynomial rate of $\log n$: $\theta^{-1}\vartheta_{k}\vartheta\gg\log^{3/2}n$
and $\theta\gg\log^{3/2}n$ (resp. $\theta\gg\log^{3/2}n$). These
conditions are less stringent than the assumptions in the prior work
\citep[e.g.,][]{BaiNg2023PCA,Jiang2023PCA,ChoiMing2024PCA} that required
the SNR to grow with a polynomial rate of $n$.

As we will show later, Theorem \ref{Thm UV 1st approx row-wise error}
is the foundation to design test statistics and conduct statistical
inference for the factors, factor loadings, and other parameters of
interest. To the best of our knowledge, the non-asymptotic first-order
approximations (\ref{U 1st-order approx hat til})--(\ref{V 1st-order approx hat til})
are new in the statistics and econometrics literature. While similar
findings exist in the fields of low-rank matrix completion and spectral
methods, the closest result to ours is Theorem 9 of \citet{yan2021inference}.
However, they assumed that the entries of the noise matrix $\bm{E}$
are independent, while we allow the presence of cross-sectional correlations
in the noise matrix $\bm{E}$.

Our results in Theorem \ref{Thm UV 1st approx row-wise error} are
nontrivial generalizations of those in \citet{yan2021inference}.
We note that their proof of row-wise error bounds via leave-one-out
(LOO) technique requires constructing an auxiliary matrix that is
independent with the target row. This construction is obtained by zeroing
the row with the same position in noise matrix. However, this approach
does not work in our case because it is still correlated with the
target row due to the cross-sectional correlation among all rows.
We overcome this challenge through applications of matrix concentration
inequalities, and our results do not need any additional structural
assumption (e.g., sparsity) on the noise covariance matrix $\bm{\Sigma}_{\varepsilon}$.
An additional advance of our theory compared with \citet{yan2021inference}
lies in that, our bounds for $\Vert(\bm{\Psi}_{U})_{k,\cdot}\Vert_{2}$
and $\Vert(\bm{\Psi}_{V})_{l,\cdot}\Vert_{2}$ are free of condition
number $\kappa:=\sigma_{1}/\sigma_{r}$, making our theory work even
when the factor loading matrix $\bm{B}$ is near-singular and ill-conditioned,
i.e., $\kappa$ is very large. In particular, our result accomodates
the heterogeneous case of \citet{BaiNg2023PCA}.

\subsection{Implications: estimation guarantees and distributional characterizations
\label{sec:F B inference estimate}}

In this section, we present the immediate consequences of Theorem
\ref{Thm UV 1st approx row-wise error} --- estimation error bounds
and distributional characterization for the PC estimators. Later in
Section \ref{sec:compare_SNR}, we will compare these results with
the recent work (\citet{BaiNg2023PCA,Jiang2023PCA,ChoiMing2024PCA}),
and highlight the advantages of our theory.

Recall that our estimators for the factors $\bm{F}$ and factor loadings
$\bm{B}$ in Algorithm \ref{alg:PCA-SVD} are given by $\widehat{\bm{F}}=\sqrt{T}\widehat{\bm{V}}$
and $\widehat{\bm{B}}=\widehat{\bm{U}}\widehat{\bm{\bm{\Sigma}}}$
respectively. Let us take a moment to see how to quantify the estimation
error for $\widehat{\bm{F}}$ in the face of rotational ambiguity.
In view of Theorem~\ref{Thm UV 1st approx row-wise error}, we know
that $\widehat{\bm{F}}\bm{R}_{V}$ should be close to $\sqrt{T}\bm{V}$
in Euclidean distance. In addition, recall that $\bm{U}\bm{\Lambda}\bm{V}^{\top}$
is the SVD of the rank-$r$ matrix $T^{-1/2}\bm{BF}^{\top}$, which
suggests the existence of an invertible, $\sigma(\bm{F})$-measurable
matrix $\bm{J}\in\mathbb{R}^{r\times r}$ such that $\bm{V}=T^{-1/2}\bm{FJ}$.
Hence, by defining $\bm{R}_{F}\coloneqq\bm{J}\bm{R}_{V}^{\top}$,
we may use $\widehat{\bm{F}}-\bm{F}\bm{R}_{F}$ to evaluate the estimation
error of $\widehat{\bm{F}}$. Similarly, we define $\bm{R}_{B}\coloneqq(\bm{R}_{F}^{-1})^{\top}$
and use $\widehat{\bm{B}}-\bm{B}\bm{R}_{B}$ to measure the estimation
error of $\widehat{\bm{B}}$.

We will show in Lemma \ref{Lemma SVD BF good event} that $\bm{R}_{F}$
and $\bm{R}_{B}$ are close to rotation matrices. Indeed, in the study
of PC estimators for factor model (e.g., \citet{BaiNg2023PCA,Jiang2023PCA,ChoiMing2024PCA}),
it is quite common to take $\bm{F}\bm{R}_{F}$ and $\bm{B}\bm{R}_{B}$,
rather than $\bm{F}$ and $\bm{B}$, as the groundtruth, though the
specific choices of $\bm{R}_{F}$ and $\bm{R}_{B}$ vary across different
works. Both the averaged and the row-wise estimation error bound can
be deduced from Theorem \ref{Thm UV 1st approx row-wise error}. We
record the results in the following corollary.

\begin{corollary}[\textsf{Estimation guarantees for factors and
factor loadings}]\label{corollary:error bound B F}Suppose that the
assumptions in Theorem \ref{Thm UV 1st approx row-wise error} hold.
Then there exists a universal constant $C_{0}>0$ such that: with
probability at least $1-O(n^{-2})$, for factors, the averaged and
the row-wise estimation error are given by 
\begin{equation}
\frac{1}{T}\big\Vert\widehat{\bm{F}}-\bm{F}\bm{R}_{F}\big\Vert_{\mathrm{F}}^{2}\leq C_{0}\frac{1}{\theta^{2}}\frac{n}{T}r\text{\qquad and\qquad}\big\Vert(\widehat{\bm{F}}-\bm{F}\bm{R}_{F})_{t,\cdot}\big\Vert_{2}\leq C_{0}\frac{1}{\theta}\sqrt{r}\log n,\forall t,\label{averaged and rowwise error F}
\end{equation}
respectively; for factor loadings, we define $\bar{\bm{U}}:=\bm{B\Sigma}^{-1}\in\mathbb{R}^{N\times r}$,
if there exists some universal constant $C_{U}>0$ such that $\Vert\bar{\bm{U}}\Vert_{2,\infty}\leq C_{U}\theta/\vartheta$,
$\Vert\bar{\bm{U}}\Vert_{2,\infty}\leq C_{U}\sqrt{T}\vartheta/n$
and $\Vert\bar{\bm{U}}\Vert_{2,\infty}\leq C_{U}\sqrt{T}\theta^{2}/(\vartheta n)$,
then the averaged estimation error bound is given by 
\begin{subequations}
\begin{equation}
\frac{1}{N}\big\Vert\widehat{\bm{B}}-\bm{B}\bm{R}_{B}\big\Vert_{\mathrm{F}}^{2}\leq C_{0}\bigl(\frac{1}{T}\Vert\bm{\Sigma}_{\varepsilon}^{1/2}\Vert_{1}^{2}+(\frac{n^{2}}{\theta^{4}T^{2}}+\frac{1}{\theta^{2}T})\sigma_{r}^{2}\bigr)r^{3}\log^{4}n,\label{averaged error B}
\end{equation}
and the row-wise estimation error bound is given by 
\begin{equation}
\big\Vert(\widehat{\bm{B}}-\bm{B}\bm{R}_{B})_{i,\cdot}\big\Vert_{2}\leq C_{0}\bigl(\frac{1}{\sqrt{T}}\Vert(\bm{\Sigma}_{\varepsilon}^{1/2})_{i,\cdot}\Vert_{1}+(\frac{n}{\theta^{2}T}+\frac{1}{\theta\sqrt{T}})\sigma_{r}\Vert\bar{\bm{U}}_{i,\cdot}\Vert_{2}\bigr)r\log^{2}n.\label{rowwise error B}
\end{equation}
\end{subequations}

\end{corollary}

Let us interpret the above bounds under some specific settings. For
simplicity, we ignore the log factors $\log n$ and consider the setting
discussed after Theorem \ref{Thm UV 1st approx row-wise error}, with
the assumption that $\Vert\bm{\Sigma}_{\varepsilon}^{1/2}\Vert_{1}\asymp\Vert\bm{\Sigma}_{\varepsilon}^{1/2}\Vert_{2}\asymp1$.
Under this setting, to make the assumptions on $\Vert\bar{\bm{U}}\Vert_{2,\infty}$
hold, it suffices to assume that $\theta\gtrsim1$, where $\theta\equiv\sigma_{r}/\Vert\bm{\Sigma}_{\varepsilon}^{1/2}\Vert_{2}$
is the SNR. For factors, the upper bounds for the averaged and row-wise
estimation error rates in (\ref{averaged and rowwise error F}) are
given by $\theta^{-2}$ and $\theta^{-1}$, respectively; for factor
loadings, the two bounds in (\ref{averaged error B})--(\ref{rowwise error B})
are given by $T^{-1}+N^{-1}\theta^{-2}$ and $T^{-1/2}+N^{-1/2}\theta^{-1}$,
respectively. All these bounds go to zero under the condition that
$\theta\gg1$, which match the minimal condition required for the
consistent estimations of the left and right singular subspaces if
there is no correlation in the noise matrix \citep[e.g.,][]{yan2021inference}.

Next, we present our results on the inference for factors and factor
loadings in Corollaries \ref{Thm F inference} and \ref{Thm B inference},
respectively.

\begin{corollary}[\textsf{Distributional theory for factors}]\label{Thm F inference}Suppose
that Assumptions \ref{Assump_Bf_identification}, \ref{Assump_noise_Z_entries},
and \ref{Assump_factor_f} hold. For any given target error level
$\delta>0$, assume that $T\geq C_{1}(r+\log n)$, $n\geq C_{1}r\log n$,
$T\geq C_{0}\kappa_{\varepsilon}\delta^{-2}r^{2}\log^{4}n$, $n\geq C_{0}\delta^{-1/2}$,
\[
\theta\geq C_{0}\delta^{-1}\sqrt{\kappa_{\varepsilon}}\frac{n}{T}r\log^{3/2}n,\qquad\text{and}\qquad\frac{\Vert\bm{\Sigma}_{\varepsilon}^{1/2}\bar{\bm{U}}\Vert_{2,\infty}}{\Vert\bm{\Sigma}_{\varepsilon}^{1/2}\Vert_{2}}\leq c_{0}\delta\kappa_{\varepsilon}^{-1/2}r^{-5/4},
\]
hold for some universal constant $c_{0},C_{0}>0$ and sufficiently
large constant $C_{1}>0$. Then, for any $t=1,2,\ldots,T$, it holds
that 
\begin{equation}
\sup_{\mathcal{C}\in\mathscr{C}^{r}}\left\vert \mathbb{P}((\widehat{\bm{F}}-\bm{F}\bm{R}_{F})_{t,\cdot}\in\mathcal{C})-\mathbb{P}(\mathcal{N}(0,\bm{\Sigma}_{F,t})\in\mathcal{C})\right\vert \leq\delta,\label{F Gaussian law result}
\end{equation}
where $\mathscr{C}^{r}$ is the collection of all convex sets in $\mathbb{R}^{r}$,
and the covariance matrix $\bm{\Sigma}_{F,t}$ is given by $\bm{\Sigma}_{F,t}=\bm{R}_{V}\bm{\Lambda}^{-1}\bm{U}^{\top}\bm{\Sigma}_{\varepsilon}\bm{U}\bm{\Lambda}^{-1}\bm{R}_{V}^{\top}$.

\end{corollary}

\begin{remark}If all the entries of the matrix $\bm{Z}$ are Gaussian,
i.e., the noise is Gaussian, then the result (\ref{F Gaussian law result})
holds without the assumption that $\Vert\bm{\Sigma}_{\varepsilon}^{1/2}\bar{\bm{U}}\Vert_{2,\infty}/\Vert\bm{\Sigma}_{\varepsilon}^{1/2}\Vert_{2}\leq c_{0}\kappa_{\varepsilon}^{-1/2}r^{-5/4}\delta$.
Indeed, in the scenario where $\Vert\bm{\Sigma}_{\varepsilon}^{1/2}\Vert_{1}/\Vert\bm{\Sigma}_{\varepsilon}^{1/2}\Vert_{2}\asymp1$
and $\bar{\bm{U}}$ is $\mu$-incoherent, 
fulfilling our assumption on $\Vert\bm{\Sigma}_{\varepsilon}^{1/2}\bar{\bm{U}}\Vert_{2,\infty}/\Vert\bm{\Sigma}_{\varepsilon}^{1/2}\Vert_{2}$
merely requires that $N\geq\mu\kappa_{\varepsilon}r^{5/2}c_{0}^{-2}\delta^{-2}$
on the growth rate of cross-sectional dimension $N$.\end{remark}

\begin{corollary}[\textsf{Distributional theory for factor loadings}]\label{Thm B inference}Suppose
that Assumptions \ref{Assump_Bf_identification}, \ref{Assump_noise_Z_entries},
and \ref{Assump_factor_f} hold. Assume that there exists a constant
$c_{\varepsilon}>0$ such that $\lambda_{\min}(\bm{\Sigma}_{\varepsilon})>c_{\varepsilon}$.
For any given target error level $\delta>0$, assume that $T\geq C_{1}(r+\log n)$,
$n\geq C_{1}r\log n$, $T\geq C_{0}\delta^{-2}(r^{2}\log n+\kappa_{\varepsilon}r\log^{4}n)$,
\[
\vartheta_{i}\geq C_{0}\delta^{-1}\frac{1}{\sqrt{c_{\varepsilon}}}\Vert\bm{\Sigma}_{\varepsilon}^{1/2}\Vert_{1}\frac{n}{T}\sqrt{r}\log^{3/2}n,\qquad\text{and}\qquad\theta\geq C_{0}\delta^{-1}\sqrt{\kappa_{\varepsilon}}\frac{n}{T}\sqrt{r}\log n,
\]
for some universal constant $C_{0}>0$ and sufficiently large constant
$C_{1}>0$. Then, for any $i=1,2,\ldots,N$, it holds that 
\[
\sup_{\mathcal{C}\in\mathscr{C}^{r}}\left\vert \mathbb{P}((\widehat{\bm{B}}-\bm{B}\bm{R}_{B})_{i,\cdot}\in\mathcal{C})-\mathbb{P}(\mathcal{N}(0,\bm{\Sigma}_{B,i})\in\mathcal{C})\right\vert \leq\delta,
\]
where $\mathscr{C}^{r}$ is the collection of all convex sets in $\mathbb{R}^{r}$,
and the covariance matrix $\bm{\Sigma}_{B,i}$ is given by $\bm{\Sigma}_{B,i}=T^{-1}(\bm{\Sigma}_{\varepsilon})_{i,i}\bm{I}_{r}$.

\end{corollary}

Let us look at the assumptions about the SNRs $\theta$ and $\vartheta_{i}$
in Corollaries \ref{Thm F inference} and \ref{Thm B inference}.
Consider the setting where $r\asymp1$, $N\asymp T$, $c_{\varepsilon}\asymp1$,
$\Vert\bm{\Sigma}_{\varepsilon}^{1/2}\Vert_{1}\asymp\Vert\bm{\Sigma}_{\varepsilon}^{1/2}\Vert_{2}\asymp1$,
$\Vert(\bm{\Sigma}_{\varepsilon}^{1/2})_{i,\cdot}\Vert_{1}\asymp1$,
and $\bm{\Sigma}_{\varepsilon}$ is well-conditioned. Then, the SNR
condition in Corollary \ref{Thm F inference} (resp. Corollary \ref{Thm B inference})
is equivalent to $\theta\gg\log^{3/2}N$ (resp. $\theta\gg\log N$
and $\vartheta_{i}\gg\log^{3/2}N$), indicating that inference of
factors and factor loadings is achievable as long as the SNR grows
faster than a polynomial rate of $\log N$. Our SNR condition is less
restrictive than the prior work \citep[e.g.,][]{BaiNg2023PCA,Jiang2023PCA,ChoiMing2024PCA}
that required the SNR to grow with a polynomial rate of $N$. Regardless
of log factors, our SNR condition for inference is equivalent to $\theta\gg1$,
which is optimal since it matches the same condition required for
consistency as commented after Corollary \ref{corollary:error bound B F}.
Also, in the special case where $\bm{\Sigma}_{\varepsilon}$ is diagonal,
i.e., there is no correlation in the noise matrix, our SNR condition
for inference matches the minimal condition required for the subspace
inference \citep[e.g.,][]{yan2021inference}.

Both Corollaries \ref{Thm F inference} and \ref{Thm B inference}
are stated in a non-asymptotic sense, and can be easily translated
to asymptotic normality results. In practice, the confidence regions
for the factors and the factor loadings can be constructed by replacing
the asymptotic covariance matrices $\bm{\Sigma}_{F,t}$ and $\bm{\Sigma}_{B,i}$
with their consistent estimators: 
\[
\widehat{\bm{\Sigma}}_{F,t}=\widehat{\bm{\Sigma}}^{-1}\widehat{\bm{U}}^{\top}\widehat{\bm{\bm{\Sigma}}}_{\varepsilon}^{\tau}\widehat{\bm{U}}\widehat{\bm{\Sigma}}^{-1}\qquad\text{and}\qquad\widehat{\bm{\Sigma}}_{B,i}=\frac{1}{T}(\widehat{\bm{\bm{\Sigma}}}_{\varepsilon}^{\tau})_{i,i}\bm{I}_{r},
\]
respectively. Here, $\widehat{\bm{\bm{\Sigma}}}_{\varepsilon}^{\tau}$
is the estimator of the noise covariance matrix which we will introduce
in detail in Section \ref{sec:noise cov estimation}. The intuition
behind the above consistent estimators is that, $\widehat{\bm{U}}$
and $\widehat{\bm{\Sigma}}$ are the consistent estimators of $\bm{U}\bm{R}_{U}$
and $\bm{R}_{U}^{\top}\bm{\Lambda}\bm{R}_{V}$, respectively, and
we will show this fact in the proof of Theorem \ref{Thm UV 1st approx row-wise error}. 

%% file: compare_SNR_rate.tex
\section{Comparison with previous work\label{sec:compare_SNR}}

In this section, we compare our main results with related ones established
in prior works \citet{BaiNg2023PCA,Jiang2023PCA,ChoiMing2024PCA}.
All of them studied the PC estimators in weak factor models, and established
estimation error bounds and asymptotic normality under different assumptions.
Distinct from these three papers, all our results are entirely non-asymptotic,
providing finite-sample characterizations for both the estimation
error and the uncertainty level of statistical inference. From a technical
viewpoint, in the regime $N\asymp T$, a notable advancement of our
theory is that, our assumptions for inference require the growth rate
of SNR to be faster than a polynomial rate of $\log N$, unlike the
polynomial rate of $N$ required in these three papers. Also, all
these three papers assume that the $r$ singular values $\{\sigma_{i}\}_{i=1}^{r}$
of factor loading $\bm{B}$ are distinct, while our theory does not
require such eigengap condition.

Having provided an overview of the advantages of our theory, we now
make the comparisons in detail. 
\begin{itemize}
\item \citet{BaiNg2023PCA} studied both the homogeneous case where $\sigma_{i}^{2}\asymp N^{\alpha}$
for $i=1,2,\ldots,r$, and the heterogeneous case where $\sigma_{i}^{2}\asymp N^{\alpha_{i}}$
with $1\geq\alpha_{1}\geq\alpha_{2}\geq\ldots\geq\alpha_{r}>0$. They
required $\alpha_{r}>1/2$ to prove the asymptotic normality for inference.
However, the case when $0<\alpha_{r}\leq1/2$ has not been covered
by their inferential theory. Our theory fills this gap and establishes
the inference results even when $\sigma_{r}$ do not grow with a polynomial
rate of $N$. Besides the inference, our theory provides row-wise
estimation error bounds for factors and factor loadings as detailed
in Section \ref{sec:F B inference estimate}, while the bounds in
\citet{BaiNg2023PCA} are Frobenius norm bounds measuring the average
error over all rows. The consistency and asymptotic normality results
in \citet{Jiang2023PCA} are similar to those in \citet{BaiNg2023PCA},
while the focus of \citet{Jiang2023PCA} is to identify the so-called
pseudo-true parameter that is consistently estimated by the PCA estimator.
We note that both their frameworks accommodate temporal correlation,
an aspect not covered by our theory. 
\item \citet{ChoiMing2024PCA} adopted the leave-one-out
analysis similar in spirit to those used in matrix completion to investigate
the PCA estimators. The setup they studied is the homogeneous case
in \citet{BaiNg2023PCA} because they assumed that $N^{-\alpha}\bm{B}^{\top}\bm{B}\rightarrow\bm{\Sigma}_{\bm{B}}$,
which requires $\sigma_{i}^{2}\asymp N^{\alpha}$ for all the singular
values $\{\sigma_{i}\}_{i=1,2,\ldots,r}$ of factor loading $\bm{B}$.
When temporal correlation does not exist, similar to our results,
they also filled the gap to establish the inference results for $0<\alpha\leq1/2$.
In comparison, our theory is fully non-asymptotic and does not assume
any asymptotic growth assumptions on the singular values. Also, the
homogeneous case they studied assumes that the condition number $\kappa=\sigma_{1}/\sigma_{r}$
satisfies that $\kappa\asymp1$, while our theory does not need this
assumption and allows any growth rate of the condition number $\kappa$. 
\end{itemize}
To enhance the clarity of the comparison of the SNR assumption for
inference, we detailed the results in Table \ref{table:SNR normality}
under a specific setup: Assumptions \ref{Assump_Bf_identification},
\ref{Assump_noise_Z_entries}, \ref{Assump_factor_f} hold; The cross-sectional
and temporal dimensions satisfy that $N\asymp T\asymp n$; the number
of factors satisfies $r\asymp1$; the noise covariance matrix satisfies
$\kappa_{\varepsilon}\asymp1$, $\left\Vert \bm{\Sigma}_{\varepsilon}\right\Vert _{2}\asymp1$,
and $\Vert\bm{\Sigma}_{\varepsilon}^{1/2}\Vert_{1}\asymp1$. Since
the scale of the noise is reflected by $\bm{\Sigma}_{\varepsilon}$
and in our current setup we have that $\left\Vert \bm{\Sigma}_{\varepsilon}\right\Vert _{2}\asymp1$
and $\Vert\bm{\Sigma}_{\varepsilon}^{1/2}\Vert_{1}\asymp1$, the assumptions
on the SNRs $\theta=\sigma_{r}/\sqrt{\left\Vert \bm{\Sigma}_{\varepsilon}\right\Vert _{2}}$
and $\vartheta=\sigma_{r}/\Vert\bm{\Sigma}_{\varepsilon}^{1/2}\Vert_{1}$
defined in Section \ref{sec:main_1st_approx} are fully represented
by the growth rate of the smallest singular value $\sigma_{r}$ of
factor loading matrix $\bm{B}$.

\begin{table}[h]
\caption{SNR assumptions for inference of factor and factor loading\label{table:SNR normality}}

\centering

\begin{tabular}{c|c|c}
\hline 
 & \multirow{2}{*}{Factor} & \multirow{2}{*}{Factor loading}\tabularnewline
 &  & \tabularnewline
\hline 
\multirow{2}{*}{\citet{BaiNg2023PCA}} & \multirow{2}{*}{$\sigma_{r}\gg N^{1/4}$} & \multirow{2}{*}{$\sigma_{r}\gg N^{1/4}$}\tabularnewline
 &  & \tabularnewline
\hline 
\multirow{2}{*}{\citet{Jiang2023PCA}} & \multirow{2}{*}{$\sigma_{r}\gg\max(N^{1/4},\kappa^{2})$} & \multirow{2}{*}{$\sigma_{r}\gg\kappa^{1/2}N^{1/4}$}\tabularnewline
 &  & \tabularnewline
\hline 
\multirow{2}{*}{\citet{ChoiMing2024PCA}} & \multirow{2}{*}{$\sigma_{r}\asymp N^{\alpha/2}\text{ with }\alpha>0$} & \multirow{2}{*}{$\sigma_{r}\gg N^{1/6}\log^{\omega/6}N$}\tabularnewline
 &  & \tabularnewline
\hline 
\multirow{2}{*}{Our Theory} & \multirow{2}{*}{$\sigma_{r}\gg\log^{3/2}N$} & \multirow{2}{*}{$\sigma_{r}\gg\log^{3/2}N$}\tabularnewline
 &  & \tabularnewline
\hline 
\end{tabular}
\end{table}

Table \ref{table:SNR normality} demonstrates the advancement of our
theory. For the conditions of \citet{ChoiMing2024PCA} in Table \ref{table:SNR normality},
the parameter $\omega>0$ is defined in Assumption B'' (iv) therein
on the noise structure and they need two additional growth rate assumptions
on $\sigma_{r}$ (see Assumption D'' (i) therein). For other regimes
where $T\ll N$ or $N\ll T$, similar comparisons can be made, so
we omit the details for the limit of space. 

%% file: applications.tex
\section{Applications in econometrics\label{sec:3_econ_applications}}

Our first-order approximations do not only advance the existing theories
on the PCA under the weak factor model, but also pave the path for
various statistical tests that are useful in economics and finance.
In this section, we show four applications of our results.

\subsection{The factor specification tests}

Recall that the factor model is given by $\bm{X}=\bm{BF}^{\top}+\bm{E}$,
where the matrix $\bm{F}$ is the realization of latent factors. Suppose
we have time series data of some observed factors, e.g., the Fama-French
factors. Our focus is on testing if the observed factor is in the
linear space spanned by the latent factors $\bm{F}$. In particular,
we examine this linear dependence in a flexible range of the whole
period $[1,T]$. Formally, we consider an index set $S=\{t_{1},t_{2},\ldots,t_{|S|}\}\subseteq\{1,2,\ldots,T\}$
of interest, and we have the data $\bm{v}=(v_{t_{1}},v_{t_{2}},\ldots,v_{t_{|S|}})^{\top}\in\mathbb{R}^{|S|}$
of an observed factor recorded at the time index set $S$. We test
the hypothesis as follows, 
\begin{equation}
H_{0}:\qquad\text{There exists }\bm{w}\in\mathbb{R}^{r}\text{ such that }\bm{v}=\bm{F}_{S,\cdot}\bm{w}.\label{Null H0 factor test subset}
\end{equation}
Under the null hypothesis $H_{0}$, we have that $v_{t}=\bm{F}_{t,\cdot}\bm{w}=\bm{f}_{t}^{\top}\bm{w}$
for any $t\in S$.

Under the strong factor model where the pervasiveness assumption holds,
\citet{BaiNg2006factor} studied this problem and designed test statistics
for the whole set, i.e., $S=\{1,2,\ldots,T\}$. Our study extends
to the case when $S$ is a subset of the whole time span $\{1,2,\ldots,T\}$
under the weak factor model. The subset $S$ can be any specific time
window of interest. This scenario is economically meaningful because
the relationship $v_{t}=\bm{f}_{t}^{\top}\bm{w}$ may only be valid
for relatively short periods, not necessarily across the entire span;
see \citet{BaiNg2006factor} for the connections between the CAPM
analysis and this problem. In Section \ref{sec:Numerical-experiments},
we show that our factor specification test results strikingly reconcile
with the economic cycles and financial crisis.

Our test statistic relies on the estimation of the noise covariance
matrix $\bm{\Sigma}_{\varepsilon}$. We denote by $\widehat{\bm{\bm{\Sigma}}}_{\varepsilon}^{\tau}$
the estimator of $\bm{\Sigma}_{\varepsilon}$ and we will elaborate
its construction in the next section. In Theorem \ref{Thm factor test plug-in Chi-sq}
below, we keep the error bound $\Vert\widehat{\bm{\bm{\Sigma}}}_{\varepsilon}^{\tau}-\bm{\Sigma}_{\varepsilon}\Vert_{2}$
in the final result.

\begin{theorem}\label{Thm factor test plug-in Chi-sq} For any given
target error level $\delta>0$, assume there exist some universal
constants $c_{0},C_{0}>0$ and sufficiently large constant $C_{1}>0$
such that, $n\geq C_{1}r\log n$, $|S|\geq C_{1}(r+\log n)$, $T\geq C_{0}|S|\delta^{-2}\kappa_{\varepsilon}^{2}r^{2}\log^{4}n$,
\begin{equation}
\theta\geq C_{0}\delta^{-1}\frac{n}{T}\sqrt{|S|}\kappa_{\varepsilon}^{4}r^{2}\log^{4}n,\qquad\text{and}\qquad\frac{\Vert\bm{\Sigma}_{\varepsilon}^{1/2}\bar{\bm{U}}\Vert_{2,\infty}}{\Vert\bm{\Sigma}_{\varepsilon}^{1/2}\Vert_{2}}\leq c_{0}\delta\kappa_{\varepsilon}^{-1}\sqrt{\frac{T}{n}}|S|^{-3/2};\label{assemble condition SNR and data size}
\end{equation}
and that for the estimator $\widehat{\bm{\bm{\Sigma}}}_{\varepsilon}^{\tau}$
of the noise covariance matrix $\bm{\Sigma}_{\varepsilon}$, it holds
\[
\frac{1}{\left\Vert \bm{\Sigma}_{\varepsilon}\right\Vert _{2}}\big\Vert\widehat{\bm{\Sigma}}_{\varepsilon}^{\tau}-\bm{\Sigma}_{\varepsilon}\big\Vert_{2}\leq c_{0}\delta|S|^{-1/2}\kappa_{\varepsilon}^{-4}r^{-1}\log^{-2}n.
\]
Then, under Assumptions \ref{Assump_Bf_identification}, \ref{Assump_noise_Z_entries},
\ref{Assump_factor_f}, and the null hypothesis $H_{0}$ in (\ref{Null H0 factor test subset}),
we have that 
\begin{equation}
\left|\mathbb{P}\left(\mathfrak{\widehat{\mathcal{T}}}(S,\bm{v})\leq\chi_{1-\alpha}^{2}(|S|-r)\right)-(1-\alpha)\right|\leq\delta,\label{Chi-square stat final inequality}
\end{equation}
where the test statistic $\mathfrak{\widehat{\mathcal{T}}}(S,\bm{v})$
is defined by 
\[
\mathfrak{\widehat{\mathcal{T}}}(S,\bm{v}):=\frac{1}{\widehat{\phi}}\bm{v}^{\top}(\bm{I}_{|S|}-\bm{P}_{\widehat{\bm{V}}_{S,\cdot}})\bm{v}\qquad\text{with}\qquad\widehat{\phi}:=\frac{1}{T}((\widehat{\bm{V}}_{S,\cdot})^{+}\bm{v})^{\top}\widehat{\bm{\Sigma}}^{-1}\widehat{\bm{U}}^{\top}\widehat{\bm{\bm{\Sigma}}}_{\varepsilon}^{\tau}\widehat{\bm{U}}\widehat{\bm{\Sigma}}^{-1}(\widehat{\bm{V}}_{S,\cdot})^{+}\bm{v}.
\]
\end{theorem}

\begin{remark}If all the entries of the noise matrix $\bm{Z}$ in
(\ref{noise matrix formula}) are Gaussian, then we do not need the
second assumption in \eqref{assemble condition SNR and data size}.
\end{remark}

The idea of the test statistic $\mathfrak{\widehat{\mathcal{T}}}(S,\bm{v})$
is to utilize the fact that, under the null hypothesis $H_{0}$, the
residual vector of $\bm{v}$ is zero after projection onto the column
space of $\bm{F}_{S,\cdot}$. Note that $\widehat{\bm{V}}_{S,\cdot}$
estimates the column space of $\bm{F}_{S,\cdot}$. So, we construct
$\mathfrak{\widehat{\mathcal{T}}}(S,\bm{v})$ by computing the $\ell_{2}$-norm
of the projection residual vector $\Vert(\bm{I}_{|S|}-\bm{P}_{\widehat{\bm{V}}_{S,\cdot}})\bm{v}\Vert_{2}^{2}=\bm{v}^{\top}(\bm{I}_{|S|}-\bm{P}_{\widehat{\bm{V}}_{S,\cdot}})\bm{v}$.
The component $\widehat{\phi}$ estimates the variance of $(\bm{I}_{|S|}-\bm{P}_{\widehat{\bm{V}}_{S,\cdot}})\bm{v}$,
and needs the estimator $\widehat{\bm{\bm{\Sigma}}}_{\varepsilon}^{\tau}$
of noise covariance matrix as input.

Our focus is on the regime where the size of subset $S$ is small. This is economically meaningful because the linear relationship between observed factors and latent factors usually holds only for a short horizon. Based on our distributional results in \eqref{Chi-square stat final inequality}, the case when $S$ is the whole period, i.e., $|S|=T$, can be handled straightforwardly by applying the Gaussian approximation to the Chi-square distribution. Our assumption on the size $|S|$ can be summarized as
$r+\log n\ll|S|\ll T/(\kappa_{\varepsilon}^{2}r^{2}\log^{4}n)$. On
the one hand, we require that $|S|$ cannot be too small so that $\bm{F}_{S,\cdot}$
is of full column rank and its singular value is close to $\sqrt{|S|}$,
which is crucial for our proof. On other hand, we require that $|S|$
cannot be too large, otherwise the null distribution $\chi^{2}(|S|-r)$
with mean $|S|-r$ would diverge in such a case, complicating the
proof of proximity between $\chi^{2}(|S|-r)$ and the law of test
statistic $\mathfrak{\widehat{\mathcal{T}}}(S,\bm{v})$.

The assumption on SNR is that $\theta\gg\frac{n}{T}\sqrt{|S|}\kappa_{\varepsilon}^{4}r^{2}\log^{4}n$.
Consider the setting where $r\asymp1$, $N\asymp T$, and $\bm{\Sigma}_{\varepsilon}$
is well-conditioned (i.e., $\kappa_{\varepsilon}\asymp1$). Then,
the SNR condition becomes $\theta\gg\sqrt{|S|}\log^{4}n$, illustrating
the effect of the subset size $|S|$ on the growth rate required for
SNR $\theta$. The $\sqrt{|S|}$ factor appears in the SNR because we use the Chi-square distribution as the null distribution to accommodate the general case when $|S|$ is small. When $|S|\asymp1$, regardless of the log factors,
the SNR condition is equivalent to $\theta\gg1$, which match the
minimal condition required for the consistent estimations of the left
and right singular subspaces if there is no correlation in the noise
matrix \citep[e.g.,][]{yan2021inference}.

\citet{BaiNg2006factor} studied the factor specification test under
the whole period case, i.e., $S=\{1,2,\ldots,T\}$. They derived the
approximate cumulative distribution function (CDF) for their test
statistic, but they did not provide a finite sample error bound for
it and their results were under the strong factor model. In comparison,
our results give a precise characterization of the error and adapt
to the weak factor model. In our case, both the low SNR in the weak
factor model and the subset $S$ in the hypothesis $H_{0}$ make it
challenging to design the test statistics. The subspace perturbation
bounds in Theorem \ref{Thm UV 1st approx row-wise error} enable us
to conquer these difficulties and establish non-asymptotic analysis
of our test statistic for any arbitrary subset $S$.

\subsubsection{Estimation of the noise covariance matrix $\bm{\Sigma}_{\varepsilon}$\label{sec:noise cov estimation}}

We accommodate the case where the cross-sectional dimension $N$ is
much larger than the temporal dimension $T$. In such case, estimating
the high-dimensional covariance matrix $\bm{\Sigma}_{\varepsilon}$
becomes challenging yet is vital for numerous statistical tests, including
the factor specification test and the two-sample test for betas discussed
in the next section. Following the approaches in, e.g., \citet{bickel2008covariance,fan2011covFsee,POET2013},
we assume that the error covariance matrix $\bm{\Sigma}_{\varepsilon}$
is sparse in a suitable sense and use the adaptive thresholding method
to estimate $\bm{\Sigma}_{\varepsilon}$. It's important to note that
the assumed sparsity of $\bm{\Sigma}_{\varepsilon}$ is solely for
our statistical tests and is not a required condition to establish
our main theories on the subspace perturbation bounds in Theorem \ref{Thm UV 1st approx row-wise error}.

\begin{assumption} \label{Assump_noise_cov_sparse}For some $q\in[0,1)$,
there exists a constant $s(\bm{\Sigma}_{\varepsilon})>0$ such that
\[
\max_{1\leq i\leq N}\sum_{j=1}^{N}((\bm{\Sigma}_{\varepsilon})_{i,i}(\bm{\Sigma}_{\varepsilon})_{j,j})^{(1-q)/2}|(\bm{\Sigma}_{\varepsilon})_{i,j}|^{q}\leq s(\bm{\Sigma}_{\varepsilon}).
\]
\end{assumption}

The assumption is slightly weaker than those in, e.g., \citet{bickel2008covariance,POET2013},
where they assumed that $\lambda_{\min}(\bm{\Sigma}_{\varepsilon})\asymp1$,
$\lambda_{\max}(\bm{\Sigma}_{\varepsilon})\asymp1$, and $\max_{1\leq i\leq N}\sum_{j=1}^{N}|(\bm{\Sigma}_{\varepsilon})_{i,j}|^{q}\leq s_{0}$
for some sparsity parameter $s_{0}$. In particular, for $q=0$, it
constrains on the maximum number of nonzero elements of $\bm{\Sigma}_{\varepsilon}$.
This sparsity assumption for the noise covariance $\bm{\Sigma}_{\varepsilon}$
is also natural in economics and finance. As the latent common factors
largely explain the co-movements in the panel data, the correlation
among individual asset's idiosyncratic noises should be close to zero.
A specific example of the sparse structure of $\bm{\Sigma}_{\varepsilon}$
arises from the remaining sector effects \citep{gaglia2016panelData}.

We estimate the idiosyncratic noise matrix using $\widehat{\bm{E}}=\bm{X}-\widehat{\bm{B}}\widehat{\bm{F}}^{\top}$,
where $\widehat{\bm{B}}$ and $\widehat{\bm{F}}$ are the PCA estimators
of factor loadings and factors defined in Section \ref{sec:F B inference estimate}.
Then, the pilot estimator of the covariance matrix $\bm{\Sigma}_{\varepsilon}$
is given by $\widehat{\bm{\bm{\Sigma}}}_{\varepsilon}=T^{-1}\widehat{\bm{E}}\widehat{\bm{E}}^{\top}$.
We will show in Lemma \ref{Lemma noise cov matrix estimate} later
that $\max_{1\leq i,j\leq N}\bigl|(\widehat{\bm{\bm{\Sigma}}}_{\varepsilon}-\bm{\Sigma}_{\varepsilon})_{i,j}\bigl|\lesssim\epsilon_{N,T}\sqrt{(\bm{\Sigma}_{\varepsilon})_{i,i}(\bm{\Sigma}_{\varepsilon})_{j,j}}$,
where $\epsilon_{N,T}>0$ is negligible as $N,T\rightarrow\infty$.
Then we apply an adaptive thresholding method \citep{bickel2008covariance}
to obtain the sparse covariance matrix estimator $\widehat{\bm{\bm{\Sigma}}}_{\varepsilon}^{\tau}$
as follows, 
\[
(\widehat{\bm{\bm{\Sigma}}}_{\varepsilon}^{\tau})_{i,j}:=\Biggl\{\begin{array}{cc}
(\widehat{\bm{\bm{\Sigma}}}_{\varepsilon})_{i,j}, & \qquad\text{if}\text{ }i=j,\\
h((\widehat{\bm{\bm{\Sigma}}}_{\varepsilon})_{i,j},\tau_{i,j}), & \qquad\text{if}\text{ }i\neq j,
\end{array}
\]
where $h(z,\tau)$ is a thresholding function with the threshold value
$\tau$. Here, the threshold value $\tau_{i,j}$ is set adaptively
to $\tau_{i,j}=C\epsilon_{N,T}\sqrt{(\widehat{\bm{\bm{\Sigma}}}_{\varepsilon})_{i,i}(\widehat{\bm{\bm{\Sigma}}}_{\varepsilon})_{j,j}}$
with some large constant $C>0$. The idea behind the adaptive thresholding
is examining the sample correlation matrix and retaining entries exceeding
$C\epsilon_{N,T}$ in magnitude. The specific examples of the function
$h$ include the hard-thresholding function $h(z,\tau)=z1_{\{|z|<\tau\}}$,
among other common thresholding functions like the soft-thresholding
function and SCAD \citep{FanLi2001SCAD}. In general, we require that
the thresholding function $h(z,\tau)$ to satisfy: (i) $h(z,\tau)=0\text{ if }|z|<\tau$;
(ii) $|h(z,\tau)-z|<\tau$. The estimation error is given as follows.

\begin{lemma}\label{Lemma noise cov matrix estimate}Suppose that
all the assumptions in Corollary \ref{corollary:error bound B F}
hold, and that there exists a sufficiently small constant $\epsilon_{N,T}>0$
such that 
\begin{equation}
\max_{1\leq i,j\leq N}\bigl(\frac{1}{\sqrt{T}}(1+(1+\gamma_{\varepsilon,j})\gamma_{i}+(1+\gamma_{\varepsilon,i})\gamma_{j})+(\beta_{i}\gamma_{j}+\beta_{j}\gamma_{i})(\frac{n}{\theta^{2}T}+\frac{1}{\theta\sqrt{T}})\bigr)\cdot r\log^{2}n\leq\epsilon_{N,T},\label{noise cov error condition}
\end{equation}
where 
\[
\beta_{i}:=\frac{\bigl\Vert\bm{b}_{i}\bigl\Vert_{2}}{\sqrt{(\bm{\Sigma}_{\varepsilon})_{i,i}}},\qquad\gamma_{i}:=\frac{\sigma_{r}\bigl\Vert\bar{\bm{U}}_{i,\cdot}\bigl\Vert_{2}}{\sqrt{(\bm{\Sigma}_{\varepsilon})_{i,i}}}\qquad\text{and}\qquad\gamma_{\varepsilon,i}:=\frac{\bigl\Vert(\bm{\Sigma}_{\varepsilon}^{1/2})_{i,\cdot}\bigl\Vert_{1}}{\bigl\Vert(\bm{\Sigma}_{\varepsilon}^{1/2})_{i,\cdot}\bigl\Vert_{2}}.
\]
Then, 
we have that, with probability at least $1-O(n^{-10})$, it holds
$\max_{1\leq i,j\leq N}((\bm{\Sigma}_{\varepsilon})_{i,i}(\bm{\Sigma}_{\varepsilon})_{j,j})^{-1/2}\bigl|(\widehat{\bm{\bm{\Sigma}}}_{\varepsilon}-\bm{\Sigma}_{\varepsilon})_{i,j}\bigl|\leq C_{0}\epsilon_{N,T}$,
for a universal constant $C_{0}>0$. Further, under Assumption \ref{Assump_noise_cov_sparse},
the generalized thresholding estimator $\widehat{\bm{\bm{\Sigma}}}_{\varepsilon}^{\tau}$
satisfies that, there exists some universal constant $C_{0}>0$ such
that, with probability at least $1-O(n^{-10})$, 
\[
\big\Vert\widehat{\bm{\bm{\Sigma}}}_{\varepsilon}^{\tau}-\bm{\Sigma}_{\varepsilon}\big\Vert_{2}\leq C_{0}(\epsilon_{N,T})^{1-q}s(\bm{\Sigma}_{\varepsilon}).
\]

\end{lemma}

The above error bounds for the estimators $\widehat{\bm{\bm{\Sigma}}}_{\varepsilon}$
and $\widehat{\bm{\bm{\Sigma}}}_{\varepsilon}^{\tau}$ differ from
the existing literature owing to our row-wise subspace perturbation
bounds in Theorem \ref{Thm UV 1st approx row-wise error}. Let us
interpret the condition (\ref{noise cov error condition}) under some
specific settings. Consider the setting where $r\asymp1$, $N\asymp T$,
$\bar{\bm{U}}$ is $\mu$-incoherent, 
$\Vert\bm{B}\Vert_{2,\infty}\lesssim1$ 
and $\bm{\Sigma}_{\varepsilon}$ is well-behaved such that $\gamma_{\varepsilon,i}\asymp1$,
$(\bm{\Sigma}_{\varepsilon})_{i,i}\asymp1$, and $\Vert\bm{\Sigma}_{\varepsilon}\Vert_{2}\asymp1$.
Then, the condition (\ref{noise cov error condition}) is equivalent
to $T\geq(\sigma_{r}\log^{2}n)/\epsilon_{N,T}$, which requires the
sample size $T$ is sufficiently large. According to Lemma \ref{Lemma noise cov matrix estimate},
fulfilling our assumption on the estimation error $\Vert\widehat{\bm{\Sigma}}_{\varepsilon}^{\tau}-\bm{\Sigma}_{\varepsilon}\Vert_{2}$
required for Theorem \ref{Thm factor test plug-in Chi-sq} merely
requires that the estimation error $\epsilon_{N,T}$ satisfies 
\[
\epsilon_{N,T}\ll\bigl((\left\Vert \bm{\Sigma}_{\varepsilon}\right\Vert _{2}/s(\bm{\Sigma}_{\varepsilon}))|S|^{-1/2}\kappa_{\varepsilon}^{-4}r^{-1}\log^{-2}n\bigr)^{1/(1-q)}.
\]
In subsequent applications, we will directly use $\epsilon_{N,T}$,
instead of $\widehat{\bm{\bm{\Sigma}}}_{\varepsilon}^{\tau}$, to
express assumptions on the estimation error of the noise covariance
matrix.

\subsection{Test for structural breaks in betas}

The discussion so far has studied the factor models with constant
betas. However, in many cases, there is a possibility of time-variant
betas, and it is important to test for structural changes in betas
\citep[e.g.,][]{StockWatson2009beta_break,breitung2011beta_break,chen2014beta_break,han2015beta_break}.
We test for two time periods $\Gamma_{1}=\{t_{1}+1,t_{1}+2,\ldots,t_{1}+T_{1}\}$
and $\Gamma_{2}=\{t_{2}+1,t_{2}+2,\ldots,t_{2}+T_{2}\}$ that are
possibly not consecutive, i.e., $t_{1}+T_{1}\leq t_{2}+1$.  The factor
models in these two periods are subject to a structural break in beta:  for each given individual unit $i$,
\begin{align*}
x_{i,t} & =\bm{f}_{t}^{\top}\bm{b}_{i}^{1}+\varepsilon_{i,t}\qquad\text{for}\qquad t\in\Gamma_{1},\\
x_{i,t} & =\bm{f}_{t}^{\top}\bm{b}_{i}^{2}+\varepsilon_{i,t}\qquad\text{for}\qquad t\in\Gamma_{2},
\end{align*}
where $i$ is the index of the cross-sectional unit of interest. Here,
we assume that the factors are the same, but the betas can be different.
 For example, we might ask if a external shock such as the 2008 financial crises has caused the changes in factor loadings before and after the shock with a transition period $(T_1, T_2)$.

The data consist of two panels, $\bm{X}^{1}$ and $\bm{X}^{2}$, during $\Gamma_{1}$ and $\Gamma_{2}$, respectively, where $\bm{X}^{1}\in\mathbb{R}^{N\times T_{1}}$
and $\bm{X}^{2}\in\mathbb{R}^{N\times T_{2}}$.
We test the hypothesis as follows:
\[
H_{0}:\ \bm{b}_{i}^{1}=\bm{b}_{i}^{2}\ \leftrightarrow\ H_{1}:\ \bm{b}_{i}^{1}\neq\bm{b}_{i}^{2}.
\]
Different from our approach, \citet{chen2014beta_break,han2015beta_break}
tested for the entire factor loading, with their null hypothesis being
$H_{0}:\ \bm{b}_{i}^{1}=\bm{b}_{i}^{2}\text{ for }\forall i$. \citet{StockWatson2009beta_break,breitung2011beta_break}
also studied the test for a given $i$ like ours, while \citet{StockWatson2009beta_break}
focused on empirical studies and \citet{breitung2011beta_break} developed
and established theories for the test statistics. Our work differs
in that, as we will show below, the effectiveness of our test statistic
does not require the pervasiveness assumption needed in \citet{breitung2011beta_break}.
In addition, our formulation allows the gap between the two periods
$\Gamma_{1}$ and $\Gamma_{2}$, which can model the special transition periods
we purposely, e.g., financial crisis, enabling us to test for
structural breaks before and after these special periods.

To construct the test statistic, we first merge the two panels $\bm{X}^{1}$
and $\bm{X}^{2}$ into a $N\times(T_{1}+T_{2})$-dimensional data
matrix $\bm{X}=(\bm{X}^{1},\bm{X}^{2})$. Next, we apply SVD as described
in Algorithm \ref{alg:PCA-SVD} to obtain an estimator of the factors:
$\widehat{\bm{F}}=\sqrt{T}\widehat{\bm{V}}$, where $T=T_{1}+T_{2}$
and $\widehat{\bm{U}}\widehat{\bm{\Sigma}}\widehat{\bm{V}}^{\top}$
is the truncated rank-$r$ SVD of $T^{-1/2}\bm{X}$ with $\widehat{\bm{U}}\in\mathbb{R}^{N\times r}$
and $\widehat{\bm{V}}\in\mathbb{R}^{T\times r}$. Then, we split the
estimated factors into $\widehat{\bm{F}}=((\widehat{\bm{F}}^{1})^{\top},(\widehat{\bm{F}}^{2})^{\top})^{\top}$
according to two time periods $\Gamma_{1}$ and $\Gamma_{2}$, where
$\widehat{\bm{F}}^{j}\in\mathbb{R}^{T_{j}\times r}$. Subsequently,
we obtain the estimator for $\bm{b}_{i}^{j}$ by regressing $\bm{X}_{i,\cdot}^{j}$
on $\widehat{\bm{F}}^{j}$:
\[
\widehat{\bm{b}}_{i}^{j}=((\widehat{\bm{F}}^{j})^{\top}\widehat{\bm{F}}^{j})^{-1}(\widehat{\bm{F}}^{j})^{\top}(\bm{X}_{i,\cdot}^{j})^{\top}\qquad\text{for}\qquad j=1,2.
\]
Finally, under null hypothesis $H_{0}$, we show that $\widehat{\bm{b}}_{i}^{1}-\widehat{\bm{b}}_{i}^{2}$
is approximately Gaussian via the first-order approximation results
in Theorem \ref{Thm UV 1st approx row-wise error}, and construct
the test statistic using the plug-in estimator of the covariance matrix.

\begin{theorem}\label{Thm beta structure test}Suppose that the assumptions
in Theorem \ref{Thm UV 1st approx row-wise error} hold, and the covariance
estimator $\widehat{\bm{\bm{\Sigma}}}_{\varepsilon}^{\tau}$ satisfies
the conditions as in Lemma \ref{Lemma noise cov matrix estimate}.
For any given target error level $\delta>0$, assume there exists
some universal constants $c_{0},C_{0},C_{U}>0$ such that $\min(T_{1},T_{2})\geq C_{0}(r+\log n)$,
$\Vert\bm{U}_{i,\cdot}\Vert_{2}\leq C_{U}\sqrt{T/(T_{1}T_{2})}$,
$T\geq C_{0}^{2}C_{U}^{2}\delta^{-2}\kappa_{\varepsilon}r\log^{3}n$,
$\vartheta_{i}\geq C_{0}C_{U}^{2}\delta^{-1}T^{-1/2}\kappa_{\varepsilon}^{2}r\log n$,
\[
\theta\geq C_{0}\{C_{U}^{2}\frac{1}{\delta}(\kappa_{\varepsilon}\frac{n}{T}+\kappa_{\varepsilon}^{2}\sqrt{\frac{n}{T}}\Vert\bm{U}_{i,\cdot}\Vert_{2})+\sqrt{T_{1}T_{2}/T}\frac{1}{\delta}\kappa_{\varepsilon}\}r\log^{3/2}n,
\]
and $C_{0}r^{3/2}\Vert\bm{w}\Vert_{\infty}/\Vert\bm{w}\Vert_{2}\leq\delta$,
where $\bm{w}$ is a vector defined as $\bm{w}:=[\bm{\Sigma}_{\varepsilon}^{1/2}(\bm{I}_{N}+\bm{U}\bm{U}^{\top})]_{\cdot,i}$;
for the parameter $\epsilon_{N,T}$ (cf. Lemma \ref{Lemma noise cov matrix estimate})
that captures the estimation error of the noise covariance matrix
$\bm{\Sigma}_{\varepsilon}$, assume that $\epsilon_{N,T}\leq c_{0}\delta/(\kappa_{\varepsilon}^{2}r\log n)$
and $\epsilon_{N,T}\leq c_{0}[\delta\Vert\bm{\Sigma}_{\varepsilon}\Vert_{2}/(s(\bm{\Sigma}_{\varepsilon})\Vert\bm{U}_{i,\cdot}\Vert_{2}\kappa_{\varepsilon}^{2}r\log n)]^{1/(1-q)}$.
Then, under the null hypothesis $H_{0}:\ \bm{b}_{i}^{1}=\bm{b}_{i}^{2}$,
we have that
\[
\left|\mathbb{P}\left(\mathscr{\widehat{B}}_{i}\leq\chi_{1-\alpha}^{2}(r)\right)-(1-\alpha)\right|\leq\delta,
\]
where the test statistic $\mathscr{\widehat{B}}_{i}$ is given by
\[
\mathscr{\widehat{B}}_{i}:=\frac{1}{T\widehat{\varphi}_{i}}(\widehat{\bm{b}}_{i}^{1}-\widehat{\bm{b}}_{i}^{2})^{\top}[\prod_{j=1}^{2}(\widehat{\bm{F}}^{j})^{\top}\widehat{\bm{F}}^{j}](\widehat{\bm{b}}_{i}^{1}-\widehat{\bm{b}}_{i}^{2}),
\]
with $\widehat{\varphi}_{i}:=(\widehat{\bm{\Sigma}}_{\varepsilon}^{\tau}+\widehat{\bm{U}}\widehat{\bm{U}}^{\top}\widehat{\bm{\Sigma}}_{\varepsilon}^{\tau}+\widehat{\bm{\Sigma}}_{\varepsilon}^{\tau}\widehat{\bm{U}}\widehat{\bm{U}}^{\top}+\widehat{\bm{U}}\widehat{\bm{U}}^{\top}\widehat{\bm{\Sigma}}_{\varepsilon}^{\tau}\widehat{\bm{U}}\widehat{\bm{U}}^{\top})_{i,i}$.

\end{theorem}\begin{remark}If all the entries of the noise matrix
$\bm{Z}$ in (\ref{noise matrix formula}) are Gaussian, then we do
not need the assumption that $C_{0}r^{3/2}\Vert\bm{w}\Vert_{\infty}/\Vert\bm{w}\Vert_{2}\leq\delta$,
where $\bm{w}=[\bm{\Sigma}_{\varepsilon}^{1/2}(\bm{I}_{N}+\bm{U}\bm{U}^{\top})]_{\cdot,i}$.\end{remark}

To illustrate our assumptions, we consider the following setting similar
to that discussed after Theorem \ref{Thm UV 1st approx row-wise error}:~$r\asymp1$,
$T_{1}\asymp T\asymp N$, $T_{1}\gg T_{2}$, $\bm{\Sigma}_{\varepsilon}$
is well-conditioned and its sparsity parameter satisfies $s(\bm{\Sigma}_{\varepsilon})\lesssim\Vert\bm{\Sigma}_{\varepsilon}\Vert_{2}$,
and the column subspace of $\bm{B}$ is incoherent. Then, fulfilling
our assumptions merely requires that, $\theta\gg\log^{3/2}n$ and
$\vartheta_{i}\gg T^{-1/2}\log n$ for the SNR, $T\gg\log^{3}n$ for
the sample size, and $\epsilon_{N,T}\ll\min(\log^{-1}n,(N^{1/2}/\log n)^{1/(1-q)})$
for the estimation error of $\bm{\Sigma}_{\varepsilon}$. In particular,
our SNR conditions are less restrictive than prior work \citep[e.g.,][]{breitung2011beta_break}
that needed the pervasiveness assumption, which required the SNR to
grow as $\theta\gg N^{1/4}$. Our results adapt to the weak factor
model with the cross-sectional correlations, and to the best of our
knowledge, no prior work has developed the test statistics for the
structural break test under this case.

\subsection{The two-sample test for betas}

In finance and economics, besides the latent common factors that drive
the co-movements of asset returns, the factor loadings, also known
as betas, are important as well, which measure the sensitivity of
asset return to the movements of the factors. Consider the example
that the panel data $\bm{X}$ is the stock return, where the $i$-th
row of $\bm{X}$ represents the time series data of the $i$-th stock.
Then, the factor loadings $\bm{B}$ assess the risk exposure of these
stocks to the latent common factors $\bm{f}_{t}$, with the $i$-th
row $\bm{b}_{i}$ being the $i$-th stock's beta.

For any distinct $i$ and $j$, we aim to test the hypothesis that
$\bm{b}_{i}=\bm{b}_{j}$, i.e., if the $i$-th and the $j$-th stocks
have the same risk exposure on the common factors. It is a statistical
approach to evaluate the similarity in risk structure between two
stocks. Our test statistic is similar to the idea of two-sample test:
we show that $\widehat{\bm{B}}_{i,\cdot}-\widehat{\bm{B}}_{j,\cdot}$
is approximately Gaussian and then derive a Chi-square test statistic.
In Theorem \ref{Thm factor test plug-in Chi-sq} below, we construct
the test statistic and show its validity.

\begin{theorem}\label{Thm two-sample test B}Suppose that the assumptions
in Theorem \ref{Thm UV 1st approx row-wise error} hold. Assume that
$s(\bm{\Sigma}_{\varepsilon})/\lambda_{\min}(\bm{\Sigma}_{\varepsilon})\leq C_{\varepsilon}$
for some universal constant $C_{\varepsilon}$, and the covariance
estimator $\widehat{\bm{\bm{\Sigma}}}_{\varepsilon}^{\tau}$ satisfies
the conditions as in Lemma \ref{Lemma noise cov matrix estimate}.
For any given target error level $\delta>0$, assume there exists
some universal constants $c_{0},C_{0}>0$ such that $T\geq C_{0}r^{2}\delta^{-2}\log n$,
\begin{subequations}
\begin{equation}
\max(\Vert\bar{\bm{U}}_{i,\cdot}\Vert_{2},\Vert\bar{\bm{U}}_{j,\cdot}\Vert_{2})\leq c_{0}\delta(\frac{n}{\theta\sqrt{T}}+1)^{-1}\kappa_{\varepsilon}^{-1/2}r^{1/2}\log^{-1}n,\label{B two-sample test SNR 1}
\end{equation}
and 
\begin{equation}
\min(\vartheta_{i},\vartheta_{j})\geq C_{0}\delta^{-1}\sqrt{\kappa_{\varepsilon}}\frac{\theta n}{\vartheta\sqrt{T}}(\big\Vert\bar{\bm{U}}\big\Vert_{2,\infty}+\frac{1}{\sqrt{n}})\sqrt{r}\log^{3/2}n;\label{B two-sample test SNR 2}
\end{equation}
\end{subequations}
 for the parameter $\epsilon_{N,T}$ (cf. Lemma \ref{Lemma noise cov matrix estimate})
that captures the estimation error of the noise covariance matrix
$\bm{\Sigma}_{\varepsilon}$, assume that $\epsilon_{N,T}\leq c_{0}(\delta r^{3}\log^{2}n)^{1/(1-q)}$.
Then, for any $i,j\in\{1,2,\ldots,N\}$ satisfying $i\neq j$, under
the null hypothesis $H_{0}:\ \bm{b}_{i}=\bm{b}_{j}$, we have that
\[
\left|\mathbb{P}\left(\mathfrak{\widehat{\mathcal{B}}}_{ij}\leq\chi_{1-\alpha}^{2}(r)\right)-(1-\alpha)\right|\leq\delta,
\]
where the test statistic $\mathfrak{\widehat{\mathcal{B}}}_{ij}$
is given by 
\[
\mathfrak{\widehat{\mathcal{B}}}_{ij}:=T((\widehat{\bm{\bm{\Sigma}}}_{\varepsilon}^{\tau})_{i,i}+(\widehat{\bm{\bm{\Sigma}}}_{\varepsilon}^{\tau})_{j,j}-2(\widehat{\bm{\bm{\Sigma}}}_{\varepsilon}^{\tau})_{i,j})^{-1}\big\Vert\widehat{\bm{B}}_{i,\cdot}-\widehat{\bm{B}}_{j,\cdot}\big\Vert_{2}^{2}.
\]

\end{theorem}

\begin{remark}If all the entries of the noise matrix $\bm{Z}$ in
(\ref{noise matrix formula}) are Gaussian, then we do not need the
assumption that $T\geq C_{0}r^{2}\delta^{-2}\log n$ on the sample
size $T$.\end{remark}

To illustrate the SNR conditions (\ref{B two-sample test SNR 1})--(\ref{B two-sample test SNR 2}),
we consider the setting discussed after Theorem \ref{Thm UV 1st approx row-wise error},
i.e.,~$r\asymp1$, $N\asymp T$, $\bm{\Sigma}_{\varepsilon}$ is
well-conditioned, and the column subspace of $\bm{B}$ is incoherent.
Then, fulfilling our assumptions (\ref{B two-sample test SNR 1})--(\ref{B two-sample test SNR 2})
merely requires that $\theta\gg\log n$ and $\theta^{-1}\vartheta\min(\vartheta_{i},\vartheta_{j})\gg\log^{3/2}n$.
These SNR conditions are less restrictive than the prior work \citep[e.g.,][]{BaiNg2023PCA,Jiang2023PCA,ChoiMing2024PCA}
that required the SNR to grow with a polynomial rate of $n$.

To the best of our knowledge, no prior work has developed the test
statistics for this two-sample test of betas. Our results are applicable
to the weak factor model with the cross-sectional correlations. The
subspace perturbation bounds in Theorem \ref{Thm UV 1st approx row-wise error}
pave the way for us to statistically assess the similarity for arbitrary
two rows of factor loadings $\bm{B}$.

\subsection{Statistical inference for the systematic risks}

In the factor model $x_{i,t}=\bm{b}_{i}^{\top}\bm{f}_{t}+\varepsilon_{i,t}$,
the risk associated with stock return $x_{i,t}$ is decomposed into
two parts \citep[e.g.,][]{Bai2003ECTA} --- systematic risk from
the common component $\bm{b}_{i}^{\top}\bm{f}_{t}$ and idiosyncratic
risk from the noise $\varepsilon_{i,t}$. Systematic risk, often referred
to as market risk, is integral in financial economics as it represents
the inherent risk affecting the entire market or market segment. This
risk, driven by broader economic forces such as inflation, political
events, and changes in interest rates, comes from the risk factors
$\bm{f}_{t}$ that explain the systematic co-movements and impacts
all the stocks.

A standard metric of the systematic risk is the variance of $\bm{b}_{i}^{\top}\bm{f}_{t}$,
which is given by $\Var(\bm{b}_{i}^{\top}\bm{f}_{t})=
\Vert\bm{b}_{i}\Vert_{2}^{2}$. Our focus is on constructing the confidence interval (CI) for systematic
risk $\Vert\bm{b}_{i}\Vert_{2}^{2}$. Similar to the idea we conduct
the inference for beta $\bm{b}_{i}$ in Corollary \ref{Thm B inference}
where the estimator of $\bm{b}_{i}$ is $\widehat{\bm{B}}_{i,\cdot}$,
we show that $\Vert\widehat{\bm{B}}_{i,\cdot}\Vert_{2}^{2}-\Vert\bm{b}_{i}\Vert_{2}^{2}$
is approximately Gaussian and then construct the CI for $\Vert\bm{b}_{i}\Vert_{2}^{2}$.
We present the CI and its validity in Theorem \ref{Thm B row norm}.

\begin{theorem}\label{Thm B row norm}Suppose that the assumptions
in Theorem \ref{Thm UV 1st approx row-wise error} hold. Assume that
covariance estimator $\widehat{\bm{\bm{\Sigma}}}_{\varepsilon}^{\tau}$
satisfies the conditions as in Lemma \ref{Lemma noise cov matrix estimate}.
For any given target error level $\delta>0$, assume there exist some
universal constants $c_{0},C_{0},C_{U}>0$ such that $T\geq C_{0}\delta^{-2}\log n$,
$\Vert\bar{\bm{U}}\Vert_{2,\infty}\leq C_{U}\theta/\vartheta$, $\Vert\bar{\bm{U}}\Vert_{2,\infty}\leq C_{U}\sqrt{T}\vartheta/n$,
$\Vert\bar{\bm{U}}\Vert_{2,\infty}\leq C_{U}\sqrt{T}\theta^{2}/(\vartheta n)$,
$\Vert\bm{b}_{i}\Vert_{2}^{2}/(\bm{\Sigma}_{\varepsilon})_{i,i}\leq c_{0}\delta^{2}r^{-1}\log^{-1}n$,
\begin{subequations}
\begin{equation}
\vartheta_{i}\geq C_{0}\delta^{-1}\frac{\sqrt{\kappa_{\varepsilon}}\theta}{\vartheta_{i}\sqrt{T}}\frac{1}{\Vert\bar{\bm{U}}_{i,\cdot}\Vert_{2}}r\log^{2}n,\qquad\vartheta\geq C_{0}\delta^{-1}\frac{\sqrt{\kappa_{\varepsilon}}\theta}{\vartheta_{i}\sqrt{T}}(\sqrt{n}+n\Vert\bar{\bm{U}}\Vert_{2,\infty})r\log^{2}n\label{B rownorm CI SNR 1}
\end{equation}
\begin{equation}
\vartheta_{i}\geq C_{0}\delta^{-1}\sqrt{\kappa_{\varepsilon}}(\frac{1}{\sqrt{T}}+\frac{n}{\theta T})(1+\frac{1}{\Vert\bar{\bm{U}}_{i,\cdot}\Vert_{2}})r\log n,\qquad\text{and}\qquad\Vert\bar{\bm{U}}_{i,\cdot}\Vert_{2}\leq c_{0}\delta\kappa_{\varepsilon}^{-1/2}(1+\frac{n}{\theta\sqrt{T}})^{-1};\label{B rownorm CI SNR 2}
\end{equation}
\end{subequations}
 for the parameter $\epsilon_{N,T}$ (cf. Lemma \ref{Lemma noise cov matrix estimate})
that captures the estimation error of the noise covariance matrix
$\bm{\Sigma}_{\varepsilon}$, assume that $\epsilon_{N,T}\leq c_{0}\delta r^{-1}\log^{-1/2}n$.
Then we have that, for any $1\leq i\leq N$, 
\[
\left\vert \mathbb{P}(\Vert\bm{b}_{i}\Vert_{2}^{2}\in\text{CI}_{i}^{\bm{B},(1-\alpha)})-(1-\alpha)\right\vert \leq\delta,
\]
where the confidence interval is constructed as 
\[
\text{CI}_{i}^{\bm{B},(1-\alpha)}:=[\big\Vert\widehat{\bm{B}}_{i,\cdot}\big\Vert_{2}^{2}-\widehat{\sigma}_{B,i}z_{1-\frac{1}{2}\alpha},\big\Vert\widehat{\bm{B}}_{i,\cdot}\big\Vert_{2}^{2}+\widehat{\sigma}_{B,i}z_{1-\frac{1}{2}\alpha}],\qquad\text{with}\qquad\widehat{\sigma}_{B,i}:=\frac{2}{\sqrt{T}}\sqrt{(\widehat{\bm{\bm{\Sigma}}}_{\varepsilon}^{\tau})_{i,i}}\big\Vert\widehat{\bm{B}}_{i,\cdot}\big\Vert_{2},
\]
and $z_{p}$ is the $p$-quantile of the standard Gaussian $\mathcal{N}(0,1)$.

\end{theorem}

\begin{remark}If all the entries of the noise matrix $\bm{Z}$ in
(\ref{noise matrix formula}) are Gaussian, then we do not need the
assumption that $T\geq C_{0}\delta^{-2}\log n$ on the sample size
$T$.\end{remark}

We now explain the assumptions under the setting discussed after Theorem
\ref{Thm two-sample test B}. In this case, to make (\ref{B rownorm CI SNR 1})
hold, it suffices to assume that $\theta^{-1}\vartheta_{i}\vartheta\gg\log^{2}n$
and $\theta^{-1}\vartheta_{i}^{2}\gg\log^{2}n$; to make (\ref{B rownorm CI SNR 2})
hold, it suffices to assume that $\theta\gg1$, $\vartheta_{i}\theta\gg\log n$,
and $\vartheta_{i}\gg T^{-1/2}\log n$. All these conditions require
only polynomial growth rates of $\log n$ for the SNRs $\vartheta_{i}$,
$\vartheta$, and $\theta$. The assumption on $\Vert\bar{\bm{U}}\Vert_{2,\infty}$
is the same with that we required in Corollary \ref{corollary:error bound B F}
to prove the estimation error for factor loading.

On one hand, as shown in Theorem \ref{Thm B row norm}, the CI width
$\widehat{\sigma}_{B,i}$ is proportional to the square root of the
product of noise level estimator $(\widehat{\bm{\bm{\Sigma}}}_{\varepsilon}^{\tau})_{i,i}$
and systematic risk estimator $\Vert\widehat{\bm{B}}_{i,\cdot}\Vert_{2}^{2}$.
On the other hand, in the proof of Theorem \ref{Thm B row norm},
we show that the bias of systematic risk estimator $\Vert\widehat{\bm{B}}_{i,\cdot}\Vert_{2}^{2}$
is proportional to $\Vert\bm{b}_{i}\Vert_{2}^{2}$. Thus, the assumption
that $\Vert\bm{b}_{i}\Vert_{2}^{2}/(\bm{\Sigma}_{\varepsilon})_{i,i}\leq c_{0}\delta^{2}r^{-1}\log^{-1}n$
is essential because the validity of inference hinges on the dominance
of the CI width over the bias. Consequently, we assume that the ratio
between systematic risk $\Vert\bm{b}_{i}\Vert_{2}^{2}$ and noise
level $(\bm{\Sigma}_{\varepsilon})_{i,i}$ cannot be too large. The
row-wise subspace perturbation bounds in Theorem \ref{Thm UV 1st approx row-wise error}
facilitate us to conduct statistical inference for the systematic
risk of any given stock in the panel data. 

%% file: main_numerical_experiments.tex
\section{Numerical experiments\label{sec:Numerical-experiments}}

In this section, we conduct Monte Carlo simulations to demonstrate
our inferential theories for the PC estimators in the weak factor
models. Additionally, our empirical studies reveal that the testing
results based on our test statistics surprisingly align with the economic
cycles and financial crisis periods.

\subsection{Monte Carlo simulations}

To make our simulations similar to the real applications, we use the
standard Fama-French three-factor model: 
\[
x_{i,t}=\bm{b}_{i}^{\top}\bm{f}_{t}+\varepsilon_{i,t},\qquad1\leq i\leq N,1\leq t\leq T,
\]
where the dimension of the latent factor $\bm{f}_{t}$ 
is set to $r=3$. The idiosyncratic noise $\varepsilon_{i,t}$ exhibits
cross-sectional correlations, and the noise covariance matrix $\bm{\Sigma}_{\varepsilon}$
is sparse.

We generate the factor loadings $\{\bm{b}_{i}\}_{i=1}^{N}$, the factors
$\{\bm{f}_{t}\}_{t=1}^{T}$, and the noise terms $\{\bm{\varepsilon}_{t}\}_{t=1}^{T}$
independently from $\mathcal{N}(0,\bm{\Sigma}_{b})$, $\mathcal{N}(0,s_{f}\bm{\Sigma}_{f})$,
and $\mathcal{N}(0,\bm{\Sigma}_{\varepsilon})$, respectively, where
$\bm{\varepsilon}_{t}=(\varepsilon_{1,t},\ldots,\varepsilon_{N,t})^{\top}$.
Both $\bm{\Sigma}_{b}$ and $\bm{\Sigma}_{f}$ are set to $r\times r$
identity matrices. We generate $\bm{\Sigma}_{\varepsilon}$ as a block-diagonal
matrix $\bm{\Sigma}_{\varepsilon}=\text{\ensuremath{\mathsf{diag}}}(\bm{A}_{1},\bm{A}_{2},\ldots,\bm{A}_{J})$,
where the number of blocks is set to $J=20$. We set $N=300$ and
$T=200$. For each $\bm{A}_{i}$, we construct it as an equi-correlation
matrix $\bm{A}_{i}=(1-\rho_{i})\bm{I}_{m}+\rho_{i}\bm{1}_{m}\bm{1}_{m}^{\top}$,
where the block size $m$ is set to $m=N/J=15$, $\bm{1}_{m}$ is
the $m$-dimensional vector of ones, and $\rho_{i}$ is drawn from
a uniform distribution on $[0,0.5]$. To validate our test statistic
for two-sample test of betas, we set $\bm{b}_{2}$ equal to $\bm{b}_{1}$
after generating the loading matrix $\bm{B}$. This slight modification
allows us to examine our test statistics under the null hypothesis
$\bm{b}_{1}=\bm{b}_{2}$. In our simulation results, we report the
values of $\theta=\sigma_{r}/\Vert\bm{\Sigma}_{\varepsilon}^{1/2}\Vert_{2}$
to reflect different levels of SNR.

First, we demonstrate the practical validity of the confidence regions
constructed using Corollaries \ref{Thm F inference}--\ref{Thm B inference}
and Theorem \ref{Thm B row norm}, for the factors, betas (i.e., factor
loadings), and systematic risks, respectively.
\begin{itemize}
\item For factors (resp. betas), we construct 95\% confidence regions by
substituting the asymptotic covariance matrix $\bm{\Sigma}_{V,t}$
(resp. $\bm{\Sigma}_{B,i}$) with their consistent estimators as commented
after Corollary \ref{Thm B inference}. We define $\widehat{\mathsf{Cov}}_{F}(t)$
(resp. $\widehat{\mathsf{Cov}}_{B}(i)$) as the empirical probability
that the constructed confidence region covers $\bm{f}_{t}^{\top}\bm{R}_{F}$
(resp. $\bm{b}_{i}^{\top}\bm{R}_{B}$) over 200 Monte Carlo trials,
where $\bm{R}_{F}$ (resp. $\bm{R}_{B}$) is the rotation matrix defined
in Corollary \ref{Thm F inference} (resp. Corollary \ref{Thm B inference}). 
\item For systematic risks, similarly, we define $\widehat{\mathsf{Cov}}_{\ell_{2}}(i)$
as the empirical probability that the 95\% confidence interval $\text{CI}_{i}^{\bm{b},0.95}$
constructed via Theorem \ref{Thm B row norm} covers $\Vert\bm{b}_{i}\Vert_{2}^{2}$
over 200 Monte Carlo trials.
\end{itemize}
Finally, we compute the mean and standard deviation for $\{\widehat{\mathsf{Cov}}_{F}(t)\}_{t=1}^{T}$,
$\{\widehat{\mathsf{Cov}}_{B}(i)\}_{i=1}^{N}$, and $\{\widehat{\mathsf{Cov}}_{\ell_{2}}(i)\}_{i=1}^{N}$,
then present these as $\mathsf{Mean}(\widehat{\mathsf{Cov}})$ and
$\mathsf{Std}(\widehat{\mathsf{Cov}})$, with the results reported
in Table \ref{table:CI cover}. As indicated in Table \ref{table:CI cover},
the coverage probabilities are close to 0.95 according to the mean
value, and exhibit stability across different rows of $\bm{F}$ and
$\bm{B}$ as evidenced by low standard deviation. As SNR $\theta$
goes down, the slight slippage of the coverage probabilities of the
systematic risks reconciles with our comments after Theorem \ref{Thm B row norm}
that, $\Vert\bm{b}_{i}\Vert_{2}^{2}/(\bm{\Sigma}_{\varepsilon})_{i,i}$,
which is close to SNR, reflects the ratio between the bias and the
CI width and cannot be too large. These favorable numerical results
persist even under low SNR, supporting our inferential theory under
the weak factor model.

\begin{table}[H]
\caption{Empirical coverage rates of 95\%-CI\label{table:CI cover}}

\centering

\begin{tabular}{c|c|c|c|c|c|c}
\hline 
$(N,T)=(300,200)$ & \multicolumn{2}{c|}{Factor} & \multicolumn{2}{c|}{Beta} & \multicolumn{2}{c}{Systematic Risk}\tabularnewline
\hline 
SNR $\theta$  & $\mathsf{Mean}(\widehat{\mathsf{Cov}})$  & $\mathsf{Std}(\widehat{\mathsf{Cov}})$  & $\mathsf{Mean}(\widehat{\mathsf{Cov}})$  & $\mathsf{Std}(\widehat{\mathsf{Cov}})$  & $\mathsf{Mean}(\widehat{\mathsf{Cov}})$  & $\mathsf{Std}(\widehat{\mathsf{Cov}})$\tabularnewline
\hline 
$4.5$  & $0.9383$  & $0.0172$  & $0.9325$  & $0.0171$  & $0.9071$  & $0.0400$\tabularnewline
\hline 
$3.5$  & $0.9298$  & $0.0190$  & $0.9264$  & $0.0184$  & $0.9192$  & $0.0323$\tabularnewline
\hline 
$2.5$  & $0.9045$  & $0.0210$  & $0.9103$  & $0.0215$  & $0.9244$  & $0.0292$\tabularnewline
\hline 
\end{tabular}
\end{table}

Next, we demonstrate the effectiveness of our test statistics in Theorems
\ref{Thm factor test plug-in Chi-sq} and \ref{Thm two-sample test B}
by showing their satisfactory size and power. 
\begin{itemize}
\item For the null hypothesis $H_{0}:\text{ }\bm{v}\in col(\bm{F}_{S,\cdot})$
in the factor specification test of Theorem \ref{Thm factor test plug-in Chi-sq},
the time index subset $S$ is chosen as $S=\{\frac{1}{2}T+1,\frac{1}{2}T+2,\ldots,\frac{1}{2}T+L\}$
with $|S|=L=12$, and the observed factors $\bm{v}$ is set as
\[
\bm{v}=\bm{F}_{S,\cdot}\bm{w}+\delta\bm{g}.
\]
The vector $\bm{w}$ is set as $\bm{w}=(w_{1},w_{2},w_{3})^{\top}=(1,1,0.5)^{\top}$,
and the parameter $\delta$ controls the deviation of the alternatives
from the null distribution. The vector $\bm{g}$ is constructed as
follows: a vector $\bm{u}$ is drawn from the standard Gaussian $\mathcal{N}(0,\bm{I}_{|S|})$
and the projection residual $\bm{u}^{\perp}=(\bm{I}_{|S|}-\bm{P}_{\bm{F}_{S,\cdot}})\bm{u}$
is computed, where $\bm{P}_{\bm{F}_{S,\cdot}}$ is the projection
matrix on the column space of $\bm{F}_{S,\cdot}$. Then, we set $\bm{g}=2\bm{u}^{\perp}/\Vert\bm{u}^{\perp}\Vert_{2}\cdot\Vert\bm{F}_{S,\cdot}\Vert_{\mathrm{F}}\Vert\bm{w}\Vert_{2}$.
The formulation of $\bm{g}$ ensures that the signal strength from
the column space of $\bm{F}_{S,\cdot}$, which is captured by $\Vert\bm{F}_{S,\cdot}\Vert_{\mathrm{F}}\Vert\bm{w}\Vert_{2}$,
is balanced with that from the space orthogonal to $\bm{F}_{S,\cdot}$.
The null hypothesis $H_{0}$ should not be rejected when $\delta=0$,
and should be rejected when $\delta>0$. 
\item For the null hypothesis $H_{0}:\ \bm{b}_{i}^{1}=\bm{b}_{i}^{2}$ in
the structural break test of betas in Theorem \ref{Thm beta structure test},
the cross-sectional unit is set as $i=1$, and the time subsets are
set as $\Gamma_{1}=\{1,2,\ldots,T_{1}\}$ and $\Gamma_{2}=\{T_{1}+1,T_{1}+2,\ldots,T_{1}+T_{2}\}$
with $T_{1}=T_{2}=T/2=100$. The beta $\bm{b}_{i}^{1}$ on period
$\Gamma_{1}$ is generated by the aforementioned procedure, while
the beta $\bm{b}_{i}^{2}$ on period $\Gamma_{2}$ is set as
\[
\bm{b}_{i}^{2}=\bm{b}_{i}^{1}+\Delta\Vert\bm{b}_{i}^{1}\Vert_{2}\mathbf{1},
\]
where $\mathbf{1}=(1,1,1)^{\top}$. The null hypothesis $H_{0}$ should
not be rejected when $\Delta=0$, and should be rejected when $\Delta>0$.
\item For the null hypothesis $H_{0}:\text{ }\bm{b}_{i}=\bm{b}_{j}$ in
the two-sample test of betas in Theorem \ref{Thm two-sample test B},
we set $i=1$ and study two cases for $j=2$ and $j=3$ respectively.
According to our simulation setup, the null hypothesis $H_{0}$ should
not be rejected when $j=2$, and should be rejected when $j=3$.
\end{itemize}
Tables \ref{table:factor tests}--\ref{table:beta tests} report
the empirical rejections rates at 5\% significance level over 200
Monte Carlo trials. Table \ref{table:factor tests} (resp. \ref{table:beta tests})
shows that for the test statistics in Theorem \ref{Thm factor test plug-in Chi-sq}
(resp. Theorems \ref{Thm beta structure test} and \ref{Thm two-sample test B}),
the results are favorable, exhibiting appropriate size and power,
even under a weak signal setup where the SNR $\theta$ is small.

\begin{table}[H]
\caption{Empirical rejection rates at level 5\% for the factor tests\label{table:factor tests}}

\centering

\begin{tabular}{c|c|c|c|c|c}
\hline 
$(N,T)=(300,200)$  & \multicolumn{5}{c}{Factor specification test}\tabularnewline
\hline 
SNR $\theta$  & $\delta=0$  & $\delta=0.25$  & $\delta=0.5$  & $\delta=0.75$  & $\delta=1$ \tabularnewline
\hline 
$5.5$  & $0.000$  & $0.000$  & $0.745$  & $0.990$  & $1.000$ \tabularnewline
\hline 
$5.0$  & $0.000$  & $0.000$  & $0.505$  & $0.945$  & $0.995$ \tabularnewline
\hline 
$4.5$  & $0.000$  & $0.000$  & $0.250$  & $0.895$  & $0.980$ \tabularnewline
\hline 
\end{tabular}
\end{table}

\begin{table}[H]
\caption{Empirical rejection rates at level 5\% for the beta tests\label{table:beta tests}}

\centering

\begin{tabular}{c|c|c|c|c|c|c|c}
\hline 
$(N,T)=(300,200)$ & \multicolumn{5}{c|}{Structural break test} & \multicolumn{2}{c}{Two-sample test}\tabularnewline
\hline 
SNR $\theta$ & $\Delta=0$ & $\Delta=0.25$ & $\Delta=0.5$ & $\Delta=0.75$ & $\Delta=1$ & $j=2$ & $j=3$\tabularnewline
\hline 
$5.5$ & $0.030$ & $0.510$ & $0.990$ & $1.000$ & $1.000$ & $0.045$ & $1.000$\tabularnewline
\hline 
$5.0$ & $0.030$ & $0.445$ & $0.970$ & $1.000$ & $1.000$ & $0.050$ & $1.000$\tabularnewline
\hline 
$4.5$ & $0.035$ & $0.355$ & $0.945$ & $1.000$ & $1.000$ & $0.050$ & $1.000$\tabularnewline
\hline 
\end{tabular}
\end{table}

\subsection{Empirical studies}

We analyze the monthly returns data of the S\&P 500 constituents from
the CRSP database for the period from January 1995 to March 2024.
We apply the factor specification test in Theorem \ref{Thm factor test plug-in Chi-sq}
and the structural break test for betas in Theorem \ref{Thm beta structure test}
to the stock returns.

First, we consider the factor specification test. The observed factors
we study are the Fama-French three factors: market (MKT), size (SMB),
and value (HML), denoted as $\bm{v}^{(1)}$, $\bm{v}^{(2)}$, and
$\bm{v}^{(3)}$, respectively. We obtain the time series data for
these factors from Kenneth French's website. We conduct our tests
via a rolling window approach, moving a 60-month window $[t+1,t+T]$
($T=60$) through the dataset. To mitigate the survival bias, we keep
the time series that have no more than 50\% missing values in each
window, i.e., the number of missing values is less than $T/2$, and
fill the missing values by the median of each time series. For each
window, the formed data matrix $\bm{X}^{t}$ is an $N\times T$ matrix,
where the cross-sectional dimension $N$ varies with $t$. We assume
that $\bm{X}^{t}$ satisfies the factor model as in (\ref{factor model matrix form}),
i.e.,
\[
\bm{X}^{t}=\bm{B}^{t}(\bm{F}^{t})^{\top}+\bm{E}^{t}.
\]
Here, the superscript $t$ is added to each matrix to emphasize that
we apply the PCA method across different time windows $[t+1,t+T]$.
The number of factors is fixed as $r=3$. We conduct the factor specification
test in Theorem \ref{Thm factor test plug-in Chi-sq} to test the
null hypothesis 
\[
\bm{v}_{S^{t},\cdot}^{(i)}=\bm{F}_{S^{t},\cdot}^{t}\bm{w}^{t,(i)},
\]
for the three factors, corresponding to $i=1,2,3$, respectively.
The time index subset $S^{t}$ is set as $S^{t}=[t+T-L+1,t+T]$ with
$|S^{t}|=L=12$.

The above procedure implies that, for each 12-month period $S^{t}=[t+T-L+1,t+T]$,
to test if the observed factors are in the column space of the latent
factors, we look back and utilize a broader historical data window
$[t+1,t+T]$ ($T=60$) to estimate the latent common factors $\bm{F}^{t}$.
Then we test the null hypothesis for this specific 12-month period
$S^{t}$, a subset at the end of the whole window $[t+1,t+T]$. Finally,
we plot in Figure \ref{fig:factor test} the test statistics against
the time index $t$ for each factor, underscoring the 95\% critical
value. As highlighted by \citet{FanLiaoYao2015power}, this rolling
window manner not only utilizes the up-to-date information in the
equity universe, but also alleviates the impacts of time-varying betas
and sampling biases.

\begin{figure}[H]
\centering

\begin{tabular}{c}
\includegraphics[scale=0.8]{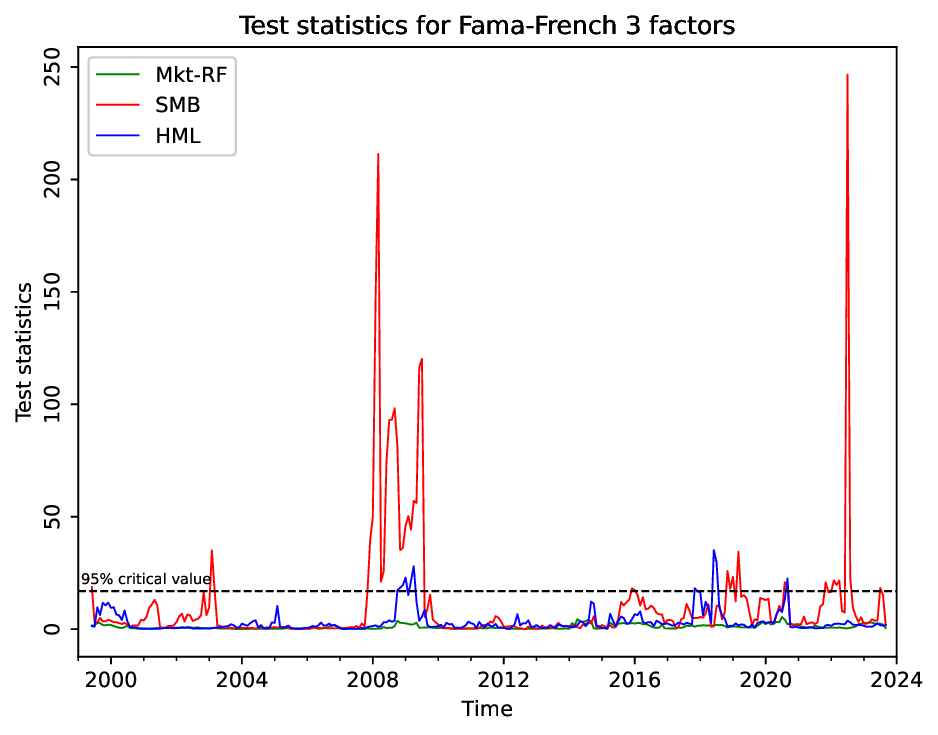}\tabularnewline
\end{tabular}

\caption{The evolution of the test statistics\label{fig:factor test}}
\end{figure}

In Figure \ref{fig:factor test}, our findings indicate that during
the financial crisis of 2007-2009, the null hypothesis that the size
factor SMB lies in the latent factors' column space, is rejected at
95\% confidence level. This suggests a diminished importance and reduced
explanatory power of the size factor SMB for stock return data in
this period. Note that the sizes of stocks can still be important
or even more important than that in the normal period, but not necessarily
captured by the variable SMB. Similar interpretations hold for the
value factor HML around 2009. During the COVID period around 2019,
both the size factor SMB and the value factor HML exhibit a loss in
explanatory power. The spike for the size factor SMB during 2022 is
probably due to the war between Russia and Ukraine that started in
February 2022, which is an unexpected economic shock to the stock
market. Notably, the market portfolio maintains its explanatory strength
throughout, indicating its stability and resilience as an explanatory
variable, even during distinct economic cycles.

Next, we consider the structural break test for betas for individual stocks. We test whether
the betas have changed before and after the three economic recessions
covered by the time horizon of our data -- the Early 2000s Recession
(Mar. 2001--Nov. 2001) due to the dot com bubble, the 2008 Great Recession due to housing bubble and
 financial crisis (Dec. 2007--Jun. 2009), and the COVID-19 Recession (Feb. 2020--Apr. 2020). Here, the
start and the end of each recession are according to the NBER's Business
Cycle Dating Committee.

For each recession period $[t_{\text{start}},t_{\text{end}}]$, we
first take the data $\bm{X}^{1}$ and $\bm{X}^{2}$ lying in the time
window $[t_{\text{start}}-T_{1},t_{\text{start}}-1]$ and $[t_{\text{end}}+1,t_{\text{end}}+T_{2}]$,
respectively, where $T_{1}=T_{2}=60$. Next, we merge the two panels
to get $\bm{X}=(\bm{X}^{1},\bm{X}^{2})$, and fill the missing values
in the same manner as in the factor specification test. We assume
that $\bm{X}$ satisfies the factor model as follows
\begin{align*}
x_{i,t} & =\bm{f}_{t}^{\top}\bm{b}_{i}^{1}+\varepsilon_{i,t}\qquad\text{for}\qquad t\in\{t_{\text{start}}-T_{1},t_{\text{start}}-T_{1}+1,\ldots,t_{\text{start}}-1\},\\
x_{i,t} & =\bm{f}_{t}^{\top}\bm{b}_{i}^{2}+\varepsilon_{i,t}\qquad\text{for}\qquad t\in\{t_{\text{end}}+1,t_{\text{end}}+2,\ldots,t_{\text{end}}+T_{2}\}.
\end{align*}
The number of factors is chosen by a scree plot of the eigenvalues
of $\bm{X}$, which is $4$, $3$, and $3$ for the three recessions,
respectively. We test the hypothesis $H_{0}:\ \bm{b}_{i}^{1}=\bm{b}_{i}^{2}\ \leftrightarrow\ H_{1}:\ \bm{b}_{i}^{1}\neq\bm{b}_{i}^{2}$
for each cross-sectional unit $i$ using the test statistic in Theorem
\ref{Thm beta structure test}. To report the test results, we group
the stocks into 11 sectors by Global Industrial Classification Standard
(GICS), and then count the numbers of stocks in each sector that reject
the null at the 95\% confidence level. The test results for the three
recessions are illustrated in Figures \ref{fig:beta test 1}, \ref{fig:beta test 2},
and \ref{fig:beta test 3}, respectively.

\begin{figure}[H]
\begin{tabular}{l}
\qquad{}\includegraphics[scale=0.65]{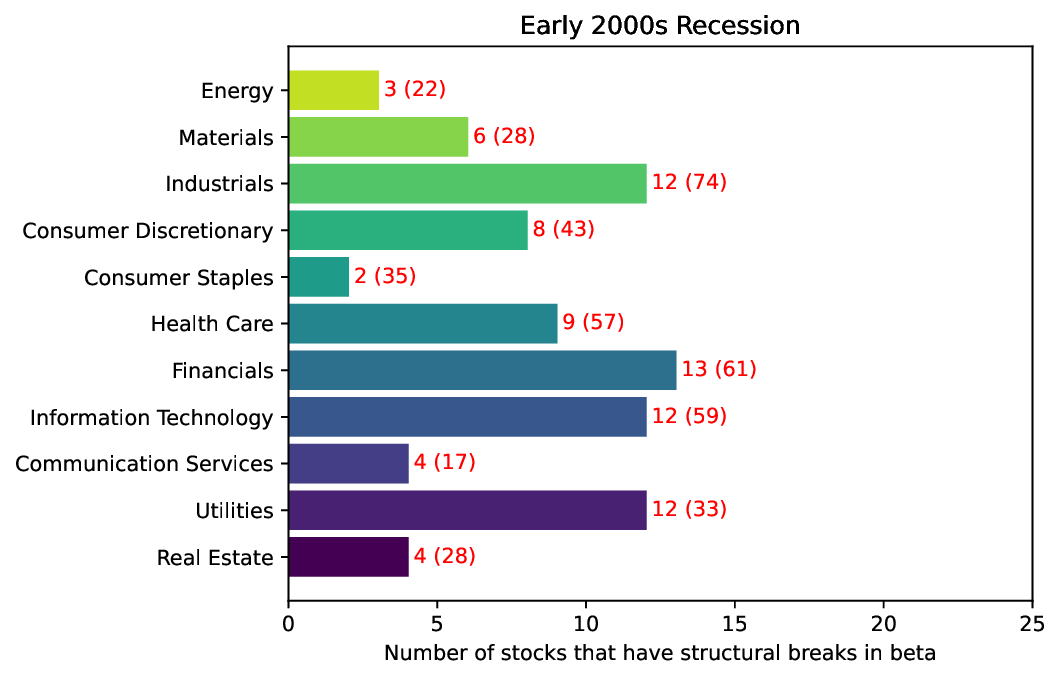}\tabularnewline
\end{tabular}

\caption{Test for breaks in betas for Early 2000s Recession (total number of
stocks is in bracket)\label{fig:beta test 1}}
\end{figure}

In Figure \ref{fig:beta test 1}, we observe that the sectors with
the highest number of stocks experiencing structural breaks in betas
are Financials, Information Technology, Industrials, and Utilities.
The impact on the Information Technology sector is likely due to the
dot-com bubble burst, which was one of the triggers of the Early 2000s
Recession. Another significant cause of this recession was the 9/11
attacks. Such severe economic shocks are possible reasons why typically
stable sectors like Industrials and Utilities were also affected.

\begin{figure}[H]
\begin{tabular}{l}
\qquad{}\includegraphics[scale=0.65]{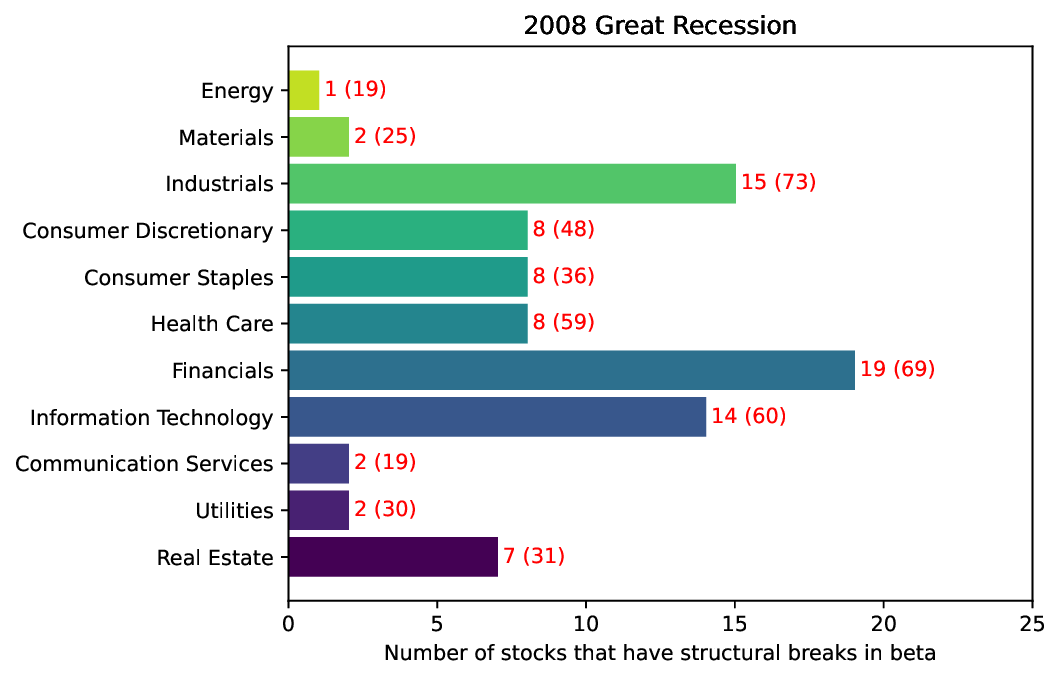}\tabularnewline
\end{tabular}

\caption{Test for breaks in betas for 2008 Great Recession (total number of
stocks is in bracket)\label{fig:beta test 2}}
\end{figure}

In Figure \ref{fig:beta test 2}, the 2007-2009 financial crisis,
marked by the subprime mortgage crisis and the collapse of the United
States housing bubble, affected many sectors. The financial sector
experienced a strong impact due to direct exposure to mortgage-backed
securities and other related financial instruments. The crisis led
to the failure or collapse of many of the United States' largest financial
institutions. Even though our analysis is inevitably influenced by
survival bias, as we can only analyze stocks that existed before and
after the crisis, we still observe that the financial sector had the
most affected stocks. The impact in Real Estate reflects the direct
consequences of the housing bubble burst. For sectors related to consumer
spending, such as Consumer Discretionary, Consumer Staples, and Health
Care, the shocks can be attributed to reduced consumer spending due
to increased unemployment and economic uncertainty.

\begin{figure}[H]
\begin{tabular}{l}
\qquad{}\includegraphics[scale=0.65]{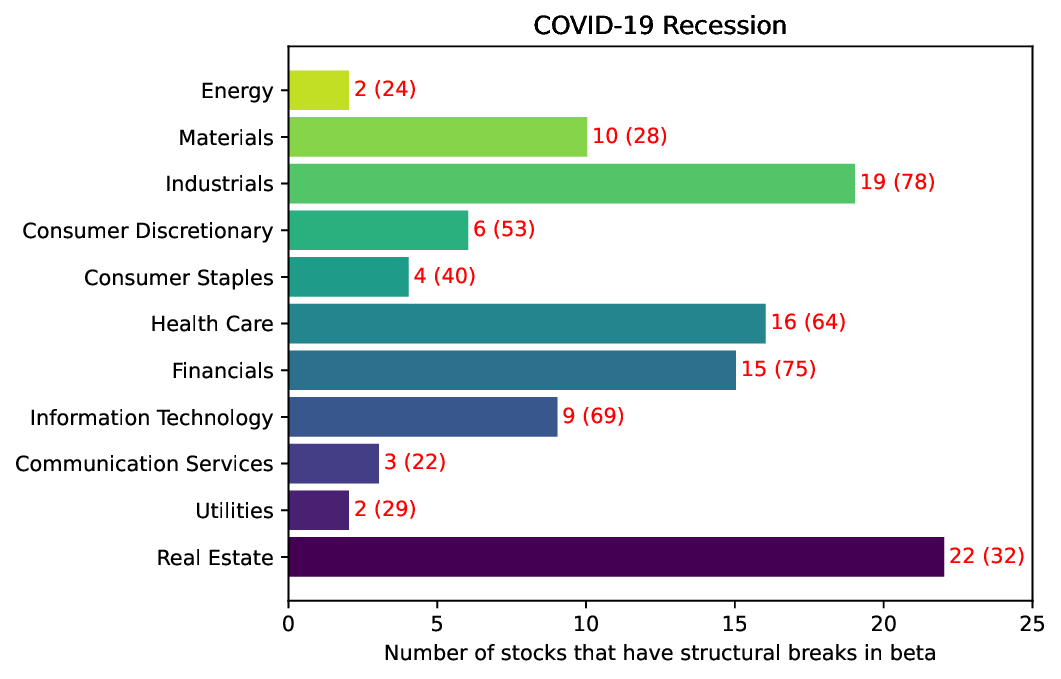}\tabularnewline
\end{tabular}

\caption{Test for breaks in betas for COVID-19 Recession (total number of stocks
is in bracket)\label{fig:beta test 3}}
\end{figure}

In Figure \ref{fig:beta test 3}, the economic effects of the pandemic
can be seen in many affected sectors, such as Real Estate, Industrials,
Health Care, and Financials. The significant changes in Real Estate
likely reflect the uncertainties brought about by lockdowns and health
crises, resulting in severe fluctuations in property values and rent
payments. The impact on the Health Care sector may be tied to the
heightened demand for medical services and supplies, alongside volatility
in biotechnology investments. Notably, Utilities and Energy showed
smaller changes, suggesting relative stability in these sectors despite
overall market volatility. During the 2008 Great Recession, these
sectors also demonstrated similar stability. This resilience could
be attributed to the essential nature of services provided by these
sectors, making them less susceptible to economic disruptions.

%% file: discussion.tex
\section{Other related works\label{sec:related work}}

The factor model is an important topic in finance and economics. The
early econometric studies on factor model can be dated back to \citet{forni2000factor,StockWatson2002PCA}.
Most previous works on factor model assumes that all factors are strong
or pervasive, that is, the SNR grows at rate $\sqrt{N}$ \citep[e.g.,][]{BaiNg2002ECTA,Bai2003ECTA,POET2013}.
When the SNR grows at a rate slower than $\sqrt{N}$, the model is
often called the weak factor model, which has been a popular research
topic in recent years. The method for determining the number of factors
under the weak factor model has been studied by a few papers \citep[e.g.,][]{Onatski2009number,Onatski2010number,frey2022weakFactorNumber}.
Several recent works have pursued the estimation and inference in
the weak factor model. To name a few examples, \citet{BaiNg2023PCA,Jiang2023PCA,ChoiMing2024PCA}
studied the consistency and asymptotic normality of PC estimators;
\citet{YoTa2022sWF_estimate,YoTa2022sWF_infer,wei2023sWF} studied
the sparsity-induced weak factor models where the low SNR is due to
the sparsity of the factor loading matrix; \citet{BKP2021weakFactor}
proposed an estimator of factor strength and established its theoretical
guarantee; \citet{Onatski2012PCAweak} showed that the PC estimators
are inconsistent in the extreme case (also known as the super-weak
factor model) where the SNR is $O(1)$.

The PCA is one of the most popular methods for the factor model, and
has been studied in many papers mentioned previously \citep[e.g.,][]{StockWatson2002PCA,BaiNg2002ECTA,Bai2003ECTA,Onatski2012PCAweak,POET2013,BaiNg2023PCA,Jiang2023PCA,ChoiMing2024PCA}.
Among the enormous literature, \citet{BaiNg2023PCA,Jiang2023PCA,ChoiMing2024PCA}
are the most recent and closest to our paper, since we all focus on
the PC estimators, especially on the inference side, under the weak
factor model. Besides the theoretical analysis, the variants of PCA
under the weak factor model have also been applied to empirical asset
pricing, macroeconomic forecasting, and many other important problems
in finance \citep[e.g.,][]{gigXiu2021testWeak,GiglioXiu2023weakFactorPredict}.
The estimation of factor model is closely related to the low-rank
matrix denoising in the statistical machine learning community. In
recent years, studying the factor model from the view of low-rank
matrix denoising has provided lots of exciting findings and understandings
(see \citet{FanLiLiao2021ARFEfactor,yan2021inference} for comprehensive
reviews), and our theory is partly inspired by this low-rank structure
of factor model.

From a technical perspective, our estimation procedure is a spectral
method, and is related to previous studies of spectral methods on
PCA and subspace estimation \citep{abbe2020ell_p,cai2019subspace,chen2020bridging,yan2021inference,zhou2023deflated}.
Our analysis relies on the leave-one-out techniques that have found
wide applications in analyzing spectral estimators and nonconvex optimization
algorithms \citep{el2015impact,abbe2017entrywise,ma2017implicit,chen2019noisy,chen2020convex};
see the recent monograph \citet{chen2021Monograph} for more details.
In addition, the problem and analysis in this paper is also related
to past works on inference for other low-rank models \citep{chen2019inference,xia2021statistical,chernozhukov2021inference,choi2023inference,choi2024inference,yan2024entrywise}.

\section{Conclusions and discussions\label{sec:Discussion}}

In this paper, we establish a novel theory for PCA under the weak
factor model, offering significant advancements on the inference of
PCA over the existing literature. The weak factor model removes the
pervasiveness assumption that requires of SNR growing at a $\sqrt{N}$
rate, where $N$ is the cross-sectional dimension. Our theory covers
both estimation and inference of factors and factor loadings. Notably,
in the regime $N\asymp T$, where $T$ is the temporal dimension,
we show that the asymptotic normality of PC estimators holds as long
as the SNR grows faster than a polynomial rate of $\log N$, while
the previous work required a polynomial rate of $N$. The optimality
of our theory is in the sense that, the required growth rate of the
SNR for consistency aligns with that for asymptotic normality, differing
only by some logarithmic factors. Our theory paves the way to design
easy-to-implement test statistics for practical applications, e.g.,
factor specification test, structural break test for betas, and build
confidence regions for crucial model parameters, e.g., betas and systematic
risks. We validate our statistical methods through extensive Monte
Carlo simulations, and conduct empirical studies to find noteworthy
correlations between our test results and specific economic cycles.

While the current scope of our theory is considerably wide-ranging,
it is possible to further widen it and there are lots of topics that
are worth pursuing. For instance, how to extend the theory to the
setting where each observed variables have only fouth moment via robustification
of the covariance input \citep{fan2021robust}? How to generalize
our inferential theory to the case where the time-serial correlation
also appears in the idiosyncratic noise? If the panel data is missing
at random, are the PC estimators still valid under the weak factor
model, and how to conduct statistical inference in this case? What
benefit could our new theory bring to forecasting methods based on
the factor-augmented regression? Among many directions, these topics
can be investigated in future research.

%% file: appendix_UV_main_results.tex
\section{\label{Proof of Thm UV 1st approx}Proof of Theorem \ref{Thm UV 1st approx row-wise error}:
First-order approximations}

For ease of exposition, we introduce some additional notation that
will be used throughout the proofs. For any $m\times n$ matrix $\bm{M}=(M_{i,j})_{1\leq i\leq m,1\leq j\leq n}$,
we let $\sigma_{i}(\bm{M})$, $\sigma_{\min}(\bm{M})$, and $\sigma_{\max}(\bm{M})$
denote the $i$-th largest, the minimum, and the maximum singular
value of $\bm{M}$, respectively. For a symmetric matrix $\bm{A}$,
we denote by tr$(\bm{A})$ its trace.

\subsection{Some useful lemmas}

To prove Theorem \ref{Thm UV 1st approx row-wise error}, we collect
some useful lemmas as preparations. We start with a lemma to reveal
the close relations between the SVD $\bm{U}\bm{\Lambda}\bm{V}^{\top}$
and the common component $\bm{BF}^{\top}$ in the factor model (\ref{factor model matrix form}),
where $\bm{B}=\bar{\bm{U}}\bm{\Sigma}$.

\begin{lemma} \label{Lemma SVD BF good event}Suppose that Assumptions
\ref{Assump_Bf_identification} and \ref{Assump_factor_f} hold. Assuming
that 
\begin{equation}
r+\log n\ll T,\label{Assump T large logn}
\end{equation}
then we have that, there exists a $\sigma(\bm{F})$-measurable event
$\mathcal{E}_{0}$ with $\mathbb{P}(\mathcal{E}_{0})>1-O(n^{-2})$,
where $\sigma(\bm{F})$ is the $\sigma$-algebra generated by $\bm{F}$,
such that, the following properties hold when $\mathcal{E}_{0}$ happens:
(i) rank$(\bm{F})=r$, 
\[
\left\Vert \frac{1}{T}\bm{F}^{\top}\bm{F}-\bm{I}_{r}\right\Vert _{2}\lesssim\frac{\sqrt{r+\log n}}{\sqrt{T}},\text{\qquad and\qquad}\left\Vert \bm{F}\right\Vert _{2,\infty}\lesssim\sqrt{\log n}.
\]

(ii) There exists a $\sigma(\bm{F})$-measurable matrix $\bm{Q}$
satisfying that $\bm{U}=\bar{\bm{U}}\bm{Q}$ and $\bm{Q}\in\mathcal{O}^{r\times r}$
is a rotation matrix, i.e., $\bm{QQ}^{\top}=\bm{Q}^{\top}\bm{Q}=\bm{I}_{r}$.

(iii) $\lambda_{i}\asymp\sigma_{i}$ for $i=1,2,\ldots,r$.

(iv) There exists a $\sigma(\bm{F})$-measurable matrix $\bm{J}$
satisfying that $\bm{V}=T^{-1/2}\bm{FJ}$ and $\bm{J}=\bm{\Sigma}\bm{Q}\bm{\Lambda}^{-1}$
is a $r\times r$ invertible matrix satisfying that $\sigma_{i}(\bm{J})\asymp1$
for $i=1,2,\ldots,r$. Further, it holds 
\[
\big\Vert\bm{V}\big\Vert_{2,\infty}\lesssim\sqrt{\frac{\log n}{T}}.
\]

\end{lemma} 
\begin{proof}
Since $\bm{F}=(\bm{f}_{1},\bm{f}_{2},\ldots,\bm{f}_{T})^{\top}$,
and $\bm{f}_{1},\bm{f}_{2},\ldots,\bm{f}_{T}$ are independent sub-Gaussian
random vectors under Assumption \ref{Assump_factor_f}, we obtain
by (4.22) in \citet{vershynin2016high} that, 
\begin{equation}
\left\Vert \frac{1}{T}\bm{F}^{\top}\bm{F}-\bm{I}_{r}\right\Vert _{2}\lesssim\sqrt{\frac{r}{T}}+\sqrt{\frac{\log n}{T}}+\frac{r}{T}+\frac{\log n}{T}\lesssim\sqrt{\frac{r+\log n}{T}},\label{FF-I_r inequality}
\end{equation}
with probability at least $1-O(n^{-2})$, where the last inequality
is owing to (\ref{Assump T large logn}) and the fact that $\kappa\geq1$.
Then using the sub-Gaussian property of $\bm{f}_{t}$ and the standard
concentration inequality, we have that 
\begin{equation}
\left\Vert \bm{F}\right\Vert _{2,\infty}\lesssim\sqrt{\log n},\label{F sup row norm inequality}
\end{equation}
with probability at least $1-O(n^{-2})$. In particular, we let $\mathcal{E}_{0}$
be the event that both (\ref{FF-I_r inequality}) and (\ref{F sup row norm inequality})
happen, then we have that $\mathbb{P}(\mathcal{E}_{0})>1-O(n^{-2})$.
In what follows, we show that $\mathcal{E}_{0}$ satisfies all the
requirements.

When (\ref{FF-I_r inequality}) happens, since (\ref{Assump T large logn})
implies that $\sqrt{\frac{r+\log n}{T}}\ll1$, we have that $\big\Vert\frac{1}{T}\bm{F}^{\top}\bm{F}-\bm{I}_{r}\big\Vert_{2}\ll1$,
and thus $\vartheta_{i}(\frac{1}{T}\bm{F}^{\top}\bm{F})\asymp1$ for
$i=1,2,\ldots,r$, implying that $\sigma_{i}(\bm{F})\asymp\sqrt{T}$
for $i=1,2,\ldots,r$. Since $\sigma_{i}(\bm{F})\asymp\sqrt{T}$ implies
that $\sigma_{\text{min}}(\bm{F})>0$, we obtain that $\text{rank}(\bm{F})=r$,
i.e., $\bm{F}$ has full column rank.

By definition, the columns of $\bm{U}$ (resp. $\bar{\bm{U}}$) are
the left singular vectors of $\bm{B}\bm{F}^{\top}$ (resp. $\bm{B}$).
When (\ref{FF-I_r inequality}) happens, since $\text{rank}(\bm{F}_{S,\cdot})=r$,
the $r$-dimensional column spaces of $\bm{U}$ and $\bar{\bm{U}}$
are the same with each other. Then since the columns of both $\bm{U}$
and $\bar{\bm{U}}$ are orthonormal, we obtain that there exists a
rotation matrix $\bm{Q}\in\mathcal{O}^{r\times r}$ such that $\bm{U}=\bar{\bm{U}}\bm{Q}$.
By construction we have that $\bm{Q}$ is $\sigma(\bm{F})$-measurable.

Next, we prove the relations for eigenvalues. By definition, we have
that $\{\lambda_{i}^{2}\}_{1\leq i\leq r}$ are the eigenvalues of
$(T^{-1/2}\bm{BF}^{\top})^{\top}T^{-1/2}\bm{BF}^{\top}=T^{-1}\bm{FB}^{\top}\bm{B\bm{F}^{\top}}$,
and $\{\sigma_{i}^{2}\}_{1\leq i\leq r}$ are the eigenvalues of $\bm{B}^{\top}\bm{B}=\bm{\Sigma}^{2}$.
By Theorem A.2 in \citet{braun2006JMLR_eigen}, we obtain that 
\[
|\lambda_{i}^{2}-\sigma_{i}^{2}|=|\vartheta_{i}(T^{-1}\bm{FB}^{\top}\bm{B\bm{F}^{\top}})-\vartheta_{i}(\bm{B}^{\top}\bm{B})|\leq\sigma_{i}^{2}\left\Vert T^{-1}\bm{\bm{F}^{\top}}\bm{F}-\bm{I}_{r}\right\Vert _{2}\lesssim\sigma_{i}^{2}\sqrt{\frac{r+\log n}{T}}\ll\sigma_{i}^{2},
\]
where the last inequality is owing to (\ref{Assump T large logn}).
So we obtain that $\lambda_{i}\asymp\sigma_{i}$ for $i=1,2,\ldots,r$.

By the property of SVD, we have that $\bm{V}=(T^{-1/2}\bm{BF}^{\top})^{\top}\bm{U}\bm{\Lambda}^{-1}=T^{-1/2}\bm{F}\bm{J}$,
where $\bm{J}$ is given by $\bm{J}:=\bm{B}^{\top}\bm{U}\bm{\Lambda}^{-1}=\bm{\Sigma}\bar{\bm{U}}^{\top}\bm{U}\bm{\Lambda}^{-1}=\bm{\Sigma}\bm{Q}\bm{\Lambda}^{-1}$
and $\bm{J}$ is invertible. So we have that $T^{-1/2}\bm{F}=\bm{V}\bm{J}^{-1}$
and thus $T^{-1}\bm{F}^{\top}\bm{F}=(\bm{J}^{-1})^{\top}\bm{V}^{\top}\bm{V}\bm{J}^{-1}=(\bm{J}^{-1})^{\top}\bm{J}^{-1}$.
Then, we obtain that $||(\bm{J}^{-1})^{\top}\bm{J}^{-1}-\bm{I}_{r}||_{2}\leq||T^{-1}\bm{F}^{\top}\bm{F}-\bm{I}_{r}||_{2}\lesssim T^{-1/2}\sqrt{r+\log n}\ll1$,
and thus we obtain that $\sigma_{i}(\bm{J})\asymp1$ for $i=1,2,\ldots,r$.
As a result, we have that 
\[
\big\Vert\bm{V}\big\Vert_{2,\infty}=\frac{1}{\sqrt{T}}\big\Vert\bm{F}\bm{J}\big\Vert_{2,\infty}\leq\frac{1}{\sqrt{T}}\big\Vert\bm{F}\big\Vert_{2,\infty}\big\Vert\bm{J}\big\Vert_{2}\overset{\text{(i)}}{\lesssim}\frac{1}{\sqrt{T}}\big\Vert\bm{F}\big\Vert_{2,\infty}\overset{\text{(ii)}}{\lesssim}\sqrt{\frac{\log n}{T}},
\]
where (i) and (ii) use $\sigma_{i}(\bm{J})\asymp1$ and $\left\Vert \bm{F}\right\Vert _{2,\infty}\lesssim\sqrt{\log n}$
in (\ref{F sup row norm inequality}), respectively. 
\end{proof}
Lemma \ref{Lemma SVD BF good event} shows how the SVD $\bm{U}\bm{\Lambda}\bm{V}^{\top}$
relates with $\bm{BF}^{\top}$ and constructs a good event $\mathcal{E}_{0}$
on which we have nice properties of the quantities in our setup. Lemma
\ref{Lemma Z norms A B} below establishes some useful results on
the random matrix $\bm{Z}$ which characterizes the noise matrix $\bm{E}$
in (\ref{noise matrix formula}).

\begin{lemma} \label{Lemma Z norms A B}Under Assumption \ref{Assump_noise_Z_entries},
for any fixed matrices $\bm{A}\in\mathbb{R}^{T\times m}$ and $\bm{B}\in\mathbb{R}^{N\times m}$
with $m\leq n$, we have that, with probability at least $1-O(n^{-2})$,
\begin{equation}
\left\Vert \bm{Z}\right\Vert _{2}\lesssim\sqrt{n},\text{ }\left\Vert \bm{E}\right\Vert _{2}\lesssim\sqrt{\left\Vert \bm{\Sigma}_{\varepsilon}\right\Vert _{2}}\sqrt{n},\label{noise E norm inequality}
\end{equation}
and 
\[
\left\Vert \bm{ZA}\right\Vert _{2,\infty}\lesssim\left\Vert \bm{A}\right\Vert _{\mathrm{F}}\log n,\text{ }\big\Vert\bm{Z}^{\top}\bm{B}\big\Vert_{2,\infty}\lesssim\left\Vert \bm{B}\right\Vert _{\mathrm{F}}\log n.
\]

\end{lemma} 
\begin{proof}
It suffices to prove that every inequality holds with probability
at least $1-O(n^{-2})$. Since the entries of $\bm{Z}$ are independent
and sub-Gaussian under Assumption \ref{Assump_noise_Z_entries}, we
obtain by the union bound argument that 
\begin{equation}
\mathbb{P}(\mathcal{E}_{Z})>1-O(n^{-2})\text{ with }\mathcal{E}_{Z}:=\left\{ \max_{1\leq k\leq N,1\leq l\leq T}|\bm{Z}_{k,l}|\leq C\sqrt{\log n}\right\} ,\label{Prob E_Z entry logn bound}
\end{equation}
where $C>0$ is a constant. Given $\mathcal{E}_{Z}$, it follows from
(3.9) in \citet{chen2021Monograph} that, $\left\Vert \bm{Z}\right\Vert _{2}\lesssim\sqrt{n}+\log n\lesssim\sqrt{n}$
with probability at least $1-O(n^{-2})$, i.e., $\mathbb{P}\left(\mathcal{E}_{1}|\mathcal{E}_{Z}\right)>1-O(n^{-2})$
where $\mathcal{E}_{1}:=\left\{ \left\Vert \bm{Z}\right\Vert _{2}\lesssim\sqrt{n}\right\} $.
So, we obtain 
\begin{equation}
\mathbb{P}\left(\mathcal{E}_{1}^{c}\right)=\mathbb{P}\left(\mathcal{E}_{1}^{c}|\mathcal{E}_{Z}\right)\mathbb{P}\left(\mathcal{E}_{Z}\right)+\mathbb{P}\left(\mathcal{E}_{1}^{c}|\mathcal{E}_{Z}^{c}\right)\mathbb{P}\left(\mathcal{E}_{Z}^{c}\right)\leq\mathbb{P}\left(\mathcal{E}_{1}^{c}|\mathcal{E}_{Z}^{c}\right)+\mathbb{P}\left(\mathcal{E}_{Z}^{c}\right)\leq O(n^{-2}),\label{condition on E_Z argument}
\end{equation}
implying that our desired result $\mathbb{P}\left(\left\Vert \bm{Z}\right\Vert _{2}\lesssim\sqrt{n}\right)=1-\mathbb{P}\left(\mathcal{E}_{1}^{c}\right)>1-O(n^{-2})$.
Indeed, the argument for $\mathcal{E}_{1}$ in (\ref{condition on E_Z argument})
adapts to all the events which we will prove to hold with probability
at least $1-O(n^{-2})$. In other words, for any event $\mathcal{E},$
to prove $\mathbb{P}\left(\mathcal{E}\right)>1-O(n^{-2})$, it suffices
to prove that $\mathbb{P}\left(\mathcal{E}|\mathcal{E}_{Z}\right)>1-O(n^{-2})$.
So, in what follows, we conduct the proofs conditioning on $\mathcal{E}_{Z}$.

It is easy to see that, $\left\Vert \bm{E}\right\Vert _{2}\lesssim\sqrt{\left\Vert \bm{\Sigma}_{\varepsilon}\right\Vert _{2}}\sqrt{n}$
follows from 
\[
\left\Vert \bm{E}\right\Vert _{2}=\big\Vert\bm{\Sigma}_{\varepsilon}^{1/2}\bm{Z}\big\Vert_{2}\leq\big\Vert\bm{\Sigma}_{\varepsilon}^{1/2}\big\Vert_{2}\left\Vert \bm{Z}\right\Vert _{2}=\sqrt{\left\Vert \bm{\Sigma}_{\varepsilon}\right\Vert _{2}}\left\Vert \bm{Z}\right\Vert _{2}\lesssim\sqrt{\left\Vert \bm{\Sigma}_{\varepsilon}\right\Vert _{2}}\sqrt{n}.
\]
Then, we prove that 
\begin{equation}
\mathbb{P}\left(\mathcal{E}_{A}|\mathcal{E}_{Z}\right)>1-O(n^{-2})\text{ with }\mathcal{E}_{A}:=\left\{ \left\Vert \bm{ZA}\right\Vert _{2,\infty}\lesssim\left\Vert \bm{A}\right\Vert _{\mathrm{F}}\log n\right\} .\label{Prob E_A given E_Z}
\end{equation}
To prove (\ref{Prob E_A given E_Z}), we note that, for $k=1,2,\ldots,N$,
$\bm{Z}_{k,\cdot}\bm{A}=\sum_{l=1}^{T}\bm{Z}_{k,l}\bm{A}_{l,\cdot}$
is a sum of independent mean zero random vectors. Given $\mathcal{E}_{Z}$,
it holds $\left\Vert \bm{Z}_{k,l}\bm{A}_{l,\cdot}\right\Vert _{2}\leq C\sqrt{\log n}\left\Vert \bm{A}_{l,\cdot}\right\Vert _{2}$
for $l=1,2,\ldots,T$. Then, using the matrix Hoeffding inequality
\citep[Theorem 1.3]{tropp2012user}, we obtain that 
\[
\mathbb{P}\left(\left\Vert \bm{Z}_{k,\cdot}\bm{A}\right\Vert _{2}\geq\delta|\mathcal{E}_{Z}\right)\leq2m\exp\left(-\frac{\delta^{2}}{8\sigma^{2}}\right),
\]
where $\sigma=[\sum_{l=1}^{T}(C\sqrt{\log n}\left\Vert \bm{A}_{l,\cdot}\right\Vert _{2})^{2}]^{1/2}=C\sqrt{\log n}\left\Vert \bm{A}\right\Vert _{\mathrm{F}}$.
Then we obtain by the union bound argument that $\mathbb{P}(\left\Vert \bm{ZA}\right\Vert _{2,\infty}\leq\delta|\mathcal{E}_{Z})>1-2Nm\exp(-\frac{\delta^{2}}{8\sigma^{2}})$.
Letting $\delta=100\sigma\sqrt{\log n}$, we obtain that $\mathbb{P}(\left\Vert \bm{ZA}\right\Vert _{2,\infty}\leq\delta|\mathcal{E}_{Z})>1-O(n^{-2})$,
which implying (\ref{Prob E_A given E_Z}). The proof for $\big\Vert\bm{Z}^{\top}\bm{B}\big\Vert_{2,\infty}\lesssim\left\Vert \bm{B}\right\Vert _{\mathrm{F}}\log n$
follows from the same manner, so we omit the details for the sake
of brevity. 
\end{proof}
Next, we consider the truncated rank-$r$ SVD $\widehat{\bm{U}}\widehat{\bm{\Sigma}}\widehat{\bm{V}}^{\top}$
of $\frac{1}{\sqrt{T}}\bm{X}=\frac{1}{\sqrt{T}}($\textbf{$\bm{B}$}$\bm{F}^{\top}\mathrm{+}\bm{E})$.
We define 
\[
\bm{H}_{U}:=\widehat{\bm{U}}^{\top}\bm{U}\text{ and }\bm{H}_{V}:=\widehat{\bm{V}}^{\top}\bm{V}.
\]
Lemma \ref{Lemma R H for U V} below give the perturbation bounds
for the singular spaces, i.e., the column spaces of $\bm{U}$ and
$\bm{V}$, under the spectral norm, and some basic facts on $\widehat{\bm{\Sigma}}$,
$\bm{H}_{U}$, and $\bm{H}_{V}$. \begin{lemma}\label{Lemma R H for U V}Suppose
that Assumptions \ref{Assump_Bf_identification}, \ref{Assump_noise_Z_entries},
and \ref{Assump_factor_f} hold. For simplicity of notations, we define
\begin{equation}
\rho:=\frac{\sqrt{\left\Vert \bm{\Sigma}_{\varepsilon}\right\Vert _{2}}\sqrt{n}}{\sigma_{r}\sqrt{T}}=\frac{\sqrt{n}}{\theta\sqrt{T}}.\label{SNR 2-norm def}
\end{equation}
Assuming that 
\begin{equation}
\rho\ll1,\label{SNR 2-norm}
\end{equation}
then we have that, with probability at least $1-O(n^{-2})$, 
\[
\max\left\{ \big\Vert\widehat{\bm{U}}\bm{R}_{U}-\bm{U}\big\Vert_{2},\big\Vert\widehat{\bm{V}}\bm{R}_{V}-\bm{V}\big\Vert_{2},\sqrt{\left\Vert \bm{H}_{U}-\bm{R}_{U}\right\Vert _{2}},\sqrt{\left\Vert \bm{H}_{V}-\bm{R}_{V}\right\Vert _{2}}\right\} \lesssim\rho,
\]
and $\max\{\big\Vert\widehat{\bm{U}}\bm{R}_{U}-\bm{U}\big\Vert{}_{\mathrm{F}},\big\Vert\widehat{\bm{V}}\bm{R}_{V}-\bm{V}\big\Vert{}_{\mathrm{F}}\}\lesssim\rho\sqrt{r}$,
as well as 
\[
\widehat{\sigma}_{i}\asymp\sigma_{i},\text{ }\sigma_{i}(\bm{H}_{U})\asymp1,\text{ }\sigma_{i}(\bm{H}_{V})\asymp1,\text{\qquad for\qquad}i=1,2,\ldots,r.
\]
Also, we have that 
\[
\max\left\{ \big\Vert\widehat{\bm{\Sigma}}(\bm{H}_{U}-\bm{R}_{U})\big\Vert_{2},\big\Vert\widehat{\bm{\Sigma}}(\bm{H}_{V}-\bm{R}_{V})\big\Vert_{2}\right\} \lesssim\sigma_{r}\rho^{2}.
\]

\end{lemma} 
\begin{proof}
Using the argument as discussed after (\ref{condition on E_Z argument}),
for any event $\mathcal{E},$ to prove $\mathbb{P}\left(\mathcal{E}\right)>1-O(n^{-2})$,
it suffices to prove that $\mathbb{P}\left(\mathcal{E}|\mathcal{E}_{Z}\cap\mathcal{E}_{0}\right)>1-O(n^{-2})$.
Here, the event $\mathcal{E}_{Z}$ is defined in (\ref{Prob E_Z entry logn bound})
and the event $\mathcal{E}_{0}$ is defined in Lemma \ref{Lemma SVD BF good event}.
To see this, we note that, as long as $\mathbb{P}\left(\mathcal{E}|\mathcal{E}_{Z}\cap\mathcal{E}_{0}\right)>1-O(n^{-2})$,
it holds 
\begin{align}
\mathbb{P}\left(\mathcal{E}^{c}\right) & =\mathbb{P}\left(\mathcal{E}^{c}|\mathcal{E}_{Z}\cap\mathcal{E}_{0}\right)\mathbb{P}\left(\mathcal{E}_{Z}\cap\mathcal{E}_{0}\right)+\mathbb{P}\left(\mathcal{E}^{c}|(\mathcal{E}_{Z}\cap\mathcal{E}_{0})^{c}\right)\mathbb{P}\left((\mathcal{E}_{Z}\cap\mathcal{E}_{0})^{c}\right)\label{condition on E_Z and F good}\\
 & \leq\mathbb{P}\left(\mathcal{E}^{c}|\mathcal{E}_{Z}\cap\mathcal{E}_{0}\right)+\mathbb{P}\left((\mathcal{E}_{Z}\cap\mathcal{E}_{0})^{c}\right)\leq O(n^{-2})+\mathbb{P}\left(\mathcal{E}_{Z}^{c}\cup\mathcal{E}_{0}^{c}\right)\nonumber \\
 & \leq O(n^{-2})+\mathbb{P}(\mathcal{E}_{Z}^{c})+\mathbb{P}(\mathcal{E}_{0}^{c})\leq O(n^{-2}).\nonumber 
\end{align}
So, in what follows, we conduct the proofs conditioning on $\mathcal{E}_{Z}\cap\mathcal{E}_{0}$.

Recall that $\widehat{\bm{U}}\widehat{\bm{\Sigma}}\widehat{\bm{V}}^{\top}$
is the truncated rank-$r$ SVD of $\frac{1}{\sqrt{T}}\bm{X}$, and
$\bm{U}\bm{\Lambda}\bm{V}^{\top}$ is the SVD of $\frac{1}{\sqrt{T}}$\textbf{$\bm{B}$}$\bm{F}^{\top}$.
Note that $\frac{1}{\sqrt{T}}\bm{X}-\frac{1}{\sqrt{T}}$\textbf{$\bm{B}$}$\bm{F}^{\top}=\frac{1}{\sqrt{T}}\bm{E}$,
and it follows from (\ref{noise E norm inequality}) in Lemma \ref{Lemma Z norms A B}
that 
\[
\frac{1}{\sqrt{T}}\left\Vert \bm{E}\right\Vert _{2}\lesssim\frac{1}{\sqrt{T}}\sqrt{\left\Vert \bm{\Sigma}_{\varepsilon}\right\Vert _{2}}\sqrt{n}\ll\sigma_{r},
\]
where the last inequality is owing to (\ref{SNR 2-norm}). So we obtain
$\frac{1}{\sqrt{T}}\left\Vert \bm{E}\right\Vert _{2}\ll\lambda_{r}$,
since $\lambda_{i}\asymp\sigma_{i}$ for $i=1,2,\ldots,r$ according
to Lemma \ref{Lemma SVD BF good event}. We write the full SVD of
$T^{-1/2}\bm{X}$ as 
\[
\frac{1}{\sqrt{T}}\bm{X}=\left[\begin{array}{cc}
\widehat{\bm{U}} & \widehat{\bm{U}}_{\perp}\end{array}\right]\left[\begin{array}{cc}
\widehat{\bm{\Sigma}} & 0\\
0 & \widehat{\bm{\Sigma}}_{\perp}
\end{array}\right]\left[\begin{array}{c}
\widehat{\bm{V}}^{\top}\\
(\widehat{\bm{V}}_{\perp})^{\top}
\end{array}\right],
\]
where $\widehat{\bm{U}}_{\perp}\in\mathbb{R}^{N\times(N-r)}$ and
$\widehat{\bm{V}}_{\perp}\in\mathbb{R}^{T\times(T-r)}$ are the orthogonal
complements of $\widehat{\bm{U}}$ and $\widehat{\bm{V}}$, and $\widehat{\bm{\Sigma}}_{\perp}\in\mathbb{R}^{(N-r)\times(T-r)}$
contains the smaller singular values and $\widehat{\bm{\Sigma}}_{\perp}$
is not necessarily a diagonal matrix. Then by Wedin's $\sin\Theta$
theorem \citep[(2.26a)\textendash (2.26b)]{chen2021Monograph}, we
have that 
\begin{equation}
\big\Vert\widehat{\bm{U}}\bm{R}_{U}-\bm{U}\big\Vert_{2}\lesssim\big\Vert\widehat{\bm{U}}\widehat{\bm{U}}^{\top}-\bm{U}\bm{U}^{\top}\big\Vert_{2}\lesssim\big\Vert\bm{U}^{\top}\widehat{\bm{U}}_{\perp}\big\Vert_{2}\lesssim\frac{\frac{1}{\sqrt{T}}\left\Vert \bm{E}\right\Vert _{2}}{\lambda_{r}-\lambda_{r+1}}\lesssim\frac{\frac{1}{\sqrt{T}}\sqrt{\left\Vert \bm{\Sigma}_{\varepsilon}\right\Vert _{2}}\sqrt{n}}{\sigma_{r}}=\rho,\label{U U_hat_cmpl bound}
\end{equation}
and $\big\Vert\widehat{\bm{V}}\bm{R}_{V}-\bm{V}\big\Vert_{2}\lesssim\big\Vert\bm{V}^{\top}\widehat{\bm{V}}_{\perp}\big\Vert_{2}\lesssim\rho$,
and similarly $\max\{\big\Vert\widehat{\bm{U}}\bm{R}_{U}-\bm{U}\big\Vert{}_{\mathrm{F}},\big\Vert\widehat{\bm{V}}\bm{R}_{V}-\bm{V}\big\Vert{}_{\mathrm{F}}\}\lesssim\rho\sqrt{r}$,
where we use the fact that $\lambda_{r}\asymp\sigma_{r}$ and $\lambda_{r+1}=0$,
since $\lambda_{i}$ is the $i$-th largest singular value of $\frac{1}{\sqrt{T}}$\textbf{$\bm{B}$}$\bm{F}^{\top}$
and rank$($\textbf{$\bm{B}$}$\bm{F}^{\top})=r$. Next, for $\bm{H}_{U}$,
using the arguments in Section C.3.1 of \citet{yan2021inference}
based on Wedin's $\sin\Theta$ theorem, we obtain that $\left\Vert \bm{H}_{U}-\bm{R}_{U}\right\Vert _{2}\lesssim\big\Vert\widehat{\bm{U}}\bm{R}_{U}-\bm{U}\big\Vert_{2}^{2}\lesssim\rho^{2}\ll1$.
Then, since $\bm{R}_{U}\in\mathcal{O}^{r\times r}$ implies that $\sigma_{i}(\bm{R}_{U})=1$
for $i=1,2,\ldots,r$, we obtain $\sigma_{i}(\bm{H}_{U})\asymp1$
for $i=1,2,\ldots,r$. The results for $\bm{H}_{V}$ can be proven
in the same manner. Also, we obtain by Weyl's inequality \citep[Lemma 2.2]{chen2021Monograph}
that, for $i=1,2,\ldots,r$, $|\widehat{\sigma}_{i}-\lambda_{i}|\leq\frac{1}{\sqrt{T}}\left\Vert \bm{E}\right\Vert _{2}\ll\sigma_{r}$,
implying that $\widehat{\sigma}_{i}\asymp\sigma_{i}$ for $i=1,2,\ldots,r$.

We now prove the results for $\big\Vert\widehat{\bm{\Sigma}}(\bm{H}_{U}-\bm{R}_{U})\big\Vert_{2}$
and $\big\Vert\widehat{\bm{\Sigma}}(\bm{H}_{V}-\bm{R}_{V})\big\Vert_{2}$.
Similar to the arguments for proving (F.33) in \citet{yan2024entrywise},
we have that 
\[
\big\Vert\widehat{\bm{\Sigma}}(\bm{H}_{U}-\bm{R}_{U})\big\Vert_{2}\lesssim\rho\big\Vert\widehat{\bm{\Sigma}}\widehat{\bm{U}}\bm{U}_{\perp}\big\Vert_{2},
\]
and $\widehat{\bm{\Sigma}}\widehat{\bm{U}}\bm{U}_{\perp}=T^{-1/2}\widehat{\bm{V}}^{\top}\bm{E}^{\top}\bm{U}_{\perp}$,
where $\bm{U}_{\perp}\in\mathbb{R}^{N\times(N-r)}$ is the orthogonal
complement of $\bm{U}$. Next, we have that 
\begin{align*}
\big\Vert\widehat{\bm{\Sigma}}\widehat{\bm{U}}\bm{U}_{\perp}\big\Vert_{2} & =\frac{1}{\sqrt{T}}\big\Vert(\bm{U}_{\perp})^{\top}\bm{E}\widehat{\bm{V}}\big\Vert_{2}\leq\frac{1}{\sqrt{T}}\big\Vert(\bm{U}_{\perp})^{\top}\bm{\Sigma}_{\varepsilon}^{1/2}\bm{Z}\bm{V}\big\Vert_{2}+\frac{1}{\sqrt{T}}\big\Vert(\bm{U}_{\perp})^{\top}\bm{E}(\widehat{\bm{V}}\bm{R}_{V}-\bm{V})\big\Vert_{2}\\
 & \overset{\text{(i)}}{\lesssim}\frac{1}{\sqrt{T}}\big\Vert\widehat{\bm{V}}^{\top}\bm{\Sigma}_{\varepsilon}^{1/2}\big\Vert_{2}\big\Vert\bm{U}_{\perp}\big\Vert_{2}\sqrt{r+(N-r)+\log n}+\frac{1}{\sqrt{T}}\big\Vert\bm{E}\big\Vert_{2}\big\Vert(\widehat{\bm{V}}\bm{R}_{V}-\bm{V})\big\Vert_{2}\\
 & \overset{\text{(ii)}}{\lesssim}\frac{1}{\sqrt{T}}\sqrt{\left\Vert \bm{\Sigma}_{\varepsilon}\right\Vert _{2}}\sqrt{N+\log n}+\frac{1}{\sqrt{T}}\sqrt{\left\Vert \bm{\Sigma}_{\varepsilon}\right\Vert _{2}}\sqrt{n}\rho\overset{\text{(iii)}}{\lesssim}\frac{1}{\sqrt{T}}\sqrt{\left\Vert \bm{\Sigma}_{\varepsilon}\right\Vert _{2}}\sqrt{n},
\end{align*}
where (i) uses (G.3) in Lemma 19 of \citet{yan2024entrywise} and
the fact that $\mathrm{rank}(\bm{U}_{\perp})=N-r$, (ii) follows from
(\ref{noise E norm inequality}) and $\big\Vert\widehat{\bm{V}}\bm{R}_{V}-\bm{V}\big\Vert_{2}\lesssim\rho$,
and (iii) is because $\rho\ll1$ as we assumed. So we conclude that
\[
\big\Vert\widehat{\bm{\Sigma}}(\bm{H}_{U}-\bm{R}_{U})\big\Vert_{2}\lesssim\rho\big\Vert\widehat{\bm{\Sigma}}\widehat{\bm{U}}\bm{U}_{\perp}\big\Vert_{2}\lesssim\rho\frac{1}{\sqrt{T}}\sqrt{\left\Vert \bm{\Sigma}_{\varepsilon}\right\Vert _{2}}\sqrt{n}=\sigma_{r}\rho^{2}.
\]
The results for $\big\Vert\widehat{\bm{\Sigma}}(\bm{H}_{V}-\bm{R}_{V})\big\Vert_{2}$
can be proven in the same manner. 
\end{proof}
Then we establish the close relation between $\widehat{\bm{\Sigma}}$
and $\bm{\Lambda}$. In particular, when we prove all of the following
results, we will use the results that have been proven to hold with
probabilities at least $1-O(n^{-2})$ in the previous lemmas. When
we use these previous results, it is equivalent to add a new event
in $\mathcal{E}_{Z}\cap\mathcal{E}_{0}$ and then repeat the argument
in (\ref{condition on E_Z and F good}). Since all of our results
are to be proven to hold with probabilities at least $1-O(n^{-2})$,
the argument in (\ref{condition on E_Z and F good}) always works
and in what follows we will use it many times without mentioning it.

\begin{lemma} \label{Lemma Sigma hat tilde H R}Suppose that (\ref{SNR 2-norm})
and Assumptions \ref{Assump_Bf_identification}, \ref{Assump_noise_Z_entries},
and \ref{Assump_factor_f} hold. We have that, with probability at
least $1-O(n^{-2})$, 
\begin{align*}
\big\Vert(\bm{H}_{U})^{\top}\widehat{\bm{\Sigma}}\bm{H}_{V}-\bm{\Lambda}\big\Vert_{2} & \lesssim\sigma_{r}\rho^{3}+\sigma_{r}\rho\sqrt{\frac{1}{n}(r+\log n)},\\
\big\Vert(\bm{R}_{U})^{\top}\widehat{\bm{\Sigma}}\bm{R}_{V}-\bm{\Lambda}\big\Vert_{2} & \lesssim\sigma_{r}\rho^{2}+\sigma_{r}\rho\sqrt{\frac{1}{n}(r+\log n)}.
\end{align*}

\end{lemma} 
\begin{proof}
Note that 
\[
\big\Vert(\bm{H}_{U})^{\top}\widehat{\bm{\Sigma}}\bm{H}_{V}-\bm{\Lambda}\big\Vert_{2}\leq\underset{=:\varepsilon_{1}}{\underbrace{\big\Vert(\bm{H}_{U})^{\top}\widehat{\bm{\Sigma}}\bm{H}_{V}-\bm{U}^{\top}(T^{-1/2}\bm{X})\bm{V}\big\Vert_{2}}}+\underset{=:\varepsilon_{2}}{\underbrace{\big\Vert\bm{U}^{\top}(T^{-1/2}\bm{X})\bm{V}-\bm{\Lambda}\big\Vert_{2}}}.
\]

For $\varepsilon_{1}$, using the full SVD of $T^{-1/2}\bm{X}$ in
the proof of Lemma \ref{Lemma R H for U V}, we have that 
\[
(\bm{H}_{U})^{\top}\widehat{\bm{\Sigma}}\bm{H}_{V}-\bm{U}^{\top}(T^{-1/2}\bm{X})\bm{V}=\bm{U}^{\top}\widehat{\bm{U}}\widehat{\bm{\Sigma}}\widehat{\bm{V}}^{\top}\bm{V}-\bm{U}^{\top}(T^{-1/2}\bm{X})\bm{V}=-\bm{U}^{\top}\widehat{\bm{U}}_{\perp}\widehat{\bm{\Sigma}}_{\perp}(\widehat{\bm{V}}_{\perp})^{\top}\bm{V}.
\]
By (\ref{noise E norm inequality}) in Lemma \ref{Lemma Z norms A B}
and Weyl's inequality \citep[Lemma 2.2]{chen2021Monograph}, we have
\[
\big\Vert\widehat{\bm{\Sigma}}_{\perp}\big\Vert_{2}\leq\sigma_{r+1}(\frac{1}{\sqrt{T}}\bm{BF}^{\top})+\frac{1}{\sqrt{T}}\left\Vert \bm{E}\right\Vert _{2}=\frac{1}{\sqrt{T}}\left\Vert \bm{E}\right\Vert _{2}\lesssim\frac{1}{\sqrt{T}}\sqrt{\left\Vert \bm{\Sigma}_{\varepsilon}\right\Vert _{2}}\sqrt{n}.
\]
By Lemmas 2.5-2.6 in \citet{chen2021Monograph} and Lemma \ref{Lemma R H for U V},
we have that $\big\Vert\bm{U}^{\top}\widehat{\bm{U}}_{\perp}\big\Vert_{2}\lesssim\big\Vert\widehat{\bm{U}}\bm{R}_{U}-\bm{U}\big\Vert_{2}\lesssim\rho$
and $\big\Vert\bm{V}^{\top}\widehat{\bm{V}}_{\perp}\big\Vert_{2}\lesssim\big\Vert\widehat{\bm{V}}\bm{R}_{V}-\bm{V}\big\Vert_{2}\lesssim\rho$.
So we obtain 
\[
\varepsilon_{1}=\big\Vert\bm{U}^{\top}\widehat{\bm{U}}_{\perp}\widehat{\bm{\Sigma}}_{\perp}(\widehat{\bm{V}}_{\perp})^{\top}\bm{V}\big\Vert_{2}\leq\big\Vert\bm{U}^{\top}\widehat{\bm{U}}_{\perp}\big\Vert_{2}\big\Vert\widehat{\bm{\Sigma}}_{\perp}\big\Vert_{2}\big\Vert(\widehat{\bm{V}}_{\perp})^{\top}\bm{V}\big\Vert_{2}\lesssim\sigma_{r}\rho^{3}.
\]

For $\varepsilon_{2}$, we start by writting 
\[
\bm{U}^{\top}(T^{-1/2}\bm{X})\bm{V}-\bm{\Lambda}=\bm{U}^{\top}(T^{-1/2}\bm{X})\bm{V}-\bm{U}^{\top}(T^{-1/2}\bm{BF}^{\top})\bm{V}=\bm{U}^{\top}(T^{-1/2}\bm{E})\bm{V}=T^{-1/2}\bm{U}^{\top}\bm{\Sigma}_{\varepsilon}^{1/2}\bm{Z}\bm{V}.
\]
Using (G.3) in Lemma 19 of \citet{yan2024entrywise}, we obtain that
\[
\varepsilon_{2}\lesssim\frac{1}{\sqrt{T}}\big\Vert\bm{U}^{\top}\bm{\Sigma}_{\varepsilon}^{1/2}\big\Vert_{2}\big\Vert\bm{V}\big\Vert_{2}\sqrt{r+r+\log n}\lesssim\frac{1}{\sqrt{T}}\sqrt{\left\Vert \bm{\Sigma}_{\varepsilon}\right\Vert _{2}}\sqrt{r+\log n}=\sigma_{r}\rho\sqrt{\frac{1}{n}(r+\log n)}.
\]
Then, we conclude that 
\[
\big\Vert(\bm{H}_{U})^{\top}\widehat{\bm{\Sigma}}\bm{H}_{V}-\bm{\Lambda}\big\Vert_{2}\leq\varepsilon_{1}+\varepsilon_{2}\lesssim\sigma_{r}\rho^{3}+\sigma_{r}\rho\sqrt{\frac{1}{n}(r+\log n)}.
\]

For $(\bm{R}_{U})^{\top}\widehat{\bm{\Sigma}}\bm{R}_{V}-\bm{\Lambda}$,
we note that 
\begin{align*}
\big\Vert(\bm{R}_{U})^{\top}\widehat{\bm{\Sigma}}\bm{R}_{V}-(\bm{H}_{U})^{\top}\widehat{\bm{\Sigma}}\bm{H}_{V}\big\Vert_{2} & \leq\big\Vert(\bm{R}_{U})^{\top}\widehat{\bm{\Sigma}}\left(\bm{R}_{V}-\bm{H}_{V}\right)\big\Vert_{2}+\big\Vert(\bm{R}_{U}-\bm{H}_{U})^{\top}\widehat{\bm{\Sigma}}\bm{H}_{V}\big\Vert_{2}\\
 & \leq\left\Vert \bm{R}_{U}\right\Vert _{2}\big\Vert\widehat{\bm{\Sigma}}(\bm{H}_{V}-\bm{R}_{V})\big\Vert_{2}+\big\Vert\widehat{\bm{\Sigma}}(\bm{H}_{U}-\bm{R}_{U})\big\Vert_{2}\left\Vert \bm{H}_{V}\right\Vert _{2}\\
 & \lesssim\sigma_{r}\rho^{2},
\end{align*}
where the last inequality is owing to Lemma \ref{Lemma R H for U V}.
So we obtain 
\begin{align*}
\big\Vert(\bm{R}_{U})^{\top}\widehat{\bm{\Sigma}}\bm{R}_{V}-\bm{\Lambda}\big\Vert_{2} & \leq\big\Vert(\bm{R}_{U})^{\top}\widehat{\bm{\Sigma}}\bm{R}_{V}-(\bm{H}_{U})^{\top}\widehat{\bm{\Sigma}}\bm{H}_{V}\big\Vert_{2}+\big\Vert(\bm{H}_{U})^{\top}\widehat{\bm{\Sigma}}\bm{H}_{V}-\bm{\Lambda}\big\Vert_{2}\\
 & \lesssim\sigma_{r}\rho^{2}+\sigma_{r}\rho^{3}+\sigma_{r}\rho\sqrt{\frac{1}{n}(r+\log n)}\\
 & \lesssim\sigma_{r}\rho^{2}+\sigma_{r}\rho\sqrt{\frac{1}{n}(r+\log n)},
\end{align*}
where the last inequality is because $\rho\ll1$ in (\ref{SNR 2-norm}). 
\end{proof}
We now prove an important lemma where the leave-one-out (LOO) technique
plays a key role.

\begin{lemma} \label{Lemma LOO UV Sig H}Suppose that Assumptions
\ref{Assump_Bf_identification}, \ref{Assump_noise_Z_entries}, and
\ref{Assump_factor_f} hold. For simplicity of notations, we define
\begin{equation}
\omega_{k}:=\frac{\left\Vert (\bm{\Sigma}_{\varepsilon}^{1/2})_{k,\cdot}\right\Vert _{1}\sqrt{n}}{\sigma_{r}\sqrt{T}}=\frac{\sqrt{n}}{\vartheta_{k}\sqrt{T}}\text{\qquad for\qquad}k=1,2,\ldots,N,\label{SNR 1-norm k-th row}
\end{equation}
as well as 
\[
\omega:=\max_{1\leq k\leq N}\omega_{k}=\frac{\left\Vert \bm{\Sigma}_{\varepsilon}^{1/2}\right\Vert _{1}\sqrt{n}}{\sigma_{r}\sqrt{T}}=\frac{\sqrt{n}}{\vartheta\sqrt{T}},
\]

(i) Assuming that (\ref{SNR 2-norm logn}) holds, i.e., 
\[
\rho\sqrt{\log n}=\frac{\sqrt{\left\Vert \bm{\Sigma}_{\varepsilon}\right\Vert _{2}}\sqrt{n\log n}}{\sigma_{r}\sqrt{T}}\ll1,
\]
then we have that, for $l=1,2,\ldots,T$, with probability at least
$1-O(n^{-2})$, 
\begin{align*}
 & \big\Vert(\widehat{\bm{V}}\widehat{\bm{\Sigma}}\bm{H}_{U}-(T^{-1/2}\bm{X})^{\top}\bm{U})_{l,\cdot}\big\Vert_{2}\\
 & \lesssim\sigma_{r}\rho^{2}\left(\sqrt{\frac{r}{n}}\log^{3/2}n+\log n\big\Vert\bm{V}_{l,\cdot}\big\Vert_{2}+\log n\big\Vert(\widehat{\bm{V}}\bm{H}_{V}-\bm{V})_{l,\cdot}\big\Vert_{2}\right).
\end{align*}

(ii) Assuming that (\ref{SNR 1-norm logn}) holds, i.e., 
\[
\omega\sqrt{\log n}=\frac{\left\Vert \bm{\Sigma}_{\varepsilon}^{1/2}\right\Vert _{1}\sqrt{n\log n}}{\sigma_{r}\sqrt{T}}\ll1,
\]
then we have that, for $k=1,2,\ldots,N$, with probability at least
$1-O(n^{-2})$, 
\begin{align*}
 & \big\Vert(\widehat{\bm{U}}\widehat{\bm{\Sigma}}\bm{H}_{V}-(T^{-1/2}\bm{X})\bm{V})_{k,\cdot}\big\Vert_{2}\\
 & \lesssim\sigma_{r}\left(\omega_{k}\rho\sqrt{\frac{r}{n}}\log^{3/2}n+\rho^{2}\big\Vert\bm{U}_{k,\cdot}\big\Vert_{2}+\omega_{k}\omega\log n\left(\big\Vert\bm{U}\big\Vert_{2,\infty}+\big\Vert\widehat{\bm{U}}\bm{H}_{U}-\bm{U}\big\Vert_{2,\infty}\right)\right),
\end{align*}
where $\omega_{k}=\left\Vert (\bm{\Sigma}_{\varepsilon}^{1/2})_{k,\cdot}\right\Vert _{1}\sqrt{n}/(\sigma_{r}\sqrt{T})$
is defined in (\ref{SNR 1-norm k-th row}).

\end{lemma} 
\begin{proof}
The proof consists of four steps. In the first (resp, third) step,
we prove the upper bound for $\big\Vert(\widehat{\bm{V}}\widehat{\bm{\Sigma}}\bm{H}_{U}-(T^{-1/2}\bm{X})^{\top}\bm{U})_{l,\cdot}\big\Vert_{2}$
(resp. $\big\Vert(\widehat{\bm{U}}\widehat{\bm{\Sigma}}\bm{H}_{V}-(T^{-1/2}\bm{X})\bm{V})_{k,\cdot}\big\Vert_{2}$),
and in the second (resp. fourth) step, we add the details of the proof
for an important inequality used in the proof for $\big\Vert(\widehat{\bm{V}}\widehat{\bm{\Sigma}}\bm{H}_{U}-(T^{-1/2}\bm{X})^{\top}\bm{U})_{l,\cdot}\big\Vert_{2}$
(resp. $\big\Vert(\widehat{\bm{U}}\widehat{\bm{\Sigma}}\bm{H}_{V}-(T^{-1/2}\bm{X})\bm{V})_{k,\cdot}\big\Vert_{2}$).

\textit{Step 1} --
By calculations, we have that, for $l=1,2,\ldots,T$, 
\[
(\widehat{\bm{V}}\widehat{\bm{\Sigma}}\bm{H}_{U}-(T^{-1/2}\bm{X})^{\top}\bm{U})_{l,\cdot}=\underset{=:\alpha_{1}}{\underbrace{\bm{V}_{l,\cdot}\bm{\Lambda}\bm{U}^{\top}(\widehat{\bm{U}}\bm{H}_{U}-\bm{U})}}+\underset{=:\alpha_{2}}{\underbrace{T^{-1/2}(\bm{Z}^{\top})_{l,\cdot}\bm{\Sigma}_{\varepsilon}^{1/2}(\widehat{\bm{U}}\bm{H}_{U}-\bm{U})}}.
\]

For $\alpha_{1}$, using the full SVD of $T^{-1/2}\bm{X}$ in the
proof of Lemma \ref{Lemma R H for U V}, we have that, similar to
Step 1 in Section F.2.3 of \citet{yan2024entrywise}, it holds 
\[
\bm{\Lambda}\bm{U}^{\top}(\widehat{\bm{U}}\bm{H}_{U}-\bm{U})=\bm{V}^{\top}\widehat{\bm{V}}_{\perp}(\widehat{\bm{\Sigma}}_{\perp})^{\top}(\widehat{\bm{U}}_{\perp})^{\top}\bm{U}-T^{-1/2}\bm{V}^{\top}\bm{E}^{\top}(\widehat{\bm{U}}\widehat{\bm{U}}^{\top}-\bm{I}_{N})\bm{U}.
\]
Using (\ref{U U_hat_cmpl bound}), the spectral norm of the first
matrix can be bounded by 
\[
\big\Vert\bm{V}^{\top}\widehat{\bm{V}}_{\perp}(\widehat{\bm{\Sigma}}_{\perp})^{\top}(\widehat{\bm{U}}_{\perp})^{\top}\bm{U}\big\Vert_{2}\leq\big\Vert\bm{V}^{\top}\widehat{\bm{V}}_{\perp}\big\Vert_{2}\big\Vert\widehat{\bm{\Sigma}}_{\perp}\big\Vert_{2}\big\Vert\bm{U}^{\top}\widehat{\bm{U}}_{\perp}\big\Vert_{2}\lesssim\sigma_{r}\rho^{3},
\]
where we use the fact that $\big\Vert\widehat{\bm{\Sigma}}_{\perp}\big\Vert_{2}=\hat{\sigma}_{r}\lesssim T^{-1/2}\left\Vert \bm{E}\right\Vert _{2}\lesssim\sigma_{r}\rho$.
The spectral norm of the second matrix can be bounded by 
\begin{align*}
\big\Vert T^{-1/2}\bm{V}^{\top}\bm{E}^{\top}(\widehat{\bm{U}}\widehat{\bm{U}}^{\top}-\bm{I}_{N})\bm{U}\big\Vert_{2} & \leq\frac{1}{\sqrt{T}}\big\Vert\bm{E}\bm{V}\big\Vert_{2}\big\Vert(\widehat{\bm{U}}\widehat{\bm{U}}^{\top}-\bm{I}_{N})\bm{U}\big\Vert_{2}=\frac{1}{\sqrt{T}}\big\Vert\bm{\Sigma}_{\varepsilon}^{1/2}\bm{Z}\bm{V}\big\Vert_{2}\big\Vert\widehat{\bm{U}}\bm{H}_{U}-\bm{U}\big\Vert_{2}\\
 & \overset{\text{(i)}}{\lesssim}\frac{1}{\sqrt{T}}\sqrt{\left\Vert \bm{\Sigma}_{\varepsilon}\right\Vert _{2}}\sqrt{N+\log n}\cdot\rho\lesssim\sigma_{r}\rho^{2},
\end{align*}
where (i) uses (G.3) in Lemma 19 of \citet{yan2024entrywise} and
the fact that $\widehat{\bm{U}}\bm{H}_{U}-\bm{U}=\widehat{\bm{U}}\widehat{\bm{U}}^{\top}\bm{U}-\bm{U}=(\widehat{\bm{U}}\widehat{\bm{U}}^{\top}-\bm{U}\bm{U}^{\top})\bm{U}$
implying $\big\Vert\widehat{\bm{U}}\bm{H}_{U}-\bm{U}\big\Vert_{2}\leq\big\Vert\widehat{\bm{U}}\widehat{\bm{U}}^{\top}-\bm{U}\bm{U}^{\top}\big\Vert_{2}\lesssim\rho$
according to (\ref{U U_hat_cmpl bound}). So, we obtain that 
\begin{equation}
\left\Vert \alpha_{1}\right\Vert _{2}\leq\big\Vert\bm{V}_{l,\cdot}\big\Vert_{2}(\sigma_{r}\rho^{3}+\sigma_{r}\rho^{2})\lesssim\sigma_{r}\rho^{2}\big\Vert\bm{V}_{l,\cdot}\big\Vert_{2}.\label{1st term for V Sig H_U}
\end{equation}

For $\alpha_{2}$, we use the leave-one-out (LOO) technique. We define
$\bm{Z}^{\{l\}}$ as the matrix obtained by replacing the $l$th column
of $\bm{Z}$ with a zero vector, and define $\bm{X}^{\{l\}}:=\bm{BF}^{\top}+\bm{\Sigma}_{\varepsilon}^{1/2}\bm{Z}^{\{l\}}$.
Then, we define $\widehat{\bm{U}}^{\{l\}},$ $\widehat{\bm{V}}^{\{l\}},$
$\bm{H}_{U}^{\{l\}},$ and $\bm{H}_{V}^{\{l\}}$ w.r.t. $T^{-1/2}\bm{X}^{\{l\}}$,
in the same way as how we define $\widehat{\bm{U}},$ $\widehat{\bm{V}},$
$\bm{H}_{U},$ and $\bm{H}_{V}$ w.r.t. $T^{-1/2}\bm{X}$. We have
that 
\begin{equation}
\alpha_{2}=\underset{=:\gamma_{1}}{\underbrace{T^{-1/2}(\bm{Z}^{\top})_{l,\cdot}\bm{\Sigma}_{\varepsilon}^{1/2}(\widehat{\bm{U}}^{\{l\}}\bm{H}_{U}^{\{l\}}-\bm{U})}}+\underset{=:\gamma_{2}}{\underbrace{T^{-1/2}(\bm{Z}^{\top})_{l,\cdot}\bm{\Sigma}_{\varepsilon}^{1/2}(\widehat{\bm{U}}\bm{H}_{U}-\widehat{\bm{U}}^{\{l\}}\bm{H}_{U}^{\{l\}})}}.\label{2nd term for V Sig H_U LOO decompose}
\end{equation}

For $\gamma_{1}$, since $(\widehat{\bm{U}}^{\{l\}}\bm{H}_{U}^{\{l\}}-\bm{U})$
is independent with $(\bm{Z}^{\top})_{l,\cdot}$, we apply Lemma \ref{Lemma Z norms A B}
to obtain that $\left\Vert \gamma_{1}\right\Vert _{2}\lesssim\frac{1}{\sqrt{T}}\big\Vert\bm{\Sigma}_{\varepsilon}^{1/2}(\widehat{\bm{U}}^{\{l\}}\bm{H}_{U}^{\{l\}}-\bm{U})\big\Vert_{\mathrm{F}}\log n$.
Note that 
\begin{align*}
\big\Vert\bm{\Sigma}_{\varepsilon}^{1/2}(\widehat{\bm{U}}^{\{l\}}\bm{H}_{U}^{\{l\}}-\bm{U})\big\Vert_{\mathrm{F}} & \leq\big\Vert\bm{\Sigma}_{\varepsilon}^{1/2}\big\Vert_{2}\big\Vert\widehat{\bm{U}}^{\{l\}}\bm{H}_{U}^{\{l\}}-\bm{U}\big\Vert_{\mathrm{F}}\\
 & \leq\sqrt{\left\Vert \bm{\Sigma}_{\varepsilon}\right\Vert _{2}}\left(\big\Vert\widehat{\bm{U}}\bm{H}_{U}-\bm{U}\big\Vert_{\mathrm{F}}+\big\Vert\widehat{\bm{U}}^{\{l\}}\bm{H}_{U}^{\{l\}}-\widehat{\bm{U}}\bm{H}_{U}\big\Vert_{\mathrm{F}}\right).
\end{align*}
Then since $\widehat{\bm{U}}\bm{H}_{U}-\bm{U}=\widehat{\bm{U}}\widehat{\bm{U}}^{\top}\bm{U}-\bm{U}=(\widehat{\bm{U}}\widehat{\bm{U}}^{\top}-\bm{U}\bm{U}^{\top})\bm{U}$,
we obtain that 
\begin{equation}
\big\Vert\widehat{\bm{U}}\bm{H}_{U}-\bm{U}\big\Vert_{\mathrm{F}}\leq\big\Vert\widehat{\bm{U}}\widehat{\bm{U}}^{\top}-\bm{U}\bm{U}^{\top}\big\Vert_{\mathrm{F}}\lesssim\big\Vert\widehat{\bm{U}}\bm{R}_{U}-\bm{U}\big\Vert_{\mathrm{F}}\lesssim\rho\sqrt{r},\label{U H_U - U F-norm}
\end{equation}
where the last inequality is from Lemma \ref{Lemma R H for U V}.
So, we conclude that 
\begin{equation}
\left\Vert \gamma_{1}\right\Vert _{2}\lesssim\frac{1}{\sqrt{T}}\sqrt{\left\Vert \bm{\Sigma}_{\varepsilon}\right\Vert _{2}}\left(\rho\sqrt{r}+\left\Vert \mu_{\{l\}}\right\Vert _{\mathrm{F}}\right)\log n\text{\qquad with\qquad}\mu_{\{l\}}:=\widehat{\bm{U}}^{\{l\}}\bm{H}_{U}^{\{l\}}-\widehat{\bm{U}}\bm{H}_{U}.\label{1st term in LOO V}
\end{equation}

For $\gamma_{2}$, using the event $\mathcal{E}_{Z}$ defined in (\ref{Prob E_Z entry logn bound}),
we have that 
\begin{equation}
\left\Vert \gamma_{2}\right\Vert _{2}\lesssim\frac{1}{\sqrt{T}}\sqrt{\left\Vert \bm{\Sigma}_{\varepsilon}\right\Vert _{2}}\sqrt{n\log n}\left\Vert \mu_{\{l\}}\right\Vert _{2}\leq\frac{1}{\sqrt{T}}\sqrt{\left\Vert \bm{\Sigma}_{\varepsilon}\right\Vert _{2}n\log n}\left\Vert \mu_{\{l\}}\right\Vert _{\mathrm{F}},\label{2nd term in LOO V}
\end{equation}
with probability at least $1-O(n^{-2})$. Combining the bounds for
$\gamma_{1}$ and $\gamma_{2}$, we obtain that 
\begin{equation}
\left\Vert \alpha_{2}\right\Vert _{2}\leq\left\Vert \gamma_{1}\right\Vert _{2}+\left\Vert \gamma_{2}\right\Vert _{2}\lesssim\frac{1}{\sqrt{T}}\sqrt{\left\Vert \bm{\Sigma}_{\varepsilon}\right\Vert _{2}}\rho\sqrt{r}\log n+\frac{1}{\sqrt{T}}\sqrt{\left\Vert \bm{\Sigma}_{\varepsilon}\right\Vert _{2}n\log n}\left\Vert \mu_{\{l\}}\right\Vert _{\mathrm{F}}.\label{2nd term for V Sig H_U}
\end{equation}

We will use self-bounding technique to prove in \textit{Step 2} that
\[
\left\Vert \mu_{\{l\}}\right\Vert _{\mathrm{F}}\lesssim\rho\sqrt{\frac{r}{n}}\log n+\rho\sqrt{\log n}\left(\big\Vert\bm{V}_{l,\cdot}\big\Vert_{2}+\big\Vert(\widehat{\bm{V}}\bm{H}_{V}-\bm{V})_{l,\cdot}\big\Vert_{2}\right).
\]
As long as we have this, pluging the above upper bound of $\left\Vert \mu_{\{l\}}\right\Vert _{\mathrm{F}}$
into (\ref{2nd term for V Sig H_U}), and combining (\ref{1st term for V Sig H_U})
and (\ref{2nd term for V Sig H_U}), we obtain the desired bound for
$(\widehat{\bm{V}}\widehat{\bm{\Sigma}}\bm{H}_{U}-(T^{-1/2}\bm{X})^{\top}\bm{U})_{l,\cdot}$
as below, 
\begin{align*}
 & \big\Vert(\widehat{\bm{V}}\widehat{\bm{\Sigma}}\bm{H}_{U}-(T^{-1/2}\bm{X})^{\top}\bm{U})_{l,\cdot}\big\Vert_{2}\\
 & \leq\left\Vert \alpha_{1}\right\Vert _{2}+\left\Vert \alpha_{2}\right\Vert _{2}\\
 & \lesssim\sigma_{r}\rho^{2}\big\Vert\bm{V}_{l,\cdot}\big\Vert_{2}+\frac{1}{\sqrt{T}}\sqrt{\left\Vert \bm{\Sigma}_{\varepsilon}\right\Vert _{2}}\rho\sqrt{r}\log n\\
 & +\frac{1}{\sqrt{T}}\sqrt{\left\Vert \bm{\Sigma}_{\varepsilon}\right\Vert _{2}n\log n}\left[\rho\sqrt{\frac{r}{n}}\log n+\rho\sqrt{\log n}\left(\big\Vert\bm{V}_{l,\cdot}\big\Vert_{2}+\big\Vert(\widehat{\bm{V}}\bm{H}_{V}-\bm{V})_{l,\cdot}\big\Vert_{2}\right)\right]\\
 & \lesssim\sigma_{r}\rho^{2}\left(\sqrt{\frac{r}{n}}\log^{3/2}n+\log n\big\Vert\bm{V}_{l,\cdot}\big\Vert_{2}+\log n\big\Vert(\widehat{\bm{V}}\bm{H}_{V}-\bm{V})_{l,\cdot}\big\Vert_{2}\right).
\end{align*}

\textit{Step 2} --
Recall that $\mu_{\{l\}}=\widehat{\bm{U}}^{\{l\}}\bm{H}_{U}^{\{l\}}-\widehat{\bm{U}}\bm{H}_{U}$
as defined in (\ref{1st term in LOO V}). We define 
\[
\nu_{\{l\}}:=\widehat{\bm{V}}^{\{l\}}\bm{H}_{V}^{\{l\}}-\widehat{\bm{V}}\bm{H}_{V}.
\]
By definition, we have that $\mu_{\{l\}}=(\widehat{\bm{U}}^{\{l\}}(\widehat{\bm{U}}^{\{l\}})^{\top}-\widehat{\bm{U}}\widehat{\bm{U}}^{\top})\bm{U}$,
so $\left\Vert \mu_{\{l\}}\right\Vert _{\mathrm{F}}\leq\big\Vert\widehat{\bm{U}}^{\{l\}}(\widehat{\bm{U}}^{\{l\}})^{\top}-\widehat{\bm{U}}\widehat{\bm{U}}^{\top}\big\Vert_{\mathrm{F}}$,
and similarly we have $\left\Vert \nu_{\{l\}}\right\Vert _{\mathrm{F}}\leq\big\Vert\widehat{\bm{V}}^{\{l\}}(\widehat{\bm{V}}^{\{l\}})^{\top}-\widehat{\bm{V}}\widehat{\bm{V}}^{\top}\big\Vert_{\mathrm{F}}$.
Denote by $E_{l}:=\frac{1}{\sqrt{T}}\bm{X}^{\{l\}}-\frac{1}{\sqrt{T}}\bm{X}=\frac{1}{\sqrt{T}}\bm{\Sigma}_{\varepsilon}^{1/2}(\bm{Z}^{\{l\}}-\bm{Z})$.
Then using (\ref{SNR 2-norm}) and our argument in proof of Lemma
\ref{Lemma R H for U V}, we have that $\left\Vert E_{l}\right\Vert \ll\sigma_{r}$,
and $\sigma_{i}(\frac{1}{\sqrt{T}}\bm{X}^{\{l\}})\asymp\sigma_{i}$
for $i=1,2,\ldots,r$, and $\sigma_{r+1}(\frac{1}{\sqrt{T}}\bm{X}^{\{l\}})\ll\sigma_{r}$.
By Wedin's $\sin\Theta$ theorem \citep[Theorem 2.9]{chen2021Monograph},
we obtain that 
\begin{align*}
 & \max\left\{ \left\Vert \mu_{\{l\}}\right\Vert _{\mathrm{F}},\left\Vert \nu_{\{l\}}\right\Vert _{\mathrm{F}}\right\} \\
 & \leq\max\left\{ \big\Vert\widehat{\bm{U}}^{\{l\}}(\widehat{\bm{U}}^{\{l\}})^{\top}-\widehat{\bm{U}}\widehat{\bm{U}}^{\top}\big\Vert_{\mathrm{F}},\big\Vert\widehat{\bm{V}}^{\{l\}}(\widehat{\bm{V}}^{\{l\}})^{\top}-\widehat{\bm{V}}\widehat{\bm{V}}^{\top}\big\Vert_{\mathrm{F}}\right\} \\
 & \lesssim\frac{1}{\sigma_{r}(\frac{1}{\sqrt{T}}\bm{X}^{\{l\}})-\sigma_{r+1}(\frac{1}{\sqrt{T}}\bm{X}^{\{l\}})-\left\Vert E_{l}\right\Vert _{2}}\max\left\{ \big\Vert E_{l}\widehat{\bm{V}}^{\{l\}}\big\Vert_{\mathrm{F}},\big\Vert(E_{l})^{\top}\widehat{\bm{U}}^{\{l\}}\big\Vert_{\mathrm{F}}\right\} \\
 & \lesssim\frac{1}{\sigma_{r}}\max\left\{ \big\Vert E_{l}\widehat{\bm{V}}^{\{l\}}\big\Vert_{\mathrm{F}},\big\Vert(E_{l})^{\top}\widehat{\bm{U}}^{\{l\}}\big\Vert_{\mathrm{F}}\right\} .
\end{align*}

Next, by calculations we have that 
\begin{align*}
\big\Vert(E_{l})^{\top}\widehat{\bm{U}}^{\{l\}}\big\Vert_{\mathrm{F}} & =\left\Vert \frac{1}{\sqrt{T}}(\bm{Z}^{\{l\}}-\bm{Z})^{\top}\bm{\Sigma}_{\varepsilon}^{1/2}\widehat{\bm{U}}^{\{l\}}\right\Vert _{\mathrm{F}}\\
 & =\frac{1}{\sqrt{T}}\left\Vert (\bm{Z}_{\cdot,l})^{\top}\bm{\Sigma}_{\varepsilon}^{1/2}\widehat{\bm{U}}^{\{l\}}\right\Vert _{2}\overset{\text{(i)}}{\lesssim}\frac{1}{\sqrt{T}}\left\Vert (\bm{Z}_{\cdot,l})^{\top}\bm{\Sigma}_{\varepsilon}^{1/2}\widehat{\bm{U}}^{\{l\}}\bm{H}_{U}^{\{l\}}\right\Vert _{2}\\
 & \leq\frac{1}{\sqrt{T}}\left\Vert (\bm{Z}_{\cdot,l})^{\top}\bm{\Sigma}_{\varepsilon}^{1/2}\bm{U}\right\Vert _{2}+\frac{1}{\sqrt{T}}\left\Vert (\bm{Z}_{\cdot,l})^{\top}\bm{\Sigma}_{\varepsilon}^{1/2}(\widehat{\bm{U}}^{\{l\}}\bm{H}_{U}^{\{l\}}-\bm{U})\right\Vert _{2}\\
 & \lesssim\frac{1}{\sqrt{T}}\sqrt{\left\Vert \bm{\Sigma}_{\varepsilon}\right\Vert _{2}r}\log n+\frac{1}{\sqrt{T}}\sqrt{\left\Vert \bm{\Sigma}_{\varepsilon}\right\Vert _{2}}\left(\rho\sqrt{r}+\left\Vert \mu_{\{l\}}\right\Vert _{\mathrm{F}}\right)\log n.
\end{align*}
Here, (i) is because $\sigma_{\min}(\bm{H}_{U}^{\{l\}})\asymp1$,
which can be proven in the same manner as we prove $\sigma_{i}(\bm{H}_{U})\asymp1$
in Lemma \ref{Lemma R H for U V}; the last inequality uses the bound
(\ref{1st term in LOO V}) for $\gamma_{1}$ defined in (\ref{2nd term for V Sig H_U LOO decompose})
and the fact that, by Lemma \ref{Lemma Z norms A B}, 
\[
\big\Vert(\bm{Z}_{\cdot,l})^{\top}\bm{\Sigma}_{\varepsilon}^{1/2}\bm{U}\big\Vert_{2}\lesssim\big\Vert\bm{\Sigma}_{\varepsilon}^{1/2}\bm{U}\big\Vert_{\mathrm{F}}\log n\leq\big\Vert\bm{\Sigma}_{\varepsilon}^{1/2}\big\Vert_{2}\big\Vert\bm{U}\big\Vert_{\mathrm{F}}\log n=\sqrt{\left\Vert \bm{\Sigma}_{\varepsilon}\right\Vert _{2}r}\log n.
\]

Also, we have that 
\begin{align*}
\big\Vert E_{l}\widehat{\bm{V}}^{\{l\}}\big\Vert_{\mathrm{F}} & =\left\Vert \frac{1}{\sqrt{T}}\bm{\Sigma}_{\varepsilon}^{1/2}(\bm{Z}^{\{l\}}-\bm{Z})\widehat{\bm{V}}^{\{l\}}\right\Vert _{\mathrm{F}}\leq\frac{1}{\sqrt{T}}\left\Vert \bm{\Sigma}_{\varepsilon}^{1/2}\right\Vert _{2}\left\Vert (\bm{Z}^{\{l\}}-\bm{Z})\widehat{\bm{V}}^{\{l\}}\right\Vert _{\mathrm{F}}\\
 & =\frac{1}{\sqrt{T}}\sqrt{\left\Vert \bm{\Sigma}_{\varepsilon}\right\Vert _{2}}\left\Vert \bm{Z}_{\cdot,l}\right\Vert _{2}\left\Vert (\widehat{\bm{V}}^{\{l\}})_{l,\cdot}\right\Vert _{2}\overset{\text{(i)}}{\lesssim}\frac{1}{\sqrt{T}}\sqrt{\left\Vert \bm{\Sigma}_{\varepsilon}\right\Vert _{2}}\left\Vert \bm{Z}_{\cdot,l}\right\Vert _{2}\left\Vert (\widehat{\bm{V}}^{\{l\}})_{l,\cdot}\bm{H}_{V}^{\{l\}}\right\Vert _{2}\\
 & \leq\frac{1}{\sqrt{T}}\sqrt{\left\Vert \bm{\Sigma}_{\varepsilon}\right\Vert _{2}}\left\Vert \bm{Z}_{\cdot,l}\right\Vert _{2}\left(\big\Vert\bm{V}_{l,\cdot}\big\Vert_{2}+\big\Vert(\widehat{\bm{V}}\bm{H}_{V}-\bm{V})_{l,\cdot}\big\Vert_{2}+\left\Vert (\widehat{\bm{V}}^{\{l\}}\bm{H}_{V}^{\{l\}}-\widehat{\bm{V}}\bm{H}_{V})_{l,\cdot}\right\Vert _{2}\right)\\
 & \leq\frac{1}{\sqrt{T}}\sqrt{\left\Vert \bm{\Sigma}_{\varepsilon}\right\Vert _{2}n\log n}\left(\big\Vert\bm{V}_{l,\cdot}\big\Vert_{2}+\big\Vert(\widehat{\bm{V}}\bm{H}_{V}-\bm{V})_{l,\cdot}\big\Vert_{2}+\left\Vert \nu_{\{l\}}\right\Vert _{\mathrm{F}}\right).
\end{align*}
Here (i) is because $\sigma_{\min}(\bm{H}_{V}^{\{l\}})\asymp1$, which
can be proven in the same manner as we prove $\sigma_{i}(\bm{H}_{V})\asymp1$
in Lemma \ref{Lemma R H for U V}; the last inequality uses the event
$\mathcal{E}_{Z}$ defined in (\ref{Prob E_Z entry logn bound}).

Combining the above bounds for $\Vert(E_{l})^{\top}\widehat{\bm{U}}^{\{l\}}\Vert_{\mathrm{F}}$
and $\Vert E_{l}\widehat{\bm{V}}^{\{l\}}\Vert_{\mathrm{F}}$, and
using $\rho\ll1$, we obtain 
\begin{align*}
 & \max\left\{ \left\Vert \mu_{\{l\}}\right\Vert _{\mathrm{F}},\left\Vert \nu_{\{l\}}\right\Vert _{\mathrm{F}}\right\} \\
 & \lesssim\frac{1}{\sigma_{r}\sqrt{T}}\sqrt{\left\Vert \bm{\Sigma}_{\varepsilon}\right\Vert _{2}r}\log n+\frac{1}{\sigma_{r}\sqrt{T}}\sqrt{\left\Vert \bm{\Sigma}_{\varepsilon}\right\Vert _{2}}\left(\rho\sqrt{r}+\left\Vert \mu_{\{l\}}\right\Vert _{\mathrm{F}}\right)\log n\\
 & +\frac{1}{\sigma_{r}\sqrt{T}}\sqrt{\left\Vert \bm{\Sigma}_{\varepsilon}\right\Vert _{2}n\log n}\left(\big\Vert\bm{V}_{l,\cdot}\big\Vert_{2}+\big\Vert(\widehat{\bm{V}}\bm{H}_{V}-\bm{V})_{l,\cdot}\big\Vert_{2}+\left\Vert \nu_{\{l\}}\right\Vert _{\mathrm{F}}\right)\\
 & \lesssim\rho\sqrt{\frac{r}{n}}\log n+\rho\sqrt{\log n}\left(\big\Vert\bm{V}_{l,\cdot}\big\Vert_{2}+\big\Vert(\widehat{\bm{V}}\bm{H}_{V}-\bm{V})_{l,\cdot}\big\Vert_{2}\right)+\rho\sqrt{\log n}\max\left\{ \left\Vert \mu_{\{l\}}\right\Vert _{\mathrm{F}},\left\Vert \nu_{\{l\}}\right\Vert _{\mathrm{F}}\right\} .
\end{align*}
Note that both sides of the above inequality include the term $\max\left\{ \left\Vert \mu_{\{l\}}\right\Vert _{\mathrm{F}},\left\Vert \nu_{\{l\}}\right\Vert _{\mathrm{F}}\right\} $.
Then, since $\rho\sqrt{\log n}\ll1$ according to (\ref{SNR 2-norm logn}),
we obtain that 
\[
\max\left\{ \left\Vert \mu_{\{l\}}\right\Vert _{\mathrm{F}},\left\Vert \nu_{\{l\}}\right\Vert _{\mathrm{F}}\right\} \lesssim\rho\sqrt{\frac{r}{n}}\log n+\rho\sqrt{\log n}\left(\big\Vert\bm{V}_{l,\cdot}\big\Vert_{2}+\big\Vert(\widehat{\bm{V}}\bm{H}_{V}-\bm{V})_{l,\cdot}\big\Vert_{2}\right).
\]

\textit{Step 3} --
For $\beta_{1}$, in the same manner as we did for $\alpha_{1}$ in
(\ref{1st term for V Sig H_U}), we have that $\left\Vert \beta_{1}\right\Vert _{2}\lesssim\sigma_{r}\rho^{2}\big\Vert\bm{U}_{k,\cdot}\big\Vert_{2}$.
For $\beta_{2}$, our starting point is the following inequality 
\begin{equation}
\left\Vert \beta_{2}\right\Vert _{2}\leq\frac{1}{\sqrt{T}}\big\Vert(\bm{\Sigma}_{\varepsilon}^{1/2})_{k,\cdot}\big\Vert_{1}\big\Vert\bm{Z}(\widehat{\bm{V}}\bm{H}_{V}-\bm{V})\big\Vert_{2,\infty}.\label{2nd term for U Sig H_V inequality0}
\end{equation}
Then, the problem boils down to deriving an upper bound for $\big\Vert\bm{Z}(\widehat{\bm{V}}\bm{H}_{V}-\bm{V})\big\Vert_{2,\infty}=\max_{1\leq m\leq N}\big\Vert\bm{Z}_{m,\cdot}(\widehat{\bm{V}}\bm{H}_{V}-\bm{V})\big\Vert_{2}$.
We use the LOO technique to handle $\bm{Z}_{m,\cdot}(\widehat{\bm{V}}\bm{H}_{V}-\bm{V})$
for any $m=1,2,\ldots,N$. We define $\bm{Z}^{(m)}$ as the matrix
obtained by replacing the $m$th row of $\bm{Z}$ with a zero vector,
and define $\bm{X}^{(m)}:=\bm{BF}^{\top}+\bm{\Sigma}_{\varepsilon}^{1/2}\bm{Z}^{(m)}$.
Then, we define $\widehat{\bm{U}}^{(m)},$ $\widehat{\bm{V}}^{(m)},$
$\bm{H}_{U}^{(m)},$ and $\bm{H}_{V}^{(m)}$ w.r.t. $\frac{1}{\sqrt{T}}\bm{X}^{(m)}$,
in the same way as how we define $\widehat{\bm{U}},$ $\widehat{\bm{V}},$
$\bm{H}_{U},$ and $\bm{H}_{V}$ w.r.t. $\frac{1}{\sqrt{T}}\bm{X}$.
We have that 
\begin{equation}
\bm{Z}_{m,\cdot}(\widehat{\bm{V}}\bm{H}_{V}-\bm{V})=\underset{=:\zeta_{1}}{\underbrace{\bm{Z}_{m,\cdot}(\widehat{\bm{V}}^{(m)}\bm{H}_{V}^{(m)}-\bm{V})}}+\underset{=:\zeta_{2}}{\underbrace{\bm{Z}_{m,\cdot}(\widehat{\bm{V}}\bm{H}_{V}-\widehat{\bm{V}}^{(m)}\bm{H}_{V}^{(m)})}}.\label{2nd term for U Sig H_V LOO decompose}
\end{equation}

For $\zeta_{1}$, since $(\widehat{\bm{V}}^{(m)}\bm{H}_{V}^{(m)}-\bm{V})$
is independent with $\bm{Z}_{m,\cdot}$, we apply Lemma \ref{Lemma Z norms A B}
to obtain that $\left\Vert \zeta_{1}\right\Vert _{2}\lesssim\big\Vert\widehat{\bm{V}}^{(m)}\bm{H}_{V}^{(m)}-\bm{V}\big\Vert_{\mathrm{F}}\log n$.
Note that 
\[
\big\Vert\widehat{\bm{V}}^{(m)}\bm{H}_{V}^{(m)}-\bm{V}\big\Vert_{\mathrm{F}}\leq\big\Vert\widehat{\bm{V}}\bm{H}_{V}-\bm{V}\big\Vert_{\mathrm{F}}+\big\Vert\widehat{\bm{V}}^{(m)}\bm{H}_{V}^{(m)}-\widehat{\bm{V}}\bm{H}_{V}\big\Vert_{\mathrm{F}}\lesssim\rho\sqrt{r}+\big\Vert\widehat{\bm{V}}^{(m)}\bm{H}_{V}^{(m)}-\widehat{\bm{V}}\bm{H}_{V}\big\Vert_{\mathrm{F}}.
\]
Here, the last inequality is because $\big\Vert\widehat{\bm{V}}\bm{H}_{V}-\bm{V}\big\Vert_{\mathrm{F}}\lesssim\rho\sqrt{r}$,
which can be proven in the same manner as (\ref{U H_U - U F-norm}).
So, we conclude that 
\begin{equation}
\left\Vert \zeta_{1}\right\Vert _{2}\lesssim\left(\rho\sqrt{r}+\left\Vert \nu_{(m)}\right\Vert _{\mathrm{F}}\right)\log n\text{\qquad with\qquad}\nu_{(m)}:=\widehat{\bm{V}}^{(m)}\bm{H}_{V}^{(m)}-\widehat{\bm{V}}\bm{H}_{V}.\label{1st term in LOO U}
\end{equation}

For $\zeta_{2}$, using the event $\mathcal{E}_{Z}$ defined in (\ref{Prob E_Z entry logn bound}),
we have that 
\[
\left\Vert \zeta_{2}\right\Vert _{2}\lesssim\sqrt{n\log n}\left\Vert \nu_{(m)}\right\Vert _{2},
\]
with probability at least $1-O(n^{-2})$. Combining the bounds for
$\zeta_{1}$ and $\zeta_{2}$, we obtain that 
\begin{equation}
\big\Vert\bm{Z}_{m,\cdot}(\widehat{\bm{V}}\bm{H}_{V}-\bm{V})\big\Vert_{2}\leq\left\Vert \zeta_{1}\right\Vert _{2}+\left\Vert \zeta_{2}\right\Vert _{2}\lesssim\rho\sqrt{r}\log n+\sqrt{n\log n}\left\Vert \nu_{(m)}\right\Vert _{\mathrm{F}}.\label{2nd term for U Sig H_V}
\end{equation}

We will use self-bounding technique to prove in \textit{Step 4} that
\[
\left\Vert \nu_{(m)}\right\Vert _{\mathrm{F}}\lesssim\rho\sqrt{\frac{r}{n}}\log n+\omega\sqrt{\log n}\left(\big\Vert\bm{U}\big\Vert_{2,\infty}+\big\Vert\widehat{\bm{U}}\bm{H}_{U}-\bm{U}\big\Vert_{2,\infty}\right).
\]
As long as we have this, pluging the above upper bound of $\left\Vert \nu_{(m)}\right\Vert _{\mathrm{F}}$
into (\ref{2nd term for U Sig H_V}), we obtain the upper bound for
$\big\Vert\bm{Z}(\widehat{\bm{V}}\bm{H}_{V}-\bm{V})\big\Vert_{2,\infty}=\max_{1\leq m\leq N}\big\Vert\bm{Z}_{m,\cdot}(\widehat{\bm{V}}\bm{H}_{V}-\bm{V})\big\Vert_{2}$,
which leads to an upper bound for $\left\Vert \beta_{2}\right\Vert _{2}$
in (\ref{2nd term for U Sig H_V inequality0}). Finally, combining
the bound for $\beta_{1}$ and the bound for $\beta_{2}$ in (\ref{2nd term for U Sig H_V inequality0}),
we obtain the desired bound for $(\widehat{\bm{V}}\widehat{\bm{\Sigma}}\bm{H}_{U}-(T^{-1/2}\bm{X})^{\top}\bm{U})_{k,\cdot}$
as below, 
\begin{align*}
 & \big\Vert(\widehat{\bm{V}}\widehat{\bm{\Sigma}}\bm{H}_{U}-(T^{-1/2}\bm{X})^{\top}\bm{U})_{k,\cdot}\big\Vert_{2}\\
 & \leq\left\Vert \beta_{1}\right\Vert _{2}+\left\Vert \beta_{2}\right\Vert _{2}\\
 & \lesssim\sigma_{r}\rho^{2}\big\Vert\bm{U}_{k,\cdot}\big\Vert_{2}+\frac{1}{\sqrt{T}}\big\Vert(\bm{\Sigma}_{\varepsilon}^{1/2})_{k,\cdot}\big\Vert_{1}\rho\sqrt{r}\log n\\
 & +\frac{1}{\sqrt{T}}\big\Vert(\bm{\Sigma}_{\varepsilon}^{1/2})_{k,\cdot}\big\Vert_{1}\sqrt{n\log n}\left[\rho\sqrt{\frac{r}{n}}\log n+\omega\sqrt{\log n}\left(\big\Vert\bm{U}\big\Vert_{2,\infty}+\big\Vert\widehat{\bm{U}}\bm{H}_{U}-\bm{U}\big\Vert_{2,\infty}\right)\right]\\
 & \lesssim\sigma_{r}\omega_{k}\rho\sqrt{\frac{r}{n}}\log^{3/2}n+\sigma_{r}\rho^{2}\big\Vert\bm{U}_{k,\cdot}\big\Vert_{2}+\sigma_{r}\omega_{k}\omega\log n\left(\big\Vert\bm{U}\big\Vert_{2,\infty}+\big\Vert\widehat{\bm{U}}\bm{H}_{U}-\bm{U}\big\Vert_{2,\infty}\right).
\end{align*}

\textit{Step 4} --
Recall that $\nu_{(m)}:=\widehat{\bm{V}}^{(m)}\bm{H}_{V}^{(m)}-\widehat{\bm{V}}\bm{H}_{V}$
as defined in (\ref{1st term in LOO U}). We define 
\[
\mu_{(m)}:=\widehat{\bm{U}}^{(m)}\bm{H}_{U}^{(m)}-\widehat{\bm{U}}\bm{H}_{U}.
\]
By definition, we have that $\mu_{(m)}=(\widehat{\bm{U}}^{(m)}(\widehat{\bm{U}}^{(m)})^{\top}-\widehat{\bm{U}}\widehat{\bm{U}}^{\top})\bm{U}$,
so $\left\Vert \mu_{(m)}\right\Vert _{\mathrm{F}}\leq\big\Vert\widehat{\bm{U}}^{(m)}(\widehat{\bm{U}}^{(m)})^{\top}-\widehat{\bm{U}}\widehat{\bm{U}}^{\top}\big\Vert_{\mathrm{F}}$,
and similarly we have $\left\Vert \nu_{(m)}\right\Vert _{\mathrm{F}}\leq\big\Vert\widehat{\bm{V}}^{(m)}(\widehat{\bm{V}}^{(m)})^{\top}-\widehat{\bm{V}}\widehat{\bm{V}}^{\top}\big\Vert_{\mathrm{F}}$.
Denote by $E_{m}:=\frac{1}{\sqrt{T}}\bm{X}^{(m)}-\frac{1}{\sqrt{T}}\bm{X}=\frac{1}{\sqrt{T}}\bm{\Sigma}_{\varepsilon}^{1/2}(\bm{Z}^{(m)}-\bm{Z})$.
Then using (\ref{SNR 2-norm}) and our argument in proof of Lemma
\ref{Lemma R H for U V}, we have that $\left\Vert E_{m}\right\Vert \ll\sigma_{r}$,
and $\sigma_{i}(\frac{1}{\sqrt{T}}\bm{X}^{(m)})\asymp\sigma_{i}$
for $i=1,2,\ldots,r$, and $\sigma_{r+1}(\frac{1}{\sqrt{T}}\bm{X}^{(m)})\ll\sigma_{r}$.
By Wedin's $\sin\Theta$ theorem \citep[Theorem 2.9]{chen2021Monograph},
we obtain that 
\begin{align*}
 & \max\left\{ \left\Vert \mu_{(m)}\right\Vert _{\mathrm{F}},\left\Vert \nu_{(m)}\right\Vert _{\mathrm{F}}\right\} \\
 & \leq\max\left\{ \big\Vert\widehat{\bm{U}}^{(m)}(\widehat{\bm{U}}^{(m)})^{\top}-\widehat{\bm{U}}\widehat{\bm{U}}^{\top}\big\Vert_{\mathrm{F}},\big\Vert\widehat{\bm{V}}^{(m)}(\widehat{\bm{V}}^{(m)})^{\top}-\widehat{\bm{V}}\widehat{\bm{V}}^{\top}\big\Vert_{\mathrm{F}}\right\} \\
 & \lesssim\frac{1}{\sigma_{r}(\frac{1}{\sqrt{T}}\bm{X}^{(m)})-\sigma_{r+1}(\frac{1}{\sqrt{T}}\bm{X}^{(m)})-\left\Vert E_{m}\right\Vert _{2}}\max\left\{ \big\Vert E_{m}\widehat{\bm{V}}^{(m)}\big\Vert_{\mathrm{F}},\big\Vert(E_{m})^{\top}\widehat{\bm{U}}^{(m)}\big\Vert_{\mathrm{F}}\right\} \\
 & \lesssim\frac{1}{\sigma_{r}}\max\left\{ \big\Vert E_{m}\widehat{\bm{V}}^{(m)}\big\Vert_{\mathrm{F}},\big\Vert(E_{m})^{\top}\widehat{\bm{U}}^{(m)}\big\Vert_{\mathrm{F}}\right\} .
\end{align*}

Next, by calculations we have that 
\begin{align*}
\big\Vert(E_{m})^{\top}\widehat{\bm{U}}^{(m)}\big\Vert_{\mathrm{F}} & =\left\Vert \frac{1}{\sqrt{T}}(\bm{Z}^{(m)}-\bm{Z})^{\top}\bm{\Sigma}_{\varepsilon}^{1/2}\widehat{\bm{U}}^{(m)}\right\Vert _{\mathrm{F}}=\frac{1}{\sqrt{T}}\left\Vert (\bm{Z}_{m,\cdot})^{\top}(\bm{\Sigma}_{\varepsilon}^{1/2}\widehat{\bm{U}}^{(m)})_{m,\cdot}\right\Vert _{\mathrm{F}}\\
 & =\frac{1}{\sqrt{T}}\left\Vert \bm{Z}_{m,\cdot}\right\Vert _{2}\left\Vert (\bm{\Sigma}_{\varepsilon}^{1/2}\widehat{\bm{U}}^{(m)})_{m,\cdot}\right\Vert _{2}\overset{\text{(i)}}{\lesssim}\frac{1}{\sqrt{T}}\left\Vert \bm{Z}_{m,\cdot}\right\Vert _{2}\left\Vert (\bm{\Sigma}_{\varepsilon}^{1/2}\widehat{\bm{U}}^{(m)})_{m,\cdot}\bm{H}_{U}^{(m)}\right\Vert _{2}\\
 & \leq\frac{1}{\sqrt{T}}\left\Vert \bm{Z}_{m,\cdot}\right\Vert _{2}\left\Vert \bm{\Sigma}_{\varepsilon}^{1/2}\widehat{\bm{U}}^{(m)}\bm{H}_{U}^{(m)}\right\Vert _{2,\infty}\leq\frac{1}{\sqrt{T}}\left\Vert \bm{Z}_{m,\cdot}\right\Vert _{2}\left\Vert \bm{\Sigma}_{\varepsilon}^{1/2}\right\Vert _{1}\left\Vert \widehat{\bm{U}}^{(m)}\bm{H}_{U}^{(m)}\right\Vert _{2,\infty}\\
 & \leq\frac{1}{\sqrt{T}}\left\Vert \bm{\Sigma}_{\varepsilon}^{1/2}\right\Vert _{1}\left\Vert \bm{Z}_{m,\cdot}\right\Vert _{2}\left(\big\Vert\bm{U}\big\Vert_{2,\infty}+\big\Vert\widehat{\bm{U}}\bm{H}_{U}-\bm{U}\big\Vert_{2,\infty}+\left\Vert \widehat{\bm{U}}^{(m)}\bm{H}_{U}^{(m)}-\widehat{\bm{U}}\bm{H}_{U}\right\Vert _{2,\infty}\right)\\
 & \lesssim\frac{1}{\sqrt{T}}\left\Vert \bm{\Sigma}_{\varepsilon}^{1/2}\right\Vert _{1}\sqrt{n\log n}\left(\big\Vert\bm{U}\big\Vert_{2,\infty}+\big\Vert\widehat{\bm{U}}\bm{H}_{U}-\bm{U}\big\Vert_{2,\infty}+\left\Vert \mu_{(m)}\right\Vert _{\mathrm{F}}\right).
\end{align*}
Here (i) is because $\sigma_{\min}(\bm{H}_{U}^{(m)})\asymp1$, which
can be proven in the same manner as we prove $\sigma_{i}(\bm{H}_{U})\asymp1$
in Lemma \ref{Lemma R H for U V}, and the last inequality uses the
event $\mathcal{E}_{Z}$ defined in (\ref{Prob E_Z entry logn bound}).

Also, we have that 
\begin{align*}
\big\Vert E_{m}\widehat{\bm{V}}^{(m)}\big\Vert_{\mathrm{F}} & =\left\Vert \frac{1}{\sqrt{T}}\bm{\Sigma}_{\varepsilon}^{1/2}(\bm{Z}^{(m)}-\bm{Z})\widehat{\bm{V}}^{(m)}\right\Vert _{\mathrm{F}}\leq\frac{1}{\sqrt{T}}\left\Vert \bm{\Sigma}_{\varepsilon}^{1/2}\right\Vert _{2}\left\Vert (\bm{Z}^{(m)}-\bm{Z})\widehat{\bm{V}}^{(m)}\right\Vert _{\mathrm{F}}\\
 & =\frac{1}{\sqrt{T}}\sqrt{\left\Vert \bm{\Sigma}_{\varepsilon}\right\Vert _{2}}\left\Vert \bm{Z}_{m,\cdot}\widehat{\bm{V}}^{(m)}\right\Vert _{2}\overset{\text{(i)}}{\lesssim}\frac{1}{\sqrt{T}}\sqrt{\left\Vert \bm{\Sigma}_{\varepsilon}\right\Vert _{2}}\left\Vert \bm{Z}_{m,\cdot}\widehat{\bm{V}}^{(m)}\bm{H}_{V}^{(m)}\right\Vert _{2}\\
 & \leq\frac{1}{\sqrt{T}}\sqrt{\left\Vert \bm{\Sigma}_{\varepsilon}\right\Vert _{2}}\left(\left\Vert \bm{Z}_{m,\cdot}\bm{V}\right\Vert _{2}+\left\Vert \bm{Z}_{m,\cdot}(\widehat{\bm{V}}^{(m)}\bm{H}_{V}^{(m)}-\bm{V})\right\Vert _{2}\right)\\
 & \lesssim\frac{1}{\sqrt{T}}\sqrt{\left\Vert \bm{\Sigma}_{\varepsilon}\right\Vert _{2}r}\log n+\frac{1}{\sqrt{T}}\sqrt{\left\Vert \bm{\Sigma}_{\varepsilon}\right\Vert _{2}}\left(\rho\sqrt{r}+\left\Vert \nu_{(m)}\right\Vert _{\mathrm{F}}\right)\log n.
\end{align*}
Here, (i) is because $\sigma_{\min}(\bm{H}_{V}^{(m)})\asymp1$, which
can be proven in the same manner as we prove $\sigma_{i}(\bm{H}_{V})\asymp1$
in Lemma \ref{Lemma R H for U V}; the last inequality uses the bound
(\ref{1st term in LOO U}) for $\zeta_{1}$ defined in (\ref{2nd term for U Sig H_V LOO decompose})
and the fact that, by Lemma \ref{Lemma Z norms A B}, $\big\Vert\bm{Z}_{m,\cdot}\bm{V}\big\Vert_{2}\lesssim\big\Vert\bm{V}\big\Vert_{\mathrm{F}}\log n=\sqrt{r}\log n$.

Combining the above bounds for $\big\Vert(E_{m})^{\top}\widehat{\bm{U}}^{(m)}\big\Vert_{\mathrm{F}}$
and $\big\Vert E_{m}\widehat{\bm{V}}^{(m)}\big\Vert_{\mathrm{F}}$,
and using $\rho\ll1$, we obtain 
\begin{align*}
 & \max\left\{ \left\Vert \mu_{(m)}\right\Vert _{\mathrm{F}},\left\Vert \nu_{(m)}\right\Vert _{\mathrm{F}}\right\} \\
 & \lesssim\frac{1}{\sigma_{r}\sqrt{T}}\sqrt{\left\Vert \bm{\Sigma}_{\varepsilon}\right\Vert _{2}r}\log n+\frac{1}{\sqrt{T}}\sqrt{\left\Vert \bm{\Sigma}_{\varepsilon}\right\Vert _{2}}\left(\rho\sqrt{r}+\left\Vert \nu_{(m)}\right\Vert _{\mathrm{F}}\right)\log n\\
 & +\frac{1}{\sigma_{r}\sqrt{T}}\left\Vert \bm{\Sigma}_{\varepsilon}^{1/2}\right\Vert _{1}\sqrt{n\log n}\left(\big\Vert\bm{U}\big\Vert_{2,\infty}+\big\Vert\widehat{\bm{U}}\bm{H}_{U}-\bm{U}\big\Vert_{2,\infty}+\left\Vert \mu_{(m)}\right\Vert _{\mathrm{F}}\right)\\
 & \lesssim\rho\sqrt{\frac{r}{n}}\log n+\omega\sqrt{\log n}\left(\big\Vert\bm{U}\big\Vert_{2,\infty}+\big\Vert\widehat{\bm{U}}\bm{H}_{U}-\bm{U}\big\Vert_{2,\infty}\right)+\omega\sqrt{\log n}\max\left\{ \left\Vert \mu_{(m)}\right\Vert _{\mathrm{F}},\left\Vert \nu_{(m)}\right\Vert _{\mathrm{F}}\right\} .
\end{align*}
Note that, both sides of the above inequality include the term $\max\left\{ \left\Vert \mu_{(m)}\right\Vert _{\mathrm{F}},\left\Vert \nu_{(m)}\right\Vert _{\mathrm{F}}\right\} $.
Then, since $\omega\sqrt{\log n}\ll1$ according to (\ref{SNR 1-norm logn}),
we obtain that 
\[
\max\left\{ \left\Vert \mu_{(m)}\right\Vert _{\mathrm{F}},\left\Vert \nu_{(m)}\right\Vert _{\mathrm{F}}\right\} \lesssim\rho\sqrt{\frac{r}{n}}\log n+\omega\sqrt{\log n}\left(\big\Vert\bm{U}\big\Vert_{2,\infty}+\big\Vert\widehat{\bm{U}}\bm{H}_{U}-\bm{U}\big\Vert_{2,\infty}\right).
\]
\end{proof}
Then we handle $\widehat{\bm{U}}\bm{H}_{U}-\bm{U}$ and $\widehat{\bm{V}}\bm{H}_{V}-\bm{V}$.

\begin{lemma} \label{Lemma UH-U VH-V}Suppose that Assumptions \ref{Assump_Bf_identification},
\ref{Assump_noise_Z_entries}, and \ref{Assump_factor_f} hold, and
\begin{equation}
\rho\ll\sqrt{\log n}\text{ and }r\log n\ll n.\label{SNR for UH-U VH-V self}
\end{equation}

(i) Assuming (\ref{SNR 2-norm logn}), i.e., $\rho\sqrt{\log n}\ll1$,
then we have that, for $l=1,2,\ldots,T$, with probability at least
$1-O(n^{-2})$, 
\[
\big\Vert(\widehat{\bm{V}}\bm{H}_{V}-\bm{V})_{l,\cdot}\big\Vert_{2}\lesssim\rho\sqrt{\frac{r}{n}}\log n+(\rho^{2}+\rho\sqrt{\frac{r}{n}})\log n\big\Vert\bm{V}_{l,\cdot}\big\Vert_{2}.
\]

(ii) Assuming (\ref{SNR 1-norm logn}), i.e., $\omega\sqrt{\log n}\ll1$,
then we have that, with probability at least $1-O(n^{-2})$, 
\[
\big\Vert\widehat{\bm{U}}\bm{H}_{U}-\bm{U}\big\Vert_{2,\infty}\lesssim\omega\sqrt{\frac{r}{n}}\log n+(\rho\sqrt{\frac{r}{n}}+\omega^{2})\log n\big\Vert\bm{U}\big\Vert_{2,\infty},
\]
and for $k=1,2,\ldots,N$, 
\[
\big\Vert(\widehat{\bm{U}}\bm{H}_{U}-\bm{U})_{k,\cdot}\big\Vert_{2}\lesssim\omega_{k}\sqrt{\frac{r}{n}}\log n+(\rho\sqrt{\frac{r}{n}}+\omega^{2})\log n\big\Vert\bm{U}_{k,\cdot}\big\Vert_{2}+\omega_{k}\omega\log n\big\Vert\bm{U}\big\Vert_{2,\infty},
\]
where $\omega_{k}=\left\Vert ((\bm{\Sigma}_{e})^{1/2})_{k,\cdot}\right\Vert _{1}\sqrt{n}/(\sigma_{r}\sqrt{T})$
is defined in (\ref{SNR 1-norm k-th row}).

\end{lemma} 
\begin{proof}
We will prove (i) and (ii) in \textit{Step 1} and \textit{Step 2},
respectively.

\textit{Step 1} --
Since $\lambda_{i}\asymp\sigma_{i}$ according to Lemma \ref{Lemma SVD BF good event},
we obtain that $\Vert(\widehat{\bm{V}}\bm{H}_{V}-\bm{V})_{l,\cdot}\Vert_{2}\lesssim\frac{1}{\sigma_{r}}\Vert(\widehat{\bm{V}}\bm{H}_{V}\bm{\Lambda}-\bm{V}\bm{\Lambda})_{l,\cdot}\Vert_{2}$.
Next, by calculations, we have that 
\[
\widehat{\bm{V}}\bm{H}_{V}\bm{\Lambda}-\bm{V}\bm{\Lambda}=\widehat{\bm{V}}\widehat{\bm{\Sigma}}\bm{H}_{U}-(T^{-1/2}\bm{X})^{\top}\bm{U}+T^{-1/2}\bm{Z}^{\top}\bm{\Sigma}_{\varepsilon}^{1/2}\bm{U}+\widehat{\bm{V}}\bm{H}_{V}\bm{\Lambda}-\widehat{\bm{V}}\widehat{\bm{\Sigma}}\bm{H}_{U}.
\]
So we obtain that 
\begin{align*}
 & \big\Vert(\widehat{\bm{V}}\bm{H}_{V}-\bm{V})_{l,\cdot}\big\Vert_{2}\\
 & \lesssim\underset{=:\bm{a}_{1}}{\underbrace{\frac{1}{\sigma_{r}}\big\Vert(\widehat{\bm{V}}\widehat{\bm{\Sigma}}\bm{H}_{U}-(T^{-1/2}\bm{X})^{\top}\bm{U})_{l,\cdot}\big\Vert_{2}}}+\underset{=:\bm{a}_{2}}{\underbrace{\frac{1}{\sigma_{r}}\big\Vert T^{-1/2}(\bm{Z}^{\top})_{l,\cdot}\bm{\Sigma}_{\varepsilon}^{1/2}\bm{U}\big\Vert_{2}}}+\underset{=:\bm{a}_{3}}{\underbrace{\frac{1}{\sigma_{r}}\big\Vert\widehat{\bm{V}}_{l,\cdot}(\bm{H}_{V}\bm{\Lambda}-\widehat{\bm{\Sigma}}\bm{H}_{U})\big\Vert_{2}}}.
\end{align*}
The term $\bm{a}_{1}$ can be bounded using Lemma \ref{Lemma LOO UV Sig H}.
For $\bm{a}_{2}$, we obtain by Lemma \ref{Lemma Z norms A B} that
\[
\bm{a}_{2}\leq\frac{1}{\sigma_{r}\sqrt{T}}\big\Vert\bm{\Sigma}_{\varepsilon}^{1/2}\bm{U}\big\Vert_{\mathrm{F}}\log n\leq\frac{1}{\sigma_{r}\sqrt{T}}\big\Vert\bm{\Sigma}_{\varepsilon}^{1/2}\big\Vert_{2}\big\Vert\bm{U}\big\Vert_{\mathrm{F}}\log n=\rho\sqrt{\frac{r}{n}}\log n.
\]

For $\bm{a}_{3}$, we have that $\big\Vert\widehat{\bm{V}}_{l,\cdot}(\bm{H}_{V}\bm{\Lambda}-\widehat{\bm{\Sigma}}\bm{H}_{U})\big\Vert_{2}\leq\big\Vert\widehat{\bm{V}}_{l,\cdot}\big\Vert_{2}\big\Vert\bm{H}_{V}\bm{\Lambda}-\widehat{\bm{\Sigma}}\bm{H}_{U}\big\Vert_{2}$
and we will handle the two terms in the right hand side. For $\big\Vert\widehat{\bm{V}}_{l,\cdot}\big\Vert_{2}$,
since $\sigma_{i}(\bm{H}_{V})\asymp1$ in Lemma \ref{Lemma R H for U V},
we have that 
\[
\big\Vert\widehat{\bm{V}}_{l,\cdot}\big\Vert_{2}\lesssim\big\Vert\widehat{\bm{V}}_{l,\cdot}\bm{H}_{V}\big\Vert_{2}=\big\Vert(\widehat{\bm{V}}\bm{H}_{V})_{l,\cdot}\big\Vert_{2}\leq\big\Vert\bm{V}_{l,\cdot}\big\Vert_{2}+\big\Vert(\widehat{\bm{V}}\bm{H}_{V}-\bm{V})_{l,\cdot}\big\Vert_{2}.
\]
For $\big\Vert\bm{H}_{V}\bm{\Lambda}-\widehat{\bm{\Sigma}}\bm{H}_{U}\big\Vert_{2}$,
using the full SVD of $T^{-1/2}\bm{X}$ in the proof of Lemma \ref{Lemma R H for U V},
we have that 
\begin{align*}
\bm{H}_{V}\bm{\Lambda}-\widehat{\bm{\Sigma}}\bm{H}_{U} & =\widehat{\bm{V}}^{\top}\bm{V}\bm{\Lambda}-\widehat{\bm{\Sigma}}\widehat{\bm{U}}^{\top}\bm{U}=\widehat{\bm{V}}^{\top}(\bm{V}\bm{\Lambda}\bm{U}^{\top}-\widehat{\bm{V}}\widehat{\bm{\Sigma}}\widehat{\bm{U}}^{\top})\bm{U}\\
 & =\widehat{\bm{V}}^{\top}(T^{-1/2}\bm{E}^{\top}-\widehat{\bm{V}}_{\perp}(\widehat{\bm{\Sigma}}_{\perp})^{\top}(\widehat{\bm{U}}_{\perp})^{\top})\bm{U}=T^{-1/2}\widehat{\bm{V}}^{\top}\bm{E}^{\top}\bm{U}.
\end{align*}
Thus, we obtain that 
\begin{align*}
\big\Vert\bm{H}_{V}\bm{\Lambda}-\widehat{\bm{\Sigma}}\bm{H}_{U}\big\Vert_{2} & \overset{\text{(i)}}{\lesssim}T^{-1/2}\big\Vert\bm{U}^{\top}\bm{E}\widehat{\bm{V}}\bm{H}_{V}\big\Vert_{2}\leq T^{-1/2}\big\Vert\bm{U}^{\top}\bm{E}\bm{V}\big\Vert_{2}+T^{-1/2}\big\Vert\bm{U}^{\top}\bm{E}(\widehat{\bm{V}}\bm{H}_{V}-\bm{V})\big\Vert_{2}\\
 & \lesssim T^{-1/2}\big\Vert\bm{U}^{\top}\bm{\Sigma}_{\varepsilon}^{1/2}\bm{Z}\bm{V}\big\Vert_{2}+T^{-1/2}\big\Vert\bm{U}^{\top}\bm{\Sigma}_{\varepsilon}^{1/2}\bm{Z}\big\Vert_{2}\big\Vert\widehat{\bm{V}}\bm{H}_{V}-\bm{V}\big\Vert_{2}\\
 & \overset{\text{(ii)}}{\lesssim}\sigma_{r}\rho\sqrt{\frac{1}{n}(r+\log n)}+\frac{1}{\sqrt{T}}\sqrt{\left\Vert \bm{\Sigma}_{\varepsilon}\right\Vert _{2}}\sqrt{r+T+\log n}\cdot\rho\\
 & \lesssim\sigma_{r}\rho\sqrt{\frac{1}{n}(r+\log n)}+\sigma_{r}\rho^{2},
\end{align*}
where (i) is because $\sigma_{i}(\bm{H}_{V})\asymp1$ in Lemma \ref{Lemma R H for U V},
(ii) uses (G.3) in Lemma 19 of \citet{yan2024entrywise} and the upper
bound for $\big\Vert\widehat{\bm{V}}\bm{H}_{V}-\bm{V}\big\Vert_{2}$
is similar to that in (\ref{U U_hat_cmpl bound}). Combing the bounds
for $\Vert\widehat{\bm{V}}_{l,\cdot}\Vert_{2}$ and $\Vert\bm{H}_{V}\bm{\Lambda}-\widehat{\bm{\Sigma}}\bm{H}_{U}\Vert_{2}$,
we obtain that 
\begin{align*}
\bm{a}_{3} & \lesssim\frac{1}{\sigma_{r}}\left(\big\Vert\bm{V}_{l,\cdot}\big\Vert_{2}+\big\Vert(\widehat{\bm{V}}\bm{H}_{V}-\bm{V})_{l,\cdot}\big\Vert_{2}\right)\left(\sigma_{r}\rho\sqrt{\frac{1}{n}(r+\log n)}+\sigma_{r}\rho^{2}\right)\\
 & \lesssim\left(\big\Vert\bm{V}_{l,\cdot}\big\Vert_{2}+\big\Vert(\widehat{\bm{V}}\bm{H}_{V}-\bm{V})_{l,\cdot}\big\Vert_{2}\right)\left(\sqrt{\frac{1}{n}(r+\log n)}+\rho\right)\rho.
\end{align*}

Finally, we combine the bounds for $\bm{a}_{1}$, $\bm{a}_{2}$, and
$\bm{a}_{3}$ to obtain that, for $l=1,2,\ldots,T$, 
\begin{align*}
 & \big\Vert(\widehat{\bm{V}}\bm{H}_{V}-\bm{V})_{l,\cdot}\big\Vert_{2}\\
 & \lesssim\rho^{2}\left(\sqrt{\frac{r}{n}}\log^{3/2}n+\log n\big\Vert\bm{V}_{l,\cdot}\big\Vert_{2}+\log n\big\Vert(\widehat{\bm{V}}\bm{H}_{V}-\bm{V})_{l,\cdot}\big\Vert_{2}\right)\\
 & +\rho\sqrt{\frac{r}{n}}\log n+\left(\big\Vert\bm{V}_{l,\cdot}\big\Vert_{2}+\big\Vert(\widehat{\bm{V}}\bm{H}_{V}-\bm{V})_{l,\cdot}\big\Vert_{2}\right)\left(\sqrt{\frac{1}{n}(r+\log n)}+\rho\right)\rho\\
 & \lesssim\rho\sqrt{\frac{r}{n}}\log n+(\rho^{2}+\rho\sqrt{\frac{r}{n}})\log n\big\Vert\bm{V}_{l,\cdot}\big\Vert_{2}\\
 & +(\rho^{2}+\rho\sqrt{\frac{r}{n}})\log n\big\Vert(\widehat{\bm{V}}\bm{H}_{V}-\bm{V})_{l,\cdot}\big\Vert_{2},
\end{align*}
where we use $\rho\sqrt{\log n}\ll1$ in (\ref{SNR 2-norm logn}).
Then using the self-bounding technique , since $(\rho^{2}+\rho\sqrt{\frac{r}{n}})\log n\ll1$
owing to (\ref{SNR 2-norm logn}) and (\ref{SNR for UH-U VH-V self}),
we conclude that 
\[
\big\Vert(\widehat{\bm{V}}\bm{H}_{V}-\bm{V})_{l,\cdot}\big\Vert_{2}\lesssim\rho\sqrt{\frac{r}{n}}\log n+(\rho^{2}+\rho\sqrt{\frac{r}{n}})\log n\big\Vert\bm{V}_{l,\cdot}\big\Vert_{2}.
\]

\textit{Step 2} --
Since $\lambda_{i}\asymp\sigma_{i}$ according to Lemma \ref{Lemma SVD BF good event},
we obtain that $\big\Vert(\widehat{\bm{U}}\bm{H}_{U}-\bm{U})_{k,\cdot}\big\Vert_{2}\lesssim\frac{1}{\sigma_{r}}\big\Vert(\widehat{\bm{U}}\bm{H}_{U}\bm{\Lambda}-\bm{U}\bm{\Lambda})_{k,\cdot}\big\Vert_{2}$.
Next, by calculations, we have that 
\[
\widehat{\bm{U}}\bm{H}_{U}\bm{\Lambda}-\bm{U}\bm{\Lambda}=\widehat{\bm{U}}\widehat{\bm{\Sigma}}\bm{H}_{V}-(T^{-1/2}\bm{X})\bm{V}+T^{-1/2}\bm{\Sigma}_{\varepsilon}^{1/2}\bm{Z}\bm{V}+\widehat{\bm{U}}\bm{H}_{U}\bm{\Lambda}-\widehat{\bm{U}}\widehat{\bm{\Sigma}}\bm{H}_{V}.
\]
So we obtain that 
\begin{align*}
 & \big\Vert(\widehat{\bm{U}}\bm{H}_{U}-\bm{U})_{k,\cdot}\big\Vert_{2}\\
 & \lesssim\underset{=:\bm{b}_{1}}{\underbrace{\frac{1}{\sigma_{r}}\big\Vert(\widehat{\bm{U}}\widehat{\bm{\Sigma}}\bm{H}_{V}-(T^{-1/2}\bm{X})\bm{V})_{k,\cdot}\big\Vert_{2}}}+\underset{=:\bm{b}_{2}}{\underbrace{\frac{1}{\sigma_{r}}\big\Vert T^{-1/2}(\bm{\Sigma}_{\varepsilon}^{1/2})_{k,\cdot}\bm{Z}\bm{V}\big\Vert_{2}}}+\underset{=:\bm{b}_{3}}{\underbrace{\frac{1}{\sigma_{r}}\big\Vert\widehat{\bm{U}}_{k,\cdot}(\bm{H}_{U}\bm{\Lambda}-\widehat{\bm{\Sigma}}\bm{H}_{V})\big\Vert_{2}}}.
\end{align*}
The term $\bm{b}_{1}$ can be bounded using Lemma \ref{Lemma LOO UV Sig H}.
For $\bm{b}_{2}$, we obtain by Lemma \ref{Lemma Z norms A B} that
\begin{align*}
\bm{b}_{2} & \leq\frac{1}{\sigma_{r}\sqrt{T}}\big\Vert(\bm{\Sigma}_{\varepsilon}^{1/2})_{k,\cdot}\big\Vert_{1}\big\Vert\bm{Z}\bm{V}\big\Vert_{2,\infty}\lesssim\frac{1}{\sigma_{r}\sqrt{T}}\big\Vert(\bm{\Sigma}_{\varepsilon}^{1/2})_{k,\cdot}\big\Vert_{1}\big\Vert\bm{V}\big\Vert_{\mathrm{F}}\log n\\
 & =\frac{1}{\sigma_{r}\sqrt{T}}\big\Vert(\bm{\Sigma}_{\varepsilon}^{1/2})_{k,\cdot}\big\Vert_{1}\sqrt{r}\log n=\omega_{k}\sqrt{\frac{r}{n}}\log n.
\end{align*}

For $\bm{b}_{3}$, we have that $\big\Vert\widehat{\bm{U}}_{k,\cdot}(\bm{H}_{U}\bm{\Lambda}-\widehat{\bm{\Sigma}}\bm{H}_{V})\big\Vert_{2}\leq\big\Vert\widehat{\bm{U}}_{k,\cdot}\big\Vert_{2}\big\Vert\bm{H}_{U}\bm{\Lambda}-\widehat{\bm{\Sigma}}\bm{H}_{V}\big\Vert_{2}$
and we will handle the two terms in the right hand side. For $\big\Vert\widehat{\bm{U}}_{k,\cdot}\big\Vert_{2}$,
since $\sigma_{i}(\bm{H}_{U})\asymp1$ in Lemma \ref{Lemma R H for U V},
we have that 
\[
\big\Vert\widehat{\bm{U}}_{k,\cdot}\big\Vert_{2}\lesssim\big\Vert\widehat{\bm{U}}_{k,\cdot}\bm{H}_{U}\big\Vert_{2}=\big\Vert(\widehat{\bm{U}}\bm{H}_{U})_{k,\cdot}\big\Vert_{2}\leq\big\Vert\bm{U}_{k,\cdot}\big\Vert_{2}+\big\Vert(\widehat{\bm{U}}\bm{H}_{U}-\bm{U})_{k,\cdot}\big\Vert_{2}.
\]
For $\big\Vert\bm{H}_{U}\bm{\Lambda}-\widehat{\bm{\Sigma}}\bm{H}_{V}\big\Vert_{2}$,
similar to the upper bound for $\big\Vert\bm{H}_{V}\bm{\Lambda}-\widehat{\bm{\Sigma}}\bm{H}_{U}\big\Vert_{2}$
in Step 1, we have that 
\[
\big\Vert\bm{H}_{U}\bm{\Lambda}-\widehat{\bm{\Sigma}}\bm{H}_{V}\big\Vert_{2}\lesssim\sigma_{r}\rho\sqrt{\frac{1}{n}(r+\log n)}+\sigma_{r}\rho^{2}.
\]
Combining the bounds for $\big\Vert\widehat{\bm{U}}_{k,\cdot}\big\Vert_{2}$
and $\big\Vert\bm{H}_{U}\bm{\Lambda}-\widehat{\bm{\Sigma}}\bm{H}_{V}\big\Vert_{2}$,
we obtain that 
\[
\bm{b}_{3}\lesssim\left(\big\Vert\bm{U}_{k,\cdot}\big\Vert_{2}+\big\Vert(\widehat{\bm{U}}\bm{H}_{U}-\bm{U})_{k,\cdot}\big\Vert_{2}\right)\left(\rho\sqrt{\frac{1}{n}(r+\log n)}+\rho^{2}\right).
\]

Finally, we combine the bounds for $\bm{b}_{1}$, $\bm{b}_{2}$, and
$\bm{b}_{3}$ to obtain that, for $k=1,2,\ldots,N$, 
\begin{align*}
 & \big\Vert(\widehat{\bm{U}}\bm{H}_{U}-\bm{U})_{k,\cdot}\big\Vert_{2}\\
 & \lesssim\left(\omega_{k}\rho\sqrt{\frac{r}{n}}\log^{3/2}n+\rho^{2}\big\Vert\bm{U}_{k,\cdot}\big\Vert_{2}+\omega_{k}\omega\log n\left(\big\Vert\bm{U}\big\Vert_{2,\infty}+\big\Vert\widehat{\bm{U}}\bm{H}_{U}-\bm{U}\big\Vert_{2,\infty}\right)\right)\\
 & +\omega_{k}\sqrt{\frac{r}{n}}\log n+\left(\big\Vert\bm{U}_{k,\cdot}\big\Vert_{2}+\left\Vert (\widehat{\bm{U}}\bm{H}_{U}-\bm{U})_{k,\cdot}\right\Vert _{2}\right)\left(\rho\sqrt{\frac{1}{n}(r+\log n)}+\rho^{2}\right)\\
 & \lesssim\omega_{k}\sqrt{\frac{r}{n}}\log n+\left(\rho\sqrt{\frac{1}{n}(r+\log n)}+\rho^{2}\right)\left(\big\Vert\bm{U}_{k,\cdot}\big\Vert_{2}+\big\Vert(\widehat{\bm{U}}\bm{H}_{U}-\bm{U})_{k,\cdot}\big\Vert_{2}\right)\\
 & +\omega_{k}\omega\log n\left(\big\Vert\bm{U}\big\Vert_{2,\infty}+\big\Vert\widehat{\bm{U}}\bm{H}_{U}-\bm{U}\big\Vert_{2,\infty}\right),
\end{align*}
where we use the fact that $\rho\sqrt{\log n}\leq\omega\sqrt{\log n}\ll1$.
Taking supremum w.r.t. $k$, we obtain that 
\[
\big\Vert\widehat{\bm{U}}\bm{H}_{U}-\bm{U}\big\Vert_{2,\infty}\lesssim\omega\sqrt{\frac{r}{n}}\log n+\left(\rho\sqrt{\frac{1}{n}(r+\log n)}+\omega^{2}\log n\right)\left(\big\Vert\bm{U}\big\Vert_{2,\infty}+\big\Vert\widehat{\bm{U}}\bm{H}_{U}-\bm{U}\big\Vert_{2,\infty}\right).
\]
Then using the self-bounding technique for $\Vert\widehat{\bm{U}}\bm{H}_{U}-\bm{U}\Vert_{2,\infty}$,
since $\omega\sqrt{\log n}\ll1$ and $\rho\sqrt{\frac{1}{n}(r+\log n)}\ll1$
and $\rho\leq\omega$, we conclude that 
\[
\big\Vert\widehat{\bm{U}}\bm{H}_{U}-\bm{U}\big\Vert_{2,\infty}\lesssim\omega\sqrt{\frac{r}{n}}\log n+\left(\rho\sqrt{\frac{r}{n}}\log n+\omega^{2}\log n\right)\big\Vert\bm{U}\big\Vert_{2,\infty}.
\]
Plugging the above bound of $\Vert\widehat{\bm{U}}\bm{H}_{U}-\bm{U}\Vert_{2,\infty}$
into the previous inequality for $\Vert(\widehat{\bm{U}}\bm{H}_{U}-\bm{U})_{k,\cdot}\Vert_{2}$,
we obtain that 
\begin{align*}
\big\Vert(\widehat{\bm{U}}\bm{H}_{U}-\bm{U})_{k,\cdot}\big\Vert_{2} & \lesssim\omega_{k}\sqrt{\frac{r}{n}}\log n+\left(\rho\sqrt{\frac{1}{n}(r+\log n)}+\rho^{2}\right)\left(\big\Vert\bm{U}_{k,\cdot}\big\Vert_{2}+\big\Vert(\widehat{\bm{U}}\bm{H}_{U}-\bm{U})_{k,\cdot}\big\Vert_{2}\right)\\
 & +\omega_{k}\omega\log n\left(\big\Vert\bm{U}\big\Vert_{2,\infty}+\omega\sqrt{\frac{r}{n}}\log n+\left(\rho\sqrt{\frac{r}{n}}+\omega^{2}\right)\log n\big\Vert\bm{U}\big\Vert_{2,\infty}\right).
\end{align*}
Then using the self-bounding technique for $\Vert(\widehat{\bm{U}}\bm{H}_{U}-\bm{U})_{k,\cdot}\Vert_{2}$,
since $\rho\sqrt{\log n}\leq\omega\sqrt{\log n}\ll1$ in (\ref{SNR 1-norm logn}),
we conclude that 
\[
\big\Vert(\widehat{\bm{U}}\bm{H}_{U}-\bm{U})_{k,\cdot}\big\Vert_{2}\lesssim\omega_{k}\sqrt{\frac{r}{n}}\log n+\left(\rho\sqrt{\frac{r}{n}}\log n+\omega^{2}\log n\right)\big\Vert\bm{U}_{k,\cdot}\big\Vert_{2}+\omega_{k}\omega\log n\big\Vert\bm{U}\big\Vert_{2,\infty}.
\]
\end{proof}

\subsection{Proof of Theorem \ref{Thm UV 1st approx row-wise error}}

We are now ready to prove Theorem \ref{Thm UV 1st approx row-wise error}.
We will prove (i) and (ii) in \textit{Step 1} and \textit{Step 2},
respectively.

\textit{Step 1 -- }
By calculations, we have that 
\[
\bm{\Psi}_{U}=(\widehat{\bm{U}}\widehat{\bm{\Sigma}}\bm{H}_{V}-(T^{-1/2}\bm{X})\bm{V})\bm{\Lambda}^{-1}+\widehat{\bm{U}}\widehat{\bm{\Sigma}}(\bm{R}_{V}-\bm{H}_{V})\bm{\Lambda}^{-1}+(\widehat{\bm{U}}\bm{R}_{U})(\bm{\Lambda}-(\bm{R}_{U})^{\top}\widehat{\bm{\Sigma}}\bm{R}_{V})\bm{\Lambda}^{-1}.
\]
Since $\Vert\widehat{\bm{U}}_{k,\cdot}\bm{R}_{U}\Vert_{2}=\Vert\widehat{\bm{U}}_{k,\cdot}\Vert_{2}$,
we obtain that 
\begin{align*}
 & \left\Vert (\bm{\Psi}_{U})_{k,\cdot}\right\Vert _{2}\\
 & \lesssim\big\Vert(\widehat{\bm{U}}\widehat{\bm{\Sigma}}\bm{H}_{V}-(T^{-1/2}\bm{X})\bm{V})_{k,\cdot}\big\Vert_{2}\big\Vert\bm{\Lambda}^{-1}\big\Vert_{2}+\big\Vert\widehat{\bm{U}}_{k,\cdot}\big\Vert_{2}\left(\big\Vert\widehat{\bm{\Sigma}}(\bm{R}_{V}-\bm{H}_{V})\bm{\Lambda}^{-1}\big\Vert_{2}+\big\Vert(\bm{\Lambda}-(\bm{R}_{U})^{\top}\widehat{\bm{\Sigma}}\bm{R}_{V})\bm{\Lambda}^{-1}\big\Vert_{2}\right)\\
 & \overset{\text{(i)}}{\lesssim}\big\Vert(\widehat{\bm{U}}\widehat{\bm{\Sigma}}\bm{H}_{V}-(T^{-1/2}\bm{X})\bm{V})_{k,\cdot}\big\Vert_{2}\frac{1}{\sigma_{r}}+\big\Vert\widehat{\bm{U}}_{k,\cdot}\big\Vert_{2}\left(\sigma_{r}\rho^{2}\frac{1}{\sigma_{r}}+\frac{1}{\sigma_{r}}\left(\sigma_{r}\rho^{2}+\sigma_{r}\rho\sqrt{\frac{1}{n}(r+\log n)}\right)\right)\\
 & \lesssim\big\Vert(\widehat{\bm{U}}\widehat{\bm{\Sigma}}\bm{H}_{V}-(T^{-1/2}\bm{X})\bm{V})_{k,\cdot}\big\Vert_{2}\frac{1}{\sigma_{r}}+\big\Vert\widehat{\bm{U}}_{k,\cdot}\big\Vert_{2}\left(\rho^{2}+\rho\sqrt{\frac{r}{n}}\sqrt{\log n}\right),
\end{align*}
where (i) is owing to the inequalities in Lemmas \ref{Lemma R H for U V}
and \ref{Lemma Sigma hat tilde H R}. Using $\sigma_{i}(\bm{H}_{U})\asymp1$
in Lemma \ref{Lemma R H for U V}, we have that $\big\Vert\widehat{\bm{U}}_{k,\cdot}\big\Vert_{2}\lesssim\big\Vert(\widehat{\bm{U}}\bm{H}_{U})_{k,\cdot}\big\Vert_{2}\leq\big\Vert\bm{U}_{k,\cdot}\big\Vert_{2}+\big\Vert(\widehat{\bm{U}}\bm{H}_{U}-\bm{U})_{k,\cdot}\big\Vert_{2}$.
Then, we obtain by Lemma \ref{Lemma LOO UV Sig H} that, 
\begin{align*}
\left\Vert (\bm{\Psi}_{U})_{k,\cdot}\right\Vert _{2} & \lesssim\left(\omega_{k}\rho\sqrt{\frac{r}{n}}\log^{3/2}n+\rho^{2}\big\Vert\bm{U}_{k,\cdot}\big\Vert_{2}+\omega_{k}\omega\log n\left(\big\Vert\bm{U}\big\Vert_{2,\infty}+\big\Vert\widehat{\bm{U}}\bm{H}_{U}-\bm{U}\big\Vert_{2,\infty}\right)\right)\\
 & +\left(\big\Vert\bm{U}_{k,\cdot}\big\Vert_{2}+\big\Vert(\widehat{\bm{U}}\bm{H}_{U}-\bm{U})_{k,\cdot}\big\Vert_{2}\right)\left(\rho^{2}+\rho\sqrt{\frac{r}{n}}\sqrt{\log n}\right).
\end{align*}

Finally, we obtain by Lemma \ref{Lemma UH-U VH-V} that, 
\begin{align*}
 & \left\Vert (\bm{\Psi}_{U})_{k,\cdot}\right\Vert _{2}\\
 & \lesssim\left(\omega_{k}\rho\sqrt{\frac{r}{n}}\log^{3/2}n+\rho^{2}\big\Vert\bm{U}_{k,\cdot}\big\Vert_{2}+\omega_{k}\omega\log n\left(\begin{array}{c}
\big\Vert\bm{U}\big\Vert_{2,\infty}+\omega\sqrt{\frac{r}{n}}\log n\\
+(\rho\sqrt{\frac{r}{n}}+\omega^{2})\log n\big\Vert\bm{U}\big\Vert_{2,\infty}
\end{array}\right)\right)\\
 & +\left(\begin{array}{c}
\big\Vert\bm{U}_{k,\cdot}\big\Vert_{2}+\omega_{k}\sqrt{\frac{r}{n}}\log n\\
+(\rho\sqrt{\frac{r}{n}}+\omega^{2})\log n\big\Vert\bm{U}_{k,\cdot}\big\Vert_{2}+\omega_{k}\omega\log n\big\Vert\bm{U}\big\Vert_{2,\infty}
\end{array}\right)\left(\rho^{2}+\rho\sqrt{\frac{r}{n}}\sqrt{\log n}\right)\\
 & \lesssim\omega_{k}\omega\sqrt{\frac{r}{n}}\log^{3/2}n+\left(\rho^{2}+\rho\sqrt{\frac{r}{n}}\log n\right)\big\Vert\bm{U}_{k,\cdot}\big\Vert_{2}+\omega_{k}\omega\log n\big\Vert\bm{U}\big\Vert_{2,\infty},
\end{align*}
where, in the last inequality, we use the assumptions that $\omega\sqrt{\log n}\ll1$,
$\rho\sqrt{\log n}\ll1$, and $r\log n\ll n$ owing to (\ref{SNR 2-norm logn})
and (\ref{SNR 1-norm logn}), and $\omega_{k}\leq\omega$ as well
as $\rho\leq\omega$.

\textit{Step 2} --
By calculations, we have that 
\[
\bm{\Psi}_{V}=(\widehat{\bm{V}}\widehat{\bm{\Sigma}}\bm{H}_{U}-(T^{-1/2}\bm{X})^{\top}\bm{U})\bm{\Lambda}^{-1}+\widehat{\bm{V}}\widehat{\bm{\Sigma}}(\bm{R}_{U}-\bm{H}_{U})\bm{\Lambda}^{-1}+(\widehat{\bm{V}}\bm{R}_{V})(\bm{\Lambda}-(\bm{R}_{V})^{\top}\widehat{\bm{\Sigma}}\bm{R}_{U})\bm{\Lambda}^{-1}.
\]
Since $\Vert\widehat{\bm{V}}_{l,\cdot}\bm{R}_{V}\Vert_{2}=\Vert\widehat{\bm{V}}_{l,\cdot}\Vert_{2}$,
we obtain that 
\begin{align*}
 & \left\Vert (\bm{\Psi}_{V})_{l,\cdot}\right\Vert _{2}\\
 & \lesssim\big\Vert(\widehat{\bm{V}}\widehat{\bm{\Sigma}}\bm{H}_{U}-(T^{-1/2}\bm{X})^{\top}\bm{U})_{l,\cdot}\big\Vert_{2}\big\Vert\bm{\Lambda}^{-1}\big\Vert_{2}+\big\Vert\widehat{\bm{V}}_{l,\cdot}\big\Vert_{2}\left(\big\Vert\widehat{\bm{\Sigma}}(\bm{R}_{U}-\bm{H}_{U})\bm{\Lambda}^{-1}\big\Vert_{2}+\big\Vert(\bm{\Lambda}-(\bm{R}_{V})^{\top}\widehat{\bm{\Sigma}}\bm{R}_{U})\bm{\Lambda}^{-1}\big\Vert_{2}\right)\\
 & \overset{\text{(i)}}{\lesssim}\big\Vert(\widehat{\bm{V}}\widehat{\bm{\Sigma}}\bm{H}_{U}-(T^{-1/2}\bm{X})^{\top}\bm{U})_{l,\cdot}\big\Vert_{2}\frac{1}{\sigma_{r}}+\big\Vert\widehat{\bm{V}}_{l,\cdot}\big\Vert_{2}\left(\sigma_{r}\rho^{2}\frac{1}{\sigma_{r}}+\frac{1}{\sigma_{r}}\left(\sigma_{r}\rho^{2}+\sigma_{r}\rho\sqrt{\frac{r}{n}}\log n\right)\right)\\
 & \lesssim\big\Vert(\widehat{\bm{V}}\widehat{\bm{\Sigma}}\bm{H}_{U}-(T^{-1/2}\bm{X})^{\top}\bm{U})_{l,\cdot}\big\Vert_{2}\frac{1}{\sigma_{r}}+\big\Vert\widehat{\bm{V}}_{l,\cdot}\big\Vert_{2}\left(\rho^{2}+\rho\sqrt{\frac{r}{n}}\log n\right),
\end{align*}
where (i) is owing to the inequalities in Lemmas \ref{Lemma R H for U V}
and \ref{Lemma Sigma hat tilde H R}, and the fact that $\big\Vert\bm{\Lambda}-(\bm{R}_{V})^{\top}\widehat{\bm{\Sigma}}\bm{R}_{U}\big\Vert_{2}=\big\Vert(\bm{\Lambda}-(\bm{R}_{V})^{\top}\widehat{\bm{\Sigma}}\bm{R}_{U})^{\top}\big\Vert_{2}=\big\Vert\bm{\Lambda}-(\bm{R}_{U})^{\top}\widehat{\bm{\Sigma}}\bm{R}_{V}\big\Vert_{2}$.
Using $\sigma_{i}(\bm{H}_{V})\asymp1$ in Lemma \ref{Lemma R H for U V},
we have that $\big\Vert\widehat{\bm{V}}_{l,\cdot}\big\Vert_{2}\lesssim\big\Vert(\widehat{\bm{V}}\bm{H}_{V})_{l,\cdot}\big\Vert_{2}\leq\big\Vert\bm{V}_{l,\cdot}\big\Vert_{2}+\big\Vert(\widehat{\bm{V}}\bm{H}_{V}-\bm{V})_{l,\cdot}\big\Vert_{2}$.
Then, we obtain by Lemma \ref{Lemma LOO UV Sig H} that, 
\begin{align*}
\left\Vert (\bm{\Psi}_{V})_{l,\cdot}\right\Vert _{2} & \lesssim\sigma_{r}\rho^{2}\left(\sqrt{\frac{r}{n}}\log^{3/2}n+\log n\big\Vert\bm{V}_{l,\cdot}\big\Vert_{2}+\log n\big\Vert(\widehat{\bm{V}}\bm{H}_{V}-\bm{V})_{l,\cdot}\big\Vert_{2}\right)\\
 & +\left(\big\Vert\bm{V}_{l,\cdot}\big\Vert_{2}+\big\Vert(\widehat{\bm{V}}\bm{H}_{V}-\bm{V})_{l,\cdot}\big\Vert_{2}\right)\left(\rho^{2}+\rho\sqrt{\frac{r}{n}}\log n\right).
\end{align*}

Finally, we obtain by Lemma \ref{Lemma UH-U VH-V} that, 
\begin{align*}
 & \left\Vert (\bm{\Psi}_{V})_{l,\cdot}\right\Vert _{2}\\
 & \lesssim\rho^{2}\left(\sqrt{\frac{r}{n}}\log^{3/2}n+\log n\big\Vert\bm{V}_{l,\cdot}\big\Vert_{2}+\log n\left(\begin{array}{c}
\rho\sqrt{\frac{r}{n}}\log n\\
+(\rho^{2}+\rho\sqrt{\frac{r}{n}})\log n\big\Vert\bm{V}_{l,\cdot}\big\Vert_{2}
\end{array}\right)\right)\\
 & +\left(\begin{array}{c}
\big\Vert\bm{V}_{l,\cdot}\big\Vert_{2}+\rho\sqrt{\frac{r}{n}}\log n\\
+(\rho^{2}+\rho\sqrt{\frac{r}{n}})\log n\big\Vert\bm{V}_{l,\cdot}\big\Vert_{2}
\end{array}\right)\left(\rho^{2}+\rho\sqrt{\frac{r}{n}}\log n\right)\\
 & \lesssim\rho^{2}\sqrt{\frac{r}{n}}\log^{3/2}n+\left(\rho^{2}\log n+\rho\sqrt{\frac{r}{n}}\log n\right)\big\Vert\bm{V}_{l,\cdot}\big\Vert_{2},
\end{align*}
where, in the last inequality, we use the assumptions that, $\rho\sqrt{\log n}\ll1$
and $r\log n\ll n$ owing to (\ref{SNR 2-norm logn}) and (\ref{SNR for UH-U VH-V self}).

\section{Proof of Corollary \ref{corollary:error bound B F}}

For the first-order terms $\bm{G}_{U}$ and $\bm{G}_{V}$ in Theorem
\ref{Thm UV 1st approx row-wise error}, we obtain by Lemma \ref{Lemma Z norms A B}
that, with probability at least $1-O(n^{-2})$, 
\begin{align*}
\left\Vert (\bm{G}_{V})_{l,\cdot}\right\Vert _{2} & =\big\Vert(T^{-1/2}\bm{E}^{\top}\bm{U}\bm{\Lambda}^{-1})_{l,\cdot}\big\Vert_{2}=\big\Vert(T^{-1/2}\bm{Z}^{\top}\bm{\Sigma}_{\varepsilon}^{1/2}\bm{U}\bm{\Lambda}^{-1})_{l,\cdot}\big\Vert_{2}\\
 & =\big\Vert T^{-1/2}(\bm{Z}^{\top})_{l,\cdot}\bm{\Sigma}_{\varepsilon}^{1/2}\bm{U}\bm{\Lambda}^{-1}\big\Vert_{2}\lesssim\frac{1}{\sqrt{T}}\big\Vert\bm{\Sigma}_{\varepsilon}^{1/2}\bm{U}\bm{\Lambda}^{-1}\big\Vert_{\mathrm{F}}\log n\\
 & \lesssim\frac{1}{\sqrt{T}}\left\Vert \bm{\Sigma}_{\varepsilon}^{1/2}\right\Vert _{2}\big\Vert\bm{U}\big\Vert_{\mathrm{F}}\frac{1}{\sigma_{r}}\log n=\rho\sqrt{\frac{r}{n}}\log n,
\end{align*}
and similarly we obtain that 
\[
\left\Vert (\bm{G}_{U})_{k,\cdot}\right\Vert _{2}\lesssim\frac{1}{\vartheta_{k}\sqrt{T}}\sqrt{r}\log n.
\]

\textit{Step 1 -- Derive the error bounds for factors.}

Using the upper bounds for $\big\Vert\widehat{\bm{U}}\bm{R}_{U}-\bm{U}\big\Vert_{\mathrm{F}}$
and $\big\Vert\widehat{\bm{V}}\bm{R}_{V}-\bm{V}\big\Vert_{\mathrm{F}}$
in Lemma \ref{Lemma R H for U V} and the upper bound for $\big\Vert(\bm{R}_{U})^{\top}\widehat{\bm{\Sigma}}\bm{R}_{V}-\bm{\Lambda}\big\Vert_{2}$
in Lemma \ref{Lemma Sigma hat tilde H R}, we have that 
\[
\frac{1}{T}\big\Vert\widehat{\bm{F}}-\bm{F}\bm{J}(\bm{R}_{V})^{\top}\big\Vert{}_{\mathrm{F}}^{2}=\big\Vert\widehat{\bm{V}}\bm{R}_{V}(\bm{R}_{V})^{\top}-\bm{V}(\bm{R}_{V})^{\top}\big\Vert{}_{\mathrm{F}}^{2}\lesssim\frac{1}{\theta^{2}}\frac{n}{T}r.
\]
For $\widehat{\bm{V}}\bm{R}_{V}-\bm{V}=\bm{G}_{V}+\bm{\Psi}_{V}$,
using the fact that $\big\Vert\bm{V}\big\Vert_{2,\infty}\lesssim\sqrt{\frac{\log n}{T}}$
in Lemma \ref{Lemma SVD BF good event}, we obtain that, 
\begin{align*}
\left\Vert \bm{\Psi}_{V}\right\Vert _{2,\infty} & =\sup_{1\leq l\leq T}\left\Vert (\bm{\Psi}_{V})_{l,\cdot}\right\Vert _{2}\lesssim\rho^{2}\sqrt{\frac{r}{n}}\log^{3/2}n+\left(\rho^{2}+\rho\sqrt{\frac{r}{n}}\right)\log n\big\Vert\bm{V}\big\Vert_{2,\infty}\\
 & \lesssim\rho^{2}\sqrt{\frac{r}{n}}\log^{3/2}n+\left(\rho^{2}+\rho\sqrt{\frac{r}{n}}\log n\right)\sqrt{\frac{\log n}{T}}\lesssim\rho\sqrt{\frac{r}{n}}\log n,
\end{align*}
where the last inequality is because $\rho\sqrt{\log n}\ll1$ and
$\log n\ll T$. Since $\left\Vert \bm{G}_{V}\right\Vert _{2,\infty}\lesssim\rho\sqrt{\frac{r}{n}}\log n$,
we obtain that $\big\Vert\widehat{\bm{V}}\bm{R}_{V}-\bm{V}\big\Vert_{2,\infty}\leq\left\Vert \bm{G}_{V}\right\Vert _{2,\infty}+\left\Vert \bm{\Psi}_{V}\right\Vert _{2,\infty}\lesssim\rho\sqrt{\frac{r}{n}}\log n$,
which implies the row-wise error bound for $\big\Vert(\widehat{\bm{F}}-\bm{F}\bm{R}_{F})_{t,\cdot}\big\Vert_{2}$.

\textit{Step 2 -- Derive the error bounds for factor loadings.}

For $\widehat{\bm{U}}\bm{R}_{U}-\bm{U}=\bm{G}_{U}+\bm{\Psi}_{U}$,
the problem boils down to comparing the upper bounds of $\left\Vert (\bm{G}_{U})_{k,\cdot}\right\Vert _{2}$
and $\left\Vert (\bm{\Psi}_{U})_{k,\cdot}\right\Vert _{2}$. For the
first term in the upper bound of $\left\Vert (\bm{\Psi}_{U})_{k,\cdot}\right\Vert _{2}$,
we obtain by $\omega\sqrt{\log n}\ll1$ that $\omega_{k}\omega\sqrt{\frac{r}{n}}\log^{3/2}n=\omega_{k}\sqrt{\frac{r}{n}}\log n(\omega\sqrt{\log n})\ll\omega_{k}\sqrt{\frac{r}{n}}\log n$.
Next, for the second and the third terms in the upper bound of $\left\Vert (\bm{\Psi}_{U})_{k,\cdot}\right\Vert _{2}$,
we have that 
\[
\frac{(\rho^{2}+\rho\sqrt{\frac{r}{n}}\log n)\big\Vert\bm{U}_{k,\cdot}\big\Vert_{2}}{\omega_{k}\sqrt{\frac{r}{n}}\log n}=(\frac{\rho}{\sqrt{\frac{r}{n}}\log n}+1)\frac{\rho}{\omega_{k}}\big\Vert\bm{U}_{k,\cdot}\big\Vert_{2}\leq(\rho\sqrt{n}+1)\frac{\rho}{\omega_{k}}\big\Vert\bar{\bm{U}}_{k,\cdot}\big\Vert_{2}\lesssim1,
\]
and 
\[
\frac{\omega_{k}\omega\log n\big\Vert\bm{U}\big\Vert_{2,\infty}}{\omega_{k}\sqrt{\frac{r}{n}}\log n}=\omega\sqrt{\frac{n}{r}}\big\Vert\bm{U}\big\Vert_{2,\infty}\leq\omega\sqrt{n}\big\Vert\bar{\bm{U}}\big\Vert_{2,\infty}\lesssim1.
\]
So we obtain that 
\[
\big\Vert(\widehat{\bm{U}}\bm{R}_{U}-\bm{U})_{k,\cdot}\big\Vert_{2}\lesssim\omega_{k}\sqrt{\frac{r}{n}}\log n\text{\qquad and\qquad}\big\Vert\widehat{\bm{U}}\bm{R}_{U}-\bm{U}\big\Vert_{2,\infty}\lesssim\omega\sqrt{\frac{r}{n}}\log n.
\]
The row-wise error bound for $\widehat{\bm{B}}-\bm{B}\bm{R}_{B}$
follows from combining the error bounds of $\widehat{\bm{U}}\bm{R}_{U}-\bm{U}$
and $(\bm{R}_{U})^{\top}\widehat{\bm{\Sigma}}\bm{R}_{V}-\bm{\Lambda}$.
Then the averaged estimation error bound follows from averaging the
row-wise error bounds.

%% file: appendix_factor_inference.tex
\section{Inference for the factors and the factor loadings}

\subsection{Proof of Corollary \ref{Thm F inference}}

\textit{Step 1 -- bounding $\mathbb{P}((\widehat{\bm{V}}\bm{R}_{V}\bm{J}^{-1}-\frac{1}{\sqrt{T}}\bm{F})_{t,\cdot}\in\mathcal{C})$.}

It follows from Lemma \ref{Lemma SVD BF good event} that, there exists
a $\sigma(\bm{F})$-measurable matrix $\bm{J}$ satisfying that $\bm{V}=\frac{1}{\sqrt{T}}\bm{FJ}$
and $\bm{J}$ is a $r\times r$ invertible matrix satisfying that
$\sigma_{i}(\bm{J})\asymp1$ for $i=1,2,\ldots,r$. We define 
\[
\bm{R}_{F}:=\bm{J}\bm{R}_{V}^{\top},
\]
then we obtain by Theorem \ref{Thm UV 1st approx row-wise error}
that 
\[
\widehat{\bm{F}}-\bm{F}\bm{R}_{F}=\sqrt{T}(\widehat{\bm{V}}\bm{R}_{V}\bm{J}^{-1}-\bm{V}\bm{J}^{-1})\bm{R}_{F}=\sqrt{T}(\bm{G}_{V}+\bm{\Psi}_{V})\bm{J}^{-1}\bm{R}_{F}=(\sqrt{T}\bm{G}_{V}\bm{J}^{-1}+\sqrt{T}\bm{\Psi}_{V}\bm{J}^{-1})\bm{R}_{F}.
\]
We also note that $\bm{\Sigma}_{F,t}=\bm{R}_{V}\bm{\Lambda}^{-1}\bm{U}^{\top}\bm{\Sigma}_{\varepsilon}\bm{U}\bm{\Lambda}^{-1}\bm{R}_{V}^{\top}=\bm{R}_{V}\bm{J}^{\top}\bm{\Sigma}^{-1}\bar{\bm{U}}^{\top}\bm{\Sigma}_{\varepsilon}\bar{\bm{U}}\bm{\Sigma}^{-1}\bm{J}\bm{R}_{V}^{\top}=\bm{R}_{F}^{\top}(\bm{\Sigma}^{-1}\bar{\bm{U}}^{\top}\bm{\Sigma}_{\varepsilon}\bar{\bm{U}}\bm{\Sigma}^{-1})\bm{R}_{F}$.
Denote by $(\bar{\bm{U}}^{\top}\bm{\Sigma}_{\varepsilon}\bar{\bm{U}})^{-1/2}$
the matrix such that $(\bar{\bm{U}}^{\top}\bm{\Sigma}_{\varepsilon}\bar{\bm{U}})^{-1/2}(\bar{\bm{U}}^{\top}\bm{\Sigma}_{\varepsilon}\bar{\bm{U}})(\bar{\bm{U}}^{\top}\bm{\Sigma}_{\varepsilon}\bar{\bm{U}})^{-1/2}=\bm{I}_{r}$.
So, we have that 
\begin{align*}
 & \sup_{\mathcal{C}\in\mathscr{C}^{r}}\left\vert \mathbb{P}((\widehat{\bm{F}}-\bm{F}\bm{R}_{F})_{t,\cdot}\in\mathcal{C})-\mathbb{P}(\mathcal{N}(0,\bm{\Sigma}_{F,t})\in\mathcal{C})\right\vert \\
= & \sup_{\mathcal{C}\in\mathscr{C}^{r}}\left\vert \mathbb{P}(\sqrt{T}(\bm{G}_{V}\bm{J}^{-1}+\bm{\Psi}_{V}\bm{J}^{-1})_{t,\cdot}\in\mathcal{C})-\mathbb{P}(\mathcal{N}(0,\bm{\Sigma}^{-1}\bar{\bm{U}}^{\top}\bm{\Sigma}_{\varepsilon}\bar{\bm{U}}\bm{\Sigma}^{-1})\in\mathcal{C})\right\vert \\
= & \sup_{\mathcal{C}\in\mathscr{C}^{r}}\left\vert \mathbb{P}(\sqrt{T}(\bar{\bm{U}}^{\top}\bm{\Sigma}_{\varepsilon}\bar{\bm{U}})^{-1/2}\bm{\Sigma}((\bm{G}_{V}\bm{J}^{-1}+\bm{\Psi}_{V}\bm{J}^{-1})_{t,\cdot})^{\top}\in\mathcal{C})-\mathbb{P}(\mathcal{N}(0,\bm{I}_{r})\in\mathcal{C})\right\vert ,
\end{align*}
where the last equality is because $\mathscr{C}^{r}$ is invariant
under nondegenerate linear transformation and the matrix $(\bar{\bm{U}}^{\top}\bm{\Sigma}_{\varepsilon}\bar{\bm{U}})^{-1/2}\bm{\Sigma}$
is not degenerate.

To study $\sqrt{T}(\bar{\bm{U}}^{\top}\bm{\Sigma}_{\varepsilon}\bar{\bm{U}})^{-1/2}\bm{\Sigma}((\bm{G}_{V}\bm{J}^{-1}+\bm{\Psi}_{V}\bm{J}^{-1})_{t,\cdot})^{\top}$,
we define the vector $\bm{K}_{V,t}$ and the scalar $\varepsilon_{1}$,
respectively, as below: 
\[
\bm{K}_{V,t}:=\sqrt{T}(\bar{\bm{U}}^{\top}\bm{\Sigma}_{\varepsilon}\bar{\bm{U}})^{-1/2}\bm{\Sigma}((\bm{G}_{V}\bm{J}^{-1})_{t,\cdot})^{\top}\text{\qquad and\qquad}\varepsilon_{1}:=\big\Vert\sqrt{T}(\bar{\bm{U}}^{\top}\bm{\Sigma}_{\varepsilon}\bar{\bm{U}})^{-1/2}\bm{\Sigma}((\bm{\Psi}_{V}\bm{J}^{-1})_{t,\cdot})^{\top}\big\Vert_{2}.
\]
Then we have the following inequality 
\[
\mathbb{P}(\bm{K}_{V,t}\in\mathcal{C}^{-\varepsilon_{1}})\leq\mathbb{P}(\sqrt{T}(\bar{\bm{U}}^{\top}\bm{\Sigma}_{\varepsilon}\bar{\bm{U}})^{-1/2}\bm{\Sigma}((\bm{G}_{V}\bm{J}^{-1}+\bm{\Psi}_{V}\bm{J}^{-1})_{t,\cdot})^{\top}\in\mathcal{C})\leq\mathbb{P}(\bm{K}_{V,t}\in\mathcal{C}^{\varepsilon_{1}}).
\]
For any non-empty convex set $\mathcal{C}\in\mathscr{C}^{r}$, we
define $\mathcal{C^{\varepsilon}}:=\{x\in\mathbb{R}^{r}:\delta_{\mathcal{C}}(x)\leq\varepsilon\}$
where $\delta_{\mathcal{C}}(x)$ is the signed distance of the point
$x$ to the set $\mathcal{C}$. We will show later that $\varepsilon_{1}$
is negligible.

\textit{Step 2 -- Proving the approximate Gaussianity of $\bm{K}_{V,t}$.}

In this step, we establish the upper bound for the constant $\tau_{V}$
defined below, 
\[
\tau_{V}:=\sup_{\mathcal{C}_{1}\in\mathscr{C}^{r}}\left\vert \mathbb{P}(\bm{K}_{V,t}\in\mathcal{C}_{1})-\mathbb{P}(\mathcal{N}(0,\bm{I}_{r})\in\mathcal{C}_{1})\right\vert .
\]
We obtain by Theorem \ref{Thm UV 1st approx row-wise error} that
$\bm{G}_{V}:=T^{-1/2}\bm{E}^{\top}\bm{U}\bm{\Lambda}^{-1}$. It follows
from Lemma \ref{Lemma SVD BF good event} that $\bm{J}=\bm{\Sigma}\bm{Q}\bm{\Lambda}^{-1}$
and $\bm{U}=\bar{\bm{U}}\bm{Q}$, where $\bm{Q}\in\mathcal{O}^{r\times r}$
is a rotation matrix, i.e., $\bm{QQ}^{\top}=\bm{Q}^{\top}\bm{Q}=\bm{I}_{r}$.
So, we obtain that $\bm{U}\bm{Q}^{-1}=\bar{\bm{U}}$ and 
\begin{align*}
\sqrt{T}\bm{G}_{V}\bm{J}^{-1} & =\bm{E}^{\top}\bm{U}\bm{\Lambda}^{-1}\bm{\Lambda}\bm{Q}^{-1}\bm{\Sigma}^{-1}=\bm{E}^{\top}\bm{U}\bm{Q}^{-1}\bm{\Sigma}^{-1}\\
 & =\bm{E}^{\top}\bar{\bm{U}}\bm{\Sigma}^{-1}=\bm{Z}^{\top}\bm{\Sigma}_{\varepsilon}^{1/2}\bar{\bm{U}}\bm{\Sigma}^{-1}=\bm{Z}^{\top}\bm{H}\text{\qquad with\qquad}\bm{H}:=\bm{\Sigma}_{\varepsilon}^{1/2}\bar{\bm{U}}\bm{\Sigma}^{-1}.
\end{align*}
Going back to our target vector, we have that 
\begin{align*}
\bm{K}_{V,t} & =(\bar{\bm{U}}^{\top}\bm{\Sigma}_{\varepsilon}\bar{\bm{U}})^{-1/2}\bm{\Sigma}(\sqrt{T}(\bm{G}_{V}\bm{J}^{-1})_{t,\cdot})^{\top}=(\bar{\bm{U}}^{\top}\bm{\Sigma}_{\varepsilon}\bar{\bm{U}})^{-1/2}\bm{\Sigma}((\bm{Z}^{\top}\bm{H})_{t,\cdot})^{\top}\\
 & =(\bar{\bm{U}}^{\top}\bm{\Sigma}_{\varepsilon}\bar{\bm{U}})^{-1/2}\bm{\Sigma}\bm{H}^{\top}\bm{Z}_{\cdot,t}=(\bar{\bm{U}}^{\top}\bm{\Sigma}_{\varepsilon}\bar{\bm{U}})^{-1/2}\bm{\Sigma}\bm{\Sigma}^{-1}(\bm{\Sigma}_{\varepsilon}^{1/2}\bar{\bm{U}})^{\top}\bm{Z}_{\cdot,t}\\
 & =\sum_{k=1}^{N}\bm{Z}_{k,t}(\bar{\bm{U}}^{\top}\bm{\Sigma}_{\varepsilon}\bar{\bm{U}})^{-1/2}((\bm{\Sigma}_{\varepsilon}^{1/2}\bar{\bm{U}})_{k,\cdot})^{\top}.
\end{align*}
Note that both $\bar{\bm{U}}$ and $\bm{\Sigma}_{\varepsilon}$ are
deterministic. So, we obtain that the vector $\bm{K}_{V,t}$ is the
summation of $N$ independent and mean zero vectors, and the covariance
matrix of $\bm{K}_{V,t}$ is given by 
\begin{align*}
cov(\bm{K}_{V,t}) & =\sum_{k=1}^{N}(\bar{\bm{U}}^{\top}\bm{\Sigma}_{\varepsilon}\bar{\bm{U}})^{-1/2}((\bm{\Sigma}_{\varepsilon}^{1/2}\bar{\bm{U}})_{k,\cdot})^{\top}(\bm{\Sigma}_{\varepsilon}^{1/2}\bar{\bm{U}})_{k,\cdot}(\bar{\bm{U}}^{\top}\bm{\Sigma}_{\varepsilon}\bar{\bm{U}})^{-1/2}\\
 & =(\bar{\bm{U}}^{\top}\bm{\Sigma}_{\varepsilon}\bar{\bm{U}})^{-1/2}(\bar{\bm{U}}^{\top}\bm{\Sigma}_{\varepsilon}\bar{\bm{U}})(\bar{\bm{U}}^{\top}\bm{\Sigma}_{\varepsilon}\bar{\bm{U}})^{-1/2}=\bm{I}_{r}.
\end{align*}

Given the above analysis of $\bm{K}_{V,t}$, we are now ready to bound
the term $\tau_{V}$:

(i) When all entries of $\bm{Z}$ are Gaussian, we have that, the
vector $\bm{K}_{V,t}$ is also a Gaussian vector because it is a linear
combination of Gaussian vectors. Since $cov(\bm{K}_{V,t})=\bm{I}_{r}$,
we obtain that the law of $\bm{K}_{V,t}$ is the $r$-dimensional
standard Gaussian law $\mathcal{N}(0,\bm{I}_{r})$. Thus, we have
that $\tau_{V}=0$ when all entries of $\bm{Z}$ are Gaussian.

(ii) When the entries of $\bm{Z}$ are sub-Gaussian, we obtain by
the Berry-Esseen theorem (Theorem 1.1 in \citet{raivc2019multivariate})
that 
\[
\tau_{V}\lesssim r^{1/4}\tau_{V}^{(0)},
\]
and 
\begin{align*}
\tau_{V}^{(0)} & =\sum_{k=1}^{N}\mathbb{E}\left[\big\Vert\bm{Z}_{k,t}(\bar{\bm{U}}^{\top}\bm{\Sigma}_{\varepsilon}\bar{\bm{U}})^{-1/2}((\bm{\Sigma}_{\varepsilon}^{1/2}\bar{\bm{U}})_{k,\cdot})^{\top}\big\Vert_{2}^{3}\right]=\sum_{k=1}^{N}\big\Vert(\bm{\Sigma}_{V,t})^{-1/2}(\bm{H}_{k,\cdot})^{\top}\big\Vert_{2}^{3}\mathbb{E}\left[\big\Vert\bm{Z}_{k,t}\big\Vert_{2}^{3}\right]\\
 & \overset{\text{(i)}}{\lesssim}\sum_{k=1}^{N}\big\Vert(\bar{\bm{U}}^{\top}\bm{\Sigma}_{\varepsilon}\bar{\bm{U}})^{-1/2}((\bm{\Sigma}_{\varepsilon}^{1/2}\bar{\bm{U}})_{k,\cdot})^{\top}\big\Vert_{2}^{3}\mathbb{E}\left[|\bm{Z}_{k,t}|_{2}^{3}\right]\\
 & \leq\max_{1\leq k\leq N}\big\Vert(\bar{\bm{U}}^{\top}\bm{\Sigma}_{\varepsilon}\bar{\bm{U}})^{-1/2}((\bm{\Sigma}_{\varepsilon}^{1/2}\bar{\bm{U}})_{k,\cdot})^{\top}\big\Vert_{2}\cdot\sum_{k=1}^{N}\big\Vert(\bar{\bm{U}}^{\top}\bm{\Sigma}_{\varepsilon}\bar{\bm{U}})^{-1/2}((\bm{\Sigma}_{\varepsilon}^{1/2}\bar{\bm{U}})_{k,\cdot})^{\top}\big\Vert_{2}^{2}\\
 & \leq(\lambda_{\min}(\bar{\bm{U}}^{\top}\bm{\Sigma}_{\varepsilon}\bar{\bm{U}}))^{-1/2}\max_{1\leq k\leq N}\big\Vert(\bm{\Sigma}_{\varepsilon}^{1/2}\bar{\bm{U}})_{k,\cdot}\big\Vert_{2}\cdot tr[(\bar{\bm{U}}^{\top}\bm{\Sigma}_{\varepsilon}\bar{\bm{U}})^{-1/2}\bar{\bm{U}}^{\top}\bm{\Sigma}_{\varepsilon}\bar{\bm{U}}(\bar{\bm{U}}^{\top}\bm{\Sigma}_{\varepsilon}\bar{\bm{U}})^{-1/2}]\\
 & =(\lambda_{\min}(\bar{\bm{U}}^{\top}\bm{\Sigma}_{\varepsilon}\bar{\bm{U}}))^{-1/2}\big\Vert\bm{\Sigma}_{\varepsilon}^{1/2}\bar{\bm{U}}\big\Vert_{2,\infty}r,
\end{align*}
where (i) is because the entries of $\bm{Z}$ are sub-Gaussian and
their sub-Gaussian norms satisfy $\left\Vert Z_{i,t}\right\Vert _{\psi_{2}}=O(1)$.
Then, since $\bar{\bm{U}}$ has full column rank, we have that 
\begin{align*}
\lambda_{\min}(\bar{\bm{U}}^{\top}\bm{\Sigma}_{\varepsilon}\bar{\bm{U}}) & =\inf_{\left\Vert \bm{x}\right\Vert _{2}=1}\bm{x}^{\top}\bar{\bm{U}}^{\top}\bm{\Sigma}_{\varepsilon}\bar{\bm{U}}\bm{x}=\inf_{\left\Vert \bm{x}\right\Vert _{2}=1}(\bar{\bm{U}}\bm{x})^{\top}\bm{\Sigma}_{\varepsilon}\bar{\bm{U}}\bm{x}=\inf_{\left\Vert \bm{x}\right\Vert _{2}=1}\bm{x}^{\top}\bm{\Sigma}_{\varepsilon}\bm{x}=\lambda_{\min}(\bm{\Sigma}_{\varepsilon}).
\end{align*}
Recall that $\rho:=\sqrt{\left\Vert \bm{\Sigma}_{\varepsilon}\right\Vert _{2}}\sqrt{n}/(\sigma_{r}\sqrt{T})$.
So, we conclude that 
\[
\tau_{V}^{(0)}\lesssim\frac{1}{\lambda_{\text{min}}(\bm{\Sigma}_{\varepsilon})}\big\Vert\bm{\Sigma}_{\varepsilon}^{1/2}\bar{\bm{U}}\big\Vert_{2,\infty}r=\sqrt{\kappa_{\varepsilon}}\frac{\big\Vert\bm{\Sigma}_{\varepsilon}^{1/2}\bar{\bm{U}}\big\Vert_{2,\infty}}{\big\Vert\bm{\Sigma}_{\varepsilon}^{1/2}\big\Vert_{2}}r,\text{\qquad and\qquad\ensuremath{\tau_{V}\lesssim r^{1/4}\tau_{V}^{(0)}\lesssim\sqrt{\kappa_{\varepsilon}}r^{5/4}\frac{\big\Vert\bm{\Sigma}_{\varepsilon}^{1/2}\bar{\bm{U}}\big\Vert_{2,\infty}}{\big\Vert\bm{\Sigma}_{\varepsilon}^{1/2}\big\Vert_{2}}}},
\]
where we use the fact that $\left\Vert \bm{\Sigma}_{\varepsilon}\right\Vert _{2}=\lambda_{\text{max}}(\bm{\Sigma}_{\varepsilon})$
and $\kappa_{\varepsilon}=\lambda_{\text{max}}(\bm{\Sigma}_{\varepsilon})/\lambda_{\text{min}}(\bm{\Sigma}_{\varepsilon})$
is the condition number of $\bm{\Sigma}_{\varepsilon}$.

\textit{Step 3 -- Establishing the validity of confidence regions.}

Since we obtained in Step 2 that $\left\vert \mathbb{P}(\bm{K}_{V,t}\in\mathcal{C}_{1})-\mathbb{P}(\mathcal{N}(0,\bm{I}_{r})\in\mathcal{C}_{1})\right\vert \leq\tau_{V}$
holds for any convex set $\mathcal{C}_{1}$, we let $\mathcal{C}_{1}$
be $\mathcal{C}^{-\varepsilon_{1}}$ and $\mathcal{C}^{\varepsilon_{1}}$,
respectively to obtain the probabilities bounds as below 
\[
\mathbb{P}(\mathcal{N}(0,\bm{I}_{r})\in\mathcal{C}^{-\varepsilon_{1}})-\tau_{V}\leq\mathbb{P}(\sqrt{T}(\bar{\bm{U}}^{\top}\bm{\Sigma}_{\varepsilon}\bar{\bm{U}})^{-1/2}\bm{\Sigma}((\bm{G}_{V}\bm{J}^{-1}+\bm{\Psi}_{V}\bm{J}^{-1})_{t,\cdot})^{\top}\in\mathcal{C})\leq\mathbb{P}(\mathcal{N}(0,\bm{I}_{r})\in\mathcal{C}^{\varepsilon_{1}})+\tau_{V}.
\]
Next, by Theorem 1.2 in \citet{raivc2019multivariate} for multi-dimensional
standard Gaussian vector, we obtain 
\begin{align*}
 & \mathbb{P}(\mathcal{N}(0,\bm{I}_{r})\in\mathcal{C})-(K_{1}r^{1/4}+K_{0})\varepsilon_{1}-\tau_{V}\\
\leq & \mathbb{P}(\sqrt{T}(\bar{\bm{U}}^{\top}\bm{\Sigma}_{\varepsilon}\bar{\bm{U}})^{-1/2}\bm{\Sigma}((\bm{G}_{V}\bm{J}^{-1}+\bm{\Psi}_{V}\bm{J}^{-1})_{t,\cdot})^{\top}\in\mathcal{C})\\
\leq & \mathbb{P}(\mathcal{N}(0,\bm{I}_{r})\in\mathcal{C})+(K_{1}r^{1/4}+K_{0})\varepsilon_{1}+\tau_{V}.
\end{align*}
Since the above inequality holds for any convex set $\mathcal{C}\in\mathscr{C}^{r}$,
we obtain that 
\begin{align*}
 & \sup_{\mathcal{C}\in\mathscr{C}^{r}}\left\vert \mathbb{P}((\widehat{\bm{F}}-\bm{F}\bm{R}_{F})_{t,\cdot}\in\mathcal{C})-\mathbb{P}(\mathcal{N}(0,\bm{\Sigma}_{V,t})\in\mathcal{C})\right\vert \\
 & =\sup_{\mathcal{C}\in\mathscr{C}^{r}}\left\vert \mathbb{P}(\sqrt{T}(\bar{\bm{U}}^{\top}\bm{\Sigma}_{\varepsilon}\bar{\bm{U}})^{-1/2}\bm{\Sigma}((\bm{G}_{V}\bm{J}^{-1}+\bm{\Psi}_{V}\bm{J}^{-1})_{t,\cdot})^{\top}\in\mathcal{C})-\mathbb{P}(\mathcal{N}(0,\bm{I}_{r})\in\mathcal{C})\right\vert \\
 & \leq(K_{1}r^{1/4}+K_{0})\varepsilon_{1}+\tau_{V}\lesssim r^{1/4}\varepsilon_{1}+\tau_{V}.
\end{align*}

Finally, we derive the upper bound for $\varepsilon_{1}:=\frac{1}{\sqrt{T}}\big\Vert(\bm{\Sigma}_{V,t})^{-1/2}((\bm{\Psi}_{V}\bm{J}^{-1})_{t,\cdot})^{\top}\big\Vert_{2}$.

Using the fact that $\sigma_{i}(\bm{J})\asymp1$ for $i=1,2,\ldots,r$
in Lemma \ref{Lemma SVD BF good event}, we obtain that, with probability
at least $1-O(n^{-2})$, it holds 
\begin{align*}
\varepsilon_{1} & =\big\Vert\sqrt{T}(\bar{\bm{U}}^{\top}\bm{\Sigma}_{\varepsilon}\bar{\bm{U}})^{-1/2}\bm{\Sigma}((\bm{\Psi}_{V}\bm{J}^{-1})_{t,\cdot})^{\top}\big\Vert_{2}\leq\frac{\sqrt{T}\big\Vert(\bm{\Psi}_{V}\bm{J}^{-1})_{t,\cdot}\bm{\Sigma}\big\Vert_{2}}{(\lambda_{\text{min}}(\bm{\Sigma}_{\varepsilon}))^{1/2}}=\frac{\sqrt{\kappa_{\varepsilon}}\sqrt{T}}{\big\Vert\bm{\Sigma}_{\varepsilon}^{1/2}\big\Vert_{2}}\big\Vert(\bm{\Psi}_{V}\bm{J}^{-1})_{t,\cdot}\bm{\Sigma}\big\Vert_{2}.
\end{align*}
Then, using the upper bound $\big\Vert\bm{V}\big\Vert_{2,\infty}\lesssim\sqrt{\frac{\log n}{T}}$
obtained in Lemma \ref{Lemma SVD BF good event} and the formula of
$\bm{\Psi}_{V}$ in the proof of Theorem \ref{Thm UV 1st approx row-wise error},
we obtain that 
\begin{align*}
\big\Vert(\bm{\Psi}_{V}\bm{J}^{-1})_{t,\cdot}\bm{\Sigma}\big\Vert_{2} & \lesssim\sigma_{r}\rho^{2}\sqrt{\frac{r}{n}}\log^{3/2}n+\sigma_{r}\left(\rho^{2}+\rho\sqrt{\frac{r}{n}}\log n\right)\big\Vert\bm{V}_{l,\cdot}\big\Vert_{2}\\
 & \lesssim\sigma_{r}\rho^{2}\sqrt{\frac{r}{n}}\log^{3/2}n+\sigma_{r}\left(\rho^{2}+\rho\sqrt{\frac{r}{n}}\log n\right)\sqrt{\frac{\log n}{T}}.
\end{align*}
So, we obtain that, with probability at least $1-O(n^{-2})$, it holds
\begin{align*}
r^{1/4}\varepsilon_{1} & \lesssim r^{1/4}\frac{\sqrt{\kappa_{\varepsilon}}\sqrt{T}\sigma_{r}}{\big\Vert\bm{\Sigma}_{\varepsilon}^{1/2}\big\Vert_{2}}\rho\cdot\left(\rho\sqrt{\frac{r}{n}}\log^{3/2}n+\left(\rho+\sqrt{\frac{r}{n}}\log n\right)\sqrt{\frac{\log n}{T}}\right)\\
 & =r^{1/4}\sqrt{\kappa_{\varepsilon}}\sqrt{n}\left(\rho\sqrt{\frac{r}{n}}\log^{3/2}n+\left(\rho+\sqrt{\frac{r}{n}}\log n\right)\sqrt{\frac{\log n}{T}}\right)\\
 & =\sqrt{\kappa_{\varepsilon}}\left(\rho+\left(\rho\sqrt{n}+1\right)\frac{1}{\sqrt{T}}\right)r\log^{3/2}n\\
 & \lesssim\sqrt{\kappa_{\varepsilon}}(\rho\sqrt{\frac{n}{T}}+\frac{1}{\sqrt{T}})r\log^{3/2}n.
\end{align*}
Thus, by a standard conditioning argument, we conclude that 
\begin{align*}
 & \sup_{\mathcal{C}\in\mathscr{C}^{r}}\left\vert \mathbb{P}((\widehat{\bm{V}}\bm{R}_{F}-\bm{F})_{t,\cdot}\in\mathcal{C})-\mathbb{P}(\mathcal{N}(0,\bm{\Sigma}_{V,t})\in\mathcal{C})\right\vert \\
 & \lesssim\sqrt{\kappa_{\varepsilon}}(\rho\sqrt{\frac{n}{T}}+\frac{1}{\sqrt{T}})r\log^{3/2}n+\tau_{V}+O(n^{-2}),
\end{align*}
which is our desired result. The sufficient conditions for $\sqrt{\kappa_{\varepsilon}}(\rho\sqrt{\frac{n}{T}}+\frac{1}{\sqrt{T}})r\log^{3/2}n\lesssim\delta$
to hold can be rewritten as two inequalities as follows: $C_{0}n^{-2}\leq\delta$,
$\sqrt{\kappa_{\varepsilon}}r^{5/4}\big\Vert\bm{\Sigma}_{\varepsilon}^{1/2}\bar{\bm{U}}\big\Vert_{2,\infty}\leq\delta\big\Vert\bm{\Sigma}_{\varepsilon}^{1/2}\big\Vert_{2}$,
\[
C_{0}\sqrt{\kappa_{\varepsilon}}\frac{\sqrt{\left\Vert \bm{\Sigma}_{\varepsilon}\right\Vert _{2}}}{\sigma_{r}}\frac{n}{T}r\log^{3/2}n\leq\delta,\qquad\text{and}\qquad C_{0}\kappa_{\varepsilon}r^{2}\log^{3}n\ll\delta^{2}T.
\]

\subsection{Proof of Corollary \ref{Thm B inference}}

The proof is similar to that for the factors. So we omit the proof
due to the limit of space. 

%% file: appendix_factor_test.tex
\section{\label{Proof of Thm factor test}Proof of Theorem \ref{Thm factor test plug-in Chi-sq}:
Factor test}

\subsection{Some useful lemmas}

To prove Theorem \ref{Thm factor test plug-in Chi-sq}, we collect
some useful lemmas as preparations. Recall that, we already showed
in Lemma \ref{Lemma SVD BF good event} that $\bm{V}=\frac{1}{\sqrt{T}}\bm{FJ}$,
implying that $\bm{\bm{V}}_{S,\cdot}=\frac{1}{\sqrt{T}}\bm{F}_{S,\cdot}\bm{J}$
and thus the column space of $\bm{F}_{S,\cdot}$ is closely related
to that of $\bm{\bm{V}}_{S,\cdot}$. Then, we have the following lemma
for some properties of $\bm{\bm{V}}_{S,\cdot}$.

\begin{lemma} \label{Lemma V subset good event}Suppose that Assumptions
\ref{Assump_Bf_identification} and \ref{Assump_factor_f} hold. Assume
that (\ref{Assump T large logn}), i.e., $r+\log n\ll T,$ and 
\begin{equation}
r+\log n\ll|S|,\label{Assump subset size large logn}
\end{equation}
then we have that, there exists a $\sigma(\bm{F})$-measurable event
$\mathcal{E}_{S}$ with $\mathbb{P}(\mathcal{E}_{S})>1-O(n^{-2})$,
such that, the following properties hold when $\mathcal{E}_{0}\cap\mathcal{E}_{S}$
happens:

(i) $\sigma_{i}(\bm{\bm{\bm{V}}}_{S,\cdot})\asymp\sqrt{\frac{|S|}{T}}$
for $i=1,2,\ldots,r$, and $\text{rank}(\bm{\bm{\bm{V}}}_{S,\cdot})=r$,
i.e., $\bm{\bm{\bm{V}}}_{S,\cdot}$ has full column rank.

(ii) $\sigma_{i}((\bm{\bm{\bm{V}}}_{S,\cdot})^{+})\asymp\sqrt{\frac{T}{|S|}}$
for $i=1,2,\ldots,r$, where $(\bm{\bm{\bm{V}}}_{S,\cdot})^{+}$ is
the generalized inverse of $\bm{\bm{\bm{V}}}_{S,\cdot}$. \end{lemma} 
\begin{proof}
Since $\bm{F}=(\bm{f}_{1},\bm{f}_{2},\ldots,\bm{f}_{T})^{\top}$,
and $\bm{f}_{1},\bm{f}_{2},\ldots,\bm{f}_{T}$ are independent sub-Gaussian
random vectors under Assumption \ref{Assump_factor_f}, we obtain
by (4.22) in \citet{vershynin2016high} that, 
\begin{equation}
\left\Vert \frac{1}{|S|}(\bm{F}_{S,\cdot})^{\top}\bm{F}_{S,\cdot}-\bm{I}_{r}\right\Vert _{2}=\left\Vert \frac{1}{|S|}\sum_{i\in S}\bm{f}_{i}(\bm{f}_{i})^{\top}-\bm{I}_{r}\right\Vert _{2}\lesssim\sqrt{\frac{r}{|S|}}+\sqrt{\frac{\log n}{|S|}}+\frac{r}{|S|}+\frac{\log n}{|S|}\lesssim\sqrt{\frac{r+\log n}{|S|}},\label{FF-I_r inequality subset}
\end{equation}
with probability at least $1-O(n^{-2})$, where the last inequality
is owing to $(r+\log n)\ll|S|$ in (\ref{Assump subset size large logn}).
In particular, we let $\mathcal{E}_{0}$ be the event that (\ref{FF-I_r inequality subset})
happens, then we show that $\mathcal{E}_{S}$ satisfies all the requirements.

When (\ref{FF-I_r inequality subset}) happens, since $(r+\log n)\ll|S|$,
we obtain that $\big\Vert\frac{1}{|S|}(\bm{F}_{S,\cdot})^{\top}\bm{F}_{S,\cdot}-\bm{I}_{r}\big\Vert_{2}\ll1$,
and thus $\lambda_{i}(\frac{1}{|S|}(\bm{F}_{S,\cdot})^{\top}\bm{F}_{S,\cdot})\asymp1$
for $i=1,2,\ldots,r$, implying that $\sigma_{i}(\bm{F}_{S,\cdot})\asymp\sqrt{|S|}$
for $i=1,2,\ldots,r$. Since $\sigma_{r}(\bm{F}_{S,\cdot})\asymp\sqrt{|S|}$
implies that $\sigma_{\text{min}}(\bm{F}_{S,\cdot})>0$, we obtain
that $\text{rank}(\bm{F}_{S,\cdot})=r$, i.e., $\bm{F}_{S,\cdot}$
has full column rank.

Then, recall that, given $\mathcal{E}_{0}$ in Lemma \ref{Lemma SVD BF good event},
it holds $\bm{V}=\frac{1}{\sqrt{T}}\bm{FJ}$ and $\bm{J}$ is a $r\times r$
invertible matrix satisfying that $\left\Vert \bm{J}-\bm{Q}\right\Vert _{2}\ll1$
and $\sigma_{i}(\bm{J})\asymp1$ for $i=1,2,\ldots,r$, where $\bm{Q}\in\mathcal{O}^{r\times r}$
is the rotation matrix given in Lemma \ref{Lemma SVD BF good event}.
When $\mathcal{E}_{0}\cap\mathcal{E}_{S}$ happens, since $\left\Vert \bm{J}-\bm{Q}\right\Vert _{2}\ll1$
and $\big\Vert\frac{1}{|S|}(\bm{F}_{S,\cdot})^{\top}\bm{F}_{S,\cdot}-\bm{I}_{r}\big\Vert_{2}\ll1$,
we have that 
\begin{align*}
\left\Vert \frac{T}{|S|}(\bm{\bm{\bm{V}}}_{S,\cdot})^{\top}\bm{\bm{\bm{V}}}_{S,\cdot}-\bm{I}_{r}\right\Vert _{2} & =\left\Vert \bm{J}^{\top}(\frac{1}{|S|}(\bm{F}_{S,\cdot})^{\top}\bm{F}_{S,\cdot})\bm{J}-\bm{Q}^{\top}\bm{I}_{r}\bm{Q}\right\Vert _{2}\\
 & \leq\left\Vert \bm{J}^{\top}(\frac{1}{|S|}(\bm{F}_{S,\cdot})^{\top}\bm{F}_{S,\cdot}-\bm{I}_{r})\bm{J}\right\Vert _{2}+\big\Vert\bm{J}^{\top}\bm{J}-\bm{Q}^{\top}\bm{Q}\big\Vert_{2}\\
 & \leq\left\Vert \frac{1}{|S|}(\bm{F}_{S,\cdot})^{\top}\bm{F}_{S,\cdot}-\bm{I}_{r}\right\Vert _{2}+\big\Vert(\bm{J}^{\top}-\bm{Q}^{\top})\bm{Q}\big\Vert_{2}+\big\Vert\bm{J}^{\top}(\bm{J}-\bm{Q})\big\Vert_{2}\\
 & \ll1.
\end{align*}
So, we obtain that $\lambda_{i}(\frac{T}{|S|}(\bm{\bm{\bm{V}}}_{S,\cdot})^{\top}\bm{\bm{\bm{V}}}_{S,\cdot})\asymp1$
for $i=1,2,\ldots,r$, and thus $\sigma_{i}(\bm{\bm{\bm{V}}}_{S,\cdot})\asymp\sqrt{\frac{|S|}{T}}$
for $i=1,2,\ldots,r$, implying that $\text{rank}(\bm{\bm{\bm{V}}}_{S,\cdot})=r$.

Given $\mathcal{E}_{0}\cap\mathcal{E}_{S}$, since $\bm{\bm{\bm{V}}}_{S,\cdot}$
has full column rank, we have that $(\bm{\bm{\bm{V}}}_{S,\cdot})^{+}=[(\bm{\bm{\bm{V}}}_{S,\cdot})^{\top}\bm{\bm{\bm{V}}}_{S,\cdot}]^{-1}(\bm{\bm{\bm{V}}}_{S,\cdot})^{\top}$.
So, on the one hand, we have that 
\begin{align*}
\sigma_{\text{max}}((\bm{\bm{\bm{V}}}_{S,\cdot})^{+}) & \leq\big\Vert[(\bm{\bm{\bm{V}}}_{S,\cdot})^{\top}\bm{\bm{\bm{V}}}_{S,\cdot}]^{-1}\big\Vert_{2}\big\Vert\bm{\bm{\bm{V}}}_{S,\cdot}\big\Vert_{2}\leq\frac{\sigma_{\text{max}}(\bm{\bm{\bm{V}}}_{S,\cdot})}{(\sigma_{\text{min}}(\bm{\bm{\bm{V}}}_{S,\cdot}))^{2}}\lesssim\left(\sqrt{\frac{|S|}{T}}\right)^{-2}\sqrt{\frac{|S|}{T}}=\sqrt{\frac{T}{|S|}}.
\end{align*}
On the other hand, similarly we have that 
\[
\sigma_{\text{min}}((\bm{\bm{\bm{V}}}_{S,\cdot})^{+})\geq\frac{\sigma_{\text{min}}(\bm{\bm{\bm{V}}}_{S,\cdot})}{(\sigma_{\text{max}}(\bm{\bm{\bm{V}}}_{S,\cdot}))^{2}}\gtrsim\left(\sqrt{\frac{|S|}{T}}\right)^{-2}\sqrt{\frac{|S|}{T}}=\sqrt{\frac{T}{|S|}}.
\]
So we conclude that $\sigma_{i}((\bm{\bm{\bm{V}}}_{S,\cdot})^{+})\asymp\sqrt{\frac{T}{|S|}}$
for $i=1,2,\ldots,r$. 
\end{proof}
Using the properties of $\bm{\bm{\bm{V}}}_{S,\cdot}$ obtained in
Lemma \ref{Lemma V subset good event}, we are able to establish the
following properties on the projection matrices.

\begin{lemma} \label{Lemma projection vector 1st-order approx}Suppose
that Assumptions \ref{Assump_Bf_identification}, \ref{Assump_noise_Z_entries},
and \ref{Assump_factor_f} hold. Assume that $r+\log n\ll T$, $r\log n\ll n$,
and $\rho\sqrt{\log n}\ll1$ as in Theorem \ref{Thm factor test plug-in Chi-sq},
and $r+\log n\ll|S|$ in (\ref{Assump subset size large logn}). Then
we have that

(i) The projection matrices have the perturbation bound as below:
with probability at least $1-O(n^{-2})$, it holds 
\[
\big\Vert\bm{P}_{\bm{\bm{\bm{V}}}_{S,\cdot}}-\bm{P}_{\widehat{\bm{V}}_{S,\cdot}}\big\Vert_{2}\lesssim\rho\sqrt{\frac{T}{n}}\sqrt{r}\log^{3/2}n.
\]

(ii) Under the null hypothesis $H_{0}$ in (\ref{Null H0 factor test subset}),
we have that, with probability at least $1-O(n^{-2})$, 
\[
\big\Vert(\bm{I}_{|S|}-\bm{P}_{\widehat{\bm{V}}_{S,\cdot}})\bm{v}\big\Vert_{2}\lesssim\sqrt{|S|}\rho\sqrt{\frac{r}{n}}\log n\big\Vert(\bm{\bm{\bm{V}}}_{S,\cdot})^{+}\bm{v}\big\Vert_{2},
\]
and the following first-order approximation 
\[
(\bm{P}_{\widehat{\bm{V}}_{S,\cdot}}-\bm{I}_{|S|})\bm{v}=\bm{G}_{P}+\bm{\Upsilon}_{P},\text{\qquad with\qquad}\bm{G}_{P}:=(\bm{I}_{|S|}-\bm{P}_{\bm{\bm{V}}_{S,\cdot}})(\bm{G}_{V})_{S,\cdot}(\bm{\bm{\bm{V}}}_{S,\cdot})^{+}\bm{v},
\]
where $\bm{G}_{V}$ is the first-order term in (\ref{V 1st-order approx hat til})
and is given by $\bm{G}_{V}=T^{-1/2}\bm{E}^{\top}\bm{U}\bm{\Lambda}^{-1}$,
and with probability at least $1-O(n^{-2})$, it holds 
\[
\left\Vert \bm{\Upsilon}_{P}\right\Vert _{2}\lesssim\sqrt{|S|}\left(\rho^{2}\frac{1}{\sqrt{T}}+\rho\frac{1}{\sqrt{nT}}\right)r\log^{2}n\big\Vert(\bm{\bm{\bm{V}}}_{S,\cdot})^{+}\bm{v}\big\Vert_{2}.
\]

\end{lemma} 
\begin{proof}
The proof consists of three steps.

\textit{Step 1 -- Projection matrix perturbation bound.}

Since the rotation matrix $\bm{R}_{V}\in\mathcal{O}^{r\times r}$
is invertible, we have that $col(\widehat{\bm{V}}_{S,\cdot})=col(\widehat{\bm{V}}_{S,\cdot}\bm{R}_{V})$,
implying that $\bm{P}_{\widehat{\bm{V}}_{S,\cdot}}=\bm{P}_{\widehat{\bm{V}}_{S,\cdot}\bm{R}_{V}}$.
By Theorem 1.2 in \citet{ch2016projection} on the perturbation bound
of projection matrix, we have that 
\begin{align*}
\big\Vert\bm{P}_{\bm{\bm{\bm{V}}}_{S,\cdot}}-\bm{P}_{\widehat{\bm{V}}_{S,\cdot}}\big\Vert_{2} & =\big\Vert\bm{P}_{\bm{\bm{\bm{V}}}_{S,\cdot}}-\bm{P}_{\widehat{\bm{V}}_{S,\cdot}\bm{R}_{V}}\big\Vert_{2}\leq\min\left\{ \big\Vert(\bm{\bm{\bm{V}}}_{S,\cdot})^{+}\big\Vert_{2},\big\Vert(\widehat{\bm{V}}_{S,\cdot}\bm{R}_{V})^{+}\big\Vert_{2}\right\} \big\Vert\widehat{\bm{V}}_{S,\cdot}\bm{R}_{V}-\bm{\bm{\bm{V}}}_{S,\cdot}\big\Vert_{2}\\
 & \leq\big\Vert(\bm{\bm{\bm{V}}}_{S,\cdot})^{+}\big\Vert_{2}\big\Vert\widehat{\bm{V}}_{S,\cdot}\bm{R}_{V}-\bm{\bm{\bm{V}}}_{S,\cdot}\big\Vert_{2}.
\end{align*}
The term $\big\Vert(\bm{\bm{\bm{V}}}_{S,\cdot})^{+}\big\Vert_{2}$
can be bounded using Lemma \ref{Lemma V subset good event}. For the
term $\big\Vert\widehat{\bm{V}}_{S,\cdot}\bm{R}_{V}-\bm{\bm{\bm{V}}}_{S,\cdot}\big\Vert_{2}$,
we have that 
\[
\big\Vert\widehat{\bm{V}}_{S,\cdot}\bm{R}_{V}-\bm{\bm{\bm{V}}}_{S,\cdot}\big\Vert_{2}\leq\sqrt{|S|}\big\Vert\widehat{\bm{V}}_{S,\cdot}\bm{R}_{V}-\bm{\bm{\bm{V}}}_{S,\cdot}\big\Vert_{2,\infty}\leq\sqrt{|S|}\big\Vert\widehat{\bm{V}}\bm{R}_{V}-\bm{V}\big\Vert_{2,\infty}\lesssim\sqrt{|S|}\rho\sqrt{\frac{r}{n}}\log n,
\]
where the last inequality uses the upper bound of $\big\Vert\widehat{\bm{V}}\bm{R}_{V}-\bm{V}\big\Vert_{2,\infty}$
obtained in the proof of Corollary \ref{corollary:error bound B F}.
So, we obtain 
\[
\big\Vert\bm{P}_{\bm{\bm{\bm{V}}}_{S,\cdot}}-\bm{P}_{\widehat{\bm{V}}_{S,\cdot}}\big\Vert_{2}\leq\big\Vert(\bm{\bm{\bm{V}}}_{S,\cdot})^{+}\big\Vert_{2}\big\Vert\widehat{\bm{V}}_{S,\cdot}\bm{R}_{V}-\bm{\bm{\bm{V}}}_{S,\cdot}\big\Vert_{2}\lesssim\sqrt{\frac{T}{|S|}}\sqrt{|S|}\rho\sqrt{\frac{r}{n}}\log n=\rho\sqrt{\frac{T}{n}}\sqrt{r}\log n.
\]

\textit{Step 2 -- First-order approximation expression.}

For ease of exposition, we denote 
\[
G:=\widehat{\bm{V}}_{S,\cdot}\bm{R}_{V}-\bm{\bm{V}}_{S,\cdot}.
\]
We obtain by the identity after (4.6) in \citet{stewart1977perturbation}
that 
\begin{align*}
\bm{P}_{\widehat{\bm{V}}_{S,\cdot}}-\bm{P}_{\bm{\bm{V}}_{S,\cdot}} & =\bm{P}_{\widehat{\bm{V}}_{S,\cdot}\bm{R}_{V}}-\bm{P}_{\bm{\bm{V}}_{S,\cdot}}\\
 & =(\widehat{\bm{V}}_{S,\cdot}\bm{R}_{V})^{+}\bm{P}_{(\widehat{\bm{V}}_{S,\cdot}\bm{R}_{V})^{\top}}G^{\top}(\bm{I}_{|S|}-\bm{P}_{\bm{\bm{V}}_{S,\cdot}})+(\bm{I}_{|S|}-\bm{P}_{\widehat{\bm{V}}_{S,\cdot}})G\bm{P}_{(\bm{\bm{\bm{V}}}_{S,\cdot})^{\top}}(\bm{\bm{\bm{V}}}_{S,\cdot})^{+}.
\end{align*}
Next, we obtain by Lemma \ref{Lemma SVD BF good event} that, there
exists a $\sigma(\bm{F})$-measurable and invertible matrix $\bm{J}$
such that $\bm{V}=\frac{1}{\sqrt{T}}\bm{FJ}$, implying that $\bm{\bm{V}}_{S,\cdot}=\frac{1}{\sqrt{T}}\bm{F}_{S,\cdot}\bm{J}$.
So, under the null hypothesis $H_{0}$ in (\ref{Null H0 factor test subset}),
since $\bm{v}\in col(\bm{F}_{S,\cdot})$ implies that $\bm{v}\in col(\bm{\bm{V}}_{S,\cdot})$,
we have that $(\bm{I}_{|S|}-\bm{P}_{\bm{\bm{V}}_{S,\cdot}})\bm{v}=0$.
Then, we obtain that 
\begin{align*}
(\bm{P}_{\widehat{\bm{V}}_{S,\cdot}}-\bm{I}_{|S|})\bm{v} & =(\bm{P}_{\widehat{\bm{V}}_{S,\cdot}\bm{R}_{V}}-\bm{P}_{\bm{\bm{V}}_{S,\cdot}})\bm{v}=(\bm{I}_{|S|}-\bm{P}_{\widehat{\bm{V}}_{S,\cdot}})G\bm{P}_{(\bm{\bm{\bm{V}}}_{S,\cdot})^{\top}}(\bm{\bm{\bm{V}}}_{S,\cdot})^{+}\bm{v}\\
 & =(\bm{I}_{|S|}-\bm{P}_{\widehat{\bm{V}}_{S,\cdot}})G(\bm{\bm{\bm{V}}}_{S,\cdot})^{+}\bm{v}.
\end{align*}
Here, the last equality is because $\bm{\bm{V}}_{S,\cdot}$ has full
column rank, implying that $\bm{P}_{(\bm{\bm{\bm{V}}}_{S,\cdot})^{\top}}=\bm{I}_{r}$.

Then, since the projection matrix $\bm{I}_{|S|}-\bm{P}_{\widehat{\bm{V}}_{S,\cdot}}$
satisfies that $\big\Vert\bm{I}_{|S|}-\bm{P}_{\widehat{\bm{V}}_{S,\cdot}}\big\Vert_{2}=1$,
we obtain that 
\[
\big\Vert(\bm{P}_{\widehat{\bm{V}}_{S,\cdot}}-\bm{I}_{|S|})\bm{v}\big\Vert_{2}\leq\left\Vert G\right\Vert _{2}\big\Vert(\bm{\bm{\bm{V}}}_{S,\cdot})^{+}\bm{v}\big\Vert_{2}\lesssim\sqrt{|S|}\rho\sqrt{\frac{r}{n}}\log n\cdot\big\Vert(\bm{\bm{\bm{V}}}_{S,\cdot})^{+}\bm{v}\big\Vert_{2},
\]
where the last inequality uses the upper bound of $\left\Vert G\right\Vert _{2}=\big\Vert\widehat{\bm{V}}_{S,\cdot}\bm{R}_{V}-\bm{\bm{\bm{V}}}_{S,\cdot}\big\Vert_{2}$
obtained in Step 1.

\textit{Step 3 -- Decomposition and upper bounds for the higher-order terms.}

To establish the first-order approximation, we do the decomposition
as follows: 
\begin{align*}
 & (\bm{P}_{\widehat{\bm{V}}_{S,\cdot}}-\bm{I}_{|S|})\bm{v}\\
= & (\bm{I}_{|S|}-\bm{P}_{\bm{\bm{V}}_{S,\cdot}})G(\bm{\bm{\bm{V}}}_{S,\cdot})^{+}\bm{v}+(\bm{P}_{\bm{\bm{V}}_{S,\cdot}}-\bm{P}_{\widehat{\bm{V}}_{S,\cdot}})G(\bm{\bm{\bm{V}}}_{S,\cdot})^{+}\bm{v}\\
= & (\bm{I}_{|S|}-\bm{P}_{\bm{\bm{V}}_{S,\cdot}})(\bm{G}_{V})_{S,\cdot}(\bm{\bm{\bm{V}}}_{S,\cdot})^{+}\bm{v}\\
 & +\underset{=:\psi_{1}}{\underbrace{(\bm{I}_{|S|}-\bm{P}_{\bm{\bm{V}}_{S,\cdot}})(\bm{\Psi}_{V})_{S,\cdot}(\bm{\bm{\bm{V}}}_{S,\cdot})^{+}\bm{v}}}+\underset{=:\psi_{2}}{\underbrace{(\bm{P}_{\bm{\bm{V}}_{S,\cdot}}-\bm{P}_{\widehat{\bm{V}}_{S,\cdot}})G(\bm{\bm{\bm{V}}}_{S,\cdot})^{+}\bm{v}}}.
\end{align*}
Here, the last equation is because $G=\widehat{\bm{V}}_{S,\cdot}\bm{R}_{V}-\bm{\bm{\bm{V}}}_{S,\cdot}=(\widehat{\bm{V}}\bm{R}_{V}-\bm{\bm{\bm{V}}})_{S,\cdot}=(\bm{G}_{V})_{S,\cdot}+(\bm{\Psi}_{V})_{S,\cdot}$.
So, we get that $\bm{\Upsilon}_{P}=\psi_{1}+\psi_{2}$, and thus the
problems boil down to deriving the upper bound of the norms of $\psi_{1}$
and $\psi_{2}$.

For $\psi_{1}$, since $\big\Vert(\bm{I}_{|S|}-\bm{P}_{\bm{\bm{V}}_{S,\cdot}})\big\Vert_{2}=1$,
we have that 
\begin{align*}
\left\Vert \psi_{1}\right\Vert _{2} & \leq\big\Vert(\bm{\Psi}_{V})_{S,\cdot}\big\Vert_{2}\big\Vert(\bm{\bm{\bm{V}}}_{S,\cdot})^{+}\bm{v}\big\Vert_{2}\leq\sqrt{|S|}\left\Vert \bm{\Psi}_{V}\right\Vert _{2,\infty}\big\Vert(\bm{\bm{\bm{V}}}_{S,\cdot})^{+}\bm{v}\big\Vert_{2}\\
 & \overset{\text{(i)}}{\lesssim}\sqrt{|S|}\left(\rho^{2}\sqrt{\frac{r}{n}}\log^{3/2}n+\left(\rho^{2}+\rho\sqrt{\frac{r}{n}}\log n\right)\big\Vert\bm{V}\big\Vert_{2,\infty}\right)\big\Vert(\bm{\bm{\bm{V}}}_{S,\cdot})^{+}\bm{v}\big\Vert_{2}\\
 & \overset{\text{(ii)}}{\lesssim}\sqrt{|S|}\left(\rho^{2}\sqrt{\frac{r}{n}}\log^{3/2}n+\left(\rho^{2}+\rho\sqrt{\frac{r}{n}}\right)\log n\sqrt{\frac{\log n}{T}}\right)\big\Vert(\bm{\bm{\bm{V}}}_{S,\cdot})^{+}\bm{v}\big\Vert_{2}.
\end{align*}
where (i) uses the upper bound for $\left\Vert (\bm{\Psi}_{V})_{l,\cdot}\right\Vert _{2}$
in Theorem \ref{Thm UV 1st approx row-wise error}, and (ii) uses
the upper bound $\big\Vert\bm{V}\big\Vert_{2,\infty}\lesssim\sqrt{\frac{\log n}{T}}$
obtained in Lemma \ref{Lemma SVD BF good event}.

For $\psi_{2}$, using the bounds of $\big\Vert\bm{P}_{\bm{\bm{\bm{V}}}_{S,\cdot}}-\bm{P}_{\widehat{\bm{V}}_{S,\cdot}}\big\Vert_{2}$
and $\left\Vert G\right\Vert _{2}=\big\Vert\widehat{\bm{V}}_{S,\cdot}\bm{R}_{V}-\bm{\bm{\bm{V}}}_{S,\cdot}\big\Vert_{2}$
obtained in Step 1, we have that 
\begin{align*}
\left\Vert \psi_{2}\right\Vert _{2} & \leq\big\Vert\bm{P}_{\bm{\bm{\bm{V}}}_{S,\cdot}}-\bm{P}_{\widehat{\bm{V}}_{S,\cdot}}\big\Vert_{2}\left\Vert G\right\Vert _{2}\big\Vert(\bm{\bm{\bm{V}}}_{S,\cdot})^{+}\bm{v}\big\Vert_{2}\\
 & \lesssim\rho\sqrt{\frac{T}{n}}\sqrt{r}\log n\cdot\sqrt{|S|}\rho\sqrt{\frac{r}{n}}\log n\cdot\big\Vert(\bm{\bm{\bm{V}}}_{S,\cdot})^{+}\bm{v}\big\Vert_{2}=\sqrt{|S|}\rho^{2}\frac{\sqrt{T}}{n}r\log^{2}n\big\Vert(\bm{\bm{\bm{V}}}_{S,\cdot})^{+}\bm{v}\big\Vert_{2}.
\end{align*}

Then, we combine the above results for $\left\Vert \psi_{1}\right\Vert _{2}$
and $\left\Vert \psi_{2}\right\Vert _{2}$ to obtain that 
\begin{align*}
\left\Vert \bm{\Upsilon}_{P}\right\Vert _{2} & \leq\left\Vert \psi_{1}\right\Vert _{2}+\left\Vert \psi_{2}\right\Vert _{2}\\
 & \lesssim\sqrt{|S|}\left(\rho^{2}\sqrt{\frac{r}{n}}\log^{3/2}n+\left(\rho^{2}+\rho\sqrt{\frac{r}{n}}\right)\log n\sqrt{\frac{\log n}{T}}+\frac{\sqrt{T}}{n}\rho^{2}r\log^{2}n\right)\big\Vert(\bm{\bm{\bm{V}}}_{S,\cdot})^{+}\bm{v}\big\Vert_{2},\\
 & \lesssim\sqrt{|S|}\left(\rho^{2}\frac{1}{\sqrt{T}}+\rho\frac{1}{\sqrt{nT}}\right)r\log^{2}n\big\Vert(\bm{\bm{\bm{V}}}_{S,\cdot})^{+}\bm{v}\big\Vert_{2}\\
 & \lesssim\sqrt{|S|}(\rho^{2}\frac{1}{\sqrt{T}}+\rho\frac{1}{\sqrt{nT}})r\log^{2}n\big\Vert(\bm{\bm{\bm{V}}}_{S,\cdot})^{+}\bm{v}\big\Vert_{2},
\end{align*}
which is our desired result. 
\end{proof}
Then, we establish Gaussian approximation for the first-order term.
For simplicity of notations, we denote $\chi^{2}(n)$ as a Chi-square
random variable with degree of freedom equal to $n$, and denote by
$\chi_{\phi}^{2}(n)$ its $\phi$-quantile.

\begin{lemma} \label{Lemma true Chi-square stat}Suppose that the
assumptions in Lemma \ref{Lemma projection vector 1st-order approx}
hold. Then, under the null hypothesis $H_{0}$ in (\ref{Null H0 factor test subset}),
we have that, for any random variable $\zeta$ satisfying $\left|\zeta\right|\ll|S|$,
it holds 
\begin{equation}
\left|\mathbb{P}\left(\mathfrak{\mathcal{T}}(S,\bm{v})+\zeta\leq\chi_{1-\alpha}^{2}(|S|-r)\right)-(1-\alpha)\right|\lesssim\sqrt{|S|}(\rho\sqrt{\frac{n}{T}}+\frac{1}{\sqrt{T}})\kappa_{\varepsilon}r\log^{2}n+\frac{1}{\sqrt{|S|}}\left|\zeta\right|+s_{3},\label{Chi-square stat True key inequality}
\end{equation}
where 
\[
\mathfrak{\mathcal{T}}(S,\bm{v}):=\frac{1}{\phi}\bm{\bm{v}}^{\top}(\bm{I}_{|S|}-\bm{P}_{\widehat{\bm{V}}_{S,\cdot}})(\bm{I}_{|S|}-\bm{P}_{\bm{\bm{V}}_{S,\cdot}})(\bm{I}_{|S|}-\bm{P}_{\widehat{\bm{V}}_{S,\cdot}})\bm{\bm{v}},
\]
with 
\[
\phi:=\frac{1}{T}((\bm{\bm{\bm{V}}}_{S,\cdot})^{+}\bm{\bm{v}})^{\top}\bm{\Lambda}^{-1}\bm{U}^{\top}\bm{\Sigma}_{\varepsilon}\bm{U}\bm{\Lambda}^{-1}(\bm{\bm{\bm{V}}}_{S,\cdot})^{+}\bm{\bm{v}},
\]
and $s_{3}\geq0$ is a constant bounded by 
\[
s_{3}\lesssim\kappa_{\varepsilon}\sqrt{\frac{n}{T}}|S|^{3/2}\frac{\big\Vert\bm{\Sigma}_{\varepsilon}^{1/2}\bar{\bm{U}}\big\Vert_{2,\infty}}{\big\Vert\bm{\Sigma}_{\varepsilon}^{1/2}\big\Vert_{2}}.
\]
Further, if all the entries of the matrix $\bm{Z}$ are Gaussian,
i.e., the noise is Gaussian, then $s_{3}$ in the inequality (\ref{Chi-square stat True key inequality})
is equal to zero, i.e., the inequality (\ref{Chi-square stat True key inequality})
holds when $s_{3}=0$.

In addition, we have that 
\[
\phi\gtrsim\frac{1}{\kappa_{\varepsilon}^{2}n}\rho^{2}\big\Vert(\bm{\bm{\bm{V}}}_{S,\cdot})^{+}\bm{\bm{v}}\big\Vert_{2}^{2}.
\]

\end{lemma} 
\begin{proof}
Our starting point is the following inequality 
\begin{equation}
\mathbb{P}\left(\mathfrak{\mathcal{T}}(S,\bm{\bm{v}})\leq\chi_{1-\alpha}^{2}(|S|-r)-\left|\zeta\right|\right)\leq\mathbb{P}\left(\mathfrak{\mathcal{T}}(S,\bm{\bm{v}})+\zeta\leq\chi_{1-\alpha}^{2}(|S|-r)\right)\leq\mathbb{P}\left(\mathfrak{\mathcal{T}}(S,\bm{\bm{v}})\leq\chi_{1-\alpha}^{2}(|S|-r)+\left|\zeta\right|\right).\label{1st inequality for Chi-square stat plus a term}
\end{equation}
Then, we will prove both the upper bound and the lower bound are close
to $(1-\alpha)$.

We will use three steps. We conduct the proof conditioning on $\mathcal{E}_{0}\cap\mathcal{E}_{S}$.

\textit{Step 1 -- Bounding $\mathbb{P}\left(\mathfrak{\mathcal{T}}(S,\bm{\bm{v}})+\zeta\leq\chi_{1-\alpha}^{2}(|S|-r)\right)$.}

We already obtained the first-order approximation of $(\bm{I}_{|S|}-\bm{P}_{\widehat{\bm{V}}_{S,\cdot}})\bm{\bm{v}}=\bm{G}_{P}+\bm{\Upsilon}_{P}$
in Lemma \ref{Lemma projection vector 1st-order approx}. Indeed,
the first-order term $\bm{G}_{P}$ can be rewritten as 
\begin{equation}
\bm{G}_{P}=(\bm{I}_{|S|}-\bm{P}_{\bm{\bm{V}}_{S,\cdot}})\bm{u}\text{\qquad with\qquad}\bm{u}:=(\bm{G}_{V})_{S,\cdot}(\bm{\bm{\bm{V}}}_{S,\cdot})^{+}\bm{\bm{v}}=(T^{-1/2}\bm{E}^{\top}\bm{U}\bm{\Lambda}^{-1})_{S,\cdot}(\bm{\bm{\bm{V}}}_{S,\cdot})^{+}\bm{\bm{v}}.\label{projection 1st term def before iid sum}
\end{equation}
Later we will show that the $|S|$-dimensional random vector $\bm{u}$
is close to a Gaussian vector, so that the desired Chi-square law
results can be proven. But now let us focus on how this term appear
at the term $\mathfrak{\mathcal{T}}(S,\bm{\bm{v}})$. To do this,
we need to decompose $(\bm{I}_{|S|}-\bm{P}_{\bm{\bm{V}}_{S,\cdot}})$
as follows.

By Lemma \ref{Lemma V subset good event}, given the event $\mathcal{E}_{0}\cap\mathcal{E}_{S}$,
$\bm{\bm{V}}_{S,\cdot}$ has rank $r$, and the $|S|\times|S|$ projection
matrix $(\bm{I}_{|S|}-\bm{P}_{\bm{\bm{V}}_{S,\cdot}})$ is idempotent
with rank $(|S|-r)$. So the eigen-decomposition of $\bm{I}_{|S|}-\bm{P}_{\bm{\bm{V}}_{S,\cdot}}$
is given by 
\begin{equation}
\bm{I}_{|S|}-\bm{P}_{\bm{\bm{V}}_{S,\cdot}}=\bm{H}\bm{I}_{|S|-r}\bm{H}^{\top}=\bm{H}\bm{H}^{\top},\label{eigen decompose I-P_V}
\end{equation}
where $\bm{H}$ is a $\sigma(\bm{F})$-measurable $|S|\times(|S|-r)$
matrix, and $\bm{H}$ has orthonormal columns, i.e., $\bm{H}^{\top}\bm{H}=\bm{I}_{|S|-r}$.
So we obtain that 
\[
\bm{G}_{P}=(\bm{I}_{|S|}-\bm{P}_{\bm{\bm{V}}_{S,\cdot}})\bm{u}=\bm{H}\bm{H}^{\top}\bm{u},
\]
and thus $\mathfrak{\mathcal{T}}(S,\bm{\bm{v}})$ can be rewritten
as 
\begin{align}
\mathfrak{\mathcal{T}}(S,\bm{\bm{v}}) & =\frac{1}{\left\Vert \bm{m}\right\Vert _{2}^{2}}(\bm{G}_{P}+\bm{\Upsilon}_{P})^{\top}\bm{H}\bm{H}^{\top}(\bm{G}_{P}+\bm{\Upsilon}_{P})=\frac{1}{\left\Vert \bm{m}\right\Vert _{2}^{2}}(\bm{H}\bm{H}^{\top}\bm{u}+\bm{\Upsilon}_{P})^{\top}\bm{H}\bm{H}^{\top}(\bm{H}\bm{H}^{\top}\bm{u}+\bm{\Upsilon}_{P})\nonumber \\
 & =\frac{1}{\left\Vert \bm{m}\right\Vert _{2}^{2}}\big\Vert\bm{H}^{\top}(\bm{H}\bm{H}^{\top}\bm{u}+\bm{\Upsilon}_{P})\big\Vert_{2}^{2}=\frac{1}{\left\Vert \bm{m}\right\Vert _{2}^{2}}\big\Vert\bm{H}^{\top}\bm{u}+\bm{H}^{\top}\bm{\Upsilon}_{P}\big\Vert_{2}^{2}=\left\Vert \frac{1}{\left\Vert \bm{m}\right\Vert _{2}}\bm{H}^{\top}\bm{u}+\frac{1}{\left\Vert \bm{m}\right\Vert _{2}}\bm{H}^{\top}\bm{\Upsilon}_{P}\right\Vert _{2}^{2}.\label{Chi-square stat 2-norm equality}
\end{align}
The expression (\ref{Chi-square stat 2-norm equality}) shows how
$\frac{1}{\left\Vert \bm{m}\right\Vert _{2}}\bm{H}^{\top}\bm{u}$
plays a role in $\mathfrak{\mathcal{T}}(S,\bm{\bm{v}})$.

To further use the probabilities about $\frac{1}{\left\Vert \bm{m}\right\Vert _{2}}\big\Vert\bm{H}^{\top}\bm{u}\big\Vert_{2}$
to bound the probabilities about $\mathfrak{\mathcal{T}}(S,\bm{\bm{v}})$,
we define 
\begin{equation}
\tau:=\sqrt{\mathfrak{\mathcal{T}}(S,\bm{\bm{v}})}-\frac{1}{\left\Vert \bm{m}\right\Vert _{2}}\big\Vert\bm{H}^{\top}\bm{u}\big\Vert_{2}.\label{diff sqrt Chi-square stat and 2-norm Gaussian}
\end{equation}
Then, we obtain that, for the upper bound in (\ref{1st inequality for Chi-square stat plus a term})
\begin{align*}
\mathbb{P}\left(\mathfrak{\mathcal{T}}(S,\bm{\bm{v}})\leq\chi_{1-\alpha}^{2}(|S|-r)+\left|\zeta\right|\right) & =\mathbb{P}\left(\frac{1}{\left\Vert \bm{m}\right\Vert _{2}}\big\Vert\bm{H}^{\top}\bm{u}\big\Vert_{2}+\tau\leq\sqrt{\chi_{1-\alpha}^{2}(|S|-r)+\left|\zeta\right|}\right)\\
 & \leq\mathbb{P}\left(\frac{1}{\left\Vert \bm{m}\right\Vert _{2}}\big\Vert\bm{H}^{\top}\bm{u}\big\Vert_{2}\leq\sqrt{\chi_{1-\alpha}^{2}(|S|-r)+\left|\zeta\right|}+\left|\tau\right|\right),
\end{align*}
and similarly for the lower bound in (\ref{1st inequality for Chi-square stat plus a term})
we have that $\mathbb{P}(\mathfrak{\mathcal{T}}(S,\bm{\bm{v}})\leq\chi_{1-\alpha}^{2}(|S|-r)-\left|\zeta\right|)\geq\mathbb{P}\bigl(\frac{1}{\left\Vert \bm{m}\right\Vert _{2}}\big\Vert\bm{H}^{\top}\bm{u}\big\Vert_{2}\leq\sqrt{\chi_{1-\alpha}^{2}(|S|-r)-\left|\zeta\right|}-\left|\tau\right|\bigr)$.
Thus, we obtain by the inequality (\ref{1st inequality for Chi-square stat plus a term})
that 
\begin{align}
 & \mathbb{P}\left(\frac{1}{\left\Vert \bm{m}\right\Vert _{2}}\big\Vert\bm{H}^{\top}\bm{u}\big\Vert_{2}\leq\sqrt{\chi_{1-\alpha}^{2}(|S|-r)-\left|\zeta\right|}-\left|\tau\right|\right)\label{2nd inequality for Chi-square stat plus a term}\\
 & \leq\mathbb{P}\left(\mathfrak{\mathcal{T}}(S,\bm{\bm{v}})+\zeta\leq\chi_{1-\alpha}^{2}(|S|-r)\right)\nonumber \\
 & \leq\mathbb{P}\left(\frac{1}{\left\Vert \bm{m}\right\Vert _{2}}\big\Vert\bm{H}^{\top}\bm{u}\big\Vert_{2}\leq\sqrt{\chi_{1-\alpha}^{2}(|S|-r)+\left|\zeta\right|}+\left|\tau\right|\right).\nonumber 
\end{align}

\textit{Step 2 -- Expressing $\frac{1}{\left\Vert \bm{m}\right\Vert _{2}}\bm{H}^{\top}\bm{u}$
as a sum to show its proximity to a Gaussian vector.}

We define 
\begin{equation}
s_{3}:=\sup_{R\geq0}\left\vert \mathbb{P}\left(\frac{1}{\left\Vert \bm{m}\right\Vert _{2}}\big\Vert\bm{H}^{\top}\bm{u}\big\Vert_{2}\leq R\right)-\mathbb{P}(\chi^{2}(|S|-r)\leq R^{2})\right\vert .\label{distance iid sum ball prob and Gaussian}
\end{equation}
We look at the random vector $\bm{u}\in\mathbb{R}^{|S|}$ first. Conditioning
on $\bm{F}$, the random vector $\bm{u}$ defined in (\ref{projection 1st term def before iid sum})
can be written as a sum of independent and mean zero $|S|$-dimensional
vectors as follows: 
\begin{align}
\bm{u} & =(T^{-1/2}\bm{E}^{\top}\bm{U}\bm{\Lambda}^{-1})_{S,\cdot}(\bm{\bm{\bm{V}}}_{S,\cdot})^{+}\bm{\bm{v}}=(T^{-1/2}\bm{Z}^{\top}\bm{\Sigma}_{\varepsilon}^{1/2}\bm{U}\bm{\Lambda}^{-1})_{S,\cdot}(\bm{\bm{\bm{V}}}_{S,\cdot})^{+}\bm{\bm{v}}\nonumber \\
 & =T^{-1/2}(\bm{Z}^{\top})_{S,\cdot}\bm{\Sigma}_{\varepsilon}^{1/2}\bm{U}\bm{\Lambda}^{-1}(\bm{\bm{\bm{V}}}_{S,\cdot})^{+}\bm{\bm{v}}\nonumber \\
 & =(\bm{Z}_{\cdot,S})^{\top}\bm{m}=\sum_{k=1}^{N}m_{k}(\bm{Z}_{k,S})^{\top}\text{\qquad with\qquad}\bm{m}:=(m_{1},m_{2},\ldots,m_{N})^{\top}=T^{-1/2}\bm{\Sigma}_{\varepsilon}^{1/2}\bm{U}\bm{\Lambda}^{-1}(\bm{\bm{\bm{V}}}_{S,\cdot})^{+}\bm{\bm{v}}.\label{projection 1st term to iid sum}
\end{align}
So we obtain that $\frac{1}{\left\Vert \bm{m}\right\Vert _{2}}\bm{H}^{\top}\bm{u}$
can be expressed as a summation as follows, 
\begin{equation}
\frac{1}{\left\Vert \bm{m}\right\Vert _{2}}\bm{H}^{\top}\bm{u}=\sum_{k=1}^{N}\frac{1}{\left\Vert \bm{m}\right\Vert _{2}}m_{k}\bm{H}^{\top}(\bm{Z}_{k,S})^{\top}.\label{projection 1st term to iid sum final}
\end{equation}
Here, the matrix $\bm{H}$ defined in (\ref{eigen decompose I-P_V})
is $\sigma(\bm{F})$-measurable; the vector $\bm{m}$ defined in (\ref{projection 1st term to iid sum})
is also is $\sigma(\bm{F})$-measurable under the Null because $\bm{\bm{v}}$
is in the column space of $\bm{F}_{S,\cdot}$. By the independence
between $\bm{F}$ and $\bm{Z}$ in Assummption \ref{Assump_noise_Z_entries},
we obtain that, the vectors $\frac{1}{\left\Vert \bm{m}\right\Vert _{2}}m_{k}\bm{H}^{\top}(\bm{Z}_{k,S})^{\top}\in\mathbb{R}^{|S|-r}$
for $k=1,2,\ldots,N$ are independent and mean zero, and the covariance
matrix is given by 
\[
\text{cov}(\frac{1}{\left\Vert \bm{m}\right\Vert _{2}}\bm{H}^{\top}\bm{u}|\bm{F})=\frac{1}{\left\Vert \bm{m}\right\Vert _{2}^{2}}\sum_{k=1}^{N}m_{k}^{2}\bm{H}^{\top}\bm{H}=\frac{1}{\left\Vert \bm{m}\right\Vert _{2}^{2}}\left\Vert \bm{m}\right\Vert _{2}^{2}\bm{I}_{|S|-r}=\bm{I}_{|S|-r}.
\]
Here, we use the fact that $\text{cov}((\bm{Z}_{k,S})^{\top}|\bm{F})=\bm{I}_{|S|-r}$,
and $\bm{H}^{\top}\bm{H}=\bm{I}_{|S|-r}$.

Given the above analysis of $\frac{1}{\left\Vert \bm{m}\right\Vert _{2}}\bm{H}^{\top}\bm{u}$,
we are now ready to bound the term $s_{3}$ defined in (\ref{distance iid sum ball prob and Gaussian}):

(i) When all entries of $\bm{Z}$ are Gaussian, we have that, conditioning
on $\bm{F}$, the vector $\frac{1}{\left\Vert \bm{m}\right\Vert _{2}}\bm{H}^{\top}\bm{u}$
is also a Gaussian vector because it is a linear combination of Gaussian
vectors. Since $\text{cov}(\frac{1}{\left\Vert \bm{m}\right\Vert _{2}}\bm{H}^{\top}\bm{u}|\bm{F})=\bm{I}_{|S|-r}$,
and $\bm{I}_{|S|-r}$ is independent with $\bm{F}$, we obtain that
the unconditional law of $\frac{1}{\left\Vert \bm{m}\right\Vert _{2}}\bm{H}^{\top}\bm{u}$
the $(|S|-r)$-dimensional standard Gaussian law $\mathcal{N}(0,\bm{I}_{|S|-r})$,
implying that the conditional law of $\frac{1}{\left\Vert \bm{m}\right\Vert _{2}}\bm{u}$
given $\bm{F}$ is the $|S|$-dimensional standard Gaussian law $\mathcal{N}(0,\bm{I}_{|S|})$.
Note that $\chi^{2}(|S|-r)$ has the same distribution with the 2-norm
of $\mathcal{N}(0,\bm{I}_{|S|-r})$. Thus, we have that $s_{3}=0$
when all entries of $\bm{Z}$ are Gaussian.

(ii) When the entries of $\bm{Z}$ are sub-Gaussian, we obtain by
Lemma \ref{Lemma Berry-Esseen balls} that 
\[
s_{3}\lesssim\sum_{k=1}^{N}\mathbb{E}\left[\big\Vert\frac{1}{\left\Vert \bm{m}\right\Vert _{2}}\bm{H}^{\top}m_{k}(\bm{Z}_{k,S})^{\top}\big\Vert_{2}^{3}|\bm{F}\right]=\sum_{k=1}^{N}\frac{1}{\left\Vert \bm{m}\right\Vert _{2}^{3}}\mathbb{E}\left[\big\Vert\bm{H}^{\top}m_{k}(\bm{Z}_{k,S})^{\top}\big\Vert_{2}^{3}|\bm{F}\right].
\]
Then, since $\left\Vert \bm{H}\right\Vert _{2}=1$ owing to $\bm{H}^{\top}\bm{H}=\bm{I}_{|S|-r}$,
we have that 
\begin{align*}
\mathbb{E}\left[\big\Vert\bm{H}^{\top}m_{k}(\bm{Z}_{k,S})^{\top}\big\Vert_{2}^{3}|\bm{F}\right] & \leq\mathbb{E}\left[\big\Vert m_{k}(\bm{Z}_{k,S})^{\top}\big\Vert_{2}^{3}|\bm{F}\right]=\mathbb{E}\left[\left|m_{k}\right|^{3}\big\Vert\bm{Z}_{k,S}\big\Vert_{2}^{3}|\bm{F}\right]\\
 & \overset{\text{(i)}}{=}\left|m_{k}\right|^{3}\mathbb{E}\left[\big\Vert\bm{Z}_{k,S}\big\Vert_{2}^{3}|\bm{F}\right]\overset{\text{(ii)}}{\leq}\left|m_{k}\right|^{2}\mathbb{E}\left[\big\Vert\bm{Z}_{k,S}\big\Vert_{2}^{3}\right]\max_{1\leq i\leq N}\left|m_{i}\right|,
\end{align*}
where (i) is because $m_{k}$ is $\sigma(\bm{F})$-measurable, (ii)
is because $\bm{Z}$ is independent with $\bm{F}$. On the one hand,
we have that 
\[
\mathbb{E}\left[\big\Vert\bm{Z}_{k,S}\big\Vert_{2}^{3}\right]=\mathbb{E}\left[\left(\sum_{i=1}^{|S|}\left|\bm{Z}_{k,i}\right|^{2}\right)^{3/2}\right]\overset{\text{(i)}}{\leq}\mathbb{E}\left[\left(\sum_{i=1}^{|S|}\left|\bm{Z}_{k,i}\right|^{3}\right)\left(\sum_{i=1}^{|S|}1\right)^{1/2}\right]=|S|^{3/2}\mathbb{E}\left[\left|\bm{Z}_{k,i}\right|^{3}\right]\overset{\text{(ii)}}{\lesssim}|S|^{3/2},
\]
where (i) uses Holder's inequality, and (ii) uses the fact that the
entries of $\bm{Z}$ are sub-Gaussian and their sub-Gaussian norms
satisfy $\left\Vert \bm{Z}_{k,i}\right\Vert _{\psi_{2}}=O(1)$ under
Assumption \ref{Assump_noise_Z_entries}. Then, On the other hand,
for $\max_{1\leq i\leq N}\left|m_{i}\right|$, by definition of the
vector $\bm{m}$ in (\ref{projection 1st term to iid sum}), we have
that 
\begin{align*}
\max_{1\leq i\leq N}\left|m_{i}\right| & =\big\Vert T^{-1/2}\bm{\Sigma}_{\varepsilon}^{1/2}\bm{U}\bm{\Lambda}^{-1}(\bm{\bm{\bm{V}}}_{S,\cdot})^{+}\bm{\bm{v}}\big\Vert_{2,\infty}\leq\frac{1}{\sqrt{T}}\big\Vert\bm{\Sigma}_{\varepsilon}^{1/2}\bm{U}\big\Vert_{2,\infty}\big\Vert\bm{\Lambda}^{-1}(\bm{\bm{\bm{V}}}_{S,\cdot})^{+}\bm{\bm{v}}\big\Vert_{2}\\
 & \leq\frac{1}{\sqrt{T}}\big\Vert\bm{\Sigma}_{\varepsilon}^{1/2}\bm{U}\big\Vert_{2,\infty}\big\Vert\bm{\Lambda}^{-1}\big\Vert_{2}\big\Vert(\bm{\bm{\bm{V}}}_{S,\cdot})^{+}\bm{\bm{v}}\big\Vert_{2}\overset{\text{(i)}}{\leq}\frac{1}{\sqrt{T}}\big\Vert\bm{\Sigma}_{\varepsilon}^{1/2}\bar{\bm{U}}\big\Vert_{2,\infty}\left\Vert \bm{Q}\right\Vert _{2}\frac{1}{\lambda_{r}}\big\Vert(\bm{\bm{\bm{V}}}_{S,\cdot})^{+}\bm{\bm{v}}\big\Vert_{2}\\
 & \overset{\text{(ii)}}{\lesssim}\frac{1}{\sqrt{T}}\big\Vert\bm{\Sigma}_{\varepsilon}^{1/2}\bar{\bm{U}}\big\Vert_{2,\infty}\frac{1}{\sigma_{r}}\big\Vert(\bm{\bm{\bm{V}}}_{S,\cdot})^{+}\bm{\bm{v}}\big\Vert_{2},
\end{align*}
where (i) uses $\bm{U}=\bar{\bm{U}}\bm{Q}$ in Lemma \ref{Lemma SVD BF good event},
(ii) uses the facts that $\lambda_{i}\asymp\sigma_{i}$, and $\bm{Q}\in\mathcal{O}^{r\times r}$
is a rotation matrix so that $\left\Vert \bm{Q}\right\Vert _{2}=1$.
Combining the above upper bounds of $\mathbb{E}[\big\Vert\bm{Z}_{k,S}\big\Vert_{2}^{3}]$
and $\max_{1\leq i\leq N}\left|m_{i}\right|$, we obtain that 
\[
s_{3}\leq\mathbb{E}\left[\big\Vert\bm{Z}_{k,S}\big\Vert_{2}^{3}\right]\max_{1\leq i\leq N}\left|m_{i}\right|\sum_{k=1}^{N}\frac{1}{\left\Vert \bm{m}\right\Vert _{2}^{3}}\left|m_{k}\right|^{2}\lesssim|S|^{3/2}\frac{1}{\sqrt{T}}\big\Vert\bm{\Sigma}_{\varepsilon}^{1/2}\bm{U}\big\Vert_{2,\infty}\frac{1}{\sigma_{r}}\big\Vert(\bm{\bm{\bm{V}}}_{S,\cdot})^{+}\bm{\bm{v}}\big\Vert_{2}\frac{1}{\left\Vert \bm{m}\right\Vert _{2}}.
\]

The lower bound of $\left\Vert \bm{m}\right\Vert _{2}$ is derived
as follows. Recall that $\bm{m}=T^{-1/2}\bm{\Sigma}_{\varepsilon}^{1/2}\bm{U}\bm{\Lambda}^{-1}(\bm{\bm{\bm{V}}}_{S,\cdot})^{+}\bm{\bm{v}}$
as defined in (\ref{projection 1st term to iid sum}). So, we obtain
that 
\begin{align*}
\left\Vert \bm{m}\right\Vert _{2} & =\sqrt{\bm{m}^{\top}\bm{m}}=\frac{1}{\sqrt{T}}\bigl[((\bm{\bm{\bm{V}}}_{S,\cdot})^{+}\bm{\bm{v}})^{\top}\bm{\Lambda}^{-1}\bm{U}\bm{\Sigma}_{\varepsilon}\bm{U}\bm{\Lambda}^{-1}(\bm{\bm{\bm{V}}}_{S,\cdot})^{+}\bm{\bm{v}}\bigr]^{1/2}\\
 & \geq\frac{1}{\sqrt{T}}\sqrt{\lambda_{\text{min}}\bigl(\bm{\Lambda}^{-1}\bm{U}^{\top}\bm{\Sigma}_{\varepsilon}\bm{U}\bm{\Lambda}^{-1}\bigr)}\big\Vert(\bm{\bm{\bm{V}}}_{S,\cdot})^{+}\bm{\bm{v}}\big\Vert_{2}.
\end{align*}
Next, since $\bm{U}^{\top}\bm{U}=\bm{I}_{r}$ and $\bm{\Lambda}=\mathsf{diag}(\lambda_{1},\lambda_{2},\ldots,\lambda_{r})$
satisfies $\lambda_{i}\asymp\sigma_{i}$ as shown in Lemma \ref{Lemma SVD BF good event},
we obtain that 
\begin{align*}
\lambda_{\text{min}}\bigl(\bm{\Lambda}^{-1}\bm{U}^{\top}\bm{\Sigma}_{\varepsilon}\bm{U}\bm{\Lambda}^{-1}\bigr) & =\inf_{\left\Vert \bm{x}\right\Vert _{2}=1}\bm{x}^{\top}\bm{\Lambda}^{-1}\bm{U}^{\top}\bm{\Sigma}_{\varepsilon}\bm{U}\bm{\Lambda}^{-1}\bm{x}\geq\frac{1}{\lambda_{r}^{2}}\inf_{\left\Vert \bm{x}\right\Vert _{2}=1}\bm{x}^{\top}\bm{U}^{\top}\bm{\Sigma}_{\varepsilon}\bm{U}\bm{x}\\
 & \gtrsim\frac{1}{\sigma_{r}^{2}}\inf_{\left\Vert \bm{x}\right\Vert _{2}=1}(\bm{U}\bm{x})^{\top}\bm{\Sigma}_{\varepsilon}\bm{U}\bm{x}\geq\frac{1}{\sigma_{r}^{2}}\lambda_{\text{min}}(\bm{\Sigma}_{\varepsilon}).
\end{align*}
Recall that $\rho:=\sqrt{\left\Vert \bm{\Sigma}_{\varepsilon}\right\Vert _{2}}\sqrt{n}/(\sigma_{r}\sqrt{T})$.
So, we conclude that 
\begin{align}
\left\Vert \bm{m}\right\Vert _{2} & \gtrsim\frac{1}{\sqrt{T}}\frac{1}{\sigma_{r}}\sqrt{\lambda_{\text{min}}(\bm{\Sigma}_{\varepsilon})}\big\Vert(\bm{\bm{\bm{V}}}_{S,\cdot})^{+}\bm{\bm{v}}\big\Vert_{2}=\frac{1}{\sqrt{n}}\frac{1}{\kappa_{\varepsilon}}\rho\big\Vert(\bm{\bm{\bm{V}}}_{S,\cdot})^{+}\bm{\bm{v}}\big\Vert_{2},\label{projection iid sum weight 2-norm}
\end{align}
where we use the fact that $\left\Vert \bm{\Sigma}_{\varepsilon}\right\Vert _{2}=\lambda_{\text{max}}(\bm{\Sigma}_{\varepsilon})$
and $\kappa_{\varepsilon}=\lambda_{\text{max}}(\bm{\Sigma}_{\varepsilon})/\lambda_{\text{min}}(\bm{\Sigma}_{\varepsilon})$
is the condition number of $\bm{\Sigma}_{\varepsilon}$. So, when
the entries of $\bm{Z}$ are sub-Gaussian, we obtain that 
\[
s_{3}\lesssim|S|^{3/2}\frac{1}{\sqrt{T}}\big\Vert\bm{\Sigma}_{\varepsilon}^{1/2}\bm{U}\big\Vert_{2,\infty}\frac{1}{\sigma_{r}}\big\Vert(\bm{\bm{\bm{V}}}_{S,\cdot})^{+}\bm{\bm{v}}\big\Vert_{2}\frac{1}{\frac{1}{\sqrt{n}}\frac{1}{\kappa_{\varepsilon}}\rho\big\Vert(\bm{\bm{\bm{V}}}_{S,\cdot})^{+}\bm{\bm{v}}\big\Vert_{2}}=\kappa_{\varepsilon}\sqrt{\frac{n}{T}}|S|^{3/2}\frac{\big\Vert\bm{\Sigma}_{\varepsilon}^{1/2}\bar{\bm{U}}\big\Vert_{2,\infty}}{\big\Vert\bm{\Sigma}_{\varepsilon}^{1/2}\big\Vert_{2}}.
\]

By definition of $s_{3}$, we obtain by the inequality (\ref{2nd inequality for Chi-square stat plus a term})
in Step 1 for $\mathbb{P}\left(\mathfrak{\mathcal{T}}(S,\bm{\bm{v}})+\zeta\leq\chi_{1-\alpha}^{2}(|S|-r)\right)$
that 
\begin{align}
 & \mathbb{P}\left(\sqrt{\chi^{2}(|S|-r)}\leq\sqrt{\chi_{1-\alpha}^{2}(|S|-r)-\left|\zeta\right|}-\left|\tau\right|\right)-s_{3}\label{3rd inequality for Chi-square stat plus a term}\\
 & \leq\mathbb{P}\left(\mathfrak{\mathcal{T}}(S,\bm{\bm{v}})+\zeta\leq\chi_{1-\alpha}^{2}(|S|-r)\right)\nonumber \\
 & \leq\mathbb{P}\left(\sqrt{\chi^{2}(|S|-r)}\leq\sqrt{\chi_{1-\alpha}^{2}(|S|-r)+\left|\zeta\right|}+\left|\tau\right|\right)+s_{3}.\nonumber 
\end{align}

\textit{Step 3 -- Establishing chi-squared distributional characterization.}

By Lemma \ref{Lemma Chi-square quantile order}, we have that $\chi_{1-\alpha}^{2}(|S|-r)\asymp(|S|-r)\asymp|S|$,
where we use the fact that $r\ll|S|$ and thus $(|S|-r)\asymp|S|$.
Since we assume that $\left|\zeta\right|\ll|S|$, we have that $\left|\zeta\right|\lesssim\chi_{1-\alpha}^{2}(|S|-r)$.
For $a>0$ and $b\in\mathbb{R}$ satisfying $|b|\ll a$, we have that
$|\sqrt{a+b}-\sqrt{a}|=\frac{|b|}{\sqrt{a+b}+\sqrt{a}}\lesssim\frac{|b|}{\sqrt{a}}$.
So we obtain that $\sqrt{\chi_{1-\alpha}^{2}(|S|-r)-\left|\zeta\right|}\geq\sqrt{\chi_{1-\alpha}^{2}(|S|-r)}-g$
and $\sqrt{\chi_{1-\alpha}^{2}(|S|-r)+\left|\zeta\right|}\leq\sqrt{\chi_{1-\alpha}^{2}(|S|-r)}+g$,
where $g>0$ and $g$ satisfies that $g\lesssim\frac{1}{\sqrt{\chi_{1-\alpha}^{2}(|S|-r)}}\left|\zeta\right|\lesssim\frac{1}{\sqrt{|S|}}\left|\zeta\right|$.
Then, using the inequality (\ref{3rd inequality for Chi-square stat plus a term}),
we conclude that, for $h=g+\left|\tau\right|$, it holds 
\begin{align}
 & \mathbb{P}\left(\sqrt{\chi^{2}(|S|-r)}\leq\sqrt{\chi_{1-\alpha}^{2}(|S|-r)}-h\right)-s_{3}\label{4th inequality for Chi-square stat plus a term}\\
 & \leq\mathbb{P}\left(\mathfrak{\mathcal{T}}(S,\bm{\bm{v}})+\zeta\leq\chi_{1-\alpha}^{2}(|S|-r)\right)\nonumber \\
 & \leq\mathbb{P}\left(\sqrt{\chi^{2}(|S|-r)}\leq\sqrt{\chi_{1-\alpha}^{2}(|S|-r)}+h\right)+s_{3}.\nonumber 
\end{align}
Here, $h>0$ satisfies that $h\lesssim\left|\tau\right|+\frac{1}{\sqrt{|S|}}\left|\zeta\right|$.

Note that, for any $R>0$, we have that $\mathbb{P}(\sqrt{\chi^{2}(n)}\leq R)=\mathbb{P}(\mathcal{N}(0,\bm{I}_{n})\in B(R))$
where $B(R)$ is the ball. So, we obtain by Lemma \ref{Lemma Chi-square two balls}
that 
\begin{align*}
 & \left|\mathbb{P}\left(\sqrt{\chi^{2}(|S|-r)}\leq\sqrt{\chi_{1-\alpha}^{2}(|S|-r)}\pm h\right)-(1-\alpha)\right|\\
 & =\left|\mathbb{P}\left(\sqrt{\chi^{2}(|S|-r)}\leq\sqrt{\chi_{1-\alpha}^{2}(|S|-r)}\pm h\right)-\mathbb{P}\left(\sqrt{\chi^{2}(|S|-r)}\leq\sqrt{\chi_{1-\alpha}^{2}(|S|-r)}\right)\right|\lesssim h.
\end{align*}
Thus, using the two-side bounds for $\mathbb{P}\left(\mathfrak{\mathcal{T}}(S,\bm{\bm{v}})+\zeta\leq\chi_{1-\alpha}^{2}(|S|-r)\right)$
in the inequality (\ref{4th inequality for Chi-square stat plus a term}),
we obtain that 
\[
\left|\mathbb{P}\left(\mathfrak{\mathcal{T}}(S,\bm{\bm{v}})+\zeta\leq\chi_{1-\alpha}^{2}(|S|-r)\right)-(1-\alpha)\right|\lesssim h+s_{3}\lesssim\left|\tau\right|+\frac{1}{\sqrt{|S|}}\left|\zeta\right|+s_{3}.
\]
So, to prove the desired result in the lemma, it suffices to establish
the upper bound for $\left|\tau\right|$, which is defined in (\ref{diff sqrt Chi-square stat and 2-norm Gaussian}).

\textit{Step 4 -- Deriving an upper bound for $\left|\tau\right|$.}

For $\tau$ defined in (\ref{diff sqrt Chi-square stat and 2-norm Gaussian}),
using (\ref{Chi-square stat 2-norm equality}) and the triangle inequality,
we obtain that 
\[
|\tau|=\frac{1}{\left\Vert \bm{m}\right\Vert _{2}}\left|\big\Vert\bm{H}^{\top}\bm{u}+\bm{H}^{\top}\bm{\Upsilon}_{P}\big\Vert_{2}-\big\Vert\bm{H}^{\top}\bm{u}\big\Vert_{2}\right|\leq\frac{1}{\left\Vert \bm{m}\right\Vert _{2}}\big\Vert\bm{H}^{\top}\bm{\Upsilon}_{P}\big\Vert_{2}\leq\frac{1}{\left\Vert \bm{m}\right\Vert _{2}}\left\Vert \bm{H}\right\Vert _{2}\left\Vert \bm{\Upsilon}_{P}\right\Vert _{2}=\frac{1}{\left\Vert \bm{m}\right\Vert _{2}}\left\Vert \bm{\Upsilon}_{P}\right\Vert _{2},
\]
where the last equality uses the fact that $\big\Vert\bm{H}\big\Vert_{2}=1$
since $\bm{H}^{\top}\bm{H}=\bm{I}_{|S|-r}$. Combining the lower bound
of $\left\Vert \bm{m}\right\Vert _{2}$ in (\ref{projection iid sum weight 2-norm})
and the upper bound of $\left\Vert \bm{\Upsilon}_{P}\right\Vert _{2}$
in Lemma \ref{Lemma projection vector 1st-order approx}, we obtain
that 
\begin{align*}
|\tau| & \leq\frac{1}{\left\Vert \bm{m}\right\Vert _{2}}\left\Vert \bm{\Upsilon}_{P}\right\Vert _{2}\\
 & \lesssim\frac{\kappa_{\varepsilon}\sqrt{n}}{\rho\big\Vert(\bm{\bm{\bm{V}}}_{S,\cdot})^{+}\bm{\bm{v}}\big\Vert_{2}}\sqrt{|S|}\left(\rho^{2}\frac{1}{\sqrt{n}}\sqrt{\frac{n}{T}}+\rho\frac{1}{\sqrt{nT}}\right)r\log^{2}n\big\Vert(\bm{\bm{\bm{V}}}_{S,\cdot})^{+}\bm{\bm{v}}\big\Vert_{2}\\
 & =\lesssim\kappa_{\varepsilon}\sqrt{|S|}\left(\rho\sqrt{\frac{n}{T}}+\frac{1}{\sqrt{T}}\right)r\log^{2}n.
\end{align*}
Finally, the lower bound of $\phi$ can be proven using the lower
bound of $\left\Vert \bm{m}\right\Vert _{2}$ in (\ref{projection iid sum weight 2-norm})
as follows, 
\begin{align*}
\phi & =\left\Vert \bm{m}\right\Vert _{2}^{2}\gtrsim\left(\frac{1}{\sqrt{n}}\frac{1}{\kappa_{\varepsilon}}\rho\big\Vert(\bm{\bm{\bm{V}}}_{S,\cdot})^{+}\bm{\bm{v}}\big\Vert_{2}\right)^{2}=\frac{1}{\kappa_{\varepsilon}^{2}}\frac{1}{n}\rho^{2}\big\Vert(\bm{\bm{\bm{V}}}_{S,\cdot})^{+}\bm{\bm{v}}\big\Vert_{2}^{2},
\end{align*}
which is our desired result. 
\end{proof}

\subsection{Proof of Theorem \ref{Thm factor test plug-in Chi-sq}}

To prove the desired results, we use Lemma \ref{Lemma true Chi-square stat}
with the small perturbation $\zeta=\mathfrak{\widehat{\mathcal{T}}}(S,\bm{\bm{v}})-\mathcal{T}(S,\bm{\bm{v}})$.
Then, the problem boils down to deriving the upper bound for $\frac{1}{\sqrt{|S|}}|\zeta|$.
We do the following decomposition: 
\begin{align*}
\zeta=\mathfrak{\widehat{\mathcal{T}}}(S,\bm{\bm{v}})-\mathcal{T}(S,\bm{\bm{v}}) & =r_{1}+r_{2},
\end{align*}
where $r_{1}$ and $r_{2}$ are respectively defined by 
\[
r_{1}=\frac{1}{\phi}\bm{\bm{v}}^{\top}(\bm{I}_{|S|}-\bm{P}_{\widehat{\bm{V}}_{S,\cdot}})\bm{\bm{v}}-\mathcal{T}(S,\bm{\bm{v}})\text{\qquad and\qquad}r_{2}=\mathfrak{\widehat{\mathcal{T}}}(S,\bm{\bm{v}})-\frac{1}{\phi}\bm{\bm{v}}^{\top}(\bm{I}_{|S|}-\bm{P}_{\widehat{\bm{V}}_{S,\cdot}})\bm{\bm{v}}.
\]

\textit{Step 1 -- Upper bound for $\frac{1}{\sqrt{|S|}}|r_{1}|$.}

Since $(\bm{I}_{|S|}-\bm{P}_{\widehat{\bm{V}}_{S,\cdot}})$ is idempotent,
we can rewrite $r_{1}$ as 
\begin{align*}
r_{1} & =\frac{1}{\phi}\bm{\bm{v}}^{\top}(\bm{I}_{|S|}-\bm{P}_{\widehat{\bm{V}}_{S,\cdot}})(\bm{I}_{|S|}-\bm{P}_{\widehat{\bm{V}}_{S,\cdot}})(\bm{I}_{|S|}-\bm{P}_{\widehat{\bm{V}}_{S,\cdot}})\bm{\bm{v}}-\frac{1}{\phi}\bm{\bm{v}}^{\top}(\bm{I}_{|S|}-\bm{P}_{\widehat{\bm{V}}_{S,\cdot}})(\bm{I}_{|S|}-\bm{P}_{\bm{\bm{V}}_{S,\cdot}})(\bm{I}_{|S|}-\bm{P}_{\widehat{\bm{V}}_{S,\cdot}})\bm{\bm{v}}\\
 & =\frac{1}{\phi}\bm{\bm{v}}^{\top}(\bm{I}_{|S|}-\bm{P}_{\widehat{\bm{V}}_{S,\cdot}})(\bm{P}_{\bm{\bm{V}}_{S,\cdot}}-\bm{P}_{\widehat{\bm{V}}_{S,\cdot}})(\bm{I}_{|S|}-\bm{P}_{\widehat{\bm{V}}_{S,\cdot}})\bm{\bm{v}}.
\end{align*}
Then we get 
\begin{align*}
|r_{1}| & \leq\frac{1}{\phi}\big\Vert\bm{P}_{\bm{\bm{V}}_{S,\cdot}}-\bm{P}_{\widehat{\bm{V}}_{S,\cdot}}\big\Vert_{2}\big\Vert(\bm{I}_{|S|}-\bm{P}_{\widehat{\bm{V}}_{S,\cdot}})\bm{\bm{v}}\big\Vert_{2}^{2}\\
 & \overset{\text{(i)}}{\lesssim}\frac{\kappa_{\varepsilon}^{2}n}{\rho^{2}\big\Vert(\bm{\bm{\bm{V}}}_{S,\cdot})^{+}\bm{\bm{v}}\big\Vert_{2}^{2}}\rho\sqrt{\frac{T}{n}}\sqrt{r}\log^{3/2}n\left(\sqrt{|S|}\rho\sqrt{\frac{r}{n}}\log n\big\Vert(\bm{\bm{\bm{V}}}_{S,\cdot})^{+}\bm{\bm{v}}\big\Vert_{2}\right)^{2}\\
 & \leq\kappa_{\varepsilon}^{2}|S|\rho\sqrt{\frac{T}{n}}r^{3/2}\log^{7/2}n,
\end{align*}
where (i) uses the results in Lemmas \ref{Lemma projection vector 1st-order approx}
and \ref{Lemma true Chi-square stat}. So we obtain 
\[
\frac{|r_{1}|}{\sqrt{|S|}}\lesssim\kappa_{\varepsilon}^{2}\rho\sqrt{|S|}\sqrt{\frac{T}{n}}r^{3/2}\log^{7/2}n\ll1.
\]

\textit{Step 2 -- Upper bound for $\frac{1}{\sqrt{|S|}}|r_{2}|$.}

Since $(\bm{I}_{|S|}-\bm{P}_{\widehat{\bm{V}}_{S,\cdot}})$ is idempotent,
we can rewrite $r_{2}$ as 
\begin{align*}
r_{2} & =\frac{1}{\widehat{\phi}}\bm{\bm{v}}^{\top}(\bm{I}_{|S|}-\bm{P}_{\widehat{\bm{V}}_{S,\cdot}})(\bm{I}_{|S|}-\bm{P}_{\widehat{\bm{V}}_{S,\cdot}})\bm{\bm{v}}-\frac{1}{\phi}\bm{\bm{v}}^{\top}(\bm{I}_{|S|}-\bm{P}_{\widehat{\bm{V}}_{S,\cdot}})(\bm{I}_{|S|}-\bm{P}_{\widehat{\bm{V}}_{S,\cdot}})\bm{\bm{v}}\\
 & =(\frac{1}{\widehat{\phi}}-\frac{1}{\phi})\big\Vert(\bm{I}_{|S|}-\bm{P}_{\widehat{\bm{V}}_{S,\cdot}})\bm{\bm{v}}\big\Vert_{2}^{2}=\frac{\phi-\widehat{\phi}}{\widehat{\phi}}\frac{\big\Vert(\bm{I}_{|S|}-\bm{P}_{\widehat{\bm{V}}_{S,\cdot}})\bm{\bm{v}}\big\Vert_{2}^{2}}{\phi}.
\end{align*}
Then, using the results in Lemmas \ref{Lemma projection vector 1st-order approx}
and \ref{Lemma true Chi-square stat}, we obtain that 
\begin{equation}
|r_{2}|\lesssim\frac{1}{\widehat{\phi}}\bigl|\widehat{\phi}-\phi\bigl|\cdot\frac{\left(\sqrt{|S|}\rho\sqrt{\frac{r}{n}}\log n\big\Vert(\bm{\bm{\bm{V}}}_{S,\cdot})^{+}\bm{\bm{v}}\big\Vert_{2}\right)^{2}}{\frac{1}{\kappa_{\varepsilon}^{2}}\frac{1}{n}\rho^{2}\big\Vert(\bm{\bm{\bm{V}}}_{S,\cdot})^{+}\bm{\bm{v}}\big\Vert_{2}^{2}}\leq\frac{1}{\widehat{\phi}}\bigl|\widehat{\phi}-\phi\bigl|\cdot|S|\kappa_{\varepsilon}^{2}r\log^{2}n.\label{scalar prepare bound for Chi-square test}
\end{equation}
So, to derive the upper bound for $|r_{2}|$, we need to derive the
upper bound of $\frac{1}{\widehat{\phi}}\bigl|\widehat{\phi}-\phi\bigl|$.

Note that $\widehat{\phi}$ and $\phi$ are nonnegative owing the
property of positive semidefinite matrix. We will derive the upper
bound for $\bigl|\widehat{\phi}-\phi\bigl|$ first, and then use it
and the lower bound of $\phi$ obtained in Lemma \ref{Lemma true Chi-square stat}
to get the lower bound of $\widehat{\phi}$, because $\widehat{\phi}\geq\phi-\bigl|\widehat{\phi}-\phi\bigl|$.
Our starting point is to rewrite $\phi$ as the following form: 
\[
\widehat{\phi}=\frac{1}{T}((\widehat{\bm{V}}_{S,\cdot}\bm{R}_{V})^{+}\bm{\bm{v}})^{\top}((\bm{R}_{U})^{\top}\widehat{\bm{\Sigma}}^{-1}\bm{R}_{V})^{\top}(\widehat{\bm{U}}\bm{R}_{U})^{\top}\widehat{\bm{\bm{\Sigma}}}_{\varepsilon}^{\tau}(\widehat{\bm{U}}\bm{R}_{U})((\bm{R}_{U})^{\top}\widehat{\bm{\Sigma}}^{-1}\bm{R}_{V})(\widehat{\bm{V}}_{S,\cdot}\bm{R}_{V})^{+}\bm{\bm{v}},
\]
where we use the fact that $(\widehat{\bm{V}}_{S,\cdot}\bm{R}_{V})^{+}=(\bm{R}_{V})^{\top}(\widehat{\bm{V}}_{S,\cdot})^{+}$
since $\widehat{\bm{V}}_{S,\cdot}$ has full column rank and $(\bm{\bm{\bm{V}}}_{S,\cdot})^{+}=[(\bm{\bm{\bm{V}}}_{S,\cdot})^{\top}\bm{\bm{\bm{V}}}_{S,\cdot}]^{-1}(\bm{\bm{\bm{V}}}_{S,\cdot})^{\top}$.
Recall that 
\[
\phi=\frac{1}{T}((\bm{\bm{\bm{V}}}_{S,\cdot})^{+}\bm{\bm{v}})^{\top}\bm{\Lambda}^{-1}\bm{U}^{\top}\bm{\Sigma}_{\varepsilon}\bm{U}\bm{\Lambda}^{-1}(\bm{\bm{\bm{V}}}_{S,\cdot})^{+}\bm{\bm{v}}.
\]
Next, for simplicity of notations, we denote 
\[
A_{0}=(\bm{\bm{\bm{V}}}_{S,\cdot})^{+},\text{ }B_{0}=\bm{\Lambda}^{-1},\text{ }C_{0}=\bm{U},\text{ }D_{0}=\bm{\Sigma}_{\varepsilon},
\]
and 
\[
A_{1}=(\widehat{\bm{V}}_{S,\cdot}\bm{R}_{V})^{+},\text{ }B_{1}=(\bm{R}_{U})^{\top}\widehat{\bm{\Sigma}}^{-1}\bm{R}_{V},\text{ }C_{1}=\widehat{\bm{U}}\bm{R}_{U},\text{ }D_{1}=\widehat{\bm{\bm{\Sigma}}}_{\varepsilon}^{\tau},
\]
so that $\widehat{\phi}$ and $\phi$ can be written as 
\[
\widehat{\phi}=\frac{1}{T}\bm{\bm{v}}^{\top}A_{1}^{\top}B_{1}^{\top}C_{1}^{\top}D_{1}C_{1}B_{1}A_{1}\bm{\bm{v}}\qquad\text{and}\qquad\phi=\frac{1}{T}\bm{\bm{v}}^{\top}A_{0}^{\top}B_{0}^{\top}C_{0}^{\top}D_{0}C_{0}B_{0}A_{0}\bm{\bm{v}},
\]
respectively.

To derive the upper bound for $\bigl|\widehat{\phi}-\phi\bigl|$,
we will establish the upper bounds for $\left\Vert A_{1}-A_{0}\right\Vert _{2}$,
$\left\Vert B_{1}-B_{0}\right\Vert _{2}$, $\left\Vert C_{1}-C_{0}\right\Vert _{2}$,
and $\left\Vert D_{1}-D_{0}\right\Vert _{2}$ respectively.

\textit{Step 2.1.}
For $\left\Vert A_{1}-A_{0}\right\Vert _{2}$, using Theorem 3.3 in
\citet{stewart1977perturbation}, we obtain that 
\begin{align*}
\left\Vert A_{1}-A_{0}\right\Vert _{2} & =\big\Vert(\widehat{\bm{V}}_{S,\cdot}\bm{R}_{V})^{+}-(\bm{\bm{\bm{V}}}_{S,\cdot})^{+}\big\Vert_{2}\\
 & \lesssim\max(\big\Vert(\bm{\bm{\bm{V}}}_{S,\cdot})^{+}\big\Vert_{2}^{2},\big\Vert(\widehat{\bm{V}}_{S,\cdot}\bm{R}_{V})^{+}\big\Vert_{2}^{2})\big\Vert\widehat{\bm{V}}_{S,\cdot}\bm{R}_{V}-\bm{\bm{V}}_{S,\cdot}\big\Vert_{2}.
\end{align*}
For the term $\big\Vert\widehat{\bm{V}}_{S,\cdot}\bm{R}_{V}-\bm{\bm{\bm{V}}}_{S,\cdot}\big\Vert_{2}$,
we have that 
\[
\big\Vert\widehat{\bm{V}}_{S,\cdot}\bm{R}_{V}-\bm{\bm{\bm{V}}}_{S,\cdot}\big\Vert_{2}\leq\sqrt{|S|}\big\Vert\widehat{\bm{V}}_{S,\cdot}\bm{R}_{V}-\bm{\bm{\bm{V}}}_{S,\cdot}\big\Vert_{2,\infty}\leq\sqrt{|S|}\big\Vert\widehat{\bm{V}}\bm{R}_{V}-\bm{V}\big\Vert_{2,\infty}\lesssim\sqrt{|S|}\rho\sqrt{\frac{r}{n}}\log n,
\]
where the last inequality uses the upper bound of $\big\Vert\widehat{\bm{V}}\bm{R}_{V}-\bm{V}\big\Vert_{2,\infty}$
obtained in the proof of Corollary \ref{corollary:error bound B F}.
We obtain by Lemma \ref{Lemma V subset good event} that $\sigma_{i}(\bm{\bm{\bm{V}}}_{S,\cdot})\asymp\sqrt{\frac{|S|}{T}}$
for $i=1,2,\ldots,r$. Next, since (\ref{assemble condition SNR and data size})
implies that $\rho\sqrt{\frac{T}{n}}\sqrt{r}\log n\ll1$, we obtain
that $\big\Vert\widehat{\bm{V}}_{S,\cdot}\bm{R}_{V}-\bm{\bm{\bm{V}}}_{S,\cdot}\big\Vert_{2}\lesssim\sqrt{|S|}\rho\sqrt{\frac{r}{n}}\log n\ll\sqrt{\frac{|S|}{T}}$,
and thus $\sigma_{i}(\widehat{\bm{V}}_{S,\cdot}\bm{R}_{V})\asymp\sqrt{\frac{|S|}{T}}$
for $i=1,2,\ldots,r$. Then, similar to the proof of $\sigma_{i}((\bm{\bm{\bm{V}}}_{S,\cdot})^{+})\asymp\sqrt{\frac{T}{|S|}}$
for $i=1,2,\ldots,r$ in Lemma \ref{Lemma V subset good event}, we
obtain that $\sigma_{i}((\widehat{\bm{V}}_{S,\cdot}\bm{R}_{V})^{+})\asymp\sqrt{\frac{T}{|S|}}$
for $i=1,2,\ldots,r$. So we obtain 
\begin{align*}
\left\Vert A_{1}-A_{0}\right\Vert _{2} & \lesssim\max(\big\Vert(\bm{\bm{\bm{V}}}_{S,\cdot})^{+}\big\Vert_{2}^{2},\big\Vert(\widehat{\bm{V}}_{S,\cdot}\bm{R}_{V})^{+}\big\Vert_{2}^{2})\big\Vert\widehat{\bm{V}}_{S,\cdot}\bm{R}_{V}-\bm{\bm{V}}_{S,\cdot}\big\Vert_{2}\\
 & \lesssim\frac{T}{|S|}\cdot\sqrt{|S|}\rho\sqrt{\frac{r}{n}}\log n=\frac{T}{\sqrt{n|S|}}\rho\sqrt{r}\log n.
\end{align*}
Also, we obtain by Lemma \ref{Lemma V subset good event} that $\left\Vert A_{0}\right\Vert _{2}=\big\Vert(\bm{\bm{\bm{V}}}_{S,\cdot})^{+}\big\Vert_{2}\lesssim\sqrt{\frac{T}{|S|}}$.
Since $\sigma_{i}((\widehat{\bm{V}}_{S,\cdot}\bm{R}_{V})^{+})\asymp\sqrt{\frac{T}{|S|}}$
for $i=1,2,\ldots,r$, we obtain that $\left\Vert A_{1}\right\Vert _{2}=\big\Vert(\widehat{\bm{V}}_{S,\cdot}\bm{R}_{V})^{+}\big\Vert_{2}\lesssim\sqrt{\frac{T}{|S|}}$.

For $\left\Vert B_{1}-B_{0}\right\Vert _{2}$, since $\bm{R}_{U}$
and $\bm{R}_{V}$ are rotation matrices, we have that 
\begin{align*}
\left\Vert B_{1}-B_{0}\right\Vert _{2} & =\big\Vert(\bm{R}_{U})^{\top}\widehat{\bm{\Sigma}}^{-1}\bm{R}_{V}-\bm{\Lambda}^{-1}\big\Vert_{2}\\
 & =\big\Vert[(\bm{R}_{U})^{\top}\widehat{\bm{\Sigma}}^{-1}-\bm{\Lambda}^{-1}(\bm{R}_{V})^{\top}]\bm{R}_{V}\big\Vert_{2}\\
 & \leq\big\Vert\bm{\Lambda}^{-1}[\bm{\Lambda}(\bm{R}_{U})^{\top}-(\bm{R}_{V})^{\top}\widehat{\bm{\Sigma}}]\widehat{\bm{\Sigma}}^{-1}\big\Vert_{2}\leq\frac{1}{\lambda_{r}}\big\Vert\bm{\Lambda}(\bm{R}_{U})^{\top}-(\bm{R}_{V})^{\top}\widehat{\bm{\Sigma}}\big\Vert_{2}\frac{1}{\widehat{\sigma}_{r}}\\
 & \overset{\text{(i)}}{\lesssim}\frac{1}{\sigma_{r}^{2}}\big\Vert[\bm{\Lambda}-(\bm{R}_{V})^{\top}\widehat{\bm{\Sigma}}\bm{R}_{U}](\bm{R}_{U})^{\top}\big\Vert_{2}\leq\frac{1}{\sigma_{r}^{2}}\big\Vert(\bm{R}_{U})^{\top}\widehat{\bm{\Sigma}}\bm{R}_{V}-\bm{\Lambda}\big\Vert_{2}\\
 & \overset{\text{(ii)}}{\lesssim}\frac{1}{\sigma_{r}^{2}}\left(\sigma_{r}\rho^{2}+\sigma_{r}\rho\sqrt{\frac{r}{n}}\sqrt{\log n}\right)=\frac{1}{\sigma_{r}}\left(\rho^{2}+\rho\sqrt{\frac{r}{n}}\sqrt{\log n}\right),
\end{align*}
where (i) uses $\widehat{\sigma}_{i}\asymp\sigma_{i}$ and $\lambda_{i}\asymp\sigma_{i}$
in Lemmas \ref{Lemma SVD BF good event} and \ref{Lemma R H for U V},
respectively, and (ii) uses Lemma \ref{Lemma Sigma hat tilde H R}.
Also, since $\left\Vert B_{0}\right\Vert _{2}=\big\Vert\bm{\Lambda}^{-1}\big\Vert_{2}\lesssim\frac{1}{\sigma_{r}}$
and the above result implies that $\left\Vert B_{1}-B_{0}\right\Vert _{2}\ll\frac{1}{\sigma_{r}}$
owing to $\rho+\sqrt{\frac{r}{n}}\sqrt{\log n}\ll1$ as assumed, so
we obtain $\left\Vert B_{1}\right\Vert _{2}\leq\left\Vert B_{1}-B_{0}\right\Vert _{2}+\left\Vert B_{0}\right\Vert _{2}\lesssim\frac{1}{\sigma_{r}}$.

For $\left\Vert C_{1}-C_{0}\right\Vert _{2}=\big\Vert\widehat{\bm{U}}\bm{R}_{U}-\bm{U}\big\Vert_{2}$,
we obtain by Lemma \ref{Lemma R H for U V} that $\left\Vert C_{1}-C_{0}\right\Vert _{2}\lesssim\rho$.
Also, since $\widehat{\bm{U}}^{\top}\widehat{\bm{U}}=\bm{U}^{\top}\bm{U}=\bm{I}_{r}$
and $\bm{R}_{U}\in\mathcal{O}^{r\times r}$ is a rotation matrix,
we have that $\left\Vert C_{1}\right\Vert _{2}=\big\Vert\widehat{\bm{U}}\bm{R}_{U}\big\Vert_{2}\leq1$
and $\left\Vert C_{0}\right\Vert _{2}=\big\Vert\bm{U}\big\Vert_{2}\leq1$.

For $\left\Vert D_{1}-D_{0}\right\Vert _{2}=\big\Vert\widehat{\bm{\bm{\Sigma}}}_{\varepsilon}^{\tau}-\bm{\Sigma}_{\varepsilon}\big\Vert_{2}$,
by assumption we have that $\left\Vert D_{1}-D_{0}\right\Vert _{2}\sqrt{|S|}\cdot\kappa_{\varepsilon}^{4}r\log^{2}n\ll\left\Vert \bm{\Sigma}_{\varepsilon}\right\Vert _{2}$.

In particular, the above results imply that, $\left\Vert A_{i}\right\Vert _{2}\lesssim K_{A}$,
$\left\Vert B_{i}\right\Vert _{2}\lesssim K_{B}$, $\left\Vert C_{i}\right\Vert _{2}\lesssim K_{C}$,
and $\left\Vert D_{i}\right\Vert _{2}\lesssim K_{D}$ for $i=0,1$,
where $K_{A}=\sqrt{\frac{T}{|S|}}$, $K_{B}=\frac{1}{\sigma_{r}}$,
$K_{C}=1$, and $K_{D}=\left\Vert \bm{\Sigma}_{\varepsilon}\right\Vert _{2}$.

\textit{Step 2.2.}
We obtain by the representations of $\widehat{\phi}$ and $\phi$
that, 
\[
\bigl|\widehat{\phi}-\phi\bigl|\lesssim\frac{1}{T}\left(\frac{\left\Vert A_{1}-A_{0}\right\Vert _{2}}{K_{A}}+\frac{\left\Vert B_{1}-B_{0}\right\Vert _{2}}{K_{B}}+\frac{\left\Vert C_{1}-C_{0}\right\Vert _{2}}{K_{C}}+\frac{\left\Vert D_{1}-D_{0}\right\Vert _{2}}{K_{D}}\right)K_{A}^{2}K_{B}^{2}K_{C}^{2}K_{D}\left\Vert \bm{\bm{v}}\right\Vert _{2}^{2}.
\]
Then using the results in Step 2.1, we obtain that 
\begin{align*}
 & \bigl|\widehat{\phi}-\phi\bigl|\\
 & \lesssim\frac{1}{T}\left|\frac{1}{\left\Vert \bm{\Sigma}_{\varepsilon}\right\Vert _{2}}\big\Vert\widehat{\bm{\bm{\Sigma}}}_{\varepsilon}^{\tau}-\bm{\Sigma}_{\varepsilon}\big\Vert_{2}+\frac{\rho}{1}+\frac{\frac{1}{\sigma_{r}}\rho\left(\rho+\sqrt{\frac{r}{n}}\sqrt{\log n}\right)}{\frac{1}{\sigma_{r}}}+\frac{\frac{T}{\sqrt{n|S|}}\rho\sqrt{r}\left\Vert \bm{\bm{v}}\right\Vert _{2}}{\sqrt{\frac{T}{|S|}}\left\Vert \bm{\bm{v}}\right\Vert _{2}}\right|\cdot\frac{T}{|S|}\cdot\frac{1}{\sigma_{r}^{2}}\cdot\left\Vert \bm{\Sigma}_{\varepsilon}\right\Vert _{2}\cdot\left\Vert \bm{\bm{v}}\right\Vert _{2}^{2}\\
 & \lesssim\left|\frac{1}{\left\Vert \bm{\Sigma}_{\varepsilon}\right\Vert _{2}}\big\Vert\widehat{\bm{\bm{\Sigma}}}_{\varepsilon}^{\tau}-\bm{\Sigma}_{\varepsilon}\big\Vert_{2}+\rho+\rho\left(\rho+\sqrt{\frac{r}{n}}\sqrt{\log n}\right)+\rho\sqrt{\frac{T}{n}}\sqrt{r}\right|\cdot\frac{1}{|S|}\rho^{2}\frac{T}{n}\left\Vert \bm{\bm{v}}\right\Vert _{2}^{2}\\
 & \overset{\text{(i)}}{\lesssim}\left|\frac{1}{\left\Vert \bm{\Sigma}_{\varepsilon}\right\Vert _{2}}\big\Vert\widehat{\bm{\bm{\Sigma}}}_{\varepsilon}^{\tau}-\bm{\Sigma}_{\varepsilon}\big\Vert_{2}+\rho+\rho\sqrt{\frac{T}{n}}\sqrt{r}\right|\cdot\frac{1}{|S|}\rho^{2}\frac{T}{n}\left\Vert \bm{\bm{v}}\right\Vert _{2}^{2}.
\end{align*}
where (i) uses in (\ref{assemble condition SNR and data size}). We
obtain by the bound of $\phi$ in Lemma \ref{Lemma true Chi-square stat}
and the bound of $\sigma_{\text{min}}((\bm{\bm{\bm{V}}}_{S,\cdot})^{+})$
in Lemma \ref{Lemma V subset good event} that 
\[
\phi\gtrsim\frac{1}{\kappa_{\varepsilon}^{2}n}\rho^{2}\big\Vert(\bm{\bm{\bm{V}}}_{S,\cdot})^{+}\bm{\bm{v}}\big\Vert_{2}^{2}\geq\frac{1}{\kappa_{\varepsilon}^{2}n}\rho^{2}[\sigma_{\text{min}}((\bm{\bm{\bm{V}}}_{S,\cdot})^{+})]^{2}\left\Vert \bm{\bm{v}}\right\Vert _{2}^{2}\gtrsim\frac{1}{\kappa_{\varepsilon}^{2}n}\rho^{2}\frac{T}{|S|}\left\Vert \bm{\bm{v}}\right\Vert _{2}^{2}.
\]
Then, we conclude that $\bigl|\widehat{\phi}-\phi\bigl|\ll\phi$ because
the ratio between the above upper bound of $\bigl|\widehat{\phi}-\phi\bigl|$
and the above lower bound of $\phi$ is given by 
\[
\frac{\left|\frac{1}{\left\Vert \bm{\Sigma}_{\varepsilon}\right\Vert _{2}}\big\Vert\widehat{\bm{\bm{\Sigma}}}_{\varepsilon}^{\tau}-\bm{\Sigma}_{\varepsilon}\big\Vert_{2}+\rho+\rho\sqrt{\frac{T}{n}}\sqrt{r}\right|\cdot\frac{1}{|S|}\rho^{2}\frac{T}{n}\left\Vert \bm{\bm{v}}\right\Vert _{2}^{2}}{\frac{1}{\kappa_{\varepsilon}^{2}n}\rho^{2}\frac{T}{|S|}\left\Vert \bm{\bm{v}}\right\Vert _{2}^{2}}=\kappa_{\varepsilon}^{2}\left|\frac{1}{\left\Vert \bm{\Sigma}_{\varepsilon}\right\Vert _{2}}\big\Vert\widehat{\bm{\bm{\Sigma}}}_{\varepsilon}^{\tau}-\bm{\Sigma}_{\varepsilon}\big\Vert_{2}+\rho+\rho\sqrt{\frac{T}{n}}\sqrt{r}\right|\ll1,
\]
where the last inequality is owing to the assumptions. So $\widehat{\phi}\geq\phi-\bigl|\widehat{\phi}-\phi\bigl|\gtrsim\phi\gtrsim\frac{1}{\kappa_{\varepsilon}^{2}n}\rho^{2}\frac{T}{|S|}\left\Vert \bm{\bm{v}}\right\Vert _{2}^{2}$.
Using this lower bound of $\widehat{\phi}$ and the above upper bound
of $\bigl|\widehat{\phi}-\phi\bigl|$, we obtain that 
\begin{equation}
\frac{1}{\widehat{\phi}}\bigl|\widehat{\phi}-\phi\bigl|\lesssim\frac{\left|R_{\varepsilon}+\rho+\rho\sqrt{\frac{T}{n}}\sqrt{r}\right|\cdot\frac{1}{|S|}\rho^{2}\frac{T}{n}\left\Vert \bm{\bm{v}}\right\Vert _{2}^{2}}{\frac{1}{\kappa_{\varepsilon}^{2}n}\rho^{2}\frac{T}{|S|}\left\Vert \bm{\bm{v}}\right\Vert _{2}^{2}}=\kappa_{\varepsilon}^{2}\left|\frac{1}{\left\Vert \bm{\Sigma}_{\varepsilon}\right\Vert _{2}}\big\Vert\widehat{\bm{\bm{\Sigma}}}_{\varepsilon}^{\tau}-\bm{\Sigma}_{\varepsilon}\big\Vert_{2}+\rho+\rho\sqrt{\frac{T}{n}}\sqrt{r}\right|.\label{scalar diff ratio bound for Chi-square test}
\end{equation}

Using (\ref{scalar prepare bound for Chi-square test}), (\ref{scalar diff ratio bound for Chi-square test}),
and the fact that $\kappa_{\varepsilon}\geq1$, we obtain that 
\begin{align*}
\frac{|r_{2}|}{\sqrt{|S|}} & \lesssim\frac{1}{\widehat{\phi}}\bigl|\widehat{\phi}-\phi\bigl|\sqrt{|S|}\kappa_{\varepsilon}^{2}r\log^{2}n\lesssim\left|\frac{1}{\left\Vert \bm{\Sigma}_{\varepsilon}\right\Vert _{2}}\big\Vert\widehat{\bm{\bm{\Sigma}}}_{\varepsilon}^{\tau}-\bm{\Sigma}_{\varepsilon}\big\Vert_{2}+\rho+\sqrt{\frac{T}{n}}\rho\sqrt{r}\right|\sqrt{|S|}\kappa_{\varepsilon}^{4}r\log^{2}n.\\
 & \lesssim\frac{1}{\left\Vert \bm{\Sigma}_{\varepsilon}\right\Vert _{2}}\big\Vert\widehat{\bm{\bm{\Sigma}}}_{\varepsilon}^{\tau}-\bm{\Sigma}_{\varepsilon}\big\Vert_{2}\sqrt{|S|}\kappa_{\varepsilon}^{4}r\log^{2}n+\rho\sqrt{|S|}\kappa_{\varepsilon}^{4}r^{2}\log^{2}n.
\end{align*}

\textit{Step 3.}
We now assemble all the above upper bounds. We obtain by Lemma \ref{Lemma true Chi-square stat}
that 
\begin{align*}
 & \left|\mathbb{P}\left(\mathfrak{\widehat{\mathcal{T}}}(S,\bm{\bm{v}})\leq\chi_{1-\alpha}^{2}(|S|-r)\right)-(1-\alpha)\right|\\
 & \lesssim\sqrt{|S|}(\rho\sqrt{\frac{n}{T}}+\frac{1}{\sqrt{T}})\kappa_{\varepsilon}r\log^{2}n+\kappa_{\varepsilon}^{2}\rho\sqrt{|S|}\sqrt{\frac{T}{n}}r^{3/2}\log^{7/2}n\\
 & +\frac{1}{\left\Vert \bm{\Sigma}_{\varepsilon}\right\Vert _{2}}\big\Vert\widehat{\bm{\bm{\Sigma}}}_{\varepsilon}^{\tau}-\bm{\Sigma}_{\varepsilon}\big\Vert_{2}\sqrt{|S|}\kappa_{\varepsilon}^{4}r\log^{2}n+\rho\sqrt{|S|}\kappa_{\varepsilon}^{4}r^{2}\log^{2}n+s_{3}\\
 & \lesssim\frac{1}{\left\Vert \bm{\Sigma}_{\varepsilon}\right\Vert _{2}}\big\Vert\widehat{\bm{\bm{\Sigma}}}_{\varepsilon}^{\tau}-\bm{\Sigma}_{\varepsilon}\big\Vert_{2}\sqrt{|S|}\kappa_{\varepsilon}^{4}r\log^{2}n+\sqrt{|S|}\rho\sqrt{\frac{n}{T}}\kappa_{\varepsilon}^{4}r^{2}\log^{4}n+\frac{\sqrt{|S|}}{\sqrt{T}}\kappa_{\varepsilon}r\log^{2}n+s_{3}
\end{align*}
where, as defined in Lemma \ref{Lemma true Chi-square stat}, 
\[
s_{3}\lesssim\kappa_{\varepsilon}\sqrt{\frac{n}{T}}|S|^{3/2}\frac{\big\Vert\bm{\Sigma}_{\varepsilon}^{1/2}\bar{\bm{U}}\big\Vert_{2,\infty}}{\big\Vert\bm{\Sigma}_{\varepsilon}^{1/2}\big\Vert_{2}},
\]
and, if all the entries of the matrix $\bm{Z}$ are Gaussian, i.e.,
the noise is Gaussian, then the above inequality for $|\mathbb{P}(\mathfrak{\widehat{\mathcal{T}}}(S,\bm{\bm{v}})\leq\chi_{1-\alpha}^{2}(|S|-r))-(1-\alpha)|$
holds when $s_{3}=0$. Then the desired non-asymptotic results follow
from the above upper bound of $|\mathbb{P}(\mathfrak{\widehat{\mathcal{T}}}(S,\bm{\bm{v}})\leq\chi_{1-\alpha}^{2}(|S|-r))-(1-\alpha)|$.

%% file: appendix_covariance_thresholding.tex
\subsection{Estimation of the idiosyncratic noise covariance matrix}
\textit{Step 1 -- Estimation error for the sample covariance matrix $\widehat{\bm{\bm{\Sigma}}}_{\varepsilon}=\frac{1}{T}\widehat{\bm{E}}\widehat{\bm{E}}^{\top}$.}

We obtain by $x_{i,t}=\bm{b}_{i}^{\top}\bm{f}_{t}+\varepsilon_{i,t}$
that $\widehat{\varepsilon}_{i,t}-\varepsilon_{i,t}=\bm{b}_{i}^{\top}\bm{f}_{t}-\widehat{\bm{b}}_{i}^{\top}\widehat{\bm{f}}_{t}$.
So, the sample covariance estimator $\widehat{\bm{\bm{\Sigma}}}_{\varepsilon}$
is given by 
\begin{align*}
\widehat{\bm{\bm{\Sigma}}}_{\varepsilon} & =\frac{1}{T}\widehat{\bm{E}}\widehat{\bm{E}}^{\top}=\frac{1}{T}(\bm{X}-\widehat{\bm{B}}\widehat{\bm{F}}^{\top})(\bm{X}-\widehat{\bm{B}}\widehat{\bm{F}}^{\top})^{\top}\\
 & =\frac{1}{T}(\bm{X}\bm{X}^{\top}-\bm{X}\widehat{\bm{B}}\widehat{\bm{F}}^{\top}-\widehat{\bm{B}}\widehat{\bm{F}}^{\top}\bm{X}+\widehat{\bm{B}}\widehat{\bm{F}}^{\top}\widehat{\bm{F}}\widehat{\bm{B}}^{\top})=\frac{1}{T}(\bm{X}\bm{X}^{\top}-T\widehat{\bm{B}}\widehat{\bm{B}}^{\top}),
\end{align*}
where we use the fact that $\frac{1}{T}\bm{X}\widehat{\bm{F}}=\frac{1}{\sqrt{T}}\bm{X}\widehat{\bm{V}}=\widehat{\bm{U}}\widehat{\bm{\Sigma}}=\widehat{\bm{B}}$.
Then, we use the identity $\bm{X}=\bm{BF}^{\top}+\bm{E}$ to obtain
that 
\begin{align*}
\widehat{\bm{\bm{\Sigma}}}_{\varepsilon} & =\frac{1}{T}((\bm{BF}^{\top}+\bm{E})(\bm{BF}^{\top}+\bm{E})^{\top}-T\widehat{\bm{B}}\widehat{\bm{B}}^{\top})\\
 & =\frac{1}{T}\bm{E}\bm{E}^{\top}+(\frac{1}{T}\bm{E}\bm{FB}^{\top}+\frac{1}{T}\bm{BF}^{\top}\bm{E}^{\top})+(\frac{1}{T}\bm{B\bm{F}^{\top}\bm{F}B}^{\top}-\widehat{\bm{B}}\widehat{\bm{B}}^{\top}).
\end{align*}
Then, using the fact that $\frac{1}{T}\bm{B\bm{F}^{\top}\bm{F}B}^{\top}=\bm{U}\bm{\Lambda}^{2}\bm{U}^{\top}$
and $\widehat{\bm{B}}\widehat{\bm{B}}^{\top}=\widehat{\bm{U}}\widehat{\bm{\Sigma}}^{2}\widehat{\bm{U}}$,
we obtain that 
\[
\widehat{\bm{\bm{\Sigma}}}_{\varepsilon}-\bm{\Sigma}_{\varepsilon}=(\frac{1}{T}\bm{E}\bm{E}^{\top}-\mathbb{E}[\frac{1}{T}\bm{E}\bm{E}^{\top}])+(\frac{1}{T}\bm{E}\bm{FB}^{\top}+\frac{1}{T}\bm{BF}^{\top}\bm{E}^{\top})+(\bm{U}\bm{\Lambda}^{2}\bm{U}^{\top}-\widehat{\bm{U}}\widehat{\bm{\Sigma}}^{2}\widehat{\bm{U}}).
\]
We now look at the estimator error $\bigl|(\widehat{\bm{\bm{\Sigma}}}_{\varepsilon}-\bm{\Sigma}_{\varepsilon})_{i,j}\bigl|$
of the $(i,j)$-th entry by analyzing the terms in the above three
brackets.

\textit{Step 1.1 -- Bound for $(\frac{1}{T}\bm{E}\bm{E}^{\top}-\mathbb{E}[\frac{1}{T}\bm{E}\bm{E}^{\top}])$.}

First, we have that 
\[
(\frac{1}{T}\bm{E}\bm{E}^{\top}-\mathbb{E}[\frac{1}{T}\bm{E}\bm{E}^{\top}])_{i,j}=\frac{1}{T}\sum_{t=1}^{T}\bm{E}_{i,t}\bm{E}_{j,t}-\mathbb{E}[\bm{E}_{i,t}\bm{E}_{j,t}]
\]
is the sum of $T$ independent zero-mean sub-Gaussian random variables.
Since $\bm{E}_{i,t}=\sum_{k=1}^{N}(\bm{\Sigma}_{\varepsilon}^{1/2})_{i,k}\bm{Z}_{k,t}$,
we obtain that $\left\Vert \bm{E}_{i,t}\right\Vert _{\psi_{2}}\lesssim(\sum_{k=1}^{N}((\bm{\Sigma}_{\varepsilon}^{1/2})_{i,k})^{2})^{1/2}=\sqrt{(\bm{\Sigma}_{\varepsilon})_{i,i}}$.
By Lemma 2.7.7 in \citet{wainwright2019high}, we obtain that $\bm{E}_{i,t}\bm{E}_{j,t}$
is sub-exponential and $\left\Vert \bm{E}_{i,t}\bm{E}_{j,t}\right\Vert _{\psi_{1}}\leq\left\Vert \bm{E}_{i,t}\right\Vert _{\psi_{2}}\left\Vert \bm{E}_{j,t}\right\Vert _{\psi_{2}}\lesssim\sqrt{(\bm{\Sigma}_{\varepsilon})_{i,i}(\bm{\Sigma}_{\varepsilon})_{j,j}}$,
and the centering $\bm{E}_{i,t}\bm{E}_{j,t}-\mathbb{E}[\bm{E}_{i,t}\bm{E}_{j,t}]$
does not hurt the sub-exponential properties by Exercise 2.7.8 in
\citet{wainwright2019high}. Then we obtain by Theorem 2.8.1 in \citet{wainwright2019high}
that, with probability at least $1-O(n^{-10})$, 
\[
\left|(\frac{1}{T}\bm{E}\bm{E}^{\top}-\mathbb{E}[\frac{1}{T}\bm{E}\bm{E}^{\top}])_{i,j}\right|\lesssim\frac{1}{\sqrt{T}}\sqrt{(\bm{\Sigma}_{\varepsilon})_{i,i}(\bm{\Sigma}_{\varepsilon})_{j,j}}\log n.
\]

\textit{Step 1.2 -- Bound for $(\frac{1}{T}\bm{E}\bm{FB}^{\top}+\frac{1}{T}\bm{BF}^{\top}\bm{E}^{\top})$.}

Next, we have that, conditioning on $\bm{F}$, 
\[
(\frac{1}{T}\bm{E}\bm{FB}^{\top})_{i,j}=\frac{1}{T}(\bm{\Sigma}_{\varepsilon}^{1/2})_{i,\cdot}\bm{Z}\bm{F}\bm{b}_{j}=\frac{1}{T}\sum_{t=1}^{T}\sum_{k=1}^{N}(\bm{\Sigma}_{\varepsilon}^{1/2})_{i,k}\bm{Z}_{k,t}\bm{f}_{t}^{\top}\bm{b}_{j}
\]
is the sum of $NT$ independent zero-mean sub-Gaussian random variables.
We obtain by Theorem 2.6.3 in \citet{wainwright2019high} that, with
probability at least $1-O(n^{-10})$, 
\begin{align*}
\left|(\frac{1}{T}\bm{E}\bm{FB}^{\top})_{i,j}\right| & \lesssim\frac{1}{T}\sqrt{\sum_{t=1}^{T}\sum_{k=1}^{N}((\bm{\Sigma}_{\varepsilon}^{1/2})_{i,k})^{2}(\bm{f}_{t}^{\top}\bm{b}_{j})^{2}}\sqrt{\log n}=\frac{1}{T}\sqrt{(\bm{\Sigma}_{\varepsilon})_{i,i}\bm{b}_{j}^{\top}\bm{F}^{\top}\bm{F}\bm{b}_{j}}\sqrt{\log n}\\
 & =\frac{1}{\sqrt{T}}\sqrt{(\bm{\Sigma}_{\varepsilon})_{i,i}\bm{b}_{j}^{\top}(\frac{1}{T}\bm{F}^{\top}\bm{F}-\bm{I}_{r}+\bm{I}_{r})\bm{b}_{j}}\sqrt{\log n}\\
 & \lesssim\frac{1}{\sqrt{T}}\sqrt{(\bm{\Sigma}_{\varepsilon})_{i,i}\bm{b}_{j}^{\top}\bm{I}_{r}\bm{b}_{j}}\sqrt{\log n}=\frac{1}{\sqrt{T}}\left\Vert \bm{b}_{j}\right\Vert _{2}\sqrt{(\bm{\Sigma}_{\varepsilon})_{i,i}}\sqrt{\log n}\\
 & \lesssim\frac{1}{\sqrt{T}}\sigma_{1}\left\Vert \bar{\bm{U}}_{j,\cdot}\right\Vert _{2}\sqrt{(\bm{\Sigma}_{\varepsilon})_{i,i}}\sqrt{\log n},
\end{align*}
where the last line is owing to the fact that $\left\Vert \bm{b}_{i}\right\Vert _{2}\lesssim\sigma_{1}\big\Vert\bar{\bm{U}}_{i,\cdot}\big\Vert_{2}$
for all $1\leq i\leq N$. Similarly we have $\bigl|(\frac{1}{T}\bm{BF}^{\top}\bm{E}^{\top})_{i,j}\bigl|=\bigl|(\frac{1}{T}\bm{E}\bm{FB}^{\top})_{j,i}\bigl|\lesssim\frac{1}{\sqrt{T}}\sigma_{1}\left\Vert \bar{\bm{U}}_{i,\cdot}\right\Vert _{2}\sqrt{(\bm{\Sigma}_{\varepsilon})_{j,j}}\sqrt{\log n}$.

\textit{Step 1.3 -- Bound for $(\bm{U}\bm{\Lambda}^{2}\bm{U}^{\top}-\widehat{\bm{U}}\widehat{\bm{\Sigma}}^{2}\widehat{\bm{U}})$.}

Finally, we have that $(\bm{U}\bm{\Lambda}^{2}\bm{U}^{\top}-\widehat{\bm{U}}\widehat{\bm{\Sigma}}^{2}\widehat{\bm{U}})_{i,j}=\widetilde{\bm{B}}_{i,\cdot}\widetilde{\bm{B}}_{j,\cdot}^{\top}-\widehat{\bm{B}}_{i,\cdot}\widehat{\bm{B}}_{j,\cdot}^{\top}$.
We obtain by Lemma \ref{Lemma SVD BF good event} that $\bm{B}\bm{R}_{B}=\bm{B}(\bm{J}^{-1})^{\top}\bm{R}_{V}^{\top}=\widetilde{\bm{B}}\bm{R}_{V}^{\top}$.
So, using the error bounds for $\big\Vert(\widehat{\bm{B}}-\bm{B}\bm{R}_{B})_{i,\cdot}\big\Vert_{2}$
in Corollary \ref{corollary:error bound B F}, we obtain that, with
probability at least $1-O(n^{-10})$ 
\begin{align*}
\left|(\bm{U}\bm{\Lambda}^{2}\bm{U}^{\top}-\widehat{\bm{U}}\widehat{\bm{\Sigma}}^{2}\widehat{\bm{U}})_{i,j}\right| & =\left|(\bm{B}\bm{R}_{B})_{i,\cdot}(\bm{B}\bm{R}_{B})_{j,\cdot}^{\top}-\widehat{\bm{B}}_{i,\cdot}\widehat{\bm{B}}_{j,\cdot}^{\top}\right|\\
 & \lesssim\big\Vert(\bm{B}\bm{R}_{B})_{i,\cdot}-\widehat{\bm{B}}_{i,\cdot}\big\Vert_{2}\big\Vert(\bm{B}\bm{R}_{B})_{j,\cdot}\big\Vert_{2}+\big\Vert(\bm{B}\bm{R}_{B})_{j,\cdot}-\widehat{\bm{B}}_{j,\cdot}\big\Vert_{2}\big\Vert\widehat{\bm{B}}_{i,\cdot}\big\Vert_{2}\\
 & \overset{\text{(i)}}{\lesssim}\bigl(\frac{1}{\sqrt{T}}\big\Vert(\bm{\Sigma}_{\varepsilon}^{1/2})_{i,\cdot}\big\Vert_{1}+(\frac{n}{\theta^{2}T}+\frac{1}{\theta\sqrt{T}})\sigma_{r}\big\Vert\bar{\bm{U}}_{i,\cdot}\big\Vert_{2}\bigr)r\log^{2}n\left\Vert \bm{b}_{j}\right\Vert _{2}\\
 & +\bigl(\frac{1}{\sqrt{T}}\big\Vert(\bm{\Sigma}_{\varepsilon}^{1/2})_{j,\cdot}\big\Vert_{1}+(\frac{n}{\theta^{2}T}+\frac{1}{\theta\sqrt{T}})\sigma_{r}\big\Vert\bar{\bm{U}}_{j,\cdot}\big\Vert_{2}\bigr)r\log^{2}n\left\Vert \bm{b}_{i}\right\Vert _{2}
\end{align*}
where (i) is because $\Vert\bm{R}_{B}\Vert_{2}\leq\Vert\bm{J}^{-1}\Vert_{2}\Vert\bm{R}_{V}\Vert_{2}\lesssim1$
by Lemma \ref{Lemma SVD BF good event}, and $\big\Vert\widehat{\bm{B}}_{i,\cdot}\big\Vert_{2}\leq\big\Vert(\bm{B}\bm{R}_{B})_{i,\cdot}\big\Vert_{2}+\big\Vert\widehat{\bm{B}}_{i,\cdot}-(\bm{B}\bm{R}_{B})_{i,\cdot}\big\Vert_{2}\lesssim\big\Vert\bm{B}_{i,\cdot}\big\Vert_{2}$
since $\big\Vert\widehat{\bm{B}}_{i,\cdot}-(\bm{B}\bm{R}_{B})_{i,\cdot}\big\Vert_{2}$
is negligble by Corollary \ref{corollary:error bound B F}.

In summary, using a standard union bound argument, we obtain that,
with probability at least $1-O(n^{-8})$, the following bound holds
simultaneously for all $1\leq i,j\leq N$: 
\begin{align*}
\bigl|(\widehat{\bm{\bm{\Sigma}}}_{\varepsilon}-\bm{\Sigma}_{\varepsilon})_{i,j}\bigl| & \lesssim\frac{1}{\sqrt{T}}\sqrt{(\bm{\Sigma}_{\varepsilon})_{i,i}(\bm{\Sigma}_{\varepsilon})_{j,j}}\log n+\frac{1}{\sqrt{T}}\left\Vert \bm{b}_{i}\right\Vert _{2}\sqrt{(\bm{\Sigma}_{\varepsilon})_{j,j}}\sqrt{\log n}+\frac{1}{\sqrt{T}}\left\Vert \bm{b}_{j}\right\Vert _{2}\sqrt{(\bm{\Sigma}_{\varepsilon})_{i,i}}\sqrt{\log n}\\
 & +\bigl(\frac{1}{\sqrt{T}}(\big\Vert(\bm{\Sigma}_{\varepsilon}^{1/2})_{i,\cdot}\big\Vert_{1}\left\Vert \bm{b}_{j}\right\Vert _{2}+\big\Vert(\bm{\Sigma}_{\varepsilon}^{1/2})_{j,\cdot}\big\Vert_{1}\left\Vert \bm{b}_{i}\right\Vert _{2})+(\frac{n}{\theta^{2}T}+\frac{1}{\theta\sqrt{T}})\sigma_{r}(\big\Vert\bar{\bm{U}}_{i,\cdot}\big\Vert_{2}\left\Vert \bm{b}_{j}\right\Vert _{2}+\big\Vert\bar{\bm{U}}_{j,\cdot}\big\Vert_{2}\left\Vert \bm{b}_{i}\right\Vert _{2})\bigr)r\log^{2}n.
\end{align*}
The desired bound for $\bigl|(\widehat{\bm{\bm{\Sigma}}}_{\varepsilon}-\bm{\Sigma}_{\varepsilon})_{i,j}\bigl|$
follows from simplifying the above bound using $\beta_{i}$, $\gamma_{i}$,
and $\gamma_{\varepsilon,i}$.

\textit{Step 2 -- Estimation error for the generalized thresholding estimator
$\widehat{\bm{\bm{\Sigma}}}_{\varepsilon}^{\tau}$.}

Since the spectral norm is not larger than the maximum $L_{1}$-norm
of rows, we obtain that 
\[
\big\Vert\widehat{\bm{\bm{\Sigma}}}_{\varepsilon}^{\tau}-\bm{\Sigma}_{\varepsilon}\big\Vert_{2}\leq\max_{1\leq i\leq N}\sum_{j=1}^{N}|(\widehat{\bm{\bm{\Sigma}}}_{\varepsilon}^{\tau})_{i,j}-(\bm{\Sigma}_{\varepsilon})_{i,j}|=\max_{1\leq i\leq N}\sum_{j=1}^{N}|h((\widehat{\bm{\bm{\Sigma}}}_{\varepsilon})_{i,j},\tau_{i,j})-(\bm{\Sigma}_{\varepsilon})_{i,j}|.
\]
For the summand $|h((\widehat{\bm{\bm{\Sigma}}}_{\varepsilon})_{i,j},\tau_{i,j})-(\bm{\Sigma}_{\varepsilon})_{i,j}|$,
we have that 
\begin{align*}
|h((\widehat{\bm{\bm{\Sigma}}}_{\varepsilon})_{i,j},\tau_{i,j})-(\bm{\Sigma}_{\varepsilon})_{i,j}| & =|h((\widehat{\bm{\bm{\Sigma}}}_{\varepsilon})_{i,j},\tau_{i,j})1_{\{|(\widehat{\bm{\bm{\Sigma}}}_{\varepsilon})_{i,j}|\geq\tau_{i,j}\}}-(\bm{\Sigma}_{\varepsilon})_{i,j}1_{\{|(\widehat{\bm{\bm{\Sigma}}}_{\varepsilon})_{i,j}|\geq\tau_{i,j}\}}-(\bm{\Sigma}_{\varepsilon})_{i,j}1_{\{|(\widehat{\bm{\bm{\Sigma}}}_{\varepsilon})_{i,j}|<\tau_{i,j}\}}|\\
 & =|h((\widehat{\bm{\bm{\Sigma}}}_{\varepsilon})_{i,j},\tau_{i,j})-(\bm{\Sigma}_{\varepsilon})_{i,j}|1_{\{|(\widehat{\bm{\bm{\Sigma}}}_{\varepsilon})_{i,j}|\geq\tau_{i,j}\}}+|(\bm{\Sigma}_{\varepsilon})_{i,j}|1_{\{|(\widehat{\bm{\bm{\Sigma}}}_{\varepsilon})_{i,j}|<\tau_{i,j}\}}.
\end{align*}
Then we derive the upper bounds for the above two terms.

For the first term, without loss of generality, we take for instance
the hard thresholding function $h(z,\tau)=z1_{\{|z|\geq\tau\}}$,
and the proof for other thresholding functions follows from the similar
ways. We have that 
\begin{align*}
|h((\widehat{\bm{\bm{\Sigma}}}_{\varepsilon})_{i,j},\tau_{i,j})-(\bm{\Sigma}_{\varepsilon})_{i,j}|1_{\{|(\widehat{\bm{\bm{\Sigma}}}_{\varepsilon})_{i,j}|\geq\tau_{i,j}\}} & =|(\widehat{\bm{\bm{\Sigma}}}_{\varepsilon})_{i,j}-(\bm{\Sigma}_{\varepsilon})_{i,j}|1_{\{|(\widehat{\bm{\bm{\Sigma}}}_{\varepsilon})_{i,j}|\geq\tau_{i,j}\}}\\
 & \leq|(\widehat{\bm{\bm{\Sigma}}}_{\varepsilon})_{i,j}-(\bm{\Sigma}_{\varepsilon})_{i,j}|\lesssim\epsilon_{N,T}\sqrt{(\bm{\Sigma}_{\varepsilon})_{i,i}(\bm{\Sigma}_{\varepsilon})_{j,j}}\\
 & \lesssim(\epsilon_{N,T}\sqrt{(\bm{\Sigma}_{\varepsilon})_{i,i}(\bm{\Sigma}_{\varepsilon})_{j,j}})^{1-q}|(\bm{\Sigma}_{\varepsilon})_{i,j}|^{q},
\end{align*}
where the last second inequality is because $\epsilon_{N,T}\sqrt{(\bm{\Sigma}_{\varepsilon})_{i,i}(\bm{\Sigma}_{\varepsilon})_{j,j}}\lesssim|(\bm{\Sigma}_{\varepsilon})_{i,j}|$.

Next, for the second term, we have that 
\begin{align*}
|(\bm{\Sigma}_{\varepsilon})_{i,j}|1_{\{|(\widehat{\bm{\bm{\Sigma}}}_{\varepsilon})_{i,j}|<\tau_{i,j}\}} & =|(\bm{\Sigma}_{\varepsilon})_{i,j}|^{q}\cdot(|(\bm{\Sigma}_{\varepsilon})_{i,j}|)^{1-q}\cdot1_{\{|(\widehat{\bm{\bm{\Sigma}}}_{\varepsilon})_{i,j}|<\tau_{i,j}\}}\\
 & \lesssim|(\bm{\Sigma}_{\varepsilon})_{i,j}|^{q}\cdot(|(\widehat{\bm{\bm{\Sigma}}}_{\varepsilon})_{i,j}|+\epsilon_{N,T}\sqrt{(\bm{\Sigma}_{\varepsilon})_{i,i}(\bm{\Sigma}_{\varepsilon})_{j,j}})^{1-q}\cdot1_{\{|(\widehat{\bm{\bm{\Sigma}}}_{\varepsilon})_{i,j}|<\tau_{i,j}\}}\\
 & \leq|(\bm{\Sigma}_{\varepsilon})_{i,j}|^{q}\cdot(\tau_{i,j}+\epsilon_{N,T}\sqrt{(\bm{\Sigma}_{\varepsilon})_{i,i}(\bm{\Sigma}_{\varepsilon})_{j,j}})^{1-q}.
\end{align*}
For $\tau_{i,j}$, since $\max_{1\leq i,j\leq N}\bigl|(\widehat{\bm{\bm{\Sigma}}}_{\varepsilon}-\bm{\Sigma}_{\varepsilon})_{i,j}\bigl|\lesssim\epsilon_{N,T}\sqrt{(\bm{\Sigma}_{\varepsilon})_{i,i}(\bm{\Sigma}_{\varepsilon})_{j,j}}$
and $\epsilon_{N,T}\ll1$, we obtain that $\max_{1\leq i\leq N}\bigl|(\widehat{\bm{\bm{\Sigma}}}_{\varepsilon})_{i,i}-(\bm{\Sigma}_{\varepsilon})_{i,i}\bigl|\lesssim\epsilon_{N,T}(\bm{\Sigma}_{\varepsilon})_{i,i}$
and thus $(\widehat{\bm{\bm{\Sigma}}}_{\varepsilon})_{i,i}\asymp(\bm{\Sigma}_{\varepsilon})_{i,i}$,
implying $\tau_{i,j}=C\epsilon_{N,T}\sqrt{(\widehat{\bm{\bm{\Sigma}}}_{\varepsilon})_{i,i}(\widehat{\bm{\bm{\Sigma}}}_{\varepsilon})_{j,j}}\lesssim\epsilon_{N,T}\sqrt{(\bm{\Sigma}_{\varepsilon})_{i,i}(\bm{\Sigma}_{\varepsilon})_{j,j}}$.
So, we obtain that 
\begin{align*}
|(\bm{\Sigma}_{\varepsilon})_{i,j}|1_{\{|(\widehat{\bm{\bm{\Sigma}}}_{\varepsilon})_{i,j}|<\tau_{i,j}\}} & \lesssim|(\bm{\Sigma}_{\varepsilon})_{i,j}|^{q}(\epsilon_{N,T}\sqrt{(\bm{\Sigma}_{\varepsilon})_{i,i}(\bm{\Sigma}_{\varepsilon})_{j,j}}+\epsilon_{N,T}\sqrt{(\bm{\Sigma}_{\varepsilon})_{i,i}(\bm{\Sigma}_{\varepsilon})_{j,j}})^{1-q}\\
 & \lesssim(\epsilon_{N,T})^{1-q}((\bm{\Sigma}_{\varepsilon})_{i,i}(\bm{\Sigma}_{\varepsilon})_{j,j})^{(1-q)/2}|(\bm{\Sigma}_{\varepsilon})_{i,j}|^{q}.
\end{align*}

Combining the bounds for the above two terms, we conclude that, using
the assumption on the sparsity of $\bm{\Sigma}_{\varepsilon}$, 
\begin{align*}
\big\Vert\widehat{\bm{\bm{\Sigma}}}_{\varepsilon}^{\tau}-\bm{\Sigma}_{\varepsilon}\big\Vert_{2} & \leq\max_{1\leq i\leq N}\sum_{j=1}^{N}|(\widehat{\bm{\bm{\Sigma}}}_{\varepsilon}^{\tau})_{i,j}-(\bm{\Sigma}_{\varepsilon})_{i,j}|\\
 & \lesssim(\epsilon_{N,T})^{1-q}\max_{1\leq i\leq N}\sum_{j=1}^{N}((\bm{\Sigma}_{\varepsilon})_{i,i}(\bm{\Sigma}_{\varepsilon})_{j,j})^{(1-q)/2}|(\bm{\Sigma}_{\varepsilon})_{i,j}|^{q}\\
 & \leq(\epsilon_{N,T})^{1-q}s(\bm{\Sigma}_{\varepsilon}).
\end{align*}

%% file: appendix_beta_structure_test.tex
\section{Proof of Theorem \ref{Thm beta structure test}: Test for structural
breaks in betas}

\subsection{Some useful lemmas}

To prove Theorem \ref{Thm factor test plug-in Chi-sq}, we collect
some useful lemmas as preparations. Recall that, we already showed
in Lemma \ref{Lemma V subset good event} for some properties of $\bm{\bm{V}}_{S,\cdot}$
for any subset $S$. Using these properties of $\bm{\bm{\bm{V}}}_{S,\cdot}$,
we are able to establish the following lemma.

\begin{lemma}\label{Lemma beta_hat 1st-order approx}Suppose that
$r+\log n\ll\min(T_{1},T_{2})$ and the assumptions in Theorem \ref{Thm UV 1st approx row-wise error}
hold. Then we have that, with probability at least $1-O(n^{-2})$:

(i) It holds $\sigma_{i}(\widehat{\bm{V}}^{j})\asymp\sigma_{i}(\bm{V}^{j})\asymp(T_{j}/T)^{1/2}$,
$\sigma_{i}((\widehat{\bm{V}}^{j})^{+})\asymp\sigma_{i}((\bm{V}^{j})^{+})\asymp(T/T_{j})^{1/2}$,
\[
\Vert(\widehat{\bm{V}}^{j})^{+}-(\bm{V}^{j}\bm{R}_{V}^{\top})^{+}\Vert_{2}\lesssim(T/T_{j})^{1/2}\frac{1}{\theta}\sqrt{r}\log n,
\]
and
\[
\big\Vert\prod_{j=1}^{2}(\widehat{\bm{V}}^{j})^{\top}\widehat{\bm{V}}^{j}-\prod_{j=1}^{2}(\bm{V}^{j}\bm{R}_{V}^{\top})^{\top}\bm{V}^{j}\bm{R}_{V}^{\top}\big\Vert_{2}\lesssim(T_{1}T_{2}/T^{2})\frac{1}{\theta}\sqrt{r}\log n.
\]

(ii) Under the null hypothesis $H_{0}:\ \bm{b}_{i}^{1}=\bm{b}_{i}^{2}$,
it holds
\[
\big\Vert\widehat{\bm{b}}_{i}^{1}-\widehat{\bm{b}}_{i}^{2}\big\Vert_{2}\lesssim\frac{1}{\theta}\sqrt{r}\log n\Vert\bm{b}_{i}^{j}\Vert_{2}+\frac{1}{\theta}\sqrt{r}\log^{3/2}n\Vert(\bm{\Sigma}_{\varepsilon}^{1/2})_{i,\cdot}\Vert_{2}+\sum_{j=1}^{2}\frac{1}{\sqrt{T_{j}}}\Vert(\bm{\Sigma}_{\varepsilon}^{1/2})_{i,\cdot}\Vert_{2}\sqrt{r\log n}
\]
and $\widehat{\bm{b}}_{i}^{1}-\widehat{\bm{b}}_{i}^{2}$ admits the
following first-order approximation
\[
\widehat{\bm{b}}_{i}^{1}-\widehat{\bm{b}}_{i}^{2}=\bm{G}_{b_{i}}+\bm{\Upsilon}_{b_{i}},\text{\qquad with\qquad}\bm{G}_{b_{i}}:=\frac{1}{\sqrt{T}}[\sum_{j=1}^{2}(-1)^{j-1}(\bm{V}^{j}\bm{R}_{V}^{\top})^{+}(\bm{Z}^{j})^{\top}]\bm{\Sigma}_{\varepsilon}^{1/2}[\bm{I}_{N}+\bm{U}\bm{U}^{\top}]_{\cdot,i},
\]
where the term $\bm{\Upsilon}_{b_{i}}$ satisfies that
\[
\big\Vert\bm{\Upsilon}_{b_{i}}\big\Vert_{2}\lesssim\frac{1}{\theta}\sqrt{r}\log^{3/2}n\cdot\sqrt{\varphi_{i}}+(\frac{n}{\theta^{2}T}+\frac{1}{\theta\sqrt{T}})\sqrt{r}\log^{3/2}n\cdot\sigma_{r}\Vert\bm{U}_{i,\cdot}\Vert_{2},
\]
with $\varphi_{i}:=[\bm{\Sigma}_{\varepsilon}+\bm{U}\bm{U}^{\top}\bm{\Sigma}_{\varepsilon}+\bm{\Sigma}_{\varepsilon}\bm{U}\bm{U}^{\top}+\bm{U}\bm{U}^{\top}\bm{\Sigma}_{\varepsilon}\bm{U}\bm{U}^{\top}]_{i,i}$.

\end{lemma}
\begin{proof}
For the first-order approximation of $\widehat{\bm{b}}_{i}^{1}-\widehat{\bm{b}}_{i}^{2}$,
we start by rewriting the regression coefficients $\widehat{\bm{b}}_{i}^{j}$
as $\widehat{\bm{b}}_{i}^{j}=T^{-1/2}(\widehat{\bm{V}}^{j})^{+}(\bm{X}_{i,\cdot}^{j})^{\top}$,
where $(\widehat{\bm{V}}^{j})^{+}=((\widehat{\bm{V}}^{j})^{\top}\widehat{\bm{V}}^{j})^{-1}(\widehat{\bm{V}}^{j})^{\top}$
is the pseudo inverse of $\widehat{\bm{F}}^{j}$. We obtain by our
master theorem that $\widehat{\bm{V}}-\bm{V}\bm{R}_{V}^{\top}=\bm{G}_{V}\bm{R}_{V}^{\top}+\bm{\Psi}_{V}\bm{R}_{V}^{\top}$.
We partition the matrices as $\widehat{\bm{V}}=((\widehat{\bm{V}}^{1})^{\top},(\widehat{\bm{V}}^{2})^{\top})^{\top}$,
$\bm{V}=((\bm{V}^{1})^{\top},(\bm{V}^{2})^{\top})^{\top}$, $\bm{E}=(\bm{E}^{1},\bm{E}^{2})=(\bm{\Sigma}_{\varepsilon}^{1/2}\bm{Z}^{1},\bm{\Sigma}_{\varepsilon}^{1/2}\bm{Z}^{2})$,
where $\bm{V}^{j},\widehat{\bm{V}}^{j}\in\mathbb{R}^{T_{j}\times r}$,
$\bm{E}^{j},\bm{Z}^{j}\in\mathbb{R}^{N\times T_{j}}$, and denote
$\bm{D}^{j}:=\widehat{\bm{V}}^{j}-\bm{V}^{j}\bm{R}_{V}^{\top}$.

\textit{Proof of (i).}
Similar to Step 1 in proof of Lemma \ref{Lemma beta_hat 1st-order approx},
we have that
\[
\big\Vert\widehat{\bm{V}}^{j}-\bm{V}^{j}\bm{R}_{V}^{\top}\big\Vert_{2}\lesssim\sqrt{T_{j}}\rho\sqrt{\frac{r}{n}}\log n=\sqrt{T_{j}/T}\frac{1}{\theta}\sqrt{r}\log n.
\]
Then using the above bounds we can prove that, with probability at
least $1-O(n^{-2})$, the ranks of $\widehat{\bm{V}}^{j}$ and $\bm{V}^{j}$
are both equal to $r$. Similar to the proof of Lemma \ref{Lemma V subset good event},
we have that $\sigma_{i}(\widehat{\bm{V}}^{j})\asymp\sigma_{i}(\bm{V}^{j})\asymp(T_{j}/T)^{1/2}$.
We obtain by Theorem 3.4 in \citet{stewart1977perturbation} that
\begin{align*}
\big\Vert(\widehat{\bm{V}}^{j})^{+}-(\bm{V}^{j}\bm{R}_{V}^{\top})^{+}\big\Vert_{2} & \lesssim\Vert(\widehat{\bm{V}}^{j})^{+}\Vert_{2}\Vert(\bm{V}^{j}\bm{R}_{V}^{\top})^{+}\Vert_{2}\Vert\widehat{\bm{V}}^{j}-\bm{V}^{j}\bm{R}_{V}^{\top}\Vert_{2}\\
 & \lesssim(\Vert(\widehat{\bm{V}}^{j})^{+}-(\bm{V}^{j}\bm{R}_{V}^{\top})^{+}\Vert_{2}+\Vert(\bm{V}^{j}\bm{R}_{V}^{\top})^{+}\Vert_{2})\Vert(\bm{V}^{j}\bm{R}_{V}^{\top})^{+}\Vert_{2}\Vert\widehat{\bm{V}}^{j}-\bm{V}^{j}\bm{R}_{V}^{\top}\Vert_{2}\\
 & \lesssim\Vert(\widehat{\bm{V}}^{j})^{+}-(\bm{V}^{j}\bm{R}_{V}^{\top})^{+}\Vert_{2}\times\Vert(\bm{V}^{j}\bm{R}_{V}^{\top})^{+}\Vert_{2}\Vert\widehat{\bm{V}}^{j}-\bm{V}^{j}\bm{R}_{V}^{\top}\Vert_{2}\\
 & +\Vert(\bm{V}^{j}\bm{R}_{V}^{\top})^{+}\Vert_{2}^{2}\Vert\widehat{\bm{V}}^{j}-\bm{V}^{j}\bm{R}_{V}^{\top}\Vert_{2}.
\end{align*}
Note that the coefficient for $\Vert(\widehat{\bm{V}}^{j})^{+}-(\bm{V}^{j}\bm{R}_{V}^{\top})^{+}\Vert_{2}$
in the right hand side of the above inequality satisfies that $\Vert(\bm{V}^{j}\bm{R}_{V}^{\top})^{+}\Vert_{2}\Vert\widehat{\bm{V}}^{j}-\bm{V}^{j}\bm{R}_{V}^{\top}\Vert_{2}\lesssim(T/T_{j})^{1/2}\sqrt{T_{j}}\rho\sqrt{\frac{r}{n}}\log n\ll1$.
Thus, by the self-bounding method, we obtain that
\[
\big\Vert(\widehat{\bm{V}}^{j})^{+}-(\bm{V}^{j}\bm{R}_{V}^{\top})^{+}\big\Vert_{2}\lesssim\Vert(\bm{V}^{j}\bm{R}_{V}^{\top})^{+}\Vert_{2}^{2}\Vert\widehat{\bm{V}}^{j}-\bm{V}^{j}\bm{R}_{V}^{\top}\Vert_{2}\lesssim\frac{T}{T_{j}}\sqrt{T_{j}}\rho\sqrt{\frac{r}{n}}\log n\lesssim(T/T_{j})^{1/2}\frac{1}{\theta}\sqrt{r}\log n.
\]
Then we obtain that $\sigma_{i}((\widehat{\bm{V}}^{j})^{+})\asymp\sigma_{i}((\bm{V}^{j})^{+})\asymp(T/T_{j})^{1/2}$.

We have that
\begin{align*}
\Vert(\widehat{\bm{V}}^{j})^{\top}\widehat{\bm{V}}^{j}-(\bm{V}^{j}\bm{R}_{V}^{\top})^{\top}\bm{V}^{j}\bm{R}_{V}^{\top}\Vert_{2} & \lesssim(\Vert\widehat{\bm{V}}^{j}\Vert_{2}+\Vert\bm{V}^{j}\Vert_{2})\Vert\widehat{\bm{V}}^{j}-\bm{V}^{j}\bm{R}_{V}^{\top}\Vert_{2}\\
 & \lesssim(T_{j}/T)^{1/2}\sqrt{T_{j}}\rho\sqrt{\frac{r}{n}}\log n\lesssim T_{j}\rho\sqrt{\frac{r}{nT}}\log n=(T_{j}/T)\frac{1}{\theta}\sqrt{r}\log n.
\end{align*}
Then, we obtain that
\begin{align*}
\big\Vert\prod_{j=1}^{2}(\widehat{\bm{V}}^{j})^{\top}\widehat{\bm{V}}^{j}-\prod_{j=1}^{2}(\bm{V}^{j}\bm{R}_{V}^{\top})^{\top}\bm{V}^{j}\bm{R}_{V}^{\top}\big\Vert_{2} & \lesssim((T_{1}/T)^{1/2})^{2}T_{2}\rho\sqrt{\frac{r}{nT}}\log n+((T_{2}/T)^{1/2})^{2}T_{1}\rho\sqrt{\frac{r}{nT}}\log n\\
 & \lesssim\frac{1}{T}T_{1}T_{2}\rho\sqrt{\frac{r}{nT}}\log n\lesssim(T_{1}T_{2}/T^{2})\frac{1}{\theta}\sqrt{r}\log n.
\end{align*}

\textit{Proof of (ii).}
For the upper bound of $\widehat{\bm{b}}_{i}^{1}-\widehat{\bm{b}}_{i}^{2}$,
we have that
\begin{align*}
\widehat{\bm{b}}_{i}^{j} & =T^{-1/2}(\widehat{\bm{V}}^{j})^{+}(\bm{X}_{i,\cdot}^{j})^{\top}=T^{-1/2}(\widehat{\bm{V}}^{j})^{+}(\bm{F}^{j}\bm{b}_{i}^{j}+(\bm{E}_{i,\cdot}^{j})^{\top})\\
 & =(\widehat{\bm{V}}^{j})^{+}\bm{V}^{j}\bm{R}_{V}^{\top}(\bm{R}_{V}\bm{J}^{-1}\bm{b}_{i}^{j})+T^{-1/2}(\widehat{\bm{V}}^{j})^{+}(\bm{E}_{i,\cdot}^{j})^{\top}\\
 & =[\bm{I}_{r}-(\widehat{\bm{V}}^{j})^{+}(\widehat{\bm{V}}^{j}-\bm{V}^{j}\bm{R}_{V}^{\top})](\bm{R}_{V}\bm{J}^{-1}\bm{b}_{i}^{j})+T^{-1/2}(\widehat{\bm{V}}^{j})^{+}(\bm{E}_{i,\cdot}^{j})^{\top}.
\end{align*}
So we obtain that, under the null hypothesis that $\bm{b}_{i}^{1}=\bm{b}_{i}^{2}$,
it holds
\begin{align*}
\widehat{\bm{b}}_{i}^{1}-\widehat{\bm{b}}_{i}^{2} & =\sum_{j=1}^{2}(-1)^{j}[(\widehat{\bm{V}}^{j})^{+}(\widehat{\bm{V}}^{j}-\bm{V}^{j}\bm{R}_{V}^{\top})(\bm{R}_{V}\bm{J}^{-1}\bm{b}_{i}^{j})-T^{-1/2}(\widehat{\bm{V}}^{j})^{+}(\bm{E}_{i,\cdot}^{j})^{\top}]\\
 & =\sum_{j=1}^{2}(-1)^{j}[(\widehat{\bm{V}}^{j})^{+}(\widehat{\bm{V}}^{j}-\bm{V}^{j}\bm{R}_{V}^{\top})(\bm{R}_{V}\bm{J}^{-1}\bm{b}_{i}^{j})]\\
 & +\sum_{j=1}^{2}(-1)^{j-1}T^{-1/2}\{[(\widehat{\bm{V}}^{j})^{+}-(\bm{V}^{j}\bm{R}_{V}^{\top})^{+}](\bm{Z}^{j})^{\top}(\bm{\Sigma}_{\varepsilon}^{1/2})_{\cdot,i}+(\bm{V}^{j}\bm{R}_{V}^{\top})^{+}(\bm{Z}^{j})^{\top}(\bm{\Sigma}_{\varepsilon}^{1/2})_{\cdot,i}\}.
\end{align*}
Then we obtain by (G.3) in Lemma 19 of \citet{yan2024entrywise} that
\begin{align*}
\big\Vert\widehat{\bm{b}}_{i}^{1}-\widehat{\bm{b}}_{i}^{2}\big\Vert_{2} & \lesssim\sum_{j=1}^{2}(T/T_{j})^{1/2}\sqrt{T_{j}/T}\frac{1}{\theta}\sqrt{r}\log n\Vert\bm{b}_{i}^{j}\Vert_{2}\\
 & +\sum_{j=1}^{2}T^{-1/2}\cdot(T/T_{j})^{1/2}\frac{1}{\theta}\sqrt{r}\log n\cdot\Vert(\bm{\Sigma}_{\varepsilon}^{1/2})_{i,\cdot}\Vert_{2}\sqrt{T_{j}\log n}+T^{-1/2}\cdot\sqrt{T/T_{j}}\Vert(\bm{\Sigma}_{\varepsilon}^{1/2})_{i,\cdot}\Vert_{2}\sqrt{r\log n}\\
 & \lesssim\frac{1}{\theta}\sqrt{r}\log n\Vert\bm{b}_{i}^{j}\Vert_{2}+\frac{1}{\theta}\sqrt{r}\log^{3/2}n\Vert(\bm{\Sigma}_{\varepsilon}^{1/2})_{i,\cdot}\Vert_{2}+\sum_{j=1}^{2}(1/T_{j})^{1/2}\Vert(\bm{\Sigma}_{\varepsilon}^{1/2})_{i,\cdot}\Vert_{2}\sqrt{r\log n}.
\end{align*}

Then, we establish the first-order approximation for $\widehat{\bm{b}}_{i}^{1}-\widehat{\bm{b}}_{i}^{2}$.

\textit{Step 1 -- compute the first-order approximation.}

We obtain by Theorem 3.2 in \citet{stewart1977perturbation} that
\[
(\widehat{\bm{V}}^{j})^{+}-(\bm{V}^{j}\bm{R}_{V}^{\top})^{+}=(\widehat{\bm{V}}^{j})^{+}\bm{D}^{j}(\bm{V}^{j}\bm{R}_{V}^{\top})^{+}+((\widehat{\bm{V}}^{j})^{\top}\widehat{\bm{V}}^{j})^{-1}(\bm{D}^{j})^{\top}(\bm{I}_{T_{j}}-\bm{P}_{\bm{V}^{j}}),
\]
for $j=1,2$. Then, we obtain that
\begin{align*}
\widehat{\bm{b}}_{i}^{j}-T^{-1/2}(\bm{V}^{j}\bm{R}_{V}^{\top})^{+}(\bm{X}_{i,\cdot}^{j})^{\top} & =T^{-1/2}((\widehat{\bm{V}}^{j})^{+}-(\bm{V}^{j}\bm{R}_{V}^{\top})^{+})(\bm{X}_{i,\cdot}^{j})^{\top}\\
 & =T^{-1/2}((\widehat{\bm{V}}^{j})^{+}-(\bm{V}^{j}\bm{R}_{V}^{\top})^{+})\bm{F}^{j}\bm{b}_{i}^{j}+T^{-1/2}((\widehat{\bm{V}}^{j})^{+}-(\bm{V}^{j}\bm{R}_{V}^{\top})^{+})(\bm{E}_{i,\cdot}^{j})^{\top}\\
 & =T^{-1/2}(\widehat{\bm{V}}^{j})^{+}\bm{D}^{j}(\bm{V}^{j}\bm{R}_{V}^{\top})^{+}\bm{F}^{j}\bm{b}_{i}^{j}+T^{-1/2}((\widehat{\bm{V}}^{j})^{+}-(\bm{V}^{j}\bm{R}_{V}^{\top})^{+})(\bm{E}_{i,\cdot}^{j})^{\top}\\
 & =(\widehat{\bm{V}}^{j})^{+}\bm{D}^{j}\bm{R}_{V}\bm{J}^{-1}\bm{b}_{i}^{j}+T^{-1/2}((\widehat{\bm{V}}^{j})^{+}-(\bm{V}^{j}\bm{R}_{V}^{\top})^{+})(\bm{E}_{i,\cdot}^{j})^{\top}.
\end{align*}
Note that the term in the left hand side can be decomposed into
\[
T^{-1/2}(\bm{V}^{j}\bm{R}_{V}^{\top})^{+}(\bm{X}_{i,\cdot}^{j})^{\top}=T^{-1/2}(\bm{V}^{j}\bm{R}_{V}^{\top})^{+}(\bm{F}^{j}\bm{b}_{i}^{j}+(\bm{E}_{i,\cdot}^{j})^{\top})=\bm{R}_{V}\bm{J}^{-1}\bm{b}_{i}^{j}+T^{-1/2}(\bm{V}^{j}\bm{R}_{V}^{\top})^{+}(\bm{E}_{i,\cdot}^{j})^{\top},
\]
and the first term in the right hand side can be decomposed into
\begin{align*}
(\widehat{\bm{V}}^{j})^{+}\bm{D}^{j}\bm{R}_{V}\bm{J}^{-1}\bm{b}_{i}^{j} & =(\widehat{\bm{V}}^{j})^{+}\bm{D}^{j}\bm{R}_{V}\bm{J}^{-1}\bm{b}_{i}^{j}=(\widehat{\bm{V}}^{j})^{+}(\bm{G}_{V}^{j}\bm{R}_{V}^{\top}+\bm{\Psi}_{V}^{j}\bm{R}_{V}^{\top})\bm{R}_{V}\bm{J}^{-1}\bm{b}_{i}^{j}\\
 & =(\widehat{\bm{V}}^{j})^{+}\bm{G}_{V}^{j}\bm{J}^{-1}\bm{b}_{i}^{j}+(\widehat{\bm{V}}^{j})^{+}\bm{\Psi}_{V}^{j}\bm{J}^{-1}\bm{b}_{i}^{j}\\
 & =(\bm{V}^{j}\bm{R}_{V}^{\top})^{+}\bm{G}_{V}^{j}\bm{J}^{-1}\bm{b}_{i}^{j}+[((\widehat{\bm{V}}^{j})^{+}-(\bm{V}^{j}\bm{R}_{V}^{\top})^{+})\bm{G}_{V}^{j}+(\widehat{\bm{V}}^{j})^{+}\bm{\Psi}_{V}^{j}]\bm{J}^{-1}\bm{b}_{i}^{j}.
\end{align*}
So, we obtain by $\bm{G}_{V}^{j}\bm{J}^{-1}\bm{b}_{i}^{j}=T^{-1/2}(\bm{E}^{j})^{\top}\bm{U}(\bm{U}_{i,\cdot})^{\top}$
that
\begin{align*}
\widehat{\bm{b}}_{i}^{j}-\bm{R}_{V}\bm{J}^{-1}\bm{b}_{i}^{j} & =T^{-1/2}(\bm{V}^{j}\bm{R}_{V}^{\top})^{+}(\bm{E}_{i,\cdot}^{j})^{\top}+(\bm{V}^{j}\bm{R}_{V}^{\top})^{+}\bm{G}_{V}^{j}\bm{J}^{-1}\bm{b}_{i}^{j}+\bm{r}_{0}^{j}\\
 & =T^{-1/2}(\bm{V}^{j}\bm{R}_{V}^{\top})^{+}[(\bm{E}_{i,\cdot}^{j})^{\top}+(\bm{E}^{j})^{\top}\bm{U}(\bm{U}_{i,\cdot})^{\top}]+\bm{r}_{0}^{j}\\
 & =T^{-1/2}(\bm{V}^{j}\bm{R}_{V}^{\top})^{+}(\bm{E}^{j})^{\top}[\bm{I}_{N}+\bm{U}\bm{U}^{\top}]_{\cdot,i}+\bm{r}_{0}^{j},
\end{align*}
where the remainder term $\bm{r}_{0}^{j}$ is given by
\begin{align*}
\bm{r}_{0}^{j} & =T^{-1/2}((\widehat{\bm{V}}^{j})^{+}-(\bm{V}^{j}\bm{R}_{V}^{\top})^{+})(\bm{E}_{i,\cdot}^{j})^{\top}+[((\widehat{\bm{V}}^{j})^{+}-(\bm{V}^{j}\bm{R}_{V}^{\top})^{+})\bm{G}_{V}^{j}+(\widehat{\bm{V}}^{j})^{+}\bm{\Psi}_{V}^{j}]\bm{J}^{-1}\bm{b}_{i}^{j}\\
 & =T^{-1/2}((\widehat{\bm{V}}^{j})^{+}-(\bm{V}^{j}\bm{R}_{V}^{\top})^{+})(\bm{E}^{j})^{\top}[\bm{I}_{N}+\bm{U}\bm{U}^{\top}]_{\cdot,i}+(\widehat{\bm{V}}^{j})^{+}\bm{\Psi}_{V}^{j}\bm{J}^{-1}\bm{b}_{i}^{j}.
\end{align*}

Under the null hypothesis $H_{0}:\ \bm{b}_{i}^{1}=\bm{b}_{i}^{2}$,
we have that $\widehat{\bm{b}}_{i}^{1}-\widehat{\bm{b}}_{i}^{2}=\bm{G}_{b_{i}}+\bm{\Upsilon}_{b_{i}}$
where the first-order term $\bm{G}_{b_{i}}$ is given by
\begin{align*}
\bm{G}_{b_{i}} & =T^{-1/2}[(\bm{V}^{1}\bm{R}_{V}^{\top})^{+}(\bm{E}^{1})^{\top}-(\bm{V}^{2}\bm{R}_{V}^{\top})^{+}(\bm{E}^{2})^{\top}][\bm{I}_{N}+\bm{U}\bm{U}^{\top}]_{\cdot,i}\\
 & =T^{-1/2}[(\bm{V}^{1}\bm{R}_{V}^{\top})^{+}(\bm{Z}^{1})^{\top}-(\bm{V}^{2}\bm{R}_{V}^{\top})^{+}(\bm{Z}^{2})^{\top}]\bm{\Sigma}_{\varepsilon}^{1/2}[\bm{I}_{N}+\bm{U}\bm{U}^{\top}]_{\cdot,i},
\end{align*}
and the remainder term $\bm{\Upsilon}_{b_{i}}$ is given by $\bm{\Upsilon}_{b_{i}}=\bm{r}_{0}^{1}-\bm{r}_{0}^{2}$.

\textit{Step 2 -- derive the upper bound for $\bm{\Upsilon}_{b_{i}}$.}

Then we have that, by the expression of $\bm{\Psi}_{V}^{j}$ and (G.3)
in Lemma 19 of \citet{yan2024entrywise}
\begin{align*}
\Vert\bm{r}_{0}^{j}\Vert_{2} & \lesssim T^{-1/2}\big\Vert(\widehat{\bm{V}}^{j})^{+}-(\bm{V}^{j}\bm{R}_{V}^{\top})^{+}\big\Vert_{2}\Vert(\bm{Z}^{j})^{\top}\bm{\Sigma}_{\varepsilon}^{1/2}[\bm{I}_{N}+\bm{U}\bm{U}^{\top}]_{\cdot,i}\Vert_{2}\\
 & +\big\Vert(\widehat{\bm{V}}^{j})^{+}\big\Vert_{2}\cdot\sqrt{T_{j}}\left(\frac{\sqrt{n}}{\theta^{2}T}\sqrt{r}\log^{3/2}n+(\frac{n}{\theta^{2}T}+\frac{1}{\theta\sqrt{T}})\sqrt{r}\log n\sqrt{\frac{\log n}{T}}\right)\cdot\sigma_{r}\Vert\bm{U}_{i,\cdot}\Vert_{2}\\
 & \lesssim T^{-1/2}\cdot(T/T_{j})^{1/2}\frac{1}{\theta}\sqrt{r}\log n\cdot\sqrt{T_{j}\log n}\sqrt{\varphi_{i}}+\sqrt{T}(\frac{\sqrt{n}}{\theta^{2}T}+(\frac{n}{\theta^{2}T}+\frac{1}{\theta\sqrt{T}})\frac{1}{\sqrt{T}})\sqrt{r}\log^{3/2}n\cdot\sigma_{r}\Vert\bm{U}_{i,\cdot}\Vert_{2}\\
 & \lesssim\frac{1}{\theta}\sqrt{r}\log^{3/2}n\cdot\sqrt{\varphi_{i}}+(\frac{n}{\theta^{2}T}+\frac{1}{\theta\sqrt{T}})\sqrt{r}\log^{3/2}n\cdot\sigma_{r}\Vert\bm{U}_{i,\cdot}\Vert_{2}.
\end{align*}
Then, the upper bound for $\big\Vert\bm{\Upsilon}_{b_{i}}\big\Vert_{2}$
follows from $\big\Vert\bm{\Upsilon}_{b_{i}}\big\Vert_{2}\leq\Vert\bm{r}_{0}^{1}\Vert_{2}+\Vert\bm{r}_{0}^{2}\Vert_{2}$.
\end{proof}
Then, we establish Gaussian approximation for the first-order term.
For simplicity of notations, we denote $\chi^{2}(n)$ as a Chi-square
random variable with degree of freedom equal to $n$, and denote by
$\chi_{\phi}^{2}(n)$ its $\phi$-quantile.

\begin{lemma} \label{Lemma beta true Chi-square stat}Suppose that
the assumptions in Lemma \ref{Lemma beta_hat 1st-order approx} hold.
Then, under the null hypothesis $H_{0}:\ \bm{b}_{i}^{1}=\bm{b}_{i}^{2}$,
we have that, for any random variable $\zeta$ satisfying $\left|\zeta\right|\ll r$,
it holds
\begin{equation}
\left|\mathbb{P}\left(\mathscr{B}_{i}+\zeta\leq\chi_{1-\alpha}^{2}(r)\right)-(1-\alpha)\right|\lesssim|\tau|+r^{-1/2}\left|\zeta\right|+s_{3},\label{beta Chi-square stat True key inequality}
\end{equation}
where $\varphi_{i}=[\bm{\Sigma}_{\varepsilon}+\bm{U}\bm{U}^{\top}\bm{\Sigma}_{\varepsilon}+\bm{\Sigma}_{\varepsilon}\bm{U}\bm{U}^{\top}+\bm{U}\bm{U}^{\top}\bm{\Sigma}_{\varepsilon}\bm{U}\bm{U}^{\top}]_{i,i}$
as defined in Lemma \ref{Lemma beta_hat 1st-order approx}, $\mathscr{B}_{i}$
is defined by
\[
\mathscr{B}_{i}:=\frac{T}{\varphi_{i}}(\widehat{\bm{b}}_{i}^{1}-\widehat{\bm{b}}_{i}^{2})^{\top}[\prod_{j=1}^{2}(\bm{V}^{j}\bm{R}_{V}^{\top})^{\top}\bm{V}^{j}\bm{R}_{V}^{\top}](\widehat{\bm{b}}_{i}^{1}-\widehat{\bm{b}}_{i}^{2}),
\]
and $|\tau|$ satisfies that
\[
|\tau|\lesssim[\frac{1}{\theta}+\sqrt{\kappa_{\varepsilon}}\Vert\bm{U}_{i,\cdot}\Vert_{2}(\frac{1}{\theta}\frac{n}{T}+\frac{1}{\sqrt{T}})]\sqrt{T_{1}T_{2}/T}\sqrt{r}\log^{3/2}n.
\]
Here, $s_{3}\geq0$ is a constant bounded by 
\[
s_{3}\lesssim\frac{r^{3/2}\max_{1\leq k\leq N}\left|[\bm{\Sigma}_{\varepsilon}^{1/2}(\bm{I}_{N}+\bm{U}\bm{U}^{\top})]_{k,i}\right|}{\Vert[\bm{\Sigma}_{\varepsilon}^{1/2}(\bm{I}_{N}+\bm{U}\bm{U}^{\top})]_{\cdot,i}\Vert_{2}}.
\]
Further, if all the entries of the matrix $\bm{Z}$ are Gaussian,
i.e., the noise is Gaussian, then $s_{3}$ in the inequality (\ref{beta Chi-square stat True key inequality})
is equal to zero, i.e., the inequality (\ref{beta Chi-square stat True key inequality})
holds when $s_{3}=0$.

\end{lemma}
\begin{proof}
Our starting point is the following inequality 
\[
\mathbb{P}\left(\mathscr{B}_{i}\leq\chi_{1-\alpha}^{2}(r)-\left|\zeta\right|\right)\leq\mathbb{P}\left(\mathscr{B}_{i}+\zeta\leq\chi_{1-\alpha}^{2}(r)\right)\leq\mathbb{P}\left(\mathscr{B}_{i}\leq\chi_{1-\alpha}^{2}(r)+\left|\zeta\right|\right).
\]
Then, we will prove both the upper bound and the lower bound are close
to $(1-\alpha)$. We will use three steps.

\textit{Step 1 -- Bounding $\mathbb{P}(\mathscr{B}_{i}+\zeta\leq\chi_{1-\alpha}^{2}(r))$.}

We already obtained the first-order approximation of in Lemma \ref{Lemma beta_hat 1st-order approx}.
Indeed, the first-order term $\bm{G}_{P}$ can be rewritten as 
\[
\bm{G}_{b_{i}}=T^{-1/2}((\bm{V}^{1}\bm{R}_{V}^{\top})^{+}(\bm{Z}^{1})^{\top}-(\bm{V}^{2}\bm{R}_{V}^{\top})^{+}(\bm{Z}^{2})^{\top})\bm{\Sigma}_{\varepsilon}^{1/2}(\bm{I}_{N}+\bm{U}\bm{U}^{\top})_{\cdot,i}.
\]
Denote $\bm{K}\in\mathbb{R}^{r\times r}$ the matrix such that $\prod_{j=1}^{2}(\bm{V}^{j}\bm{R}_{V}^{\top})^{\top}\bm{V}^{j}\bm{R}_{V}^{\top}=\bm{K}^{\top}\bm{K}$.
Then $\mathscr{B}_{i}$ can be rewritten as
\[
\mathscr{B}_{i}=\frac{T}{\varphi_{i}}(\bm{G}_{b_{i}}+\bm{\Upsilon}_{b_{i}})^{\top}\bm{K}^{\top}\bm{K}(\bm{G}_{b_{i}}+\bm{\Upsilon}_{b_{i}})=\Vert(T/\varphi_{i})^{1/2}\bm{K}(\bm{G}_{b_{i}}+\bm{\Upsilon}_{b_{i}})\Vert_{2}^{2}.
\]
We denote
\[
\tau:=\sqrt{\mathscr{B}_{i}}-\big\Vert\sqrt{\frac{T}{\varphi_{i}}}\bm{K}\bm{G}_{b_{i}}\big\Vert_{2}=\big\Vert\sqrt{\frac{T}{\varphi_{i}}}\bm{K}(\bm{G}_{b_{i}}+\bm{\Upsilon}_{b_{i}})\big\Vert_{2}-\big\Vert\sqrt{\frac{T}{\varphi_{i}}}\bm{K}\bm{G}_{b_{i}}\big\Vert_{2}.
\]
Then we get
\[
\mathbb{P}\left(\mathscr{B}_{i}\leq\chi_{1-\alpha}^{2}(r)+\left|\zeta\right|\right)=\mathbb{P}\left(\big\Vert\sqrt{\frac{T}{\varphi_{i}}}\bm{K}\bm{G}_{b_{i}}\big\Vert_{2}+\tau\leq\sqrt{\chi_{1-\alpha}^{2}(r)+\left|\zeta\right|}\right)\leq\mathbb{P}\left(\big\Vert\sqrt{\frac{T}{\varphi_{i}}}\bm{K}\bm{G}_{b_{i}}\big\Vert_{2}\leq\sqrt{\chi_{1-\alpha}^{2}(r)+\left|\zeta\right|}+\left|\tau\right|\right),
\]
and thus
\[
\mathbb{P}\left(\big\Vert\sqrt{\frac{T}{\varphi_{i}}}\bm{K}\bm{G}_{b_{i}}\big\Vert_{2}\leq\sqrt{\chi_{1-\alpha}^{2}(r)-\left|\zeta\right|}-\left|\tau\right|\right)\leq\mathbb{P}\left(\mathscr{B}_{i}+\zeta\leq\chi_{1-\alpha}^{2}(r)\right)\leq\mathbb{P}\left(\big\Vert\sqrt{\frac{T}{\varphi_{i}}}\bm{K}\bm{G}_{b_{i}}\big\Vert_{2}\leq\sqrt{\chi_{1-\alpha}^{2}(r)+\left|\zeta\right|}+\left|\tau\right|\right).
\]

\textit{Step 2 -- Expressing $\bm{K}\bm{G}_{b_{i}}$ as a sum to show its
proximity to a Gaussian vector.}

We define 
\[
s_{3}:=\sup_{R\geq0}\left\vert \mathbb{P}(\big\Vert\sqrt{\frac{T}{\varphi_{i}}}\bm{K}\bm{G}_{b_{i}}\big\Vert_{2}\leq R)-\mathbb{P}(\chi^{2}(r)\leq R^{2})\right\vert.
\]
Conditioning on $\bm{F}$, the random vector $\bm{G}_{b_{i}}$ can
be written as a sum of $N$ independent and mean zero $|S|$-dimensional
vectors as follows: 
\[
\bm{G}_{b_{i}}=\sum_{k=1}^{N}\bm{g}_{k}\text{\qquad with\qquad}\bm{g}_{k}=T^{-1/2}(\sum_{j=1}^{2}(-1)^{j-1}(\bm{V}^{j}\bm{R}_{V}^{\top})^{+}(\bm{Z}_{k,\cdot}^{j})^{\top})[\bm{\Sigma}_{\varepsilon}^{1/2}(\bm{I}_{N}+\bm{U}\bm{U}^{\top})]_{k,i}.
\]
By calculations, we have that $\bm{K}^{\top}\bm{K}=\prod_{j=1}^{2}(\bm{V}^{j}\bm{R}_{V}^{\top})^{\top}\bm{V}^{j}\bm{R}_{V}^{\top}=[\sum_{j=1}^{2}((\bm{V}^{j}\bm{R}_{V}^{\top})^{\top}\bm{V}^{j}\bm{R}_{V}^{\top})^{-1}]^{-1}$,
the covariance matrix is given by $\text{cov}(\bm{G}_{b_{i}}|\bm{F})=\sum_{k=1}^{N}\text{cov}(\bm{g}_{k}|\bm{F})=(\varphi_{i}/T)\sum_{j=1}^{2}((\bm{V}^{j}\bm{R}_{V}^{\top})^{\top}\bm{V}^{j}\bm{R}_{V}^{\top})^{-1}=(\varphi_{i}/T)(\bm{K}^{\top}\bm{K})^{-1}=(\varphi_{i}/T)\bm{K}^{-1}(\bm{K}^{\top})^{-1}$,
and thus
\[
\text{cov}((T/\varphi_{i})^{1/2}\bm{K}\bm{G}_{b_{i}}|\bm{F})=(T/\varphi_{i})\bm{K}\text{cov}(\bm{G}_{b_{i}}|\bm{F})\bm{K}^{\top}=\bm{K}\bm{K}^{-1}(\bm{K}^{\top})^{-1}\bm{K}^{\top}=\bm{I}_{r}.
\]
Then we consider two cases.

(i) When all entries of $\bm{Z}$ are Gaussian, we have that $s_{3}=0$.

(ii) When the entries of $\bm{Z}$ are sub-Gaussian, we obtain by
Lemma \ref{Lemma Berry-Esseen balls} that
\begin{align*}
s_{3} & \lesssim\sum_{k=1}^{N}\mathbb{E}\left[\Vert(T/\varphi_{i})^{1/2}\bm{K}\bm{g}_{k}\Vert_{2}^{3}|\bm{F}\right]\leq(1/\varphi_{i})^{3/2}\mathbb{E}\left[\big\Vert\bm{K}[(\bm{V}^{1}\bm{R}_{V}^{\top})^{+}(\bm{Z}_{k,\cdot}^{1})^{\top}-(\bm{V}^{2}\bm{R}_{V}^{\top})^{+}(\bm{Z}_{k,\cdot}^{2})^{\top}]\big\Vert_{2}^{3}|\bm{F}\right]\\
 & \times\max_{1\leq k\leq N}\left|[\bm{\Sigma}_{\varepsilon}^{1/2}(\bm{I}_{N}+\bm{U}\bm{U}^{\top})]_{k,i}\right|\cdot\sum_{k=1}^{N}([\bm{\Sigma}_{\varepsilon}^{1/2}(\bm{I}_{N}+\bm{U}\bm{U}^{\top})]_{k,i})^{2}\\
 & \lesssim\varphi_{i}^{-1/2}\max_{1\leq k\leq N}\left|[\bm{\Sigma}_{\varepsilon}^{1/2}(\bm{I}_{N}+\bm{U}\bm{U}^{\top})]_{k,i}\right|\mathbb{E}\left[\big\Vert\bm{K}[(\bm{V}^{1}\bm{R}_{V}^{\top})^{+}(\bm{Z}_{k,\cdot}^{1})^{\top}-(\bm{V}^{2}\bm{R}_{V}^{\top})^{+}(\bm{Z}_{k,\cdot}^{2})^{\top}]\big\Vert_{2}^{3}|\bm{F}\right].
\end{align*}
Recall that $\prod_{j=1}^{2}(\bm{V}^{j}\bm{R}_{V}^{\top})^{\top}\bm{V}^{j}\bm{R}_{V}^{\top}=\bm{K}^{\top}\bm{K}$,
so we obtain that $\Vert\bm{K}\Vert_{2}\lesssim\prod_{j=1}^{2}(T_{j}/T)^{1/2}\leq\sqrt{T_{1}T_{2}}/T$.
since $\bm{K}[(\bm{V}^{1}\bm{R}_{V}^{\top})^{+}(\bm{Z}_{k,\cdot}^{1})^{\top}-(\bm{V}^{2}\bm{R}_{V}^{\top})^{+}(\bm{Z}_{k,\cdot}^{2})^{\top}]$
is a $r$-dimensional sub-Gaussian vector whose covariance matrix
is equal to $\bm{K}[\sum_{j=1}^{2}((\bm{V}^{j}\bm{R}_{V}^{\top})^{\top}\bm{V}^{j}\bm{R}_{V}^{\top})^{-1}]\bm{K}^{\top}=\bm{K}(\bm{K}^{\top}\bm{K})^{-1}\bm{K}^{\top}=\bm{I}_{r}$.
Then we have that, similar to Step 2 in the proof of Lemma \ref{Lemma true Chi-square stat},
it holds $\mathbb{E}\left[\big\Vert\bm{K}[(\bm{V}^{1}\bm{R}_{V}^{\top})^{+}(\bm{Z}_{k,\cdot}^{1})^{\top}-(\bm{V}^{2}\bm{R}_{V}^{\top})^{+}(\bm{Z}_{k,\cdot}^{2})^{\top}]\big\Vert_{2}^{3}|\bm{F}\right]\lesssim r^{3/2}$.
So we obtain
\[
s_{3}\lesssim r^{3/2}\varphi_{i}^{-1/2}\max_{1\leq k\leq N}\left|[\bm{\Sigma}_{\varepsilon}^{1/2}(\bm{I}_{N}+\bm{U}\bm{U}^{\top})]_{k,i}\right|.
\]

\textit{Step 3 -- Establishing chi-squared distributional characterization.}

By definition of $s_{3}$, we obtain that
\[
\mathbb{P}\left(\sqrt{\chi^{2}(r)}\leq\sqrt{\chi_{1-\alpha}^{2}(r)-\left|\zeta\right|}-\left|\tau\right|\right)-s_{3}\leq\mathbb{P}\left(\mathscr{B}_{i}+\zeta\leq\chi_{1-\alpha}^{2}(r)\right)\leq\mathbb{P}\left(\sqrt{\chi^{2}(r)}\leq\sqrt{\chi_{1-\alpha}^{2}(r)+\left|\zeta\right|}+\left|\tau\right|\right)+s_{3}.
\]

Similar to Step 3 in the proof of Lemma \ref{Lemma true Chi-square stat},
we obtain that 
\[
\left|\mathbb{P}\left(\mathscr{B}_{i}+\zeta\leq\chi_{1-\alpha}^{2}(r)\right)-(1-\alpha)\right|\lesssim\left|\tau\right|+r^{-1/2}\left|\zeta\right|+s_{3}.
\]
So, to prove the desired result in the lemma, it suffices to establish
the upper bound for $\left|\tau\right|$, which is defined in Step
1. Using (\ref{Chi-square stat 2-norm equality}) and the triangle
inequality, we obtain that 
\[
|\tau|=\left|\Vert\sqrt{\frac{T}{\varphi_{i}}}\bm{K}(\bm{G}_{b_{i}}+\bm{\Upsilon}_{b_{i}})\Vert_{2}-\Vert\sqrt{\frac{T}{\varphi_{i}}}\bm{K}\bm{G}_{b_{i}}\Vert_{2}\right|\leq\Vert\sqrt{\frac{T}{\varphi_{i}}}\bm{K}\bm{\Upsilon}_{b_{i}}\Vert_{2}.
\]
Finally, using the upper bound of $\Vert\bm{\Upsilon}_{b_{i}}\Vert_{2}$,
we obtain the upper bound for $|\tau|$ as
\begin{align*}
|\tau| & \lesssim(T/\varphi_{i})^{1/2}\cdot\frac{1}{T}\sqrt{T_{1}T_{2}}\cdot\Vert\bm{\Upsilon}_{b_{i}}\Vert_{2}=(1/\varphi_{i})^{1/2}\sqrt{T_{1}T_{2}/T}\cdot\Vert\bm{\Upsilon}_{b_{i}}\Vert_{2}\\
 & \lesssim\frac{1}{\theta}\sqrt{T_{1}T_{2}/T}\sqrt{r}\log^{3/2}n+\sqrt{\kappa_{\varepsilon}}\Vert\bm{U}_{i,\cdot}\Vert_{2}(\frac{n}{\theta T}+\frac{1}{\sqrt{T}})\sqrt{T_{1}T_{2}/T}\sqrt{r}\log^{3/2}n.
\end{align*}
\end{proof}

\subsection{Proof of Theorem \ref{Thm beta structure test}}

To prove the desired results, we use Lemma \ref{Lemma beta true Chi-square stat}
with the small perturbation $\zeta=\mathscr{\widehat{B}}_{i}-\mathscr{B}_{i}$.
Then, the problem boils down to deriving the upper bound for $r^{-1/2}|\zeta|$.
Rewrite the test statistic as $\mathscr{\widehat{B}}_{i}=(T/\varphi_{i})(\widehat{\bm{b}}_{i}^{1}-\widehat{\bm{b}}_{i}^{2})^{\top}[\prod_{j=1}^{2}(\widehat{\bm{V}}^{j})^{\top}\widehat{\bm{V}}^{j}](\widehat{\bm{b}}_{i}^{1}-\widehat{\bm{b}}_{i}^{2})$.
We do the decomposition $\zeta=\mathscr{\widehat{B}}_{i}-\mathscr{B}_{i}=r_{1}+r_{2}$,
where $r_{1}$ and $r_{2}$ are respectively defined by 
\[
r_{1}:=\frac{T}{\varphi_{i}}(\widehat{\bm{b}}_{i}^{1}-\widehat{\bm{b}}_{i}^{2})^{\top}[\prod_{j=1}^{2}(\widehat{\bm{V}}^{j})^{\top}\widehat{\bm{V}}^{j}](\widehat{\bm{b}}_{i}^{1}-\widehat{\bm{b}}_{i}^{2})-\mathscr{B}_{i},
\]
and
\[
r_{2}:=\mathscr{\widehat{B}}_{i}-\frac{T}{\varphi_{i}}(\widehat{\bm{b}}_{i}^{1}-\widehat{\bm{b}}_{i}^{2})^{\top}[\prod_{j=1}^{2}(\widehat{\bm{V}}^{j})^{\top}\widehat{\bm{V}}^{j}](\widehat{\bm{b}}_{i}^{1}-\widehat{\bm{b}}_{i}^{2}).
\]

\textit{Step 1 -- Upper bound for $|r_{1}|$.}

We have that
\[
r_{1}=\frac{T}{\varphi_{i}}(\widehat{\bm{b}}_{i}^{1}-\widehat{\bm{b}}_{i}^{2})^{\top}[\prod_{j=1}^{2}(\widehat{\bm{V}}^{j})^{\top}\widehat{\bm{V}}^{j}-\prod_{j=1}^{2}(\bm{V}^{j}\bm{R}_{V}^{\top})^{\top}\bm{V}^{j}\bm{R}_{V}^{\top}](\widehat{\bm{b}}_{i}^{1}-\widehat{\bm{b}}_{i}^{2}).
\]
Then using the results in Lemma \ref{Lemma beta_hat 1st-order approx},
we obtain
\begin{align*}
|r_{1}| & \leq\frac{T}{\varphi_{i}}\big\Vert\prod_{j=1}^{2}(\widehat{\bm{V}}^{j})^{\top}\widehat{\bm{V}}^{j}-\prod_{j=1}^{2}(\bm{V}^{j}\bm{R}_{V}^{\top})^{\top}\bm{V}^{j}\bm{R}_{V}^{\top}\big\Vert_{2}\big\Vert\widehat{\bm{b}}_{i}^{1}-\widehat{\bm{b}}_{i}^{2}\big\Vert_{2}^{2}\\
 & \lesssim\frac{T}{\varphi_{i}}\cdot(T_{1}T_{2}/T^{2})\frac{1}{\theta}\sqrt{r}\log n\cdot\big\Vert\widehat{\bm{b}}_{i}^{1}-\widehat{\bm{b}}_{i}^{2}\big\Vert_{2}^{2}=\frac{1}{\theta}\sqrt{r}\log n\cdot\frac{1}{\varphi_{i}}\frac{T_{1}T_{2}}{T}\big\Vert\widehat{\bm{b}}_{i}^{1}-\widehat{\bm{b}}_{i}^{2}\big\Vert_{2}^{2}.
\end{align*}

\textit{Step 2 -- Upper bound for $|r_{2}|$.}

By definition, we have that
\[
r_{2}=T(\frac{1}{\widehat{\varphi}_{i}}-\frac{1}{\varphi_{i}})\cdot(\widehat{\bm{b}}_{i}^{1}-\widehat{\bm{b}}_{i}^{2})^{\top}[\prod_{j=1}^{2}(\widehat{\bm{V}}^{j})^{\top}\widehat{\bm{V}}^{j}](\widehat{\bm{b}}_{i}^{1}-\widehat{\bm{b}}_{i}^{2}).
\]
Then we have that
\begin{align*}
 & |r_{2}/(\frac{1}{\widehat{\varphi}_{i}}-\frac{1}{\varphi_{i}})|=T(\widehat{\bm{b}}_{i}^{1}-\widehat{\bm{b}}_{i}^{2})^{\top}[\prod_{j=1}^{2}(\widehat{\bm{V}}^{j})^{\top}\widehat{\bm{V}}^{j}](\widehat{\bm{b}}_{i}^{1}-\widehat{\bm{b}}_{i}^{2})\\
 & \lesssim T\prod_{j=1}^{2}((T_{j}/T)^{1/2})^{2}\big\Vert\widehat{\bm{b}}_{i}^{1}-\widehat{\bm{b}}_{i}^{2}\big\Vert_{2}^{2}\lesssim\frac{T_{1}T_{2}}{T}\big\Vert\widehat{\bm{b}}_{i}^{1}-\widehat{\bm{b}}_{i}^{2}\big\Vert_{2}^{2}.
\end{align*}
So the problem boils down to bounding $|1/\widehat{\varphi}_{i}-1/\varphi_{i}|=|\widehat{\varphi}_{i}-\varphi_{i}|/|\widehat{\varphi}_{i}\varphi_{i}|$.
We obtain by Lemma \ref{Lemma noise cov matrix estimate} that $\big\Vert(\widehat{\bm{\Sigma}}_{\varepsilon}^{\tau})_{i,\cdot}-(\bm{\Sigma}_{\varepsilon})_{i,\cdot}\big\Vert_{2}\leq\big\Vert\widehat{\bm{\Sigma}}_{\varepsilon}^{\tau}-\bm{\Sigma}_{\varepsilon}\big\Vert_{2}\lesssim(\epsilon_{N,T})^{1-q}s(\bm{\Sigma}_{\varepsilon})$
and $|(\widehat{\bm{\Sigma}}_{\varepsilon}^{\tau})_{i,i}-(\bm{\Sigma}_{\varepsilon})_{i,i}|=|(\widehat{\bm{\Sigma}}_{\varepsilon})_{i,i}-(\bm{\Sigma}_{\varepsilon})_{i,i}|\lesssim(\bm{\Sigma}_{\varepsilon})_{i,i}\epsilon_{N,T}$.
We obtain by Lemma \ref{Lemma R H for U V} that $\big\Vert\widehat{\bm{U}}\bm{R}_{U}-\bm{U}\big\Vert_{2}\lesssim\rho$,
and we obtain by the proof of Corollary \ref{corollary:error bound B F}
that $\big\Vert(\widehat{\bm{U}}\bm{R}_{U}-\bm{U})_{i,\cdot}\big\Vert_{2}\lesssim\omega_{i}\sqrt{\frac{r}{n}}\log n$.
So we obtain that
\begin{align*}
|\widehat{\varphi}_{i}-\varphi_{i}| & \lesssim\epsilon_{N,T}(\bm{\Sigma}_{\varepsilon})_{i,i}+(1+\Vert\bm{U}_{i,\cdot}\Vert_{2})\{\Vert\bm{U}_{i,\cdot}\Vert_{2}(\epsilon_{N,T})^{1-q}s(\bm{\Sigma}_{\varepsilon})\\
 & +\omega_{i}\sqrt{\frac{r}{n}}\log n\Vert(\bm{\Sigma}_{\varepsilon})_{i,\cdot}\Vert_{2}+\Vert\bm{U}_{i,\cdot}\Vert_{2}\rho\Vert(\bm{\Sigma}_{\varepsilon})_{i,\cdot}\Vert_{2}\}.
\end{align*}
Then we have that
\begin{align*}
\varphi_{i} & =(\bm{I}_{N}+\bm{U}\bm{U}^{\top})_{i,\cdot}\bm{\Sigma}_{\varepsilon}(\bm{I}_{N}+\bm{U}\bm{U}^{\top})_{\cdot,i}\geq\lambda_{\min}(\bm{\Sigma}_{\varepsilon})\Vert(\bm{I}_{N}+\bm{U}\bm{U}^{\top})_{\cdot,i}\Vert_{2}^{2}=\lambda_{\min}(\bm{\Sigma}_{\varepsilon})(1+3\Vert\bm{U}_{i,\cdot}\Vert_{2}^{2}),\\
 & \gtrsim\frac{1}{\kappa_{\varepsilon}}\Vert\bm{\Sigma}_{\varepsilon}\Vert_{2}(1+\Vert\bm{U}_{i,\cdot}\Vert_{2}^{2}).
\end{align*}
Finally, we obtain that
\begin{align*}
|1/\widehat{\varphi}_{i}-1/\varphi_{i}| & =\frac{|\widehat{\varphi}_{i}-\varphi_{i}|}{|\widehat{\varphi}_{i}||\varphi_{i}|}\leq\frac{1}{\varphi_{i}}\frac{1}{\frac{1}{\kappa_{\varepsilon}}\Vert\bm{\Sigma}_{\varepsilon}\Vert_{2}(1+\Vert\bm{U}_{i,\cdot}\Vert_{2}^{2})}\\
 & \times\{\epsilon_{N,T}(\bm{\Sigma}_{\varepsilon})_{i,i}+(1+\Vert\bm{U}_{i,\cdot}\Vert_{2})[\Vert\bm{U}_{i,\cdot}\Vert_{2}(\epsilon_{N,T})^{1-q}s(\bm{\Sigma}_{\varepsilon})+\omega_{i}\sqrt{\frac{r}{n}}\log n\Vert(\bm{\Sigma}_{\varepsilon})_{i,\cdot}\Vert_{2}+\Vert\bm{U}_{i,\cdot}\Vert_{2}\rho\Vert(\bm{\Sigma}_{\varepsilon})_{i,\cdot}\Vert_{2}]\}
\end{align*}
We assemble all the above error bounds to obtain
\begin{align*}
|r_{2}| & \lesssim\frac{1}{\varphi_{i}}\frac{T_{1}T_{2}}{T}\big\Vert\widehat{\bm{b}}_{i}^{1}-\widehat{\bm{b}}_{i}^{2}\big\Vert_{2}^{2}\cdot\frac{\kappa_{\varepsilon}}{\Vert\bm{\Sigma}_{\varepsilon}\Vert_{2}}\\
 & \times\{\epsilon_{N,T}(\bm{\Sigma}_{\varepsilon})_{i,i}+\sqrt{r}[\Vert\bm{U}_{i,\cdot}\Vert_{2}(\epsilon_{N,T})^{1-q}s(\bm{\Sigma}_{\varepsilon})+\omega_{i}\sqrt{\frac{r}{n}}\log n\Vert(\bm{\Sigma}_{\varepsilon})_{i,\cdot}\Vert_{2}+\Vert\bm{U}_{i,\cdot}\Vert_{2}\rho\Vert(\bm{\Sigma}_{\varepsilon})_{i,\cdot}\Vert_{2}]\}\\
 & \lesssim\frac{1}{\varphi_{i}}\frac{T_{1}T_{2}}{T}\big\Vert\widehat{\bm{b}}_{i}^{1}-\widehat{\bm{b}}_{i}^{2}\big\Vert_{2}^{2}\cdot\kappa_{\varepsilon}\{\epsilon_{N,T}+(\epsilon_{N,T})^{1-q}\sqrt{r}\Vert\bm{U}_{i,\cdot}\Vert_{2}\frac{s(\bm{\Sigma}_{\varepsilon})}{\Vert\bm{\Sigma}_{\varepsilon}\Vert_{2}}+\frac{1}{\vartheta_{i}\sqrt{T}}r\log n+\sqrt{r}\Vert\bm{U}_{i,\cdot}\Vert_{2}\frac{1}{\theta}\sqrt{\frac{n}{T}}\}
\end{align*}
where we use the fact that $(\bm{\Sigma}_{\varepsilon})_{i,i}\leq\Vert\bm{\Sigma}_{\varepsilon}\Vert_{2}$
and $\Vert(\bm{\Sigma}_{\varepsilon})_{i,\cdot}\Vert_{2}\leq\Vert\bm{\Sigma}_{\varepsilon}\Vert_{2}$.

\textit{Step 3 -- Combining all the bounds.}

We obtain that
\begin{align*}
\left|\zeta\right| & \leq|r_{1}|+|r_{2}|\lesssim\frac{1}{\varphi_{i}}\frac{T_{1}T_{2}}{T}\big\Vert\widehat{\bm{b}}_{i}^{1}-\widehat{\bm{b}}_{i}^{2}\big\Vert_{2}^{2}\\
 & \times\{\frac{1}{\theta}\sqrt{r}\log n+\kappa_{\varepsilon}[\epsilon_{N,T}+(\epsilon_{N,T})^{1-q}\sqrt{r}\Vert\bm{U}_{i,\cdot}\Vert_{2}\frac{s(\bm{\Sigma}_{\varepsilon})}{\Vert\bm{\Sigma}_{\varepsilon}\Vert_{2}}+\frac{1}{\vartheta_{i}\sqrt{T}}r\log n+\sqrt{r}\Vert\bm{U}_{i,\cdot}\Vert_{2}\frac{1}{\theta}\sqrt{\frac{n}{T}}]\}.
\end{align*}
Then, the desired results follow from the inequality (\ref{beta Chi-square stat True key inequality})
in Lemma \ref{Lemma beta true Chi-square stat}:
\begin{align*}
 & \left|\mathbb{P}\left(\mathscr{\widehat{B}}_{i}\leq\chi_{1-\alpha}^{2}(r)\right)-(1-\alpha)\right|\lesssim|\tau|+r^{-1/2}\left|\zeta\right|+s_{3}\\
 & \lesssim[\frac{1}{\theta}+\sqrt{\kappa_{\varepsilon}}\Vert\bm{U}_{i,\cdot}\Vert_{2}(\frac{1}{\theta}\frac{n}{T}+\frac{1}{\sqrt{T}})]\sqrt{T_{1}T_{2}/T}\sqrt{r}\log^{3/2}n+\frac{1}{\varphi_{i}}\frac{T_{1}T_{2}}{T}\big\Vert\widehat{\bm{b}}_{i}^{1}-\widehat{\bm{b}}_{i}^{2}\big\Vert_{2}^{2}\\
 & \times\{\frac{1}{\theta}\sqrt{r}\log n+\kappa_{\varepsilon}\epsilon_{N,T}+\kappa_{\varepsilon}(\epsilon_{N,T})^{1-q}\sqrt{r}\Vert\bm{U}_{i,\cdot}\Vert_{2}\frac{s(\bm{\Sigma}_{\varepsilon})}{\Vert\bm{\Sigma}_{\varepsilon}\Vert_{2}}+\kappa_{\varepsilon}\frac{1}{\vartheta_{i}\sqrt{T}}r\log n+\kappa_{\varepsilon}\sqrt{r}\Vert\bm{U}_{i,\cdot}\Vert_{2}\frac{1}{\theta}\sqrt{\frac{n}{T}}\}\\
 & \lesssim\frac{1}{\theta}\sqrt{T_{1}T_{2}/T}\sqrt{r}\log^{3/2}n+\sqrt{\kappa_{\varepsilon}}\Vert\bm{U}_{i,\cdot}\Vert_{2}(\frac{1}{\theta}\frac{n}{T}+\frac{1}{\sqrt{T}})\sqrt{T_{1}T_{2}/T}\sqrt{r}\log^{3/2}n\\
 & +\{\kappa_{\varepsilon}\Vert\bm{U}_{i,\cdot}\Vert_{2}^{2}\frac{T_{1}T_{2}}{T}r\log n+\kappa_{\varepsilon}\frac{1}{\theta^{2}}\frac{T_{1}T_{2}}{T}r\log n+\kappa_{\varepsilon}r\log n\}\\
 & \times\{\frac{1}{\theta}+\kappa_{\varepsilon}\epsilon_{N,T}+\kappa_{\varepsilon}(\epsilon_{N,T})^{1-q}\Vert\bm{U}_{i,\cdot}\Vert_{2}\frac{s(\bm{\Sigma}_{\varepsilon})}{\Vert\bm{\Sigma}_{\varepsilon}\Vert_{2}}+\kappa_{\varepsilon}\frac{1}{\vartheta_{i}\sqrt{T}}+\kappa_{\varepsilon}\Vert\bm{U}_{i,\cdot}\Vert_{2}\frac{1}{\theta}\sqrt{\frac{n}{T}}\}.
\end{align*}

%% file: appendix_two_rows_of_B.tex
\section{Proof of Theorem \ref{Thm two-sample test B}: Two-sample test for
betas}

\subsection{A useful lemma}

We prove a lemma that will be useful when establishing Theorem \ref{Thm two-sample test B}.
\begin{lemma}\label{Lemma B two sample prepare}Assume
that the assumptions in Theorem \ref{Thm two-sample test B} hold.
For any $i\neq j$ and $i,j\in\{1,2,\ldots,N\}$, under the null hypothesis
$\bm{b}_{i}=\bm{b}_{j}$, we have that, for any random variable $\zeta$
satisfying $\left|\zeta\right|\ll r$, it holds 
\begin{align*}
\left|\mathbb{P}\left(\mathcal{B}_{ij}+\zeta\leq\chi_{1-\alpha}^{2}(r)\right)-(1-\alpha)\right| & \lesssim\frac{\left|\zeta\right|}{\sqrt{r}}+\tau_{U}^{(i,j)}\\
 & +\frac{\sqrt{T}\sigma_{r}(\omega_{i}+\omega_{j})\omega\sqrt{\frac{r}{n}}\log^{3/2}n+\sqrt{T}\sigma_{r}(\omega_{i}+\omega_{j})\omega\log n\big\Vert\bar{\bm{U}}\big\Vert_{2,\infty}}{\sqrt{(\bm{\Sigma}_{\varepsilon})_{i,i}+(\bm{\Sigma}_{\varepsilon})_{j,j}-2(\bm{\Sigma}_{\varepsilon})_{i,j}}}\\
 & +\frac{\sqrt{T}\sigma_{r}\left(\rho^{2}+\rho\sqrt{\frac{r}{n}}\log n\right)(\big\Vert\bar{\bm{U}}_{i,\cdot}\big\Vert_{2}+\big\Vert\bar{\bm{U}}_{j,\cdot}\big\Vert_{2})}{\sqrt{(\bm{\Sigma}_{\varepsilon})_{i,i}+(\bm{\Sigma}_{\varepsilon})_{j,j}-2(\bm{\Sigma}_{\varepsilon})_{i,j}}},
\end{align*}
where $\mathcal{B}_{ij}$ is defined as 
\begin{align*}
\mathcal{B}_{ij}:= & T((\bm{\Sigma}_{\varepsilon})_{i,i}+(\bm{\Sigma}_{\varepsilon})_{j,j}-2(\bm{\Sigma}_{\varepsilon})_{i,j})^{-1}\big\Vert(\widehat{\bm{U}}_{i,\cdot}-\widehat{\bm{U}}_{j,\cdot})\widehat{\bm{\Sigma}}\big\Vert_{2}^{2},
\end{align*}
and $\tau_{U}^{(i,j)}$ is a constant bounded by $\tau_{U}^{(i,j)}\lesssim r\sqrt{\frac{\log n}{T}}$.
Further, if all the entries of the matrix $\bm{Z}$ are Gaussian,
i.e., the noise is Gaussian, then result holds even when $\tau_{U}^{(i,j)}=0$.\end{lemma} 
\begin{proof}
\textit{Step 1 -- Expand $(\widehat{\bm{U}}_{i,\cdot}-\widehat{\bm{U}}_{j,\cdot})\widehat{\bm{\Sigma}}$
and develop the whole roadmap for the proof.}

We already have the expansion of $\widehat{\bm{U}}\bm{R}_{U}-\bm{U}$
and the error bound for $\varepsilon_{\bm{\Sigma}}:=(\bm{R}_{U})^{\top}\widehat{\bm{\Sigma}}\bm{R}_{V}-\bm{\Lambda}$,
so we obtain that 
\begin{align*}
\widehat{\bm{U}}\widehat{\bm{\Sigma}}\bm{R}_{V} & =\widehat{\bm{U}}\bm{R}_{U}(\bm{R}_{U})^{\top}\widehat{\bm{\Sigma}}\bm{R}_{V}=(\bm{U}+\bm{G}_{U}+\bm{\Psi}_{U})(\bm{\Lambda}+\varepsilon_{\bm{\Sigma}})=(\bm{U}+\bm{G}_{U}+\bm{\Psi}_{U})\bm{\Lambda}+\widehat{\bm{U}}\bm{R}_{U}\varepsilon_{\bm{\Sigma}}\\
 & =\bm{U}\bm{\Lambda}+T^{-1/2}\bm{\Sigma}_{\varepsilon}^{1/2}\bm{Z}\bm{V}+(\widehat{\bm{U}}\bm{R}_{U}\varepsilon_{\bm{\Sigma}}+\bm{\Psi}_{U}\bm{\Lambda}).
\end{align*}
Note that, by Lemma \ref{Lemma SVD BF good event}, conditioning on
the good event $\mathcal{E}_{0}$, we have that $\bm{U}_{i,\cdot}-\bm{U}_{j,\cdot}=0$
under the null $\bm{b}_{i}=\bm{b}_{j}$. The reason is that, there
exists a rotation matrix $\bm{Q}$ such that $\bm{U}=\bar{\bm{U}}\bm{Q}$,
and by definition we have that $\bar{\bm{U}}_{i,\cdot}=\bar{\bm{U}}_{j,\cdot}$
since $\bar{\bm{U}}=\bm{\bm{B}\Sigma}^{-1}$ and $\bm{b}_{i}=\bm{b}_{j}$,
i.e., $\bm{B}_{i,\cdot}=\bm{B}_{j,\cdot}$.

The following derivations are conditioning on $\mathcal{E}_{0}$ and
we will take the effect of $\mathcal{E}_{0}$ into consideration in
the final non-asymptotic probability bound. Then we have that 
\begin{align*}
 & (\widehat{\bm{U}}_{i,\cdot}-\widehat{\bm{U}}_{j,\cdot})\widehat{\bm{\Sigma}}\bm{R}_{V}\\
 & =(\bm{U}_{i,\cdot}-\bm{U}_{j,\cdot})\bm{\Lambda}+T^{-1/2}((\bm{\Sigma}_{\varepsilon}^{1/2})_{i,\cdot}-(\bm{\Sigma}_{\varepsilon}^{1/2})_{j,\cdot})\bm{Z}\bm{V}+(\widehat{\bm{U}}_{i,\cdot}-\widehat{\bm{U}}_{j,\cdot})\bm{R}_{U}\varepsilon_{\bm{\Sigma}}+((\bm{\Psi}_{U})_{i,}-(\bm{\Psi}_{U})_{j,})\bm{\Lambda}\\
 & =T^{-1/2}\bm{a}\bm{Z}\bm{V}+((\bm{\Psi}_{U})_{i,}-(\bm{\Psi}_{U})_{j,})\bm{\Lambda}=T^{-1/2}\sum_{t=1}^{T}(\bm{a}\bm{Z}_{\cdot,t})\bm{V}_{t,\cdot}+((\bm{\Psi}_{U})_{i,}-(\bm{\Psi}_{U})_{j,})\bm{\Lambda},
\end{align*}
where, for simplicity of notations, we denote $a_{k}:=T^{-1/2}(\bm{\Sigma}_{\varepsilon}^{1/2})_{i,k}-(\bm{\Sigma}_{\varepsilon}^{1/2})_{j,k}$
and denote $\bm{a}=(a_{1},a_{2},\ldots,a_{N})$ as a row vector satisfying
that 
\[
\left\Vert \bm{a}\right\Vert _{2}^{2}=\left\Vert \bm{a}\right\Vert _{\mathrm{F}}^{2}=\sum_{k=1}^{N}a_{k}^{2}=\frac{1}{T}((\bm{\Sigma}_{\varepsilon})_{i,i}+(\bm{\Sigma}_{\varepsilon})_{j,j}-2(\bm{\Sigma}_{\varepsilon})_{i,j}).
\]

We will show proximity between the law of $\sum_{t=1}^{T}(\bm{a}\bm{Z}_{\cdot,t})\bm{V}_{t,\cdot}$
and Gaussian, and prove that the term 
\[
\varepsilon_{2}:=\frac{1}{\left\Vert \bm{a}\right\Vert _{2}}\big\Vert((\bm{\Psi}_{U})_{i,}-(\bm{\Psi}_{U})_{j,})\bm{\Lambda}\big\Vert_{2}
\]
is negligible.

We will establish the upper bound for the parameter 
\[
\tau_{U}^{(i,j)}:=\sup_{R\geq0}\left\vert \mathbb{P}\left(\frac{1}{\left\Vert \bm{a}\right\Vert _{2}}\big\Vert\sum_{t=1}^{T}(\bm{a}\bm{Z}_{\cdot,t})\bm{V}_{t,\cdot}\big\Vert_{2}\leq R\right)-\mathbb{P}(\chi^{2}(r)\leq R^{2})\right\vert .
\]

The motivations for establishing the upper bound for $\tau_{U}^{(i,j)}$
are as follows: since $\left\Vert \bm{a}\right\Vert _{2}^{2}=\frac{1}{T}((\bm{\Sigma}_{\varepsilon})_{i,i}+(\bm{\Sigma}_{\varepsilon})_{j,j}-2(\bm{\Sigma}_{\varepsilon})_{i,j})$,
we can rewrite $\mathcal{B}_{ij}$ as 
\begin{align*}
\mathcal{B}_{ij} & :=\frac{1}{\left\Vert \bm{a}\right\Vert _{2}}\big\Vert\sum_{t=1}^{T}(\bm{a}\bm{Z}_{\cdot,t})\bm{V}_{t,\cdot}+((\bm{\Psi}_{U})_{i,}-(\bm{\Psi}_{U})_{j,})\bm{\Lambda}\big\Vert_{2}^{2}.
\end{align*}
Similar to the techniques we used for projection matrix of factors
in Lemma \ref{Lemma true Chi-square stat}, we have that, for any
random variable $\zeta$ satisfying $\left|\zeta\right|\ll r$, 
\begin{align*}
 & \mathbb{P}\left(\frac{1}{\left\Vert \bm{a}\right\Vert _{2}}\big\Vert\sum_{t=1}^{T}(\bm{a}\bm{Z}_{\cdot,t})\bm{V}_{t,\cdot}\big\Vert_{2}\leq\sqrt{\chi_{1-\alpha}^{2}(r)-\left|\zeta\right|}-\varepsilon_{2}\right)\\
 & \leq\mathbb{P}\left(\mathcal{B}_{ij}+\zeta\leq\chi_{1-\alpha}^{2}(r)\right)\\
 & \leq\mathbb{P}\left(\frac{1}{\left\Vert \bm{a}\right\Vert _{2}}\big\Vert\sum_{t=1}^{T}(\bm{a}\bm{Z}_{\cdot,t})\bm{V}_{t,\cdot}\big\Vert_{2}\leq\sqrt{\chi_{1-\alpha}^{2}(r)+\left|\zeta\right|}+\varepsilon_{2}\right).
\end{align*}
Then by definition of $\tau_{U}^{(i,j)}$, we have that 
\[
\mathbb{P}\left(\sqrt{\chi^{2}(r)}\leq\sqrt{\chi_{1-\alpha}^{2}(r)-\left|\zeta\right|}-\varepsilon_{2}\right)-\tau_{U}^{(i,j)}\leq\mathbb{P}\left(\mathcal{B}_{ij}+\zeta\leq\chi_{1-\alpha}^{2}(r)\right)\leq\mathbb{P}\left(\sqrt{\chi^{2}(r)}\leq\sqrt{\chi_{1-\alpha}^{2}(r)+\left|\zeta\right|}+\varepsilon_{2}\right)+\tau_{U}^{(i,j)}.
\]
Similar to Step 3 of the proof of Lemma \ref{Lemma true Chi-square stat},
since $\left|\zeta\right|\ll r\asymp\chi_{1-\alpha}^{2}(r)$, we have
that the left bound of the above inequality satisfies that 
\[
\mathbb{P}\left(\chi^{2}(r)\leq\chi_{1-\alpha}^{2}(r)\right)-h\lesssim\mathbb{P}\left(\sqrt{\chi^{2}(r)}\leq\sqrt{\chi_{1-\alpha}^{2}(r)}-h\right)\leq\mathbb{P}\left(\sqrt{\chi^{2}(r)}\leq\sqrt{\chi_{1-\alpha}^{2}(r)-\left|\zeta\right|}-\varepsilon_{2}\right),
\]
where $h>0$ satisfies that $h\lesssim\left|\varepsilon_{2}\right|+\frac{1}{\sqrt{r}}\left|\zeta\right|$.
Similar inequality holds for the right bound $\mathbb{P}(\sqrt{\chi^{2}(r)}\leq\sqrt{\chi_{1-\alpha}^{2}(r)+\left|\zeta\right|}+\varepsilon_{2})$
of $\mathbb{P}(\mathcal{B}_{ij}+\zeta\leq\chi_{1-\alpha}^{2}(r))$,
so we obtain that 
\[
\left|\mathbb{P}\left(\mathcal{B}_{ij}+\zeta\leq\chi_{1-\alpha}^{2}(r)\right)-(1-\alpha)\right|\lesssim\varepsilon_{2}+\frac{1}{\sqrt{r}}\left|\zeta\right|+\tau_{U}^{(i,j)}.
\]
Thus, the problems boil down to proving the upper bounds for $\tau_{U}^{(i,j)}$
and $\varepsilon_{2}$.

\textit{Step 2 -- Prove the upper bound for $\tau_{U}^{(i,j)}$ to show the
proximity of $\sum_{t=1}^{T}(\bm{a}\bm{Z}_{\cdot,t})\bm{V}_{t,\cdot}$
with a Gaussian vector.}

Conditioning on $\bm{F}$, the variance of $\sum_{t=1}^{T}(\bm{a}\bm{Z}_{\cdot,t})\bm{V}_{t,\cdot}$
is given by 
\[
\sum_{t=1}^{T}\mathbb{E}[(\bm{a}\bm{Z}_{\cdot,t})^{2}](\bm{V}_{t,\cdot})^{\top}\bm{V}_{t,\cdot}=\left\Vert \bm{a}\right\Vert _{2}^{2}\bm{V}^{\top}\bm{V}=\left\Vert \bm{a}\right\Vert _{2}^{2}\bm{I}_{r},
\]
where we use the fact that $\bm{V}$ is $\sigma(\bm{F})$-measurable
and $\bm{Z}$ is independent with $\bm{F}$. Then we have that

(i) When all entries of $\bm{Z}$ are Gaussian, we have that, the
law of $\frac{1}{\left\Vert \bm{a}\right\Vert _{2}}\sum_{t=1}^{T}(\bm{a}\bm{Z}_{\cdot,t})\bm{V}_{t,\cdot}$
is $r$-dimensional standard Gaussian $\mathcal{N}(0,\bm{I}_{r})$,
because it is a linear combination of Gaussian vectors. Thus, we have
that $\tau_{U}^{(i,j)}=0$ when all entries of $\bm{Z}$ are Gaussian.

(ii) When the entries of $\bm{Z}$ are sub-Gaussian, we obtain by
Lemma \ref{Lemma Berry-Esseen balls} that 
\begin{align*}
\tau_{U}^{(i,j)} & \lesssim\frac{1}{\left\Vert \bm{a}\right\Vert _{2}^{3}}\sum_{t=1}^{T}\mathbb{E}\left[\big\Vert(\bm{a}\bm{Z}_{\cdot,t})\bm{V}_{t,\cdot}\big\Vert_{2}^{3}\right]=\frac{1}{\left\Vert \bm{a}\right\Vert _{2}^{3}}\mathbb{E}\left[\big\Vert(\bm{a}\bm{Z}_{\cdot,t})\big\Vert_{2}^{3}\right]\max_{1\leq t\leq T}\big\Vert\bm{V}_{t,\cdot}\big\Vert_{2}\sum_{t=1}^{T}\big\Vert\bm{V}_{t,\cdot}\big\Vert_{2}^{2}\\
 & \lesssim\frac{1}{\left\Vert \bm{a}\right\Vert _{2}^{3}}\left\Vert \bm{a}\right\Vert _{2}^{3}\cdot\max_{1\leq t\leq T}\big\Vert\bm{V}_{t,\cdot}\big\Vert_{2}\cdot r\lesssim r\sqrt{\frac{\log n}{T}},
\end{align*}
where we used the fact that $\mathbb{E}\left[\big\Vert(\bm{a}\bm{Z}_{\cdot,t})\big\Vert_{2}^{3}\right]\lesssim\left\Vert \bm{a}\right\Vert _{2}^{3}$.
We now prove it as follows. By Exercise 6.3.5 in \citet{vershynin2016high},
we have that, since $\big\Vert\bm{a}\big\Vert_{2}=\big\Vert\bm{a}\big\Vert_{\mathrm{F}}$,
there exist constants $C_{1}>0$, $C_{2}>0$ such that, for any $t\geq0$,
it holds 
\[
\mathbb{P}\left(|\bm{a}\bm{Z}_{\cdot,t}|\geq C_{1}\big\Vert\bm{a}\big\Vert_{2}+t\right)\leq\exp(-C_{2}\frac{t^{2}}{\big\Vert\bm{a}\big\Vert_{2}^{2}}),\text{ i.e.},\mathbb{P}\left(|\bm{a}\bm{Z}_{\cdot,t}|\geq C_{1}\big\Vert\bm{a}\big\Vert_{2}(1+t)\right)\leq\exp\left(-C_{2}t^{2}\right)
\]
So similarly we have 
\begin{align*}
\mathbb{E}\left[|\bm{a}\bm{Z}_{\cdot,t}|^{3}\right] & =\int_{0}^{+\infty}3s^{2}\mathbb{P}\left(|\bm{a}\bm{Z}_{\cdot,t}|\geq s\right)ds\\
 & =\int_{0}^{C_{1}\big\Vert\bm{a}\big\Vert_{2}}3s^{2}\mathbb{P}\left(|\bm{a}\bm{Z}_{\cdot,t}|\geq s\right)ds\\
 & +\int_{0}^{+\infty}3(C_{1}\big\Vert\bm{a}\big\Vert_{2}(1+t))^{2}\mathbb{P}\left(|\bm{a}\bm{Z}_{\cdot,t}|\geq C_{1}\big\Vert\bm{a}\big\Vert_{2}(1+t)\right)dt\\
 & \leq\int_{0}^{C_{1}\big\Vert\bm{a}\big\Vert_{2}}3s^{2}ds+\int_{0}^{+\infty}3(C_{1}\big\Vert\bm{a}\big\Vert_{2}(1+t))^{2}\exp\left(-C_{2}t^{2}\right)dt\\
 & \lesssim\left\Vert \bm{a}\right\Vert _{2}^{3}.
\end{align*}

\textit{Step 3 -- Prove the upper bound for $\varepsilon_{2}$.}

Recall that 
\[
\varepsilon_{2}=\frac{1}{\left\Vert \bm{a}\right\Vert _{2}}\big\Vert((\bm{\Psi}_{U})_{i,}-(\bm{\Psi}_{U})_{j,})\bm{\Lambda}\big\Vert_{2}=\frac{\sqrt{T}}{\sqrt{(\bm{\Sigma}_{\varepsilon})_{i,i}+(\bm{\Sigma}_{\varepsilon})_{j,j}-2(\bm{\Sigma}_{\varepsilon})_{i,j}}}\big\Vert((\bm{\Psi}_{U})_{i,}-(\bm{\Psi}_{U})_{j,})\bm{\Lambda}\big\Vert_{2}.
\]
Then using the decomposition of $\bm{\Psi}_{U}$ in the proof of Theorem
\ref{Thm UV 1st approx row-wise error}, we obtain that 
\begin{align*}
\big\Vert((\bm{\Psi}_{U})_{i,}-(\bm{\Psi}_{U})_{j,})\bm{\Lambda}\big\Vert_{2} & \lesssim\sigma_{r}(\omega_{i}+\omega_{j})\omega\sqrt{\frac{r}{n}}\log^{3/2}n+\sigma_{r}(\omega_{i}+\omega_{j})\omega\log n\big\Vert\bar{\bm{U}}\big\Vert_{2,\infty}\\
 & +\sigma_{r}\left(\rho^{2}+\rho\sqrt{\frac{r}{n}}\log n\right)(\big\Vert\bar{\bm{U}}_{i,\cdot}\big\Vert_{2}+\big\Vert\bar{\bm{U}}_{j,\cdot}\big\Vert_{2}),
\end{align*}
Finally, the desird result follows from assembling the upper bounds
of $\tau_{U}^{(i,j)}$ and $\varepsilon_{2}$. 
\end{proof}

\subsection{Proof of Theorem \ref{Thm two-sample test B}:}

We are now ready to prove the main theorem. Define $\bm{e}_{ij}\in\mathbb{R}^{N}$
as the vector where the $i$-th (resp. $j$-th) entry is $1$ (resp.
$-1$) and the other entries are all zero. Note that $T\left\Vert \bm{a}\right\Vert _{2}^{2}=(\bm{\Sigma}_{\varepsilon})_{i,i}+(\bm{\Sigma}_{\varepsilon})_{j,j}-2(\bm{\Sigma}_{\varepsilon})_{i,j}=\bm{e}_{ij}^{\top}\bm{\Sigma}_{\varepsilon}\bm{e}_{ij}\geq\lambda_{\min}(\bm{\Sigma}_{\varepsilon})\left\Vert \bm{e}_{ij}\right\Vert _{2}^{2}=2\lambda_{\min}(\bm{\Sigma}_{\varepsilon})$,
$|\bm{e}_{ij}^{\top}(\widehat{\bm{\bm{\Sigma}}}_{\varepsilon}^{\tau}-\bm{\Sigma}_{\varepsilon})\bm{e}_{ij}|\lesssim(\epsilon_{N,T})^{1-q}s(\bm{\Sigma}_{\varepsilon})\ll1$,
and thus $(\widehat{\bm{\bm{\Sigma}}}_{\varepsilon}^{\tau})_{i,i}+(\widehat{\bm{\bm{\Sigma}}}_{\varepsilon}^{\tau})_{j,j}-2(\widehat{\bm{\bm{\Sigma}}}_{\varepsilon}^{\tau})_{i,j}=\bm{e}_{ij}^{\top}\bm{\Sigma}_{\varepsilon}\bm{e}_{ij}+\bm{e}_{ij}^{\top}(\widehat{\bm{\bm{\Sigma}}}_{\varepsilon}^{\tau}-\bm{\Sigma}_{\varepsilon})\bm{e}_{ij}\gtrsim\lambda_{\min}(\bm{\Sigma}_{\varepsilon})$.
So we obtain that 
\begin{align*}
|1\bigl/\bm{e}_{ij}^{\top}\widehat{\bm{\bm{\Sigma}}}_{\varepsilon}^{\tau}\bm{e}_{ij}-1\bigl/\bm{e}_{ij}^{\top}\bm{\Sigma}_{\varepsilon}\bm{e}_{ij}| & =\frac{|\bm{e}_{ij}^{\top}(\widehat{\bm{\bm{\Sigma}}}_{\varepsilon}^{\tau}-\bm{\Sigma}_{\varepsilon})\bm{e}_{ij}|}{(\bm{e}_{ij}^{\top}\widehat{\bm{\bm{\Sigma}}}_{\varepsilon}^{\tau}\bm{e}_{ij})(\bm{e}_{ij}^{\top}\bm{\Sigma}_{\varepsilon}\bm{e}_{ij})}\lesssim(\epsilon_{N,T})^{1-q}s(\bm{\Sigma}_{\varepsilon})\frac{1}{\lambda_{\min}(\bm{\Sigma}_{\varepsilon})T\left\Vert \bm{a}\right\Vert _{2}^{2}}.
\end{align*}
We let 
\begin{align*}
\zeta & :=\mathfrak{\widehat{\mathcal{B}}}_{ij}-\mathcal{B}_{ij}=T(1\bigl/\bm{e}_{ij}^{\top}\widehat{\bm{\bm{\Sigma}}}_{\varepsilon}^{\tau}\bm{e}_{ij}-1\bigl/\bm{e}_{ij}^{\top}\bm{\Sigma}_{\varepsilon}\bm{e}_{ij})\big\Vert(\widehat{\bm{U}}_{i,\cdot}-\widehat{\bm{U}}_{j,\cdot})\widehat{\bm{\Sigma}}\big\Vert_{2}^{2}.
\end{align*}
We derive the upper bound for $\big\Vert(\widehat{\bm{U}}_{i,\cdot}-\widehat{\bm{U}}_{j,\cdot})\widehat{\bm{\Sigma}}\big\Vert_{2}^{2}$
first. Recall that 
\[
(\widehat{\bm{U}}_{i,\cdot}-\widehat{\bm{U}}_{j,\cdot})\widehat{\bm{\Sigma}}\bm{R}_{V}=\sum_{t=1}^{T}(\bm{a}\bm{Z}_{\cdot,t})\bm{V}_{t,\cdot}+((\bm{\Psi}_{U})_{i,}-(\bm{\Psi}_{U})_{j,})\bm{\Lambda}.
\]
Note that $\sum_{t=1}^{T}(\bm{a}\bm{Z}_{\cdot,t})\bm{V}_{t,\cdot}$
is the sum of independent mean zero random vectors. Given the event
$\mathcal{E}_{Z}$ defined in Lemma \ref{Lemma Z norms A B}, it holds
$\big\Vert(\bm{a}\bm{Z}_{\cdot,t})\bm{V}_{t,\cdot}\big\Vert_{2}\leq C\sqrt{\log n}\big\Vert\bm{a}\big\Vert_{2}\big\Vert\bm{V}_{t,\cdot}\big\Vert_{2}$
for $t=1,2,\ldots,T$. Then, we obtain by the matrix Hoeffding inequality
\citep[Theorem 1.3]{tropp2012user} that $\big\Vert\sum_{t=1}^{T}(\bm{a}\bm{Z}_{\cdot,t})\bm{V}_{t,\cdot}\big\Vert_{2}\lesssim\sigma\sqrt{\log n}$,
where $\sigma=[\sum_{l=1}^{T}(C\sqrt{\log n}\big\Vert\bm{a}\big\Vert_{2}\big\Vert\bm{V}_{t,\cdot}\big\Vert_{2})^{2}]^{1/2}=C\sqrt{\log n}\big\Vert\bm{a}\big\Vert_{2}r$.
So, we obtain that, with probability at least $1-O(n^{-2})$, it holds
\[
\big\Vert\sum_{t=1}^{T}(\bm{a}\bm{Z}_{\cdot,t})\bm{V}_{t,\cdot}\big\Vert_{2}\lesssim\big\Vert\bm{a}\big\Vert_{2}r\log n.
\]

Then we obtain that 
\begin{align*}
\big\Vert(\widehat{\bm{U}}_{i,\cdot}-\widehat{\bm{U}}_{j,\cdot})\widehat{\bm{\Sigma}}\big\Vert_{2}^{2} & \lesssim\big\Vert\sum_{t=1}^{T}(\bm{a}\bm{Z}_{\cdot,t})\bm{V}_{t,\cdot}\big\Vert_{2}^{2}+\big\Vert((\bm{\Psi}_{U})_{i,}-(\bm{\Psi}_{U})_{j,})\bm{\Lambda}\big\Vert_{2}^{2}\\
 & \lesssim\big\Vert\bm{a}\big\Vert_{2}^{2}r^{2}\log^{2}n+\left\Vert \bm{a}\right\Vert _{2}^{2}\varepsilon_{2}^{2},
\end{align*}
where $\varepsilon_{2}$ is defined in the proof of Lemma \ref{Lemma B two sample prepare}.
Thus, we obtain that 
\begin{align*}
|\zeta| & =T(1\bigl/\bm{e}_{ij}^{\top}\widehat{\bm{\bm{\Sigma}}}_{\varepsilon}^{\tau}\bm{e}_{ij}-1\bigl/\bm{e}_{ij}^{\top}\bm{\Sigma}_{\varepsilon}\bm{e}_{ij})\big\Vert(\widehat{\bm{U}}_{i,\cdot}-\widehat{\bm{U}}_{j,\cdot})\widehat{\bm{\Sigma}}\big\Vert_{2}^{2}\\
 & \lesssim T(\epsilon_{N,T})^{1-q}s(\bm{\Sigma}_{\varepsilon})\frac{1}{\lambda_{\min}(\bm{\Sigma}_{\varepsilon})T\left\Vert \bm{a}\right\Vert _{2}^{2}}(\big\Vert\bm{a}\big\Vert_{2}^{2}r^{2}\log^{2}n+\left\Vert \bm{a}\right\Vert _{2}^{2}\varepsilon_{2}^{2})\\
 & \lesssim(\epsilon_{N,T})^{1-q}(r^{2}\log^{2}n+\varepsilon_{2}^{2}),
\end{align*}
where we used the assumption that $s(\bm{\Sigma}_{\varepsilon})/\lambda_{\min}(\bm{\Sigma}_{\varepsilon})\leq C_{\varepsilon}$
for some universal constant $C_{\varepsilon}$. Plugging the above
bound of $|\zeta|$ into Lemma \ref{Lemma B two sample prepare},
we obtain that 
\begin{align*}
\left|\mathbb{P}\left(\mathfrak{\widehat{\mathcal{B}}}_{ij}\leq\chi_{1-\alpha}^{2}(r)\right)-(1-\alpha)\right| & \lesssim(\epsilon_{N,T})^{1-q}r^{3/2}\log^{2}n+\tau_{U}^{(i,j)}+\frac{1}{\sqrt{r}}(\epsilon_{N,T})^{1-q}\varepsilon_{2}^{2}+\varepsilon_{2}\\
 & \lesssim(\epsilon_{N,T})^{1-q}r^{3/2}\log^{2}n+\tau_{U}^{(i,j)}+\varepsilon_{2},
\end{align*}
where we use the assumption that $\epsilon_{N,T}\ll1$ and $\varepsilon_{2}\ll1$.
Using the upper bound of $\varepsilon_{2}$ obtained in the proof
of Lemma \ref{Lemma B two sample prepare}, we have that 
\begin{align*}
\varepsilon_{2} & \lesssim\frac{\sqrt{T}\sigma_{r}(\omega_{i}+\omega_{j})\omega\sqrt{\frac{r}{n}}\log^{3/2}n+\sqrt{T}\sigma_{r}(\omega_{i}+\omega_{j})\omega\log n\big\Vert\bar{\bm{U}}\big\Vert_{2,\infty}}{\sqrt{(\bm{\Sigma}_{\varepsilon})_{i,i}+(\bm{\Sigma}_{\varepsilon})_{j,j}-2(\bm{\Sigma}_{\varepsilon})_{i,j}}}\\
 & +\frac{\sqrt{T}\sigma_{r}\left(\rho^{2}+\rho\sqrt{\frac{r}{n}}\log n\right)(\big\Vert\bar{\bm{U}}_{i,\cdot}\big\Vert_{2}+\big\Vert\bar{\bm{U}}_{j,\cdot}\big\Vert_{2})}{\sqrt{(\bm{\Sigma}_{\varepsilon})_{i,i}+(\bm{\Sigma}_{\varepsilon})_{j,j}-2(\bm{\Sigma}_{\varepsilon})_{i,j}}}\\
 & \lesssim\frac{\sqrt{\left\Vert \bm{\Sigma}_{\varepsilon}\right\Vert _{2}}}{\sqrt{\lambda_{\min}(\bm{\Sigma}_{\varepsilon})}}\sqrt{T}\theta(\omega_{i}+\omega_{j})\omega(\sqrt{\frac{r}{n}}\log^{3/2}n+\log n\big\Vert\bar{\bm{U}}\big\Vert_{2,\infty})\\
 & +\frac{\sqrt{\left\Vert \bm{\Sigma}_{\varepsilon}\right\Vert _{2}}}{\sqrt{\lambda_{\min}(\bm{\Sigma}_{\varepsilon})}}\sqrt{T}\theta\left(\rho^{2}+\rho\sqrt{\frac{r}{n}}\log n\right)(\big\Vert\bar{\bm{U}}_{i,\cdot}\big\Vert_{2}+\big\Vert\bar{\bm{U}}_{j,\cdot}\big\Vert_{2})\\
 & \lesssim\sqrt{\kappa_{\varepsilon}}\frac{\theta n}{\vartheta\sqrt{T}}(\frac{1}{\vartheta_{i}}+\frac{1}{\vartheta_{j}})(\big\Vert\bar{\bm{U}}\big\Vert_{2,\infty}+\frac{1}{\sqrt{n}})\sqrt{r}\log^{3/2}n+\sqrt{\kappa_{\varepsilon}}(\frac{n}{\theta\sqrt{T}}+1)(\big\Vert\bar{\bm{U}}_{i,\cdot}\big\Vert_{2}+\big\Vert\bar{\bm{U}}_{j,\cdot}\big\Vert_{2})\sqrt{r}\log n.
\end{align*}
The final result follows from the above upper bound.

%% file: appendix_row_norm_of_B.tex
\section{Proof of Theorem \ref{Thm B row norm}: Inference for the systematic
risks}

We start by rewriting our target error as follows, 
\[
\left\vert \mathbb{P}(\frac{1}{\widehat{\sigma}_{B,i}}(\big\Vert\widetilde{\bm{B}}_{i,\cdot}\big\Vert_{2}^{2}-\big\Vert(\widehat{\bm{U}}\widehat{\bm{\Sigma}})_{i,\cdot}\big\Vert_{2}^{2})\in[-z_{1-\frac{1}{2}\alpha},z_{1-\frac{1}{2}\alpha}])-(1-\alpha)\right\vert .
\]

\textit{Step 1 -- Rewrite the first-order approximation to show the Gaussian term.}

Our starting point is to prove that, for any random variable $\zeta$
satisfying $\left|\zeta\right|\ll1$, it holds 
\[
\left\vert \mathbb{P}(\left|\frac{1}{\sigma_{B,i}}(\big\Vert\widetilde{\bm{B}}_{i,\cdot}\big\Vert_{2}^{2}-\big\Vert(\widehat{\bm{U}}\widehat{\bm{\Sigma}})_{i,\cdot}\big\Vert_{2}^{2})+\zeta\right|\leq z_{1-\frac{1}{2}\alpha})-(1-\alpha)\right\vert \lesssim|\varepsilon_{B,i}|+|\zeta|+\tau_{B}^{i},
\]
where $\sigma_{B,i}=\frac{2}{\sqrt{T}}\sqrt{(\bm{\Sigma}_{\varepsilon})_{i,i}}\big\Vert\widetilde{\bm{B}}_{i,\cdot}\big\Vert_{2}$.
Then we let 
\[
\zeta=(\frac{1}{\widehat{\sigma}_{B,i}}-\frac{1}{\sigma_{B,i}})(\big\Vert\widetilde{\bm{B}}_{i,\cdot}\big\Vert_{2}^{2}-\big\Vert(\widehat{\bm{U}}\widehat{\bm{\Sigma}})_{i,\cdot}\big\Vert_{2}^{2})=(\frac{1}{\sqrt{(\widehat{\bm{\bm{\Sigma}}}_{\varepsilon}^{\tau})_{i,i}}}-\frac{1}{\sqrt{(\bm{\Sigma}_{\varepsilon})_{i,i}}})\frac{\sqrt{T}}{2\big\Vert\widetilde{\bm{B}}_{i,\cdot}\big\Vert_{2}}(\big\Vert\widetilde{\bm{B}}_{i,\cdot}\big\Vert_{2}^{2}-\big\Vert(\widehat{\bm{U}}\widehat{\bm{\Sigma}})_{i,\cdot}\big\Vert_{2}^{2}),
\]
and later we will derive an upper bound for $|\zeta|$. We already
have the expansion of $\widehat{\bm{U}}\bm{R}_{U}-\bm{U}$ and the
error bound for $\varepsilon_{\bm{\Sigma}}:=(\bm{R}_{U})^{\top}\widehat{\bm{\Sigma}}\bm{R}_{V}-\bm{\Lambda}$,
so we obtain that 
\begin{align*}
\widehat{\bm{U}}\widehat{\bm{\Sigma}}\bm{R}_{V} & =\widehat{\bm{U}}\bm{R}_{U}\cdot(\bm{R}_{U})^{\top}\widehat{\bm{\Sigma}}\bm{R}_{V}=(\bm{U}+\bm{G}_{U}+\bm{\Psi}_{U})(\bm{\Lambda}+\varepsilon_{\bm{\Sigma}})\\
 & =\widetilde{\bm{B}}+\bm{G}_{U}\bm{\Lambda}+\widehat{\bm{U}}\bm{R}_{U}\varepsilon_{\bm{\Sigma}}+\bm{\Psi}_{U}\bm{\Lambda}.
\end{align*}
Since $\bm{R}_{V}$ is a rotation matrix, we obtain 
\begin{align*}
\big\Vert(\widehat{\bm{U}}\widehat{\bm{\Sigma}})_{i,\cdot}\big\Vert_{2}^{2} & =\big\Vert(\widehat{\bm{U}}\widehat{\bm{\Sigma}})_{i,\cdot}\bm{R}_{V}\big\Vert_{2}^{2}=(\widehat{\bm{U}}\widehat{\bm{\Sigma}}\bm{R}_{V})_{i,\cdot}((\widehat{\bm{U}}\widehat{\bm{\Sigma}}\bm{R}_{V})_{i,\cdot})^{\top}\\
 & =(\widehat{\bm{U}}\bm{R}_{U})_{i,\cdot}(\bm{\Lambda}^{2}+\varepsilon_{\bm{\Sigma}}\bm{\Lambda}+\bm{\Lambda}(\varepsilon_{\bm{\Sigma}})^{\top}+\varepsilon_{\bm{\Sigma}}(\varepsilon_{\bm{\Sigma}})^{\top})((\widehat{\bm{U}}\bm{R}_{U})_{i,\cdot})^{\top}\\
 & =(\widehat{\bm{U}}\bm{R}_{U})_{i,\cdot}(\bm{\Lambda}^{2})((\widehat{\bm{U}}\bm{R}_{U})_{i,\cdot})^{\top}+\varepsilon_{1}=(\bm{U}+\bm{G}_{U}+\bm{\Psi}_{U})_{i,\cdot}(\bm{\Lambda}^{2})((\bm{U}+\bm{G}_{U}+\bm{\Psi}_{U})_{i,\cdot})^{\top}+\varepsilon_{1}\\
 & =\big\Vert\widetilde{\bm{B}}_{i,\cdot}\big\Vert_{2}^{2}+(\bm{U}_{i,\cdot}\bm{\Lambda}^{2}((\bm{G}_{U})_{i,\cdot})^{\top}+(\bm{G}_{U})_{i,\cdot}\bm{\Lambda}^{2}(\bm{U}_{i,\cdot})^{\top})+\varepsilon_{1}+\varepsilon_{2}\\
 & =\big\Vert\bm{B}_{i,\cdot}\big\Vert_{2}^{2}+(\bm{U}_{i,\cdot}\bm{\Lambda}^{2}((\bm{G}_{U})_{i,\cdot})^{\top}+(\bm{G}_{U})_{i,\cdot}\bm{\Lambda}^{2}(\bm{U}_{i,\cdot})^{\top})+\varepsilon_{1}+\varepsilon_{2}+\varepsilon_{J},
\end{align*}
where $\varepsilon_{1}$, $\varepsilon_{2}$, and $\varepsilon_{J}$
are scalars defined by 
\begin{align*}
\varepsilon_{1} & :=(\widehat{\bm{U}}\bm{R}_{U})_{i,\cdot}(\varepsilon_{\bm{\Sigma}}\bm{\Lambda}+\bm{\Lambda}(\varepsilon_{\bm{\Sigma}})^{\top}+\varepsilon_{\bm{\Sigma}}(\varepsilon_{\bm{\Sigma}})^{\top})((\widehat{\bm{U}}\bm{R}_{U})_{i,\cdot})^{\top},\\
\varepsilon_{2} & :=\bm{U}_{i,\cdot}\bm{\Lambda}^{2}((\bm{\bm{\Psi}}_{U})_{i,\cdot})^{\top}+(\bm{\bm{\Psi}}_{U})_{i,\cdot}\bm{\Lambda}^{2}(\bm{U}_{i,\cdot})^{\top}+(\bm{G}_{U}+\bm{\Psi}_{U})_{i,\cdot}(\bm{\Lambda}^{2})((\bm{G}_{U}+\bm{\Psi}_{U})_{i,\cdot})^{\top},\\
\varepsilon_{J} & :=\big\Vert\widetilde{\bm{B}}_{i,\cdot}\big\Vert_{2}^{2}-\big\Vert\bm{B}_{i,\cdot}\big\Vert_{2}^{2}=\bm{B}_{i,\cdot}((\bm{J}^{-1})^{\top}\bm{J}^{-1}-\bm{I}_{r})(\bm{B}_{i,\cdot})^{\top}.
\end{align*}
The expression of $\varepsilon_{J}$ uses the fact that $\bm{B}=\widetilde{\bm{B}}\bm{J}^{\top}$.
Note that we have 
\begin{align*}
\bm{U}_{i,\cdot}\bm{\Lambda}^{2}((\bm{G}_{U})_{i,\cdot})^{\top} & =(\bm{G}_{U})_{i,\cdot}\bm{\Lambda}^{2}(\bm{U}_{i,\cdot})^{\top}=T^{-1/2}(\bm{\Sigma}_{\varepsilon}^{1/2})_{i,\cdot}\bm{Z}\bm{V}\bm{\Lambda}(\bm{U}_{i,\cdot})^{\top}\\
 & =T^{-1/2}(\bm{\Sigma}_{\varepsilon}^{1/2})_{i,\cdot}\bm{Z}\bm{V}(\widetilde{\bm{B}}_{i,\cdot})^{\top}\\
 & =T^{-1/2}\sum_{t=1}^{T}(\bm{\Sigma}_{\varepsilon}^{1/2})_{i,\cdot}\bm{Z}_{\cdot,t}\bm{V}_{t,\cdot}(\widetilde{\bm{B}}_{i,\cdot})^{\top}
\end{align*}
is a scalar. So we obtain 
\[
\frac{1}{\sigma_{B,i}}(\big\Vert(\widehat{\bm{U}}\widehat{\bm{\Sigma}})_{i,\cdot}\big\Vert_{2}^{2}-\big\Vert\widetilde{\bm{B}}_{i,\cdot}\big\Vert_{2}^{2})=\bm{K}_{B,i}+\varepsilon_{B,i},\text{\qquad with\qquad}\bm{K}_{B,i}:=\frac{1}{\sigma_{B,i}}2T^{-1/2}\sum_{t=1}^{T}(\bm{\Sigma}_{\varepsilon}^{1/2})_{i,\cdot}\bm{Z}_{\cdot,t}\bm{V}_{t,\cdot}(\widetilde{\bm{B}}_{i,\cdot})^{\top},
\]
and 
\[
|\varepsilon_{B,i}|\leq\frac{1}{\sigma_{B,i}}(|\varepsilon_{1}|+|\varepsilon_{2}|+|\varepsilon_{J}|).
\]

\textit{Step 2 -- Show the proximity of $\bm{K}_{B,i}$ to a Gaussian vector
and simplify the inequality.}

Note that, the variance of $\bm{K}_{B,i}$ is equal to 
\[
\frac{4}{\sigma_{B,i}^{2}T}\sum_{t=1}^{T}Var[(\bm{\Sigma}_{\varepsilon}^{1/2})_{i,\cdot}\bm{Z}_{\cdot,t}\bm{V}_{t,\cdot}(\widetilde{\bm{B}}_{i,\cdot})^{\top}]=\frac{4}{\sigma_{B,i}^{2}T}(\bm{\Sigma}_{\varepsilon})_{i,i}\widetilde{\bm{B}}_{i,\cdot}(\widetilde{\bm{B}}_{i,\cdot})^{\top}=1,
\]
since $\sigma_{B,i}=\frac{2}{\sqrt{T}}\sqrt{(\bm{\Sigma}_{\varepsilon})_{i,i}}\big\Vert\widetilde{\bm{B}}_{i,\cdot}\big\Vert_{2}$.
We will establish the upper bound for the parameter 
\[
\tau_{B}^{i}:=\sup_{R\geq0}\left\vert \mathbb{P}\left(\left\vert \bm{K}_{B,i}\right\vert \leq R\right)-\mathbb{P}(\left\vert \mathcal{N}(0,1)\right\vert \leq R)\right\vert .
\]
Then by the Berry-Esseen theorem, we obtain that 
\begin{align*}
\tau_{B}^{i} & \leq\frac{2\sqrt{2}}{\sigma_{B,i}^{3}T^{3/2}}\sum_{t=1}^{T}\mathbb{E}[\big\Vert(\bm{\Sigma}_{\varepsilon}^{1/2})_{i,\cdot}\bm{Z}_{\cdot,t}\bm{V}_{t,\cdot}(\widetilde{\bm{B}}_{i,\cdot})^{\top}\big\Vert_{2}^{3}]\\
 & \leq\frac{2\sqrt{2}}{\sigma_{B,i}^{3}T^{3/2}}\mathbb{E}[|(\bm{\Sigma}_{\varepsilon}^{1/2})_{i,\cdot}\bm{Z}_{\cdot,t}|_{2}^{3}]\max_{1\leq t\leq T}|\bm{V}_{t,\cdot}(\widetilde{\bm{B}}_{i,\cdot})^{\top}|\sum_{t=1}^{T}|\bm{V}_{t,\cdot}(\widetilde{\bm{B}}_{i,\cdot})^{\top}|_{2}^{2}\\
 & =\frac{2\sqrt{2}}{\sigma_{B,i}^{3}T^{3/2}}((\bm{\Sigma}_{\varepsilon})_{i,i})^{3/2}\max_{1\leq t\leq T}|\bm{V}_{t,\cdot}(\widetilde{\bm{B}}_{i,\cdot})^{\top}|\sum_{t=1}^{T}|\bm{V}_{t,\cdot}(\widetilde{\bm{B}}_{i,\cdot})^{\top}|_{2}^{2}\\
 & \lesssim\sqrt{\frac{\log n}{T}},
\end{align*}
where we used the fact that 
\[
\sum_{t=1}^{T}|\bm{V}_{t,\cdot}(\widetilde{\bm{B}}_{i,\cdot})^{\top}|_{2}^{2}=\widetilde{\bm{B}}_{i,\cdot}[\sum_{t=1}^{T}(\bm{V}_{t,\cdot})^{\top}\bm{V}_{t,\cdot}](\widetilde{\bm{B}}_{i,\cdot})^{\top}=\widetilde{\bm{B}}_{i,\cdot}\bm{I}_{r}(\widetilde{\bm{B}}_{i,\cdot})^{\top}=\widetilde{\bm{B}}_{i,\cdot}(\widetilde{\bm{B}}_{i,\cdot})^{\top}=\big\Vert\widetilde{\bm{B}}_{i,\cdot}\big\Vert_{2}^{2},
\]
and 
\[
\mathbb{E}[|(\bm{\Sigma}_{\varepsilon}^{1/2})_{i,\cdot}\bm{Z}_{\cdot,t}|_{2}^{3}]\lesssim\big\Vert(\bm{\Sigma}_{\varepsilon}^{1/2})_{i,\cdot}\big\Vert_{2}^{3}=((\bm{\Sigma}_{\varepsilon})_{i,i})^{3/2},
\]
which follows from the same argument as we proved $\mathbb{E}\left[\big\Vert(\bm{a}\bm{Z}_{\cdot,t})\big\Vert_{2}^{3}\right]\lesssim\left\Vert \bm{a}\right\Vert _{2}^{3}$
in the proof of Lemma \ref{Lemma B two sample prepare}.

Recall that our target is to prove the upper bound of $|\varepsilon_{B,i}|$
and the inequality 
\[
\left\vert \mathbb{P}(|\bm{K}_{B,i}+\varepsilon_{B,i}+\zeta|\leq z_{1-\frac{1}{2}\alpha})-(1-\alpha)\right\vert \lesssim|\varepsilon_{B,i}|+|\zeta|+\tau_{B}^{i}.
\]
Note that 
\[
\mathbb{P}(|\bm{K}_{B,i}|\leq z_{1-\frac{1}{2}\alpha}-|\varepsilon_{B,i}|-|\zeta|)<\mathbb{P}(|\bm{K}_{B,i}+\varepsilon_{B,i}+\zeta|\leq z_{1-\frac{1}{2}\alpha})\leq\mathbb{P}(|\bm{K}_{B,i}|\leq z_{1-\frac{1}{2}\alpha}+|\varepsilon_{B,i}|+|\zeta|).
\]
Then by definition of $\tau_{B}^{i}$, we obtain 
\[
\mathbb{P}(|\mathcal{N}(0,1)|\leq z_{1-\frac{1}{2}\alpha}-|\varepsilon_{B,i}|-|\zeta|)-\tau_{B}^{i}<\mathbb{P}(|\bm{K}_{B,i}+\varepsilon_{B,i}+\zeta|\leq z_{1-\frac{1}{2}\alpha})\leq\mathbb{P}(|\mathcal{N}(0,1)|\leq z_{1-\frac{1}{2}\alpha}+|\varepsilon_{B,i}|+|\zeta|)+\tau_{B}^{i}.
\]
Then using the property of Gaussian density, or by Lemma \ref{Lemma Chi-square two balls},
we obtain that, there exists a constant $c_{G}>0$ such that 
\[
\mathbb{P}(|\mathcal{N}(0,1)|\leq z_{1-\frac{1}{2}\alpha})-c_{G}(|\varepsilon_{B,i}|+|\zeta|)-\tau_{B}^{i}\leq\mathbb{P}(|\bm{K}_{B,i}+\varepsilon_{B,i}+\zeta|\leq z_{1-\frac{1}{2}\alpha})\leq\mathbb{P}(|\mathcal{N}(0,1)|\leq z_{1-\frac{1}{2}\alpha})+c_{G}(|\varepsilon_{B,i}|+|\zeta|)+\tau_{B}^{i}.
\]
Then since $\mathbb{P}(|\mathcal{N}(0,1)|\leq z_{1-\frac{1}{2}\alpha})=1-\alpha$,
we obtain that 
\[
\left\vert \mathbb{P}(|\bm{K}_{B,i}+\varepsilon_{B,i}+\zeta|<\tau_{1-\frac{1}{2}\alpha})-(1-\alpha)\right\vert \lesssim c_{G}(|\varepsilon_{B,i}|+|\zeta|)+\tau_{B}^{i}\lesssim|\varepsilon_{B,i}|+|\zeta|+\tau_{B}^{i}.
\]

\textit{Step 3 -- Prove the upper bound for $|\varepsilon_{B,i}|$.}

By definition, we have that $|\varepsilon_{B,i}|\leq\frac{1}{\sigma_{B,i}}(|\varepsilon_{1}|+|\varepsilon_{2}|+|\varepsilon_{J}|)$.

For $\varepsilon_{1}=(\widehat{\bm{U}}\bm{R}_{U})_{i,\cdot}(\varepsilon_{\bm{\Sigma}}\bm{\Lambda}+\bm{\Lambda}(\varepsilon_{\bm{\Sigma}})^{\top}+\varepsilon_{\bm{\Sigma}}(\varepsilon_{\bm{\Sigma}})^{\top})((\widehat{\bm{U}}\bm{R}_{U})_{i,\cdot})^{\top}$,
since $\big\Vert\varepsilon_{\bm{\Sigma}}\big\Vert_{2}\lesssim\sigma_{r}(\rho^{2}+\rho\sqrt{\frac{r}{n}}\log n)\lesssim\sigma_{r}\lesssim\big\Vert\bm{\Lambda}\big\Vert_{2}$,
we have that 
\begin{align*}
\frac{1}{\sigma_{B,i}}|\varepsilon_{1}| & \lesssim\frac{\sqrt{T}}{\sqrt{(\bm{\Sigma}_{\varepsilon})_{i,i}}\big\Vert\widetilde{\bm{B}}_{i,\cdot}\big\Vert_{2}}\cdot\big\Vert(\widehat{\bm{U}}_{i,\cdot}\bm{R}_{U}-\bm{U}_{i,\cdot})\bm{\Lambda}+\bm{U}_{i,\cdot}\bm{\Lambda}\big\Vert_{2}\cdot\big\Vert\varepsilon_{\bm{\Sigma}}\big\Vert_{2}\cdot\big\Vert(\widehat{\bm{U}}\bm{R}_{U})_{i,\cdot}\big\Vert_{2}\\
 & \overset{\text{(i)}}{\lesssim}\frac{\sqrt{T}}{\sqrt{(\bm{\Sigma}_{\varepsilon})_{i,i}}\big\Vert\bar{\bm{U}}_{i,\cdot}\big\Vert_{2}\sigma_{r}}\cdot(\omega_{i}\sqrt{\frac{r}{n}}\log n\cdot\sigma_{r})\cdot\sigma_{r}(\rho^{2}+\rho\sqrt{\frac{r}{n}}\log n)\cdot\sqrt{r}\\
 & +\frac{\sqrt{T}}{\sqrt{(\bm{\Sigma}_{\varepsilon})_{i,i}}\big\Vert\widetilde{\bm{B}}_{i,\cdot}\big\Vert_{2}}\big\Vert\bm{U}_{i,\cdot}\bm{\Lambda}\big\Vert_{2}\sigma_{r}(\rho^{2}+\rho\sqrt{\frac{r}{n}}\log n)\cdot\big\Vert(\widehat{\bm{U}}_{i,\cdot}\bm{R}_{U}-\bm{U}_{i,\cdot})+\bm{U}_{i,\cdot}\big\Vert_{2}\\
 & \overset{\text{(ii)}}{\lesssim}\frac{\sqrt{T}}{\sqrt{(\bm{\Sigma}_{\varepsilon})_{i,i}}\big\Vert\bar{\bm{U}}_{i,\cdot}\big\Vert_{2}\sigma_{r}}\cdot(\omega_{i}\sqrt{\frac{r}{n}}\log n\cdot\sigma_{r})\cdot\sigma_{r}(\rho^{2}+\rho\sqrt{\frac{r}{n}}\log n)\sqrt{r}\\
 & +\frac{\sqrt{T}}{\sqrt{(\bm{\Sigma}_{\varepsilon})_{i,i}}}\sigma_{r}(\rho^{2}+\rho\sqrt{\frac{r}{n}}\log n)\cdot(\omega_{i}\sqrt{\frac{r}{n}}\log n+\big\Vert\bar{\bm{U}}_{i,\cdot}\big\Vert_{2}),
\end{align*}
where (i) uses the fact that $\big\Vert\widetilde{\bm{B}}_{i,\cdot}\big\Vert_{2}=\big\Vert\bm{U}_{i,\cdot}\bm{\Lambda}\big\Vert_{2}\gtrsim\big\Vert\bm{U}_{i,\cdot}\big\Vert_{2}\sigma_{r}$
and the bound $\big\Vert(\widehat{\bm{U}}_{i,\cdot}\bm{R}_{U}-\bm{U}_{i,\cdot})\bm{\Lambda}\big\Vert_{2}\lesssim\omega_{i}\sqrt{\frac{r}{n}}\log n\cdot\sigma_{r}$
obtained by the decomposition of $(\widehat{\bm{U}}_{i,\cdot}\bm{R}_{U}-\bm{U}_{i,\cdot})$
as we did in proof of Corollary \ref{corollary:error bound B F},
and (ii) uses the fact that $\big\Vert\bm{U}_{i,\cdot}\big\Vert_{2}=\big\Vert\bar{\bm{U}}_{i,\cdot}\big\Vert_{2}$.

For $\varepsilon_{2}=\bm{U}_{i,\cdot}\bm{\Lambda}^{2}((\bm{\bm{\Psi}}_{U})_{i,\cdot})^{\top}+(\bm{\bm{\Psi}}_{U})_{i,\cdot}\bm{\Lambda}^{2}(\bm{U}_{i,\cdot})^{\top}+(\bm{G}_{U}+\bm{\Psi}_{U})_{i,\cdot}(\bm{\Lambda}^{2})((\bm{G}_{U}+\bm{\Psi}_{U})_{i,\cdot})^{\top}$,
we have that 
\begin{align*}
\frac{1}{\sigma_{B,i}}|\varepsilon_{2}| & \lesssim\frac{\sqrt{T}}{2\sqrt{(\bm{\Sigma}_{\varepsilon})_{i,i}}\big\Vert\widetilde{\bm{B}}_{i,\cdot}\big\Vert_{2}}(\big\Vert\widetilde{\bm{B}}_{i,\cdot}\big\Vert_{2}\big\Vert(\bm{\bm{\Psi}}_{U})_{i,\cdot}\bm{\Lambda}\big\Vert_{2}+\big\Vert(\widehat{\bm{U}}_{i,\cdot}\bm{R}_{U}-\bm{U}_{i,\cdot})\bm{\Lambda}\big\Vert_{2}^{2})\\
 & \lesssim\frac{\sqrt{T}\sigma_{r}}{\sqrt{(\bm{\Sigma}_{\varepsilon})_{i,i}}}\left(\omega_{i}\omega\sqrt{\frac{r}{n}}\log^{3/2}n+\left(\rho^{2}+\rho\sqrt{\frac{r}{n}}\log n\right)\big\Vert\bar{\bm{U}}_{i,\cdot}\big\Vert_{2}+\omega_{i}\omega\log n\big\Vert\bar{\bm{U}}\big\Vert_{2,\infty}\right)\\
 & +\frac{\sqrt{T}(\omega_{i}\sqrt{\frac{r}{n}}\log n)^{2}\sigma_{r}}{\sqrt{(\bm{\Sigma}_{\varepsilon})_{i,i}}\big\Vert\bar{\bm{U}}_{i,\cdot}\big\Vert_{2}},
\end{align*}
where we use the decomposition and the upper bound for $(\bm{\bm{\Psi}}_{U})_{i,\cdot}$
in the proof of Theorem \ref{Thm UV 1st approx row-wise error}, and
the bound $\big\Vert(\widehat{\bm{U}}_{i,\cdot}\bm{R}_{U}-\bm{U}_{i,\cdot})\bm{\Lambda}\big\Vert_{2}\lesssim\omega_{i}\sqrt{\frac{r}{n}}\log n\cdot\sigma_{r}$
obtained by the decomposition of $(\widehat{\bm{U}}_{i,\cdot}\bm{R}_{U}-\bm{U}_{i,\cdot})$
as we did in proof of Corollary \ref{corollary:error bound B F}.

For $\varepsilon_{J}=\bm{B}_{i,\cdot}((\bm{J}^{-1})^{\top}\bm{J}^{-1}-\bm{I}_{r})(\bm{B}_{i,\cdot})^{\top}$,
we have that 
\[
\frac{1}{\sigma_{B,i}}|\varepsilon_{J}|\lesssim\frac{\sqrt{T}}{2\sqrt{(\bm{\Sigma}_{\varepsilon})_{i,i}}\big\Vert\bm{B}_{i,\cdot}\big\Vert_{2}}\cdot\big\Vert\bm{B}_{i,\cdot}\big\Vert_{2}^{2}\sqrt{\frac{r+\log n}{T}}\lesssim\frac{\Vert\bm{b}_{i}\Vert_{2}}{\sqrt{(\bm{\Sigma}_{\varepsilon})_{i,i}}}\sqrt{r+\log n}.
\]

\textit{Step 4 -- Prove the upper bound for $|\zeta|$ and assemble all the bounds to get the final result.}

We have that 
\[
\zeta=(((\widehat{\bm{\bm{\Sigma}}}_{\varepsilon}^{\tau})_{i,i})^{-1/2}-((\bm{\Sigma}_{\varepsilon})_{i,i})^{-1/2})\frac{\sqrt{T}}{2\big\Vert\widetilde{\bm{B}}_{i,\cdot}\big\Vert_{2}}(\big\Vert\widetilde{\bm{B}}_{i,\cdot}\big\Vert_{2}^{2}-\big\Vert(\widehat{\bm{U}}\widehat{\bm{\Sigma}})_{i,\cdot}\big\Vert_{2}^{2}).
\]
Note that $|(\widehat{\bm{\bm{\Sigma}}}_{\varepsilon}^{\tau})_{i,i}-(\bm{\Sigma}_{\varepsilon})_{i,i}|\lesssim\epsilon_{N,T}(\bm{\Sigma}_{\varepsilon})_{i,i}$
and $\epsilon_{N,T}\ll1$, so we obtain by the mean value theorem
that 
\[
|((\widehat{\bm{\bm{\Sigma}}}_{\varepsilon}^{\tau})_{i,i})^{-1/2}-((\bm{\Sigma}_{\varepsilon})_{i,i})^{-1/2}|\lesssim((\bm{\Sigma}_{\varepsilon})_{i,i})^{-3/2}\cdotp\epsilon_{N,T}(\bm{\Sigma}_{\varepsilon})_{i,i}=((\bm{\Sigma}_{\varepsilon})_{i,i})^{-1/2}\cdotp\epsilon_{N,T}.
\]
So we obtain that 
\begin{align*}
|\zeta| & \lesssim\epsilon_{N,T}\cdotp\frac{\sqrt{T}}{\sqrt{(\bm{\Sigma}_{\varepsilon})_{i,i}}\big\Vert\widetilde{\bm{B}}_{i,\cdot}\big\Vert_{2}}|\big\Vert\widetilde{\bm{B}}_{i,\cdot}\big\Vert_{2}^{2}-\big\Vert(\widehat{\bm{U}}\widehat{\bm{\Sigma}})_{i,\cdot}\big\Vert_{2}^{2}|\\
 & \lesssim\epsilon_{N,T}\cdotp|\frac{1}{\sigma_{B,i}}\bm{K}_{B,i}|+\epsilon_{N,T}\cdotp\frac{1}{\sigma_{B,i}}(|\varepsilon_{1}|+|\varepsilon_{2}|+|\varepsilon_{J}|).
\end{align*}
Using (G.3) in Lemma 19 of \citet{yan2024entrywise}, the leading
order term is bounded by 
\begin{align*}
\epsilon_{N,T}\cdotp\frac{1}{\sigma_{B,i}}\bm{K}_{B,i} & \lesssim\epsilon_{N,T}\cdotp\frac{\sqrt{T}}{\sqrt{(\bm{\Sigma}_{\varepsilon})_{i,i}}\big\Vert\widetilde{\bm{B}}_{i,\cdot}\big\Vert_{2}}\frac{1}{\sqrt{T}}\big\Vert(\bm{\Sigma}_{\varepsilon}^{1/2})_{i,\cdot}\big\Vert_{2}\big\Vert\widetilde{\bm{B}}_{i,\cdot}\big\Vert_{2}\sqrt{r}\sqrt{1+r+\log n}\\
 & \lesssim\epsilon_{N,T}\cdotp r\sqrt{\log n},
\end{align*}
where we use the fact that $\big\Vert(\bm{\Sigma}_{\varepsilon}^{1/2})_{i,\cdot}\big\Vert_{2}^{2}=(\bm{\Sigma}_{\varepsilon})_{i,i}$.

Combining the above bounds, we obtain that 
\[
\left\vert \mathbb{P}(\frac{1}{\widehat{\sigma}_{B,i}}(\big\Vert\widetilde{\bm{B}}_{i,\cdot}\big\Vert_{2}^{2}-\big\Vert(\widehat{\bm{U}}\widehat{\bm{\Sigma}})_{i,\cdot}\big\Vert_{2}^{2})\in[-z_{1-\frac{1}{2}\alpha},z_{1-\frac{1}{2}\alpha}])-(1-\alpha)\right\vert \lesssim|\varepsilon_{B,i}|+|\zeta|+\tau_{B}^{i},
\]
where $\tau_{B}^{i}\lesssim\sqrt{\frac{\log n}{T}}$, and since $\epsilon_{N,T}\ll1$,
we have that 
\begin{align*}
|\varepsilon_{B,i}|+|\zeta| & \lesssim\frac{1}{\sigma_{B,i}}(|\varepsilon_{1}|+|\varepsilon_{2}|+|\varepsilon_{J}|)+\epsilon_{N,T}\cdotp|\frac{1}{\sigma_{B,i}}\bm{K}_{B,i}|+\epsilon_{N,T}\cdotp\frac{1}{\sigma_{B,i}}(|\varepsilon_{1}|+|\varepsilon_{2}|+|\varepsilon_{J}|)\\
 & \lesssim\frac{1}{\sigma_{B,i}}(|\varepsilon_{1}|+|\varepsilon_{2}|+|\varepsilon_{J}|)+\epsilon_{N,T}\cdotp|\frac{1}{\sigma_{B,i}}\bm{K}_{B,i}|\\
 & \lesssim\frac{\sqrt{T}\sigma_{r}}{\sqrt{(\bm{\Sigma}_{\varepsilon})_{i,i}}}\left(\omega_{i}\omega\sqrt{\frac{r}{n}}\log^{3/2}n+\left(\rho^{2}+\rho\sqrt{\frac{r}{n}}\log n\right)\big\Vert\bar{\bm{U}}_{i,\cdot}\big\Vert_{2}+\omega_{i}\omega\log n\big\Vert\bar{\bm{U}}\big\Vert_{2,\infty}\right)\\
 & +\frac{\sqrt{T}(\omega_{i}\sqrt{\frac{r}{n}}\log n)^{2}\sigma_{r}}{\sqrt{(\bm{\Sigma}_{\varepsilon})_{i,i}}\big\Vert\bar{\bm{U}}_{i,\cdot}\big\Vert_{2}}+\frac{\sqrt{T}}{\sqrt{(\bm{\Sigma}_{\varepsilon})_{i,i}}\big\Vert\bar{\bm{U}}_{i,\cdot}\big\Vert_{2}\sigma_{r}}\cdot\omega_{i}\sqrt{\frac{r}{n}}\log n\cdot\sigma_{r}\cdot\sigma_{r}(\rho^{2}+\rho\sqrt{\frac{r}{n}}\log n)\sqrt{r}\\
 & +\frac{\sqrt{T}}{\sqrt{(\bm{\Sigma}_{\varepsilon})_{i,i}}}\sigma_{r}(\rho^{2}+\rho\sqrt{\frac{r}{n}}\log n)\cdot(\omega_{i}\sqrt{\frac{r}{n}}\log n+\big\Vert\bar{\bm{U}}_{i,\cdot}\big\Vert_{2})+\frac{\Vert\bm{b}_{i}\Vert_{2}}{\sqrt{(\bm{\Sigma}_{\varepsilon})_{i,i}}}\sqrt{r+\log n}+\epsilon_{N,T}\cdotp r\sqrt{\log n}.
\end{align*}
The desired result follows from the above bound.

%% file: appendix_technical.tex
\section{Technical lemmas}

We collect some technical lemmas in this section. First, we establish
the results on the quantile of Chi-square distribution. Recall that,
we denote $\chi^{2}(n)$ as a random variable follows Chi-square distribution
with degree of freedom $n$, and denote by $\chi_{q}^{2}(n)$ its
$q$-quantile.

\begin{lemma} \label{Lemma Chi-square quantile order}For any fixed
$q\in(0,1)$, the $q$-quantile of Chi-square distribution satisfies
that $\chi_{q}^{2}(n)\asymp n$. \end{lemma} 
\begin{proof}
Consider $n$ i.i.d.~standard Gaussian $\mathcal{N}(0,1)$ random
variables $X_{1},X_{2},\bm{\ldots},X_{n}$. Define $S_{n}:=\sum_{i=1}^{n}X_{i}^{2}$,
which follows $\chi^{2}(n)$ distribution. Using the Berry-Esseen
theorem \citep[Theorem 2.1.3]{vershynin2016high}, we have that 
\[
\sup_{t\in\mathbb{R}}\left|\mathbb{P}(\frac{S_{n}-n}{\sqrt{2n}}\leq t)-\int_{-\infty}^{t}\frac{1}{\sqrt{2\pi}}e^{-\frac{1}{2}x^{2}}dx\right|<c\frac{1}{\sqrt{n}},
\]
where $c$ is a constant independent with $n$. We denote by $z_{q}$
the $q$-quantile of standard Gaussian $\mathcal{N}(0,1)$. Letting
$t=z_{q}$ in the above inequality, we obtain that 
\[
q-c\frac{1}{\sqrt{n}}<\mathbb{P}(S_{n}\leq n+z_{q}\sqrt{2n})<q+c\frac{1}{\sqrt{n}}.
\]
Note that $\chi_{q}^{2}(n)$ is the $q$-quantile of $S_{n}$. By
the monotonicity of cdf, we obtain that 
\[
n+z_{q-c\frac{1}{\sqrt{n}}}\sqrt{2n}\leq\chi_{q}^{2}(n)\leq n+z_{q+c\frac{1}{\sqrt{n}}}\sqrt{2n}.
\]
Then, we prove that $z_{q-c\frac{1}{\sqrt{n}}}$ and $z_{q+c\frac{1}{\sqrt{n}}}$
are close to $z_{q}$.

On the one hand, when $c\frac{1}{\sqrt{n}}<\frac{1-q}{2}$, i.e.,
$n>\frac{4c^{2}}{(1-q)^{2}}$, we have that 
\begin{align*}
c\frac{1}{\sqrt{n}} & =(q+c\frac{1}{\sqrt{n}})-q=\int_{z_{q}}^{z_{q+c\frac{1}{\sqrt{n}}}}\frac{1}{\sqrt{2\pi}}e^{-\frac{1}{2}x^{2}}dx\\
 & \overset{\text{(i)}}{=}(z_{q+c\frac{1}{\sqrt{n}}}-z_{q})\frac{1}{\sqrt{2\pi}}e^{-\frac{1}{2}x^{2}}|_{x=\vartheta z_{q+c\frac{1}{\sqrt{n}}}+(1-\vartheta)z_{q}}\\
 & \geq(z_{q+c\frac{1}{\sqrt{n}}}-z_{q})\frac{1}{\sqrt{2\pi}}e^{-\frac{1}{2}x^{2}}|_{x=z_{(1+q)/2}},
\end{align*}
where (i) uses the mean value theorem for integral. On the other hand,
similarly we have that 
\[
c\frac{1}{\sqrt{n}}\geq(z_{q}-z_{q-c\frac{1}{\sqrt{n}}})\frac{1}{\sqrt{2\pi}}e^{-\frac{1}{2}x^{2}}|_{x=z_{q}}.
\]
So, there exists a constant $c_{1}$ which does not depend on $n$,
such that $\max(z_{q+c\frac{1}{\sqrt{n}}}-z_{q},z_{q}-z_{q-c\frac{1}{\sqrt{n}}})\leq c_{1}\frac{1}{\sqrt{n}}$
holds for sufficiently large $n$. As a result, we obtain 
\[
n+z_{q}\sqrt{2n}-\sqrt{2}c_{1}=n+(z_{q}-c_{1}\frac{1}{\sqrt{n}})\sqrt{2n}\leq\chi_{q}^{2}(n)\leq n+(z_{q}+c_{1}\frac{1}{\sqrt{n}})\sqrt{2n}\leq n+z_{q}\sqrt{2n}+\sqrt{2}c_{1}.
\]
For simplicity of notations, we write the above results for $\chi_{q}^{2}(n)$
as follows: for sufficiently large $n$, it holds 
\[
\left|\chi_{q}^{2}(n)-(n+C_{1}\sqrt{n})\right|<C_{0},
\]
where $C_{1},C_{0}>0$ are constants independent with $n$. So we
obtain that $\chi_{q}^{2}(n)\asymp n$. 
\end{proof}
Next, we consider another results on Chi-square distribution, which
can be regarded as the probability that a Gaussian vector lies in
a spherical shell.

\begin{lemma} \label{Lemma Chi-square two balls}There exists a constant
$c_{B}$ which does not depend on the dimension $m$, such that, for
any $\varepsilon$ satisfying $|\varepsilon|<R$, it holds 
\[
\sup_{R\geq0}\left|\mathbb{P}(\chi^{2}(m)\leq(R+\varepsilon)^{2})-\mathbb{P}(\chi^{2}(m)\leq R^{2})\right|\leq c_{B}|\varepsilon|.
\]

\end{lemma} 
\begin{proof}
Denote by $B_{m}(R)$ the ball in $\mathbb{R}^{m}$ centered at the
origin with the radius $R$. By (2.3) in \citet{sazonov1972CLTbound}
and the comments after it, we obtain that $|\mathbb{P}(\mathcal{N}(0,\bm{I}_{m})\in B_{m}(R+\varepsilon))-\mathbb{P}(\mathcal{N}(0,\bm{I}_{m})\in B_{m}(R))|\leq c_{B}|\varepsilon|$,
where $c_{B}>0$ is a constant and $c_{B}$ does not depend on $m$,
$R$, and $\varepsilon$. So the desired result follows from the fact
that $\mathbb{P}(\chi^{2}(m)\leq R^{2})=\mathbb{P}(\mathcal{N}(0,\bm{I}_{m})\in B(R))$. 
\end{proof}
Then we consider a special case of multi-dimensional Berry-Esseen
theorem.

\begin{lemma} \label{Lemma Berry-Esseen balls}Let $\xi_{1},\xi_{2},\ldots,\xi_{n}$
be independent $d$-dimensional random vectors with zero means, and
let $W=\sum_{k=1}^{n}\xi_{k}$. Suppose that $cov(W)=\bm{I}_{d}$.
Then we have that 
\[
\sup_{R\geq0}\left\vert \mathbb{P}\left(\left\Vert W\right\Vert _{2}\leq R\right)-\mathbb{P}(\chi^{2}(d)\leq R^{2})\right\vert \leq c_{G}\sum_{k=1}^{n}\mathbb{E}\left[\left\Vert \xi_{k}\right\Vert _{2}^{3}\right],
\]
where $c_{G}$ is a constant and $c_{G}$ does not depend on $n$
and $d$. \end{lemma} 
\begin{proof}
The proof is based on the results in \citet{raivc2019multivariate}
and we will directly adopt the notations therein to conduct the proof.
Note that $\mathbb{P}(\chi^{2}(d)\leq R^{2})=\mathbb{P}(\mathcal{N}(0,\bm{I}_{d})\in B(R))$,
where $B_{d}(R)$ is the ball in $\mathbb{R}^{d}$ centered at the
origin with the radius $R$. According to (1.4) in Theorem 1.3 and
Example 1.2, we have that 
\begin{align}
\sup_{R\geq0}\left\vert \mathbb{P}\left(\left\Vert W\right\Vert _{2}\leq R\right)-\mathbb{P}(\chi^{2}(d)\leq R^{2})\right\vert  & =\sup_{R\geq0}\left\vert \mathbb{P}\left(W\in B_{d}(R)\right)-\mathbb{P}(\mathcal{N}(0,\bm{I}_{d})\in B_{d}(R))\right\vert \nonumber \\
 & \lesssim(1+\gamma^{*}(\mathscr{A}|\rho))\sum_{k=1}^{n}\mathbb{E}\left[\left\Vert \xi_{k}\right\Vert _{2}^{3}\right].\label{Berry-Esseen from Raivc}
\end{align}
Here, for $\gamma^{*}(\mathscr{A}|\rho)$ defined in \citet{raivc2019multivariate},
following Example 1.2 therein for the class of all balls, $\mathscr{A}$
is the set of all balls in $\mathbb{R}^{d}$, and $\rho_{A}=\delta_{A}$.
Then we obtain $A^{t|\rho}=\{x;\rho_{A}(x)\leq t\}=\{x;\delta_{A}(x)\leq t\}=A^{t}$,
and thus by definition we get $\gamma^{*}(\mathscr{A}|\rho)=\gamma^{*}(\mathscr{A})$.
Since $\mathscr{A}$ is the set of all balls in $\mathbb{R}^{d}$,
we obtain by definition of $\gamma^{*}(\mathscr{A})$ that 
\[
\gamma^{*}(\mathscr{A})=\sup_{v\in\mathbb{R}^{d},R\geq0,\varepsilon\neq0}\frac{1}{|\varepsilon|}\left\vert \mathbb{P}(\mathcal{N}(0,\bm{I}_{d})\in B_{d}(R+\varepsilon;v))-\mathbb{P}(\mathcal{N}(0,\bm{I}_{d})\in B_{d}(R;v))\right\vert ,
\]
where $B_{d}(R;v)$ is the ball in $\mathbb{R}^{d}$ centered at $v$
with the radius $R$. By (2.3) in \citet{sazonov1972CLTbound} and
the comments after it, we obtain that, there exists a constant $c_{B}$
which does not depend on the dimension $d$, such that $\gamma^{*}(\mathscr{A})\leq c_{B}$.
Then the desired result follows from (\ref{Berry-Esseen from Raivc})
and the fact that $\gamma^{*}(\mathscr{A}|\rho)=\gamma^{*}(\mathscr{A})\leq c_{B}$. 
\end{proof}